\documentclass[preprint,1p,times]{elsarticle}
\usepackage[
    type={CC},
    modifier={by-nc-nd},
    version={4.0},
]{doclicense}

\newif\ifcomment \commentfalse
\newif\iffull \fulltrue
\fulltrue

\usepackage{url}
\usepackage{color}
\usepackage{alltt}
\usepackage{bcprules}
\usepackage{bcpsymbols}
\usepackage{amsmath}
\usepackage{amssymb}
\usepackage{listings}
\usepackage{my-macro}
\usepackage{src2tex}
\usepackage{float}
\usepackage{tikz}
\usepackage{proof}

\newtheorem{theorem}{Theorem}
\newtheorem{lemma}{Lemma}
\newcommand{\fsname}{\ensuremath{\text{Context\-FJ}_{{<:}}}}

\newcommand{\LT}{\kw{LT}}

\newcommand\AI[1]{\ifcomment{\color{red}[\emph{#1} -- AI]}\fi}
\newcommand\HI[1]{\ifcomment{\color{blue}[\emph{#1} -- HI]}\fi}

\renewcommand\appendixname{}

\newcommand{\npr}{\kw{npr}}

\usepackage[T1]{fontenc}
\usepackage{lmodern}

\journal{Science of Computer Programming}

\begin{document}
\lstset{
  breaklines = true,
  language=Java,
  basicstyle={\ttfamily\small},
  framesep=5pt,
  showstringspaces=false,
  tabsize=2,
  commentstyle={\slshape},
  stringstyle={\itshape},
  keywordstyle={\bfseries},
  classoffset=1,
  keywordstyle={\bfseries\slshape},
  morekeywords={layer,with,without,proceed,superproceed,requires,swap,swappable},
}

\begin{frontmatter}



\title{A Type System for First-Class Layers with Inheritance, Subtyping, and Swapping\tnoteref{label1}}
\tnotetext[label1]{This is a revised and extended version of the paper ``A Sound Type System for Layer Subtyping and Dynamically Activated First-Class Layers,'' presented at the 13th Asian Symposium on Programming Languages and Systems (APLAS 2015).}

\author{Hiroaki Inoue}
\ead{hinoue@kuis.kyoto-u.ac.jp}
\author{Atsushi Igarashi}
\ead{igarashi@kuis.kyoto-u.ac.jp}
\address{Graduate School of Informatics \\ Kyoto University}

\begin{abstract}
  Context-Oriented Programming (COP) is a programming paradigm to encourage modularization of 
  con\-text-de\-pend\-ent software.  Key features of COP are
  \emph{layers}---modules to describe context-dependent behavioral
  variations of a software system---and their \emph{dynamic
    activation}, which can modify the behavior of multiple objects
  that have already been instantiated.  Typechecking programs written
  in a COP language is difficult because the activation of a layer can
  even change objects' interfaces.  Inoue et al. have informally
  discussed how to make JCop, an extension of Java for COP by
  Appeltauer et al., type-safe.
  
  In this article, we formalize a small COP language called \fsname{}
  with its operational semantics and type system and show its type
  soundness.  The language models main features of the type-safe
  version of JCop, including dynamically activated \emph{first-class}
  layers, \emph{inheritance} of layer definitions, layer
  \emph{subtyping}, and layer \emph{swapping}.
\end{abstract}

%

\begin{keyword}
  Context-oriented programming \sep
  dynamic layer composition \sep
  first-class layers \sep
  layer inheritance \sep
  type systems



\end{keyword}

\end{frontmatter}

\doclicenseThis

\section{Introduction}
\label{sec:intro}

Software is much more interactive than it used to be: it interacts with
not only users but also external resources such as network and sensors
and changes its behavior according to inputs from these resources.
For example, an e-mail reader may switch to a text-based mode when
network throughput is low.  Such external information that affects the
behavior of software is often referred to as \emph{contexts} and 
software that is aware of contexts as context-dependent software.
However, context-dependent software is hard to develop and
maintain, because the description of context-dependent behavior, which
we desire to be modularized, often crosscuts with the dominating
module structure.  To address such a problem from a
programming-language perspective, Context-Oriented Programming
(COP) has been proposed by Hirschfeld et al~\cite{hirschfeld2008cop}.

The main language constructs for COP are \emph{layers}, which are
modules to specify context-dependent behavior, and their \emph{dynamic
  layer activation}.  A layer is basically a collection of what are
called \emph{partial methods}, which add new behavior to existing
objects or override existing methods.  When a layer is activated at
run time by a designated construct, the partial methods defined in it
become effective, changing the behavior of objects until the
activation ends.  Roughly speaking, a layer abstracts a context and
dynamic layer activation abstracts change of contexts.

The JCop language~\cite{appeltauer2012jcop} is an extension of Java with
language constructs for COP.  Not only does it support basic COP
constructs described above, but also it introduces many advanced
features such as inheritance of layer implementations and first-class
layers.  However, typechecking implemented in the JCop compiler does
not take into account the fact that layer activation can change
objects' interface by partial methods that add new methods and, as a
result, not all ``method not found'' errors are prevented statically.
In our
previous work~\cite{inoue2014towards}, we have studied this problem,
proposed a type-safe version of JCop (we call Safe JCop in this
paper) with informal discussions on its type system.

In this paper, we formalize most of the ideas proposed in the previous
work and prove that they really make the language sound.  More concretely,
we develop a small COP language called \fsname{}, which extends
ContextFJ by Igarashi, Hirschfeld, and Masuhara~\cite{contextfj2011,DynamicLayer2012contextfj}
to layer inheritance,
subtyping of layer types, first-class layers, and a type-safe layer deactivation mechanism called layer swapping~\cite{inoue2014towards}; and
we prove a type soundness theorem for \fsname.
Main issues we have to deal with are (1) the semantics of layer
inheritance, which adds another ``dimension'' to the space of method
lookup, (2) sound subtyping for first-class layers, which led us to
two kinds of subtyping relation, and (3) layer swapping.
A preliminary version of this work has been presented
elsewhere~\cite{inoue2015sound} under the title ``A Sound Type System
for Layer Subtyping and Dynamically Activated First-Class Layers.''
We have extended \fsname{} given there with \ensuremath{\itbox{superproceed}} calls, which have
been omitted, added proofs, and substantially revised the paper.

The rest of the article is organized as follows.  After informally
reviewing features of Safe JCop in~\secref{cop}, we develop \fsname{}
with its syntax, operational semantics, and type system in \secref{calculus};
and we prove type soundness in \secref{type_soundness}.  In \secref{relwork},
we discuss related work and then conclude in \secref{discussion}.

\section{Language Constructs of Safe JCop}
\label{sec:cop}

In this section, we review language constructs of Safe JCop, first
described in~\cite{inoue2014towards}, including first-class layers,
layer inheritance/subtyping, and layer swapping along informal
discussions about the type system.

As a running example, we consider programming a graphical computer
game called \emph{RetroAdventure}~\cite{appeltauer2013declarative}.
In this game, a player has a character ``hero'' that wanders around
the game world.  Here, we introduce class \ensuremath{\itbox{Hero}} that represents the
hero, which has method \ensuremath{\itbox{move}} to walk around, and class \ensuremath{\itbox{World}}
that represents the game world.

\begin{lstlisting}[label={lst:hero}]
public class Hero {
  Position pos;
  public void move(Direction dir){
    pos = /* changes pos according to dir */;
  }
}
public class World { ... }
\end{lstlisting}

\subsection{Layers and Partial Methods}

As mentioned already, a first distinctive feature of COP is
\emph{layers}---collections of \emph{partial methods} to modify the
behavior of existing objects.  A partial method is syntactically
similar to an ordinary method declared in a class, except that the
name is given in a qualified form \ensuremath{\itbox{Hero.move()}}; this means the
partial method is going to override method \ensuremath{\itbox{move}} defined in \ensuremath{\itbox{Hero}} or
(if it does not exist) add to \ensuremath{\itbox{Hero}}.  A layer can contain partial
methods for different classes, so, when it is activated, it can affect
objects from various classes at once.  Similarly to \ensuremath{\itbox{super}} calls in
Java, the body of a partial method can contain \ensuremath{\itbox{proceed}} calls to
invoke the original method overridden by this partial method.

Here, suppose that the hero's behavior is influenced by weather conditions
in the game world.  For example, in a foggy weather, the hero gets
slow and, in a stormy weather, the hero cannot move as he likes.  Here
are layers that denote weathers of the game world.

\begin{lstlisting}[label={lst:firstex}]
public layer Foggy {
  /* partial method */
  public void Hero.move(Direction dir){
    pos = /* the distance of move is shorter */;
  }
}
public layer Stormy {
  /* partial method */
  public void Hero.move(Direction dir){
    proceed(randomDirection(dir));
  }
  /* baseless partial method */
  public Direction Hero.randomDirection(Direction dir){
    return /* add randomness to dir */;
  }
}
public layer Sunny { ... }
\end{lstlisting}
\ensuremath{\itbox{Foggy}} and \ensuremath{\itbox{Stormy}} have the definitions of \ensuremath{\itbox{Hero.move}}, which change
the behavior of the original definition in different ways.  In
particular, \ensuremath{\itbox{Hero.move}} in \ensuremath{\itbox{Stormy}} uses \ensuremath{\itbox{proceed}}, replacing the
arguments to calls to \ensuremath{\itbox{move}}.  It also has \ensuremath{\itbox{Hero.randomDirection}}, 
used to determine a new randomized direction to which the hero is
going to move.

Methods defined in classes are often referred to as \emph{base methods} and
partial methods without corresponding base methods as \emph{baseless
partial methods}.  Notice that activating a layer with baseless partial
methods extends object interfaces and \ensuremath{\itbox{proceed}} in a baseless partial method
is unsafe unless another layer activation provides a baseless partial method
of the same signature.

\subsection{Layer Activation and First-Class Layers}

In Safe JCop, a layer can be activated by using a layer instance
(created by a \ensuremath{\itbox{new}} expression, just as an ordinary Java object, from
a layer definition) in a \ensuremath{\itbox{with}} statement.  The following code snippet
shows how \ensuremath{\itbox{Foggy}} can be activated. (\ensuremath{\itbox{hero}} is an object of the class
\ensuremath{\itbox{Hero}}).

\begin{lstlisting}[label={lst:with}]
with(new Foggy()){
	hero.move(); /* The hero will get slow by Rainy weather. */
}
\end{lstlisting}

\noindent
Inside the body of \ensuremath{\itbox{with}}, dynamic method dispatch is affected by the
activated layers so that partial methods are looked up first.
So, \ensuremath{\itbox{move}}ment of the \ensuremath{\itbox{hero}} will be slow.

Layer activation has a dynamic extent in the sense that the behavior
of objects changes even in methods called from inside \ensuremath{\itbox{with}}.  If
more than one layer is activated, a more recent activation has
precedence and a \ensuremath{\itbox{proceed}} call in a more recently activated layer may
call another partial method (of the same name) in another layer.

In Safe JCop, a layer instance is a first-class citizen and can be
stored in a variable, passed to, or returned from a method.  A layer
name can be used as a type.  Combining with layer subtyping discussed
later, we can switch layers to activate by a run-time condition.  For
example, suppose that the game has 
\emph{difficulty} levels, determined at run time according to
some parameters, and each level is represented by 
an instance of a sublayer of \ensuremath{\itbox{Difficulty}}.  Then, we can
set the initial difficulty level by code like this:

\begin{lstlisting}
Difficulty diff = /* an expression to compute difficulty */ ;
with(diff){ ... }
\end{lstlisting}

Moreover, a layer can declare own fields and methods (although we do not model them
in layers in this article).  So, first-class layers significantly
enhance expressiveness of the language.

\subsection{Dependencies between Layers}

Baseless partial methods and layer activation that has dynamic extent
pose a challenge on typechecking because activation of a layer
including baseless partial methods can change object interfaces.  So, a method invocation,
including a \ensuremath{\itbox{proceed}} call, may or may not be safe depending on what
layers are activated at the program point.  Safe JCop adopts
\ensuremath{\itbox{requires}} clauses~\cite{DynamicLayer2012contextfj} for layer
definitions to express which layers should have been activated before
activating each layer (instance).  The type system checks whether each
activation satisfies the \ensuremath{\itbox{requires}} clause associated to the activated
layer and also uses \ensuremath{\itbox{requires}} clauses to estimate interfaces of
objects at every program point.

For example, consider another layer \ensuremath{\itbox{ThunderInStorm}}, which expresses an
event in a game.  It affects the way how the hero's direction is
randomized during a storm and includes a baseless partial method
with a \ensuremath{\itbox{proceed}} call.  To prevent \ensuremath{\itbox{ThunderInStorm}} from being activated
in a weather other than a storm, the layer \ensuremath{\itbox{requires}} \ensuremath{\itbox{Stormy}}
as follows:

\begin{lstlisting}[label={lst:req}]
public layer ThunderInStorm requires Stormy {
	public Direction Hero.randomDirection(Direction dir){
		Direction tmpd = proceed(dir);
		... /* change tmpd to speed up */
		return tmpd;
	}
}
\end{lstlisting}
An attempt at activating \ensuremath{\itbox{ThunderInStorm}} without activating \ensuremath{\itbox{Stormy}}
will be rejected by the type system (unless the activation appears in
a layer requiring \ensuremath{\itbox{Stormy}}).  Thanks to the \ensuremath{\itbox{requires}} clause, the
type system knows that the \ensuremath{\itbox{proceed}} call will not fail.  (It will
call the partial method of the same name in \ensuremath{\itbox{Stormy}} or some other
depending on what layers are activated at run time.)

\subsection{Layer Inheritance and Subtyping}

In Safe JCop, a layer can inherit definitions from another layer by
using the keyword \ensuremath{\itbox{extends}} and the \ensuremath{\itbox{extends}} relation between layers
yields subtyping, just like Java classes.  If weather layers have many
definitions in common, it is a good idea to define a superlayer
\ensuremath{\itbox{Weather}} and concrete weather layers as its sublayers.

\begin{lstlisting}[label={lst:subtyping}]
public layer Weather {
  public Text People.sayWeather(){ return new Text(""); }
  ...
}
public layer Stormy extends Weather {
  public Text People.sayWeather(){
    Text buf = superproceed();
    buf.setText("It's stormy today.");
    return buf;
  } ...
}
public layer Foggy extends Weather {
  public Text People.sayWeather(){ ... }
  ...
}
\end{lstlisting}

\begin{sloppypar}
\noindent Here, \ensuremath{\itbox{Weather}} provides (baseless) partial method
\ensuremath{\itbox{sayWeather}} to the class \ensuremath{\itbox{People}}, which returns \ensuremath{\itbox{Text}} data that
people say about weather condition.  The implementation of
\ensuremath{\itbox{People.sayWeather}} just returns an empty \ensuremath{\itbox{Text}} and
sublayers of the \ensuremath{\itbox{Weather}} override it.
Safe JCop provides \ensuremath{\itbox{superproceed}}, which calls a partial method
overridden because of layer inheritance.  The partial method of
\ensuremath{\itbox{Stormy}} sets the contents of the text using \ensuremath{\itbox{superproceed}}.
\end{sloppypar}

Since class subtyping equals to the reflexive and transitive closure
of the \ensuremath{\itbox{extends}} relation, we expect layer subtyping to be the same;
an instance of a sublayer can be substituted for that of its
superlayer.  However, substitutability is more subtle than one might
expect and we are led to distinguishing two kinds of substitutability
and introducing two kinds of subtyping relation, called weak and
normal subtyping.  The difference arises from \ensuremath{\itbox{requires}} clauses.  To
explain the issue, we define layer \ensuremath{\itbox{Thunder}}, which is the superlayer of
\ensuremath{\itbox{ThunderInStorm}} and \ensuremath{\itbox{ThunderInFog}} and a sublayer of a marker layer
\ensuremath{\itbox{Event}}.

\begin{lstlisting}[label={lst:superproceed}]
public layer Event { ... }
public layer Thunder extends Event requires Weather {
  public void change_font(Text label){ label.setFont("Italic"); }
  public Text People.sayWeather(){
	  Text txt = proceed();
	  change_font( txt );
	  return txt;
  }
}
public layer ThunderInStorm extends Thunder requires Stormy {
  public Text People.sayWeather(){
    Text buf = superproceed();
    buf.setText("Escape from here right now!!"); 
    return buf;
  }
  ...
}
public layer ThunderInFog extends Thunder requires Foggy {
  public Text People.sayWeather(){ ... }
}
\end{lstlisting}
\ensuremath{\itbox{Thunder}} changes the font of the text of what \ensuremath{\itbox{People}} say.  It
seems natural to set the \ensuremath{\itbox{requires}} clause of \ensuremath{\itbox{Thunder}} to be
\ensuremath{\itbox{Weather}}, since its two sublayers require \ensuremath{\itbox{Stormy}} and \ensuremath{\itbox{Foggy}}
respectively.

\paragraph{Weak subtyping}
An instance of a sublayer can be used where a superlayer is
\ensuremath{\itbox{require}}d, since a sublayer defines more partial methods than its
superlayer.  For example, to activate the following layer called
\ensuremath{\itbox{Thunder}}, which \ensuremath{\itbox{requires}} \ensuremath{\itbox{Weather}}, it suffices to activate
\ensuremath{\itbox{Foggy}}, a sublayer of \ensuremath{\itbox{Weather}}, beforehand.

\begin{lstlisting}
with(new Foggy()){
  // Thunder requires Weather and Foggy extends Weather
  with(new Thunder()){ ... }
}
\end{lstlisting}

We will formalize substitutability about \ensuremath{\itbox{requires}} as \emph{weak
  subtyping}, which is the reflexive transitive closure of the
\ensuremath{\itbox{extends}} relation between layer types.  For the weak subtyping to
work, we require that a sublayer declare, at least, what its
superlayer \ensuremath{\itbox{requires}} because partial methods inherited from the
superlayer may depend on them.  We could relax this condition if a
sublayer overrides all the partial methods but such a case is expected
to be rare and so not taken into account.\footnote{%
  Re-typechecking inherited methods under the new \texttt{requires}
  clause would be another way to relax this condition but this is
  against modular checking.}

\paragraph{Normal subtyping}
The above notion of subtyping is called weak because it does \emph{not}
guarantee safe substitutability for \emph{first-class} layers.
Consider layer \ensuremath{\itbox{Difficulty}} again and assume that it requires no other
layers and has sublayers \ensuremath{\itbox{Easy}} and \ensuremath{\itbox{Hard}}.  In the following code
snippet, the activation of \ensuremath{\itbox{diff}} appears safe because its static type
\ensuremath{\itbox{Difficulty}} does not require any layers to have been activated.
\begin{lstlisting}
Difficulty diff = someCondition() ? new Easy() : new Hard();
with(diff){ ... }
\end{lstlisting}
However, the case where \ensuremath{\itbox{Easy}} or \ensuremath{\itbox{Hard}} requires some layers breaks
the expected invariant that the dependency expressed by the \ensuremath{\itbox{requires}}
clauses is satisfied at run time.  So, for assignments and parameter
passing, we need one more condition for subtyping, namely, \ensuremath{\itbox{requires}}
of a sublayer must be the same as that of its superlayer.  We call
this strong notion of subtyping \emph{normal subtyping}.

\begin{figure}[htbp]
  \includegraphics[width=13cm,clip,page=2]{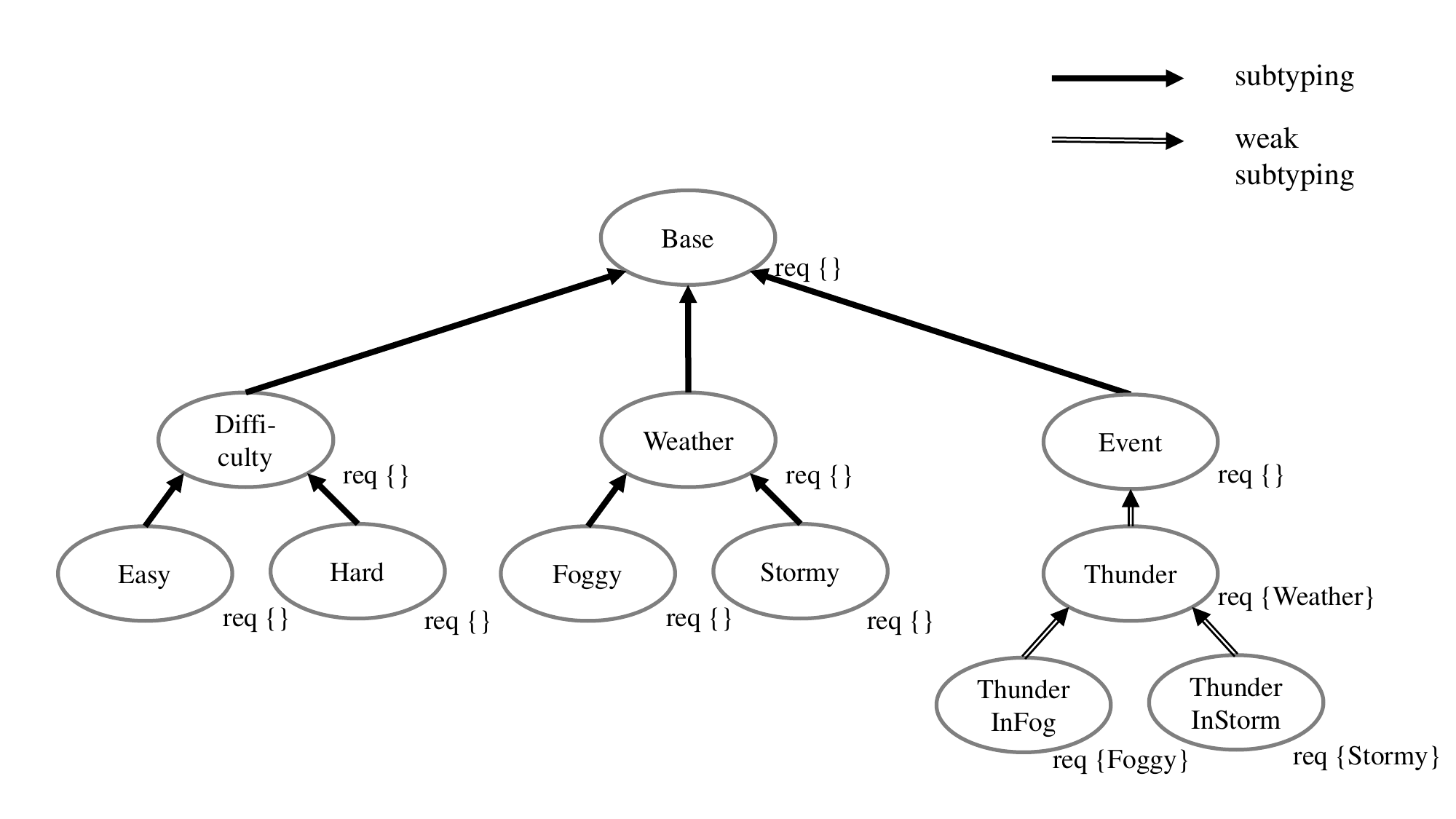}
  \caption{An example of layer subtyping hierarchy.}
  \label{fig:tree}
\end{figure}

In \figref{tree}, we show the layer subtyping hierarchy of the examples so far.
An oval means a layer and the notation \ensuremath{\itbox{req {\char'173}X{\char'175}}} beside an oval means its
\ensuremath{\itbox{requir}}\-ing layers.
Just like \ensuremath{\itbox{Object}} in Java, there is \ensuremath{\itbox{Base}}, which is a superlayer of
all layers, in Safe JCop.  If a layer omits the \ensuremath{\itbox{extends}} clause,
it is implicitly assumed that the layer \ensuremath{\itbox{extends Base}}.

\subsection{Layer Swapping and Deactivation}

The original JCop provides constructs to \emph{de}activate layers.
However, only with \ensuremath{\itbox{requires}}, it is not easy to guarantee that layer
deactivation does not lead to an error.  For safe deactivation, it has
to be checked that there is no layer that \ensuremath{\itbox{requires}} the deactivated
layer, but the type system is not designed to keep track of the
\emph{absence} of certain layers.  Instead of general-purpose layer
deactivation mechanisms, Safe JCop introduces a special construct to
express one important idiom that uses deactivation, namely
\textit{layer swapping} to deactivate some layers and activate a layer
at once.

In Safe JCop, we can define a layer as \emph{swappable}, which means that
all its sublayers can be swapped with each other, by adding the
modifier \ensuremath{\itbox{swappable}}. The \ensuremath{\itbox{swap}} statement for layer swapping is of the following
form:

\begin{center}
  \texttt{swap(}\textit{activation\_layer}, \textit{deactivation\_lay\-er\_type}\texttt{)\{ \ldots\ \}}
\end{center}

\noindent The \textit{activation\_layer} is an expression whose static
type must be a sublayer of \textit{deactivation\_lay\-er\_type}, which in
turn has to be swappable.  It deactivates \emph{all} instances of
\textit{deactivation\_layer\_type} (and its sublayers), and activates
the \textit{activation\_layer}.

Let's consider \ensuremath{\itbox{Difficulty}} once again.  We could define \ensuremath{\itbox{Difficulty}}
as a \ensuremath{\itbox{swappable}} layer and use \ensuremath{\itbox{swap}} to switch to another mode
temporarily.

\begin{lstlisting}[keywords={swap,layer,swappable},label={lst:swap}]
swappable layer Difficulty { ... }
...
Difficulty diff = someCondition() ? new Easy() : new Hard();
with(diff){
	...
	swap(new Hard(), Difficulty){
	  ... // Enforce hard mode
	}
}
\end{lstlisting}

For type safety, the necessary restriction for layer swapping was
wrong and has to be stronger than discussed in the previous
work~\cite{inoue2014towards}.  Specifically, we need the following restrictions:
\begin{itemize}
\item No sublayer of a swappable layer can be \ensuremath{\itbox{require}}d by any other layers.
\item Every sublayer of a swappable layer has to have the same interface (namely, set of public methods) and \ensuremath{\itbox{requires}} clause
  as the swappable layer.
\end{itemize}
The second condition was overlooked in the previous work.
\HI{for REVIEWER 1}
\AI{How about this?}
\AI{Maybe, it's a good idea to explain what the previous work overlooked: ``The second condition was overlooked in the previous work.'' or something like that.}
\HI{Sounds good.}

\subsection{Method Lookup}\label{sec:method_lookup}

We informally explain how Safe JCop's method lookup mechanism works,
before proceeding to the formal calculus.

When method \ensuremath{\itbox{m}} is invoked on an instance of class \ensuremath{\itbox{C}} with layers
$\ensuremath{\itbox{L}_{1}\itbox{;}}\cdots\ensuremath{\itbox{;L}_{n}\itbox{}}$ activated, the corresponding method definition is
sought as follows: first, the activated layers \ensuremath{\itbox{L}_{n}\itbox{}}, \ensuremath{\itbox{L}_{{n-1}}\itbox{}}, down
to \ensuremath{\itbox{L}_{1}\itbox{}} are searched (in this order) for a partial method named
\ensuremath{\itbox{C.m}}; if \ensuremath{\itbox{C.m}} is not found, the base class \ensuremath{\itbox{C}} is searched for the
base definition; if \ensuremath{\itbox{m}} is not found, similar search continues on the
\ensuremath{\itbox{C}}'s superclass \ensuremath{\itbox{D}}---namely, the activated layers are searched again
for a partial method named \ensuremath{\itbox{D.m}} and the base class \ensuremath{\itbox{D}} is searched
for the base definition, and so on.  In addition to the usual
inheritance chain in class-based object-oriented languages, COP adds
another dimension to the space of method lookup.  Actually, there is
yet another dimension in (Safe) JCop because of layer inheritance:
When \ensuremath{\itbox{L}_{i}\itbox{}} is searched for a partial method, its superlayers are
searched, too, before going to \ensuremath{\itbox{L}_{{i-1}}\itbox{}}.  For example, under the
following class and layer definitions
\begin{lstlisting}
class C extends D   { }
class D extends E   { void m(){ ... } }
class E             { void m(){ ... } }
layer L1            { void D.m(){ ... } }
layer L2 extends L3 { void E.m(){ ... } }
layer L3            { void C.m(){ ... } }
\end{lstlisting}
the following statement
\begin{lstlisting}
with(new L1()) {
  with(new L2()){
    new C().m();
  }
}
\end{lstlisting}
will execute partial method \ensuremath{\itbox{C.m}} defined in \ensuremath{\itbox{L3}} (we will use notation \ensuremath{\itbox{L.C.m}} to mean the partial method
\ensuremath{\itbox{C.m}} defined in layer \ensuremath{\itbox{L}} hereafter), whereas
the statement
\begin{lstlisting}
with(new L1()) { new C().m(); }
\end{lstlisting}
will execute \ensuremath{\itbox{L1.D.m}}.  

Now, we turn our attention to the semantics of \ensuremath{\itbox{super}}, \ensuremath{\itbox{proceed}}, and
\ensuremath{\itbox{superproceed}}.  When a \ensuremath{\itbox{super}}, \ensuremath{\itbox{proceed}} or \ensuremath{\itbox{superproceed}} call is
encountered during execution of a (partial) method, it continues to
look for a method definition of the same name as follows.

Suppose that \ensuremath{\itbox{C.m}} is found in layer \ensuremath{\itbox{L}_{i}\itbox{}} with layers
$\ensuremath{\itbox{}\overline{\itbox{L}}\itbox{}} = \ensuremath{\itbox{L}_{1}\itbox{;}}\cdots\ensuremath{\itbox{;L}_{n}\itbox{}}$ activated ($0 < i \le n$) and that \ensuremath{\itbox{D}} is
a superclass of \ensuremath{\itbox{C}}.
\begin{itemize}
  \item A call \ensuremath{\itbox{super.m()}} starts looking for a partial method \ensuremath{\itbox{D.m}} from
    \ensuremath{\itbox{L}_{n}\itbox{}} and so on.
  \item A \ensuremath{\itbox{proceed}} call starts looking for a partial method \ensuremath{\itbox{C.m}} from
    \ensuremath{\itbox{L}_{{i-1}}\itbox{}} or the base method of class \ensuremath{\itbox{C}} (when $i = 1$), and so
    on.
  \item A \ensuremath{\itbox{superproceed}} call starts looking for
    \ensuremath{\itbox{C.m}} in \ensuremath{\itbox{L}_{i}\itbox{{\ensuremath{'}}}} (where \ensuremath{\itbox{L}_{i}\itbox{{\ensuremath{'}}}} is the superlayer of \ensuremath{\itbox{L}_{i}\itbox{}}), \ensuremath{\itbox{L}_{i}\itbox{{\ensuremath{'}}{\ensuremath{'}}}}
    (where \ensuremath{\itbox{L}_{i}\itbox{{\ensuremath{'}}{\ensuremath{'}}}} is the superlayer of \ensuremath{\itbox{L}_{i}\itbox{{\ensuremath{'}}}}), and so on.  If \ensuremath{\itbox{C.m}}
    is not found in the superlayers, it is a run-time error (which the
    type system will prevent).
\end{itemize}

For example, consider the following class and layer definitions and
suppose \ensuremath{\itbox{L1}}, \ensuremath{\itbox{L2}} and \ensuremath{\itbox{L3}} are activated in this order.  
\begin{lstlisting}
class C extends D { }
class D extends E { void m(){ return super.m();} }
class E           { void m(){ return; } }
layer L1 {
  void C.m(){ ... super.m(); ... proceed(); ... }
  void D.m(){ ... super.m(); ... proceed(); ... }
}
layer L2 {
  void C.m(){ ... super.m(); ... proceed(); ... } 
}
layer L4 {
  void C.m(){ ... super.m(); ... proceed(); ... }
  void E.m(){ ... super.m(); ... proceed(); ... }
}
layer L3 extends L4 {
  void C.m(){ ... super.m(); ... proceed(); ... }
}
\end{lstlisting}

\begin{itemize}
\item \ensuremath{\itbox{super.m}} calls from \ensuremath{\itbox{L4.C.m}} and \ensuremath{\itbox{L1.C.m}} will invoke \ensuremath{\itbox{L1.D.m}};
  and those from \ensuremath{\itbox{L1.D.m}} and \ensuremath{\itbox{D.m}} will invoke \ensuremath{\itbox{L4.E.m}}, since \ensuremath{\itbox{L3}}
  inherits \ensuremath{\itbox{E.m}} from \ensuremath{\itbox{L4}}.
\item a \ensuremath{\itbox{proceed}} call from \ensuremath{\itbox{L4.C.m}} will invoke \ensuremath{\itbox{L2.C.m}}
and that from \ensuremath{\itbox{L1.C.m}} will invoke \ensuremath{\itbox{L1.D.m}}.
\item a \ensuremath{\itbox{superproceed}} call from \ensuremath{\itbox{L3.C.m}} will invoke \ensuremath{\itbox{L4.C.m}}.
\end{itemize}

\begin{figure}[htbp]
  \includegraphics[width=13cm,clip,page=1]{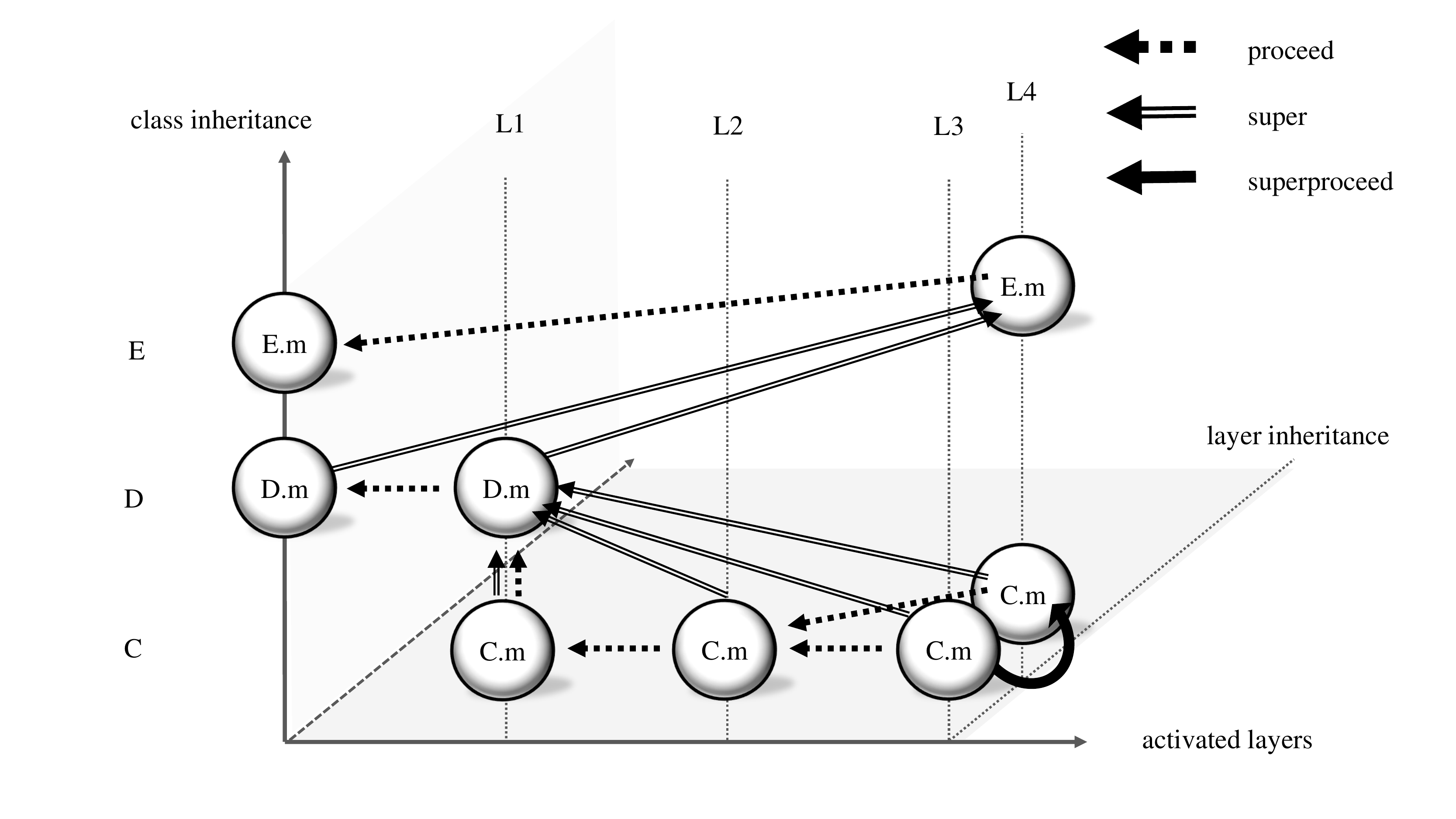}
  \caption{Method Lookup Example.}
  \label{fig:method-lookup}
\end{figure}

\figref{method-lookup} summarizes how \ensuremath{\itbox{super}}, \ensuremath{\itbox{proceed}}, and
\ensuremath{\itbox{superproceed}} calls are resolved.  Each ball represents a (partial)
method definition and its location where it is put.  The three axes
stands for class inheritance (\ensuremath{\itbox{C}} extends \ensuremath{\itbox{D}} and \ensuremath{\itbox{D}} extends \ensuremath{\itbox{E}}),
activated layers (\ensuremath{\itbox{L1}}, \ensuremath{\itbox{L2}}, and \ensuremath{\itbox{L3}} are activated in this order),
and layer inheritance (\ensuremath{\itbox{L3}} extends \ensuremath{\itbox{L4}}).  Dotted arrows represent how
\ensuremath{\itbox{proceed}} calls at each method definition are resolved.  For example,
the top-most long dotted arrow means that \ensuremath{\itbox{proceed}} from \ensuremath{\itbox{L4.E.m}} will
invoke \ensuremath{\itbox{E.m}}.  Double-line arrows represent \ensuremath{\itbox{super}} and thick arrows \ensuremath{\itbox{superproceed}}.

Finally, we should note that, for \ensuremath{\itbox{super}}, \ensuremath{\itbox{proceed}}, and
\ensuremath{\itbox{superproceed}} calls, the activated layers are the same as those when
the current method is found.  So, \ensuremath{\itbox{with}} or \ensuremath{\itbox{swap}} around \ensuremath{\itbox{super}},
\ensuremath{\itbox{proceed}}, and \ensuremath{\itbox{superproceed}} does not affect which definition is
invoked; only method invocations are affected by \ensuremath{\itbox{with}} and \ensuremath{\itbox{swap}}.

\section{\fsname} \label{sec:calculus}

In this section, we formalize a core functional subset of Safe JCop as
\fsname{} with its syntax, operational semantics and type system.
\fsname{}, a descendant of Featherweight Java
(FJ)~\cite{IgarashiPierceWadler01TOPLAS_FJ}, extends
ContextFJ~\cite{contextfj2011,DynamicLayer2012contextfj} with layer
inheritance, \ensuremath{\itbox{superproceed}}, layer subtyping, first-class layers, and
swappable layers.  JCop features that \fsname{} does \emph{not} model
for simplicity include: fields and (ordinary) methods in layers,
special variable \ensuremath{\itbox{thislayer}} to refer to the current layer instance,
\ensuremath{\itbox{superlayer}} to invoke an ordinary method in a superlayer, and
declarative layer composition.

\subsection{Syntax}

Let metavariables \ensuremath{\itbox{C}}, \ensuremath{\itbox{D}} and \ensuremath{\itbox{E}} range over class names; \ensuremath{\itbox{L}} over
layer names; \ensuremath{\itbox{f}} and \ensuremath{\itbox{g}} over field names; \ensuremath{\itbox{m}} over method names; \ensuremath{\itbox{x}}
and \ensuremath{\itbox{y}} over variables, which contains special variable \ensuremath{\itbox{this}}.  The
abstract syntax of \fsname{} is given in \figref{CFJ:syntax}.

\begin{figure}[h]
  \noindent
  \[
  \begin{array}{lllr}
    \ensuremath{\itbox{T}} & ::= & \ensuremath{\itbox{C}} \mid \ensuremath{\itbox{L}} & (\textit{types}) \\
    \ensuremath{\itbox{CL}} & ::= & \ensuremath{\itbox{class C \(\triangleleft\) C {\char'173} }\overline{\itbox{T}}\itbox{ }\overline{\itbox{f}}\itbox{; K }\overline{\itbox{M}}\itbox{ {\char'175}}}   & (\textit{classes}) \\
    \ensuremath{\itbox{LA}} & ::= & [\ensuremath{\itbox{swappable}}]\ensuremath{\itbox{ layer L \(\triangleleft\) L req }\overline{\itbox{L}}\itbox{ {\char'173} }\overline{\itbox{PM}}\itbox{ {\char'175}}} &(\textit{layers}) \\
    \ensuremath{\itbox{K}} & ::= & 
    \ensuremath{\itbox{C(}\overline{\itbox{T}}\itbox{ }\overline{\itbox{f}}\itbox{){\char'173} super(}\overline{\itbox{f}}\itbox{); this.}\overline{\itbox{f}}\itbox{ = }\overline{\itbox{f}}\itbox{; {\char'175}}}  
    & (\textit{constructors}) \\  
    \ensuremath{\itbox{M}} &  ::= & \ensuremath{\itbox{T m(}\overline{\itbox{T}}\itbox{ }\overline{\itbox{x}}\itbox{){\char'173} return e; {\char'175}}}  &(\textit{methods}) \\
    \ensuremath{\itbox{PM}} &  ::= & \ensuremath{\itbox{T C.m(}\overline{\itbox{T}}\itbox{ }\overline{\itbox{x}}\itbox{){\char'173} return e; {\char'175}}}  &(\textit{partial methods}) \\
    \ensuremath{\itbox{e}}, \ensuremath{\itbox{d}} & ::= &  \ensuremath{\itbox{x}} \mid  \ensuremath{\itbox{e.f}} \mid \ensuremath{\itbox{e.m(}\overline{\itbox{e}}\itbox{)}} \mid \ensuremath{\itbox{new T(}\overline{\itbox{e}}\itbox{)}}     \mid \ensuremath{\itbox{with e e}} \mid \ensuremath{\itbox{swap (e,L) e}} 
  & (\textit{expressions}) \\ & &
    \mid \ensuremath{\itbox{proceed(}\overline{\itbox{e}}\itbox{)}} \mid \ensuremath{\itbox{super.m(}\overline{\itbox{e}}\itbox{)}} \mid \ensuremath{\itbox{superproceed(}\overline{\itbox{e}}\itbox{)}} \\
    & & \mid \ensuremath{\itbox{new C(}\overline{\itbox{v}}\itbox{)<C,}\overline{\itbox{L}}\itbox{,}\overline{\itbox{L}}\itbox{>.m(}\overline{\itbox{e}}\itbox{)}}
    \mid \ensuremath{\itbox{new C(}\overline{\itbox{v}}\itbox{)<C,L,}\overline{\itbox{L}}\itbox{,}\overline{\itbox{L}}\itbox{>.m(}\overline{\itbox{e}}\itbox{)}}\\
    \ensuremath{\itbox{v}}, \ensuremath{\itbox{w}} & ::= &  \ensuremath{\itbox{new C(}\overline{\itbox{v}}\itbox{)}} \mid \ensuremath{\itbox{new L()}} & (\textit{values})
  \end{array}
  \]
  \caption{\fsname: Syntax.}
  \label{fig:CFJ:syntax}
\end{figure}

Following FJ, we use overlines to denote sequences: So, \ensuremath{\itbox{}\overline{\itbox{f}}\itbox{}} stands
for a possibly empty sequence $\ensuremath{\itbox{f}_{1}\itbox{}},\cdots,\ensuremath{\itbox{f}_{n}\itbox{}}$ and similarly for
\ensuremath{\itbox{}\overline{\itbox{T}}\itbox{}}, \ensuremath{\itbox{}\overline{\itbox{x}}\itbox{}}, \ensuremath{\itbox{}\overline{\itbox{e}}\itbox{}}, and so on.  The empty sequence is denoted by
$\bullet$.  Concatenation of sequences is often denoted by a comma
except for layer names, for which we use a semicolon.  We also
abbreviate pairs of sequences, writing ``\ensuremath{\itbox{}\overline{\itbox{T}}\itbox{ }\overline{\itbox{f}}\itbox{}}'' for
``$\ensuremath{\itbox{T}_{1}\itbox{ f}_{1}\itbox{}},\cdots,\ensuremath{\itbox{T}_{n}\itbox{ f}_{n}\itbox{}}$'', where $n$ is the length of \ensuremath{\itbox{}\overline{\itbox{T}}\itbox{}} and
\ensuremath{\itbox{}\overline{\itbox{f}}\itbox{}}, and similarly ``\ensuremath{\itbox{}\overline{\itbox{T}}\itbox{ }\overline{\itbox{f}}\itbox{;}}'' as shorthand for the sequence of
declarations ``\ensuremath{\itbox{T}_{1}\itbox{ f}_{1}\itbox{;}}\ldots\ensuremath{\itbox{T}_{n}\itbox{ f}_{n}\itbox{;}}'' and ``\ensuremath{\itbox{this.}\overline{\itbox{f}}\itbox{=}\overline{\itbox{f}}\itbox{;}}'' for
``\ensuremath{\itbox{this.f}_{1}\itbox{=f}_{1}\itbox{;}}\ldots\ensuremath{\itbox{;this.f}_{n}\itbox{=f}_{n}\itbox{;}}''.  Given layer sequence \ensuremath{\itbox{}\overline{\itbox{L}}\itbox{}}, We write $\set{\ensuremath{\itbox{}\overline{\itbox{L}}\itbox{}}}$ for the set of layers (obtained by ignoring the order).  Sequences of field declarations,
parameter names, layer names, and method declarations are assumed to
contain no duplicate names.

We briefly explain the syntax, focusing on COP-related constructs.  A
layer definition \ensuremath{\itbox{LA}} consists of optional modifier \ensuremath{\itbox{swappable}}, its
name, its superlayer name, layers that it \ensuremath{\itbox{requires}}, and partial
methods.  A partial method (defined as \ensuremath{\itbox{PM}}) is similar to a method
but specifies which \ensuremath{\itbox{m}} to modify by qualifying the simple
method name with a class name \ensuremath{\itbox{C}}.

Instantiation can be a layer instance \ensuremath{\itbox{new L()}}, as well as a class
instance \ensuremath{\itbox{new C(}\overline{\itbox{e}}\itbox{)}}.  Note that arguments to \ensuremath{\itbox{new L}} are always empty
because \fsname{} does not model fields of layer instances.  In the
expression \ensuremath{\itbox{with e}_{1}\itbox{ e}_{2}\itbox{}}, \ensuremath{\itbox{e}_{1}\itbox{}} stands for the layer to be activated
and \ensuremath{\itbox{e}_{2}\itbox{}} the body of \ensuremath{\itbox{with}}.
In the expression \ensuremath{\itbox{swap (e}_{1}\itbox{, L) e}_{2}\itbox{}}, \ensuremath{\itbox{e}_{1}\itbox{}} means the layer to be
activated, \ensuremath{\itbox{L}} the swappable layer, \ensuremath{\itbox{e}_{2}\itbox{}} the body of \ensuremath{\itbox{swap}}.  By this
expression, during the evaluation of \ensuremath{\itbox{e}_{2}\itbox{}}, all instances of the
swappable layer \ensuremath{\itbox{L}} and its sublayers are deactivated, and \ensuremath{\itbox{e}_{1}\itbox{}} is
activated.
\ensuremath{\itbox{super.m(}\overline{\itbox{e}}\itbox{)}}, \ensuremath{\itbox{proceed(}\overline{\itbox{e}}\itbox{)}} and \ensuremath{\itbox{superproceed(}\overline{\itbox{e}}\itbox{)}} are keywords to
invoke methods of the superclass, a previously activated layer, and
the superlayer, respectively.

Expressions $\ensuremath{\itbox{new C(}\overline{\itbox{v}}\itbox{)<D,}\overline{\itbox{L}}\itbox{{\ensuremath{'}},}\overline{\itbox{L}}\itbox{>.m(}\overline{\itbox{e}}\itbox{)}}$ and
$\ensuremath{\itbox{new C(}\overline{\itbox{v}}\itbox{)<D,L,}\overline{\itbox{L}}\itbox{{\ensuremath{'}},}\overline{\itbox{L}}\itbox{>.m(}\overline{\itbox{e}}\itbox{)}}$ are special run-time
expressions that are related to method invocation mechanism of COP,
and not supposed to appear in classes and layers.  They basically mean
that \ensuremath{\itbox{m}} is going to be invoked on \ensuremath{\itbox{new C(}\overline{\itbox{v}}\itbox{)}}.  The annotation
$\ensuremath{\itbox{<D,}\overline{\itbox{L}}\itbox{{\ensuremath{'}},}\overline{\itbox{L}}\itbox{>}}$ is used to model \ensuremath{\itbox{super}} and 
\ensuremath{\itbox{proceed}} whereas $\ensuremath{\itbox{<D,L,}\overline{\itbox{L}}\itbox{{\ensuremath{'}},}\overline{\itbox{L}}\itbox{>}}$ is used for \ensuremath{\itbox{superproceed}}.  \ensuremath{\itbox{}\overline{\itbox{L}}\itbox{}} stands for a sequence of activated layers and \ensuremath{\itbox{D}}, \ensuremath{\itbox{L}} and \ensuremath{\itbox{}\overline{\itbox{L}}\itbox{{\ensuremath{'}}}} (which is assumed to be a prefix
of \ensuremath{\itbox{}\overline{\itbox{L}}\itbox{}}) play a role of a ``cursor'' where the method
lookup starts from.  We explain how they work in detail in \secref{semantics}.

\paragraph{Program}
A \fsname{} program $(\CT,\LT,\ensuremath{\itbox{e}})$ consists of a class table $\CT$, a
layer table $\LT$ and an expression \ensuremath{\itbox{e}}, which stands for the body of
the \ensuremath{\itbox{main}} method.  $\CT$ maps a class name to a class definition and
$\LT$ a layer name to a layer definition.  A layer definition can
be regarded as a function that maps a partial method name \ensuremath{\itbox{C.m}} to a
partial method definition.  So, we can view $\LT$ as a Curried
function, and we often write $\LT(\ensuremath{\itbox{L}})(\ensuremath{\itbox{C.m}})$ for the partial method
\ensuremath{\itbox{C.m}} in \ensuremath{\itbox{L}} in a program.  We assume that the domains of $\CT$ and $\LT$
are finite.
Precisely speaking, the semantics and type
system are parameterized over $\CT$ and $\LT$ but, to lighten the
notation, we assume them to be fixed and omit from judgments.

Given $\CT$ and $LT$, \ensuremath{\itbox{extends}} and \ensuremath{\itbox{requires}} clauses are considered
relations, written \ensuremath{\itbox{\(\triangleleft\)}} and \ensuremath{\itbox{req}}, respectively, over class/layer
names.  Namely, we write \ensuremath{\itbox{L req L}_{i}\itbox{}} if \(\LT(\ensuremath{\itbox{L}}) = \ensuremath{\itbox{layer L req }\overline{\itbox{L}}\itbox{}}\)
and $\ensuremath{\itbox{L}_{i}\itbox{}} \in \ensuremath{\itbox{}\overline{\itbox{L}}\itbox{}}$.
We also write $\ensuremath{\itbox{L req }}\set{\ensuremath{\itbox{}\overline{\itbox{L}}\itbox{}}}$ if \(\LT(\ensuremath{\itbox{L}}) = \ensuremath{\itbox{layer L req }\overline{\itbox{L}}\itbox{}}\).\footnote{Note that \texttt{L${}_1$ req L${}_2$} and $\texttt{L${}_1$ req } \set{\texttt{L}_2}$ have slightly different meanings; the former means \texttt{L}${}_2$ is one of the layers required by \texttt{L}${}_1$, whereas the latter means \texttt{L}${}_2$ is the only layer required by \texttt{L}${}_1$.}
As usual, we write \(\mathcal{R}^+\)
for the transitive closure of relation \(\mathcal{R}\);
similarly for \(\mathcal{R}^{*}\)
for the reflexive transitive closure of \(\mathcal{R}\).
We write \ensuremath{\itbox{L swappable}} if \(\LT(\ensuremath{\itbox{L}})\)
is defined with the \ensuremath{\itbox{swappable}} modifier.

We assume the
following sanity conditions are satisfied by a given program:

\begin{enumerate}
\item $\CT(\ensuremath{\itbox{C}}) = \ensuremath{\itbox{class C ...}}$ for any $\ensuremath{\itbox{C}} \in \dom(\CT)$.
\item $\ensuremath{\itbox{Object}} \not \in \dom(\CT)$.
\item For every class name \ensuremath{\itbox{C}} (except \ensuremath{\itbox{Object}}) appearing anywhere
  in $\CT$, $\ensuremath{\itbox{C}} \in \dom(\CT)$.
\item $\LT(\ensuremath{\itbox{L}}) = \ensuremath{\itbox{... layer L ...}}$ for any $\ensuremath{\itbox{L}} \in \dom(\LT)$.
\item $\ensuremath{\itbox{Base}} \not \in \dom(\LT)$.
\item\label{lang:base} For every layer name \ensuremath{\itbox{L}} (except \ensuremath{\itbox{Base}})
  appearing anywhere in $\LT$, $\ensuremath{\itbox{L}} \in \dom(\LT)$.
\item\label{lang:extends} Both for classes and layers, there are no
  cycles in the transitive closure of the \ensuremath{\itbox{extends}} clauses.
\item $\LT(\ensuremath{\itbox{L}})(\ensuremath{\itbox{C.m}}) = \ensuremath{\itbox{... C.m(...){\char'173}...{\char'175}}}$ for any $\ensuremath{\itbox{L}} \in
  \dom(\LT)\) and \(\ensuremath{\itbox{C}} \neq \ensuremath{\itbox{Object}}\) and \((\ensuremath{\itbox{C.m}}) \in \dom(\LT(\ensuremath{\itbox{L}}))$.
\end{enumerate}
These sanity conditions are an extension of those of FJ: conditions
for layers (4--7) are similar to those for classes (1--3, 7).  In
Condition \ref{lang:base}, like \ensuremath{\itbox{Object}} of classes, layer \ensuremath{\itbox{Base}} is
defined as the root of the layer inheritance/subtyping hierarchy.  In
the condition (8), \(\ensuremath{\itbox{C}} \neq \ensuremath{\itbox{Object}}\)
means that a layer cannot introduce a method to \ensuremath{\itbox{Object}}, which has no
base methods.  We could allow a layer to add methods to \ensuremath{\itbox{Object}} but
doing so would just clutter presentation---there are more rules to
deal with the fact that \ensuremath{\itbox{super}} calls cannot be made in partial
methods for \ensuremath{\itbox{Object}}.

\subsection{Operational Semantics}\label{sec:semantics}
\paragraph{Lookup Functions}
We need a few auxiliary lookup functions to define operational
semantics and they are defined in \figref{CFJ:auxfuns}.  The function $\fields(\ensuremath{\itbox{C}})$
returns a sequence \ensuremath{\itbox{}\overline{\itbox{T}}\itbox{ }\overline{\itbox{f}}\itbox{}} of pairs of a field name and its type by
collecting all field declarations from \ensuremath{\itbox{C}} and its superclasses.

\begin{figure}[h]
     \noindent
     \fbox{$\fields(\ensuremath{\itbox{C}}) = \ensuremath{\itbox{}\overline{\itbox{T}}\itbox{ }\overline{\itbox{f}}\itbox{}}$}
     \iffull
     \typicallabel{MB-NextLayer}
     \fi
     \infax[F-Object]{
     \fields(\ensuremath{\itbox{Object}}) = \bullet
     }
     \infrule[F-Class]{
     \ensuremath{\itbox{class C \(\triangleleft\) D {\char'173} }\overline{\itbox{T}}\itbox{ }\overline{\itbox{f}}\itbox{; ... {\char'175}}} \andalso
     \fields(\ensuremath{\itbox{D}}) = \ensuremath{\itbox{}\overline{\itbox{S}}\itbox{ }\overline{\itbox{g}}\itbox{}}
     }{
     \fields(\ensuremath{\itbox{C}}) = \ensuremath{\itbox{}\overline{\itbox{S}}\itbox{ }\overline{\itbox{g}}\itbox{}}, \ensuremath{\itbox{}\overline{\itbox{T}}\itbox{ }\overline{\itbox{f}}\itbox{}}
     }
     \fbox{$\pmbody(\ensuremath{\itbox{m}},\ensuremath{\itbox{C}},\ensuremath{\itbox{L}}) = \ensuremath{\itbox{}\overline{\itbox{x}}\itbox{.e}} \IN \ensuremath{\itbox{L}_{0}\itbox{}}$}
     \infrule[PMB-Layer]{
     \LT(\ensuremath{\itbox{L}})(\ensuremath{\itbox{C.m}}) = \ensuremath{\itbox{T}_{0}\itbox{ C.m(}\overline{\itbox{T}}\itbox{ }\overline{\itbox{x}}\itbox{){\char'173} return e; {\char'175}}}
     }{
     \pmbody(\ensuremath{\itbox{m}},\ensuremath{\itbox{C}},\ensuremath{\itbox{L}}) = \ensuremath{\itbox{}\overline{\itbox{x}}\itbox{.e}} \IN \ensuremath{\itbox{L}}
     }
     \infrule[PMB-Super]{
       \LT(\ensuremath{\itbox{L}})(\ensuremath{\itbox{C.m}}) \undf \andalso
       \ensuremath{\itbox{L \(\triangleleft\) L{\ensuremath{'}}}} \andalso
       \pmbody(\ensuremath{\itbox{m}},\ensuremath{\itbox{C}},\ensuremath{\itbox{L{\ensuremath{'}}}}) = \ensuremath{\itbox{}\overline{\itbox{x}}\itbox{.e}} \IN \ensuremath{\itbox{L}_{0}\itbox{}}
     }{
     \pmbody(\ensuremath{\itbox{m}},\ensuremath{\itbox{C}},\ensuremath{\itbox{L}}) = \ensuremath{\itbox{}\overline{\itbox{x}}\itbox{.e}} \IN \ensuremath{\itbox{L}_{0}\itbox{}}
     }
     \fbox{$\mbody(\ensuremath{\itbox{m}},\ensuremath{\itbox{C}},\ensuremath{\itbox{}\overline{\itbox{L}}\itbox{{\ensuremath{'}}}},\ensuremath{\itbox{}\overline{\itbox{L}}\itbox{}}) = \ensuremath{\itbox{}\overline{\itbox{x}}\itbox{.e}} \IN \ensuremath{\itbox{D}},\ensuremath{\itbox{}\overline{\itbox{L}}\itbox{{\ensuremath{'}}{\ensuremath{'}}}}$}
     \infrule[MB-Class]{
       \ensuremath{\itbox{class C \(\triangleleft\) D {\char'173} ... T}_{0}\itbox{ m(}\overline{\itbox{T}}\itbox{ }\overline{\itbox{x}}\itbox{){\char'173} return e; {\char'175} ... {\char'175}}}
     }{
     \mbody(\ensuremath{\itbox{m}},\ensuremath{\itbox{C}},\bullet,\ensuremath{\itbox{}\overline{\itbox{L}}\itbox{}}) = \ensuremath{\itbox{}\overline{\itbox{x}}\itbox{.e}} \IN \ensuremath{\itbox{C}}, \bullet
     }
     \infrule[MB-Layer]{
     \pmbody(\ensuremath{\itbox{m}},\ensuremath{\itbox{C}},\ensuremath{\itbox{L}_{0}\itbox{}}) = \ensuremath{\itbox{}\overline{\itbox{x}}\itbox{.e}} \IN \ensuremath{\itbox{L}_{1}\itbox{}}
     }{
     \mbody(\ensuremath{\itbox{m}},\ensuremath{\itbox{C}},(\ensuremath{\itbox{}\overline{\itbox{L}}\itbox{{\ensuremath{'}}}};\ensuremath{\itbox{L}_{0}\itbox{}}),\ensuremath{\itbox{}\overline{\itbox{L}}\itbox{}}) = \ensuremath{\itbox{}\overline{\itbox{x}}\itbox{.e}} \IN \ensuremath{\itbox{C}}, (\ensuremath{\itbox{}\overline{\itbox{L}}\itbox{{\ensuremath{'}}}};\ensuremath{\itbox{L}_{0}\itbox{}})
     }
     \infrule[MB-Super]{
     \ensuremath{\itbox{class C \(\triangleleft\) D {\char'173} ... }\overline{\itbox{M}}\itbox{ {\char'175}}} \andalso \ensuremath{\itbox{m}} \not \in \ensuremath{\itbox{}\overline{\itbox{M}}\itbox{}} \andalso
     \mbody(\ensuremath{\itbox{m}},\ensuremath{\itbox{D}},\ensuremath{\itbox{}\overline{\itbox{L}}\itbox{}},\ensuremath{\itbox{}\overline{\itbox{L}}\itbox{}}) = \ensuremath{\itbox{}\overline{\itbox{x}}\itbox{.e}} \IN \ensuremath{\itbox{E}},\ensuremath{\itbox{}\overline{\itbox{L}}\itbox{{\ensuremath{'}}}}
     }{
     \mbody(\ensuremath{\itbox{m}},\ensuremath{\itbox{C}},\bullet,\ensuremath{\itbox{}\overline{\itbox{L}}\itbox{}}) = \ensuremath{\itbox{}\overline{\itbox{x}}\itbox{.e}} \IN \ensuremath{\itbox{E}},\ensuremath{\itbox{}\overline{\itbox{L}}\itbox{{\ensuremath{'}}}}
     }
     \infrule[MB-NextLayer]{
     \pmbody(\ensuremath{\itbox{m}},\ensuremath{\itbox{C}},\ensuremath{\itbox{L}_{0}\itbox{}}) \undf \andalso
     \mbody(\ensuremath{\itbox{m}},\ensuremath{\itbox{C}},\ensuremath{\itbox{}\overline{\itbox{L}}\itbox{{\ensuremath{'}}}},\ensuremath{\itbox{}\overline{\itbox{L}}\itbox{}}) = \ensuremath{\itbox{}\overline{\itbox{x}}\itbox{.e}} \IN \ensuremath{\itbox{D}},\ensuremath{\itbox{}\overline{\itbox{L}}\itbox{{\ensuremath{'}}{\ensuremath{'}}}}
     }{
     \mbody(\ensuremath{\itbox{m}},\ensuremath{\itbox{C}},(\ensuremath{\itbox{}\overline{\itbox{L}}\itbox{{\ensuremath{'}}}};\ensuremath{\itbox{L}_{0}\itbox{}}),\ensuremath{\itbox{}\overline{\itbox{L}}\itbox{}}) = \ensuremath{\itbox{}\overline{\itbox{x}}\itbox{.e}} \IN \ensuremath{\itbox{D}}, \ensuremath{\itbox{}\overline{\itbox{L}}\itbox{{\ensuremath{'}}{\ensuremath{'}}}}
     }
 \caption{\fsname: Lookup functions.}
 \label{fig:CFJ:auxfuns}
\end{figure}

The function $\pmbody(\ensuremath{\itbox{m}},\ensuremath{\itbox{C}},\ensuremath{\itbox{L}})$ returns the parameters and body \ensuremath{\itbox{}\overline{\itbox{x}}\itbox{.e}}
of the partial method \ensuremath{\itbox{C.m}} defined in layer \ensuremath{\itbox{L}}.  It also returns the
layer name \ensuremath{\itbox{L}_{0}\itbox{}} at which \ensuremath{\itbox{C.m}} is found, which will be used in
reduction rules to deal with \ensuremath{\itbox{superproceed}}.  If partial method \ensuremath{\itbox{C.m}}
is not found in \ensuremath{\itbox{L}}, its superlayer \ensuremath{\itbox{L{\ensuremath{'}}}} is searched and so on.
The function $\mbody(\ensuremath{\itbox{m}},\ensuremath{\itbox{C}},\ensuremath{\itbox{}\overline{\itbox{L}}\itbox{}_{1}\itbox{}},\ensuremath{\itbox{}\overline{\itbox{L}}\itbox{}_{2}\itbox{}})$ returns the parameters
and body \ensuremath{\itbox{}\overline{\itbox{x}}\itbox{.e}} of method \ensuremath{\itbox{m}} in class \ensuremath{\itbox{C}} when the search starts from
\ensuremath{\itbox{}\overline{\itbox{L}}\itbox{}_{1}\itbox{}}; the other sequence \ensuremath{\itbox{}\overline{\itbox{L}}\itbox{}_{2}\itbox{}} keeps track of the layers that are
activated when the search initially started.
It also returns \ensuremath{\itbox{D}} and \ensuremath{\itbox{}\overline{\itbox{L}}\itbox{{\ensuremath{'}}{\ensuremath{'}}}} (which will be a prefix of \ensuremath{\itbox{}\overline{\itbox{L}}\itbox{}_{2}\itbox{}}),
information on where the method has been found.  For example, in the
rule \rn{MB-Layer}, which means that the method is found in class \ensuremath{\itbox{C}}
and layer \ensuremath{\itbox{L}_{0}\itbox{}} (or its superlayers), $\mbody$ returns \ensuremath{\itbox{C}} and
$(\ensuremath{\itbox{}\overline{\itbox{L}}\itbox{{\ensuremath{'}}}};\ensuremath{\itbox{L}_{0}\itbox{}})$.  Such information will be used in reduction rules to
deal with \ensuremath{\itbox{proceed}} and \ensuremath{\itbox{super}}.  Readers familiar with ContextFJ will
notice that the rules for $\mbody$ are mostly the same as those in
ContextFJ, except that $\pmbody(\ensuremath{\itbox{m}}, \ensuremath{\itbox{C}}, \ensuremath{\itbox{L}})$ is substituted for
$\PT(\ensuremath{\itbox{m}}, \ensuremath{\itbox{C}}, \ensuremath{\itbox{L}})$ to take layer inheritance into account.  By
reading the four rules defining the two functions in a bottom-up
manner, it is not hard to see the correspondence with the method
lookup procedure, informally described in \secref{method_lookup}.

\paragraph{Reduction}
The operational semantics of \fsname{} is given by a reduction
relation of the form $\reduceto{\ensuremath{\itbox{}\overline{\itbox{L}}\itbox{}}}{\ensuremath{\itbox{e}}}{\ensuremath{\itbox{e{\ensuremath{'}}}}}$, read ``expression
\ensuremath{\itbox{e}} reduces to \ensuremath{\itbox{e{\ensuremath{'}}}} under the activated layers \ensuremath{\itbox{}\overline{\itbox{L}}\itbox{}}.''  The sequence
\ensuremath{\itbox{}\overline{\itbox{L}}\itbox{}} of layer names stands for nesting of \ensuremath{\itbox{with}} and the rightmost
name stands for the most recently activated layer.  As for other
sequences, \ensuremath{\itbox{}\overline{\itbox{L}}\itbox{}} do not contain duplicate names.  Note that we put a
sequence of layer names \ensuremath{\itbox{}\overline{\itbox{L}}\itbox{}} rather than layer instances because layer
instances have no fields and \ensuremath{\itbox{new L()}} and \ensuremath{\itbox{L}} can be identified.  If
we modelled fields in layer instances, we would have to put instances
for layer names.

\begin{figure}
  \leavevmode
  \centering
  \typicallabel{R-InvkArg}
  \infrule[R-Field]{
    \fields(\ensuremath{\itbox{C}}) = \ensuremath{\itbox{}\overline{\itbox{C}}\itbox{ }\overline{\itbox{f}}\itbox{}}
  }{
    \reducetoL{\ensuremath{\itbox{new C(}\overline{\itbox{v}}\itbox{).f}_{i}\itbox{}}}{\ensuremath{\itbox{v}_{i}\itbox{}}}
  }
  \infrule[R-Invk]{
    \reducetoL{\ensuremath{\itbox{new C(}\overline{\itbox{v}}\itbox{)<C,}\overline{\itbox{L}}\itbox{,}\overline{\itbox{L}}\itbox{>.m(}\overline{\itbox{w}}\itbox{)}}}{\ensuremath{\itbox{e{\ensuremath{'}}}}}
  }{
    \reducetoL{\ensuremath{\itbox{new C(}\overline{\itbox{v}}\itbox{).m(}\overline{\itbox{w}}\itbox{)}}}{\ensuremath{\itbox{e{\ensuremath{'}}}}}
  }
  \infrule[R-InvkB]{
    \mbody(\ensuremath{\itbox{m}},\ensuremath{\itbox{C{\ensuremath{'}}}},\ensuremath{\itbox{}\overline{\itbox{L}}\itbox{{\ensuremath{'}}{\ensuremath{'}}}},\ensuremath{\itbox{}\overline{\itbox{L}}\itbox{{\ensuremath{'}}}})= \ensuremath{\itbox{}\overline{\itbox{x}}\itbox{.e}_{0}\itbox{}} \IN \ensuremath{\itbox{C{\ensuremath{'}}{\ensuremath{'}}}},\bullet \andalso
    \ensuremath{\itbox{class C{\ensuremath{'}}{\ensuremath{'}}\(\triangleleft\) D{\char'173}...{\char'175}}}
  }{
    \reduceto{\ensuremath{\itbox{}\overline{\itbox{L}}\itbox{}}}{
      \ensuremath{\itbox{new C(}\overline{\itbox{v}}\itbox{)<C{\ensuremath{'}},}\overline{\itbox{L}}\itbox{{\ensuremath{'}}{\ensuremath{'}},}\overline{\itbox{L}}\itbox{{\ensuremath{'}}>.m(}\overline{\itbox{w}}\itbox{)}}
    }{} \\ \qquad \qquad
    \left[\begin{array}{l@{/}l}
        \ensuremath{\itbox{new C(}\overline{\itbox{v}}\itbox{)}} & \ensuremath{\itbox{this}}, \\
        \ensuremath{\itbox{}\overline{\itbox{w}}\itbox{}} & \ensuremath{\itbox{}\overline{\itbox{x}}\itbox{}}, \\
        \ensuremath{\itbox{new C(}\overline{\itbox{v}}\itbox{)<D,}\overline{\itbox{L}}\itbox{{\ensuremath{'}},}\overline{\itbox{L}}\itbox{{\ensuremath{'}}>}} & \ensuremath{\itbox{super}}
      \end{array}\right]\ensuremath{\itbox{e}_{0}\itbox{}}
  }
  \infrule[R-InvkP]{
    \mbody(\ensuremath{\itbox{m}},\ensuremath{\itbox{C{\ensuremath{'}}}},\ensuremath{\itbox{}\overline{\itbox{L}}\itbox{{\ensuremath{'}}{\ensuremath{'}}}},\ensuremath{\itbox{}\overline{\itbox{L}}\itbox{{\ensuremath{'}}}})= \ensuremath{\itbox{}\overline{\itbox{x}}\itbox{.e}_{0}\itbox{}} \IN \ensuremath{\itbox{C{\ensuremath{'}}{\ensuremath{'}}}},(\ensuremath{\itbox{}\overline{\itbox{L}}\itbox{{\ensuremath{'}}{\ensuremath{'}}{\ensuremath{'}}}};\ensuremath{\itbox{L}_{0}\itbox{}}) \andalso
    \ensuremath{\itbox{class C{\ensuremath{'}}{\ensuremath{'}}\(\triangleleft\) D{\char'173}...{\char'175}}} \andalso
    \ensuremath{\itbox{layer L}_{0}\itbox{\(\triangleleft\) L}_{1}\itbox{}}
  }{
    \reduceto{\ensuremath{\itbox{}\overline{\itbox{L}}\itbox{}}}{
      \ensuremath{\itbox{new C(}\overline{\itbox{v}}\itbox{)<C{\ensuremath{'}},}\overline{\itbox{L}}\itbox{{\ensuremath{'}}{\ensuremath{'}},}\overline{\itbox{L}}\itbox{{\ensuremath{'}}>.m(}\overline{\itbox{w}}\itbox{)}}
    }{} \\  \qquad \qquad \qquad \quad
    \left[
      \begin{array}{l@{/}l}
        \ensuremath{\itbox{new C(}\overline{\itbox{v}}\itbox{)}} & \ensuremath{\itbox{this}}, \\
        \ensuremath{\itbox{}\overline{\itbox{w}}\itbox{}} & \ensuremath{\itbox{}\overline{\itbox{x}}\itbox{}}, \\
        \ensuremath{\itbox{new C(}\overline{\itbox{v}}\itbox{)<C{\ensuremath{'}}{\ensuremath{'}},}\overline{\itbox{L}}\itbox{{\ensuremath{'}}{\ensuremath{'}}{\ensuremath{'}},}\overline{\itbox{L}}\itbox{{\ensuremath{'}}>.m}} & \ensuremath{\itbox{proceed}}, \\
        \ensuremath{\itbox{new C(}\overline{\itbox{v}}\itbox{)<D,}\overline{\itbox{L}}\itbox{{\ensuremath{'}},}\overline{\itbox{L}}\itbox{{\ensuremath{'}}>}} & \ensuremath{\itbox{super}}, \\
        \ensuremath{\itbox{new C(}\overline{\itbox{v}}\itbox{)<C{\ensuremath{'}}{\ensuremath{'}},L}_{1}\itbox{,(}\overline{\itbox{L}}\itbox{{\ensuremath{'}}{\ensuremath{'}}{\ensuremath{'}};L}_{0}\itbox{),}\overline{\itbox{L}}\itbox{{\ensuremath{'}}>.m}} & \ensuremath{\itbox{superproceed}}
      \end{array}\right]\ensuremath{\itbox{e}_{0}\itbox{}}
  }
  \infrule[R-InvkSP]{
    \pmbody(\ensuremath{\itbox{m}},\ensuremath{\itbox{C{\ensuremath{'}}}},\ensuremath{\itbox{L}_{1}\itbox{}})= \ensuremath{\itbox{}\overline{\itbox{x}}\itbox{.e}_{0}\itbox{}} \IN \ensuremath{\itbox{L}_{2}\itbox{}} \andalso
    \ensuremath{\itbox{class C{\ensuremath{'}}\(\triangleleft\) D{\char'173}...{\char'175}}} \andalso
    \ensuremath{\itbox{layer L}_{2}\itbox{\(\triangleleft\) L}_{3}\itbox{}}
  }{
    \reduceto{\ensuremath{\itbox{}\overline{\itbox{L}}\itbox{}}}{
      \ensuremath{\itbox{new C(}\overline{\itbox{v}}\itbox{)<C{\ensuremath{'}},L}_{1}\itbox{,(}\overline{\itbox{L}}\itbox{{\ensuremath{'}}{\ensuremath{'}};L}_{0}\itbox{),}\overline{\itbox{L}}\itbox{{\ensuremath{'}}>.m(}\overline{\itbox{w}}\itbox{)}}
    }{} \\  \qquad \qquad \qquad \quad
    \left[
      \begin{array}{l@{/}l}
        \ensuremath{\itbox{new C(}\overline{\itbox{v}}\itbox{)}} & \ensuremath{\itbox{this}}, \\
        \ensuremath{\itbox{}\overline{\itbox{w}}\itbox{}} & \ensuremath{\itbox{}\overline{\itbox{x}}\itbox{}}, \\
        \ensuremath{\itbox{new C(}\overline{\itbox{v}}\itbox{)<C{\ensuremath{'}},}\overline{\itbox{L}}\itbox{{\ensuremath{'}}{\ensuremath{'}},}\overline{\itbox{L}}\itbox{{\ensuremath{'}}>.m}} & \ensuremath{\itbox{proceed}}, \\
        \ensuremath{\itbox{new C(}\overline{\itbox{v}}\itbox{)<D,}\overline{\itbox{L}}\itbox{{\ensuremath{'}},}\overline{\itbox{L}}\itbox{{\ensuremath{'}}>}} & \ensuremath{\itbox{super}}, \\
        \ensuremath{\itbox{new C(}\overline{\itbox{v}}\itbox{)<C{\ensuremath{'}},L}_{3}\itbox{,(}\overline{\itbox{L}}\itbox{{\ensuremath{'}}{\ensuremath{'}};L}_{0}\itbox{),}\overline{\itbox{L}}\itbox{{\ensuremath{'}}>.m}} & \ensuremath{\itbox{superproceed}}
      \end{array}\right]\ensuremath{\itbox{e}_{0}\itbox{}}
  }

  \caption{\fsname: Reduction Rules 1.}
  \label{fig:CFJ:reduction1}
\end{figure}

\begin{figure}
  \leavevmode
  \centering
  \typicallabel{R-InvkArgAAAA}
  \infrule[RC-With]{
    \with(\ensuremath{\itbox{L}},\ensuremath{\itbox{}\overline{\itbox{L}}\itbox{}}) = \ensuremath{\itbox{}\overline{\itbox{L}}\itbox{{\ensuremath{'}}}} \andalso
    \reduceto{\ensuremath{\itbox{}\overline{\itbox{L}}\itbox{{\ensuremath{'}}}}}{\ensuremath{\itbox{e}}}{\ensuremath{\itbox{e{\ensuremath{'}}}}}
  }{
    \reduceto{\ensuremath{\itbox{}\overline{\itbox{L}}\itbox{}}}{\ensuremath{\itbox{with new L() e}}}{\ensuremath{\itbox{with new L() e{\ensuremath{'}}}}}
  }
  \infrule[RC-Swap]{
    \swap(\ensuremath{\itbox{L}},\ensuremath{\itbox{L}_{{sw}}\itbox{}}, \ensuremath{\itbox{}\overline{\itbox{L}}\itbox{}}) = \ensuremath{\itbox{}\overline{\itbox{L}}\itbox{{\ensuremath{'}}}} \andalso
    \reduceto{\ensuremath{\itbox{}\overline{\itbox{L}}\itbox{{\ensuremath{'}}}}}{\ensuremath{\itbox{e}}}{\ensuremath{\itbox{e{\ensuremath{'}}}}}
  }{
    \reducetoL{\ensuremath{\itbox{swap (new L(),L}_{{sw}}\itbox{) e}}}{\ensuremath{\itbox{swap (new L(),L}_{{sw}}\itbox{) e{\ensuremath{'}}}}}
  }
  \infrule[RC-WithArg]{
    \reduceto{\ensuremath{\itbox{}\overline{\itbox{L}}\itbox{}}}{\ensuremath{\itbox{e}_{l}\itbox{}}}{\ensuremath{\itbox{e}_{l}\itbox{{\ensuremath{'}}}}}
  }{
    \reduceto{\ensuremath{\itbox{}\overline{\itbox{L}}\itbox{}}}{\ensuremath{\itbox{with e}_{l}\itbox{ e}}}{\ensuremath{\itbox{with e}_{l}\itbox{{\ensuremath{'}} e}}}
  }
  \infrule[RC-SwapArg]{
    \reduceto{\ensuremath{\itbox{}\overline{\itbox{L}}\itbox{}}}{\ensuremath{\itbox{e}_{l}\itbox{}}}{\ensuremath{\itbox{e}_{l}\itbox{{\ensuremath{'}}}}}
  }{
    \reduceto{\ensuremath{\itbox{}\overline{\itbox{L}}\itbox{}}}{\ensuremath{\itbox{swap (e}_{l}\itbox{,L}_{{sw}}\itbox{) e}}}{\ensuremath{\itbox{swap (e}_{l}\itbox{{\ensuremath{'}},L}_{{sw}}\itbox{) e}}}
  }
  \infrule[R-WithVal]{
  }{
    \reducetoL{\ensuremath{\itbox{with new L() v}}}{\ensuremath{\itbox{v}}}
  }
  \infrule[R-SwapVal]{
  }{
    \reducetoL{\ensuremath{\itbox{swap (new L(), L}_{{sw}}\itbox{) v}}}{\ensuremath{\itbox{v}}}
  }
  \infrule[RC-Field]{
    \reducetoL{\ensuremath{\itbox{e}_{0}\itbox{}}}{\ensuremath{\itbox{e}_{0}\itbox{{\ensuremath{'}}}}}
  }{
    \reducetoL{\ensuremath{\itbox{e}_{0}\itbox{.f}}}{\ensuremath{\itbox{e}_{0}\itbox{{\ensuremath{'}}.f}}}
  }
  \infrule[RC-InvkArg]{
    \reducetoL{\ensuremath{\itbox{e}_{i}\itbox{}}}{\ensuremath{\itbox{e}_{i}\itbox{{\ensuremath{'}}}}}
  }{
    \reducetoL{\ensuremath{\itbox{e}_{0}\itbox{.m(..,e}_{i}\itbox{,..)}}}{\ensuremath{\itbox{e}_{0}\itbox{.m(..,e}_{i}\itbox{{\ensuremath{'}},..)}}}
  }    
  \infrule[RC-InvkRecv]{
    \reducetoL{\ensuremath{\itbox{e}_{0}\itbox{}}}{\ensuremath{\itbox{e}_{0}\itbox{{\ensuremath{'}}}}}
  }{
    \reducetoL{\ensuremath{\itbox{e}_{0}\itbox{.m(}\overline{\itbox{e}}\itbox{)}}}{\ensuremath{\itbox{e}_{0}\itbox{{\ensuremath{'}}.m(}\overline{\itbox{e}}\itbox{)}}}
  }
  \infrule[RC-New]{
    \reducetoL{\ensuremath{\itbox{e}_{i}\itbox{}}}{\ensuremath{\itbox{e}_{i}\itbox{{\ensuremath{'}}}}}
  }{
    \reducetoL{\ensuremath{\itbox{new C(..,e}_{i}\itbox{,..)}}}{\ensuremath{\itbox{new C(..,e}_{i}\itbox{{\ensuremath{'}},..)}}}
  }
  \infrule[RC-InvkAArg1]{
    \reducetoL{\ensuremath{\itbox{e}_{i}\itbox{}}}{\ensuremath{\itbox{e}_{i}\itbox{{\ensuremath{'}}}}}
  }{
    \reducetoL{\ensuremath{\itbox{new C(}\overline{\itbox{v}}\itbox{)<C{\ensuremath{'}},}\overline{\itbox{L}}\itbox{{\ensuremath{'}}{\ensuremath{'}},}\overline{\itbox{L}}\itbox{{\ensuremath{'}}>.m(..,e}_{i}\itbox{,..)}}}
	          {\ensuremath{\itbox{new C(}\overline{\itbox{v}}\itbox{)<C{\ensuremath{'}},}\overline{\itbox{L}}\itbox{{\ensuremath{'}},}\overline{\itbox{L}}\itbox{{\ensuremath{'}}>.m(..,e}_{i}\itbox{{\ensuremath{'}},..)}}}
  }
  \infrule[RC-InvkAArg2]{
    \reducetoL{\ensuremath{\itbox{e}_{i}\itbox{}}}{\ensuremath{\itbox{e}_{i}\itbox{{\ensuremath{'}}}}}
  }{
    \reducetoL{\ensuremath{\itbox{new C(}\overline{\itbox{v}}\itbox{)<C{\ensuremath{'}},L,}\overline{\itbox{L}}\itbox{{\ensuremath{'}}{\ensuremath{'}},}\overline{\itbox{L}}\itbox{{\ensuremath{'}}>.m(..,e}_{i}\itbox{,..)}}}
	          {\ensuremath{\itbox{new C(}\overline{\itbox{v}}\itbox{)<C{\ensuremath{'}},L,}\overline{\itbox{L}}\itbox{{\ensuremath{'}},}\overline{\itbox{L}}\itbox{{\ensuremath{'}}>.m(..,e}_{i}\itbox{{\ensuremath{'}},..)}}}
  }
  \caption{\fsname: Reduction Rules 2.}
  \label{fig:CFJ:reduction2}
\end{figure}

Reduction rules are found in \figref{CFJ:reduction1} and \figref{CFJ:reduction2}.  \rn{R-Field} is
for field access and four rules \rn{R-InvkXX} are for method
invocation: \rn{R-Invk} initializes the cursor according to the
currently activated layers \ensuremath{\itbox{}\overline{\itbox{L}}\itbox{}}; the rules \rn{R-InvkB} and 
\rn{R-InvkP} represent invocation of a base and partial method, respectively, depending on which kind is found by $\mbody$; the rule \rn{R-InvkSP} deals with
the case where the cursor in the receiver object is a quadruple, which
occurs when the entire expression was a \ensuremath{\itbox{superproceed}} call.
In the last case, $\pmbody$ is used to find a method body because \ensuremath{\itbox{superproceed}} denotes a partial method in one of the superlayers.

Note how \ensuremath{\itbox{this}}, \ensuremath{\itbox{proceed}}, \ensuremath{\itbox{super}} and \ensuremath{\itbox{superproceed}} are replaced
with the receiver with different cursor locations.  For \ensuremath{\itbox{proceed}}, the
cursor of triple moves one layer to the left and, for \ensuremath{\itbox{super}}, the
cursor moves one level up in the direction of class inheritance,
resetting the layers.  Thanks to Sanity Condition (8), the superclass
\ensuremath{\itbox{D}} is always found.  If we allowed a layer to add baseless partial
methods to \ensuremath{\itbox{Object}}, we would have to have special rules, in which
there is no substitution for \ensuremath{\itbox{super}} (and typing rules to disallow the
use of \ensuremath{\itbox{super}} in such partial methods).  Igarashi et
al.~\cite{DynamicLayer2012contextfj} (as well as the conference
version of this article~\cite{inoue2015sound}) have overlooked this
subtlety.  For \ensuremath{\itbox{superproceed}}, the cursor moves one level up in the
direction of layer inheritance (generating a quadruple from a triple
in \rn{R-InvkP}).
For example, we show how cursors of a triple and a quadruple work
using example in \secref{method_lookup}.  Let \ensuremath{\itbox{e}} be \ensuremath{\itbox{new C().m()}}.
Then, the derivation of $\reduceto{\ensuremath{\itbox{L1}};\ensuremath{\itbox{L2}};\ensuremath{\itbox{L3}}}{\ensuremath{\itbox{e}}}{\ensuremath{\itbox{e{\ensuremath{'}}}}}$ will take the form:
\[
        \infer[\rn{R-Invk}]{
          \reduceto{\ensuremath{\itbox{L1}};\ensuremath{\itbox{L2}};\ensuremath{\itbox{L3}}}{\ensuremath{\itbox{new C().m()}}}{\ensuremath{\itbox{e{\ensuremath{'}}}}}
        }{
          \infer[\rn{R-InvkP}]{
             \reduceto{\ensuremath{\itbox{L1}};\ensuremath{\itbox{L2}};\ensuremath{\itbox{L3}}}{\ensuremath{\itbox{new C<C,(L1;L2;L3),(L1;L2;L3)>().m()}}}{\ensuremath{\itbox{e{\ensuremath{'}}}}}
          }{
             \mbody(\ensuremath{\itbox{m}},\ensuremath{\itbox{C}},(\ensuremath{\itbox{L1;L2;L3}}),(\ensuremath{\itbox{L1;L2;L3}})) = \bullet\ensuremath{\itbox{.e4 in C}}, (\ensuremath{\itbox{L1;L2;L3}})
          }
        }
\]
where \ensuremath{\itbox{e{\ensuremath{'}}}} is
$$
 \left[
   \begin{array}{l@/l}
     \ensuremath{\itbox{new C<C,(L1;L2;L3),(L1;L2;L3)>()}} & \ensuremath{\itbox{this}} \\
     \ensuremath{\itbox{new C<D,(L1;L2;L3),(L1;L2;L3)>()}} & \ensuremath{\itbox{super}} \\
     \ensuremath{\itbox{new C<C,(L2;L3),(L1;L2;L3)>().m}} & \ensuremath{\itbox{proceed}} \\
     \ensuremath{\itbox{new C<C,L4,(L1;L2;L3),(L1;L2;L3)>().m}} & \ensuremath{\itbox{superproceed}} \\
   \end{array}
 \right] \ensuremath{\itbox{e4}} .
$$

Now, we go back to \figref{CFJ:reduction2}.  The rules \rn{RC-With} and \rn{RC-Swap} express layer activation and swapping, respectively.  The auxiliary functions
$\with(\ensuremath{\itbox{L}}, \ensuremath{\itbox{}\overline{\itbox{L}}\itbox{}})$ and $\swap(\ensuremath{\itbox{L}} ,\ensuremath{\itbox{L}_{{sw}}\itbox{}},\ensuremath{\itbox{}\overline{\itbox{L}}\itbox{}})$ for context
manipulation are defined by:
\[
\begin{array}{lcr@{\qquad}lcr}
  \with(\ensuremath{\itbox{L}},\ensuremath{\itbox{}\overline{\itbox{L}}\itbox{}}) &=& (\ensuremath{\itbox{}\overline{\itbox{L}}\itbox{}} \setminus \set{\ensuremath{\itbox{L}}});\ensuremath{\itbox{L}} &
  \swap(\ensuremath{\itbox{L}},\ensuremath{\itbox{L}_{{sw}}\itbox{}},\ensuremath{\itbox{}\overline{\itbox{L}}\itbox{}}) &=& (\ensuremath{\itbox{}\overline{\itbox{L}}\itbox{}}\setminus 
  \set{\ensuremath{\itbox{L{\ensuremath{'}}}} \mid \ensuremath{\itbox{L{\ensuremath{'}} \(\triangleleft\)}^{\itbox{{*}}}\itbox{ L}_{{sw}}\itbox{}}});\ensuremath{\itbox{L}}
\end{array}
\]
The function $\with$ removes \ensuremath{\itbox{L}} (if exists) from layer
sequence \ensuremath{\itbox{}\overline{\itbox{L}}\itbox{}} and adds \ensuremath{\itbox{L}} to the end of \ensuremath{\itbox{}\overline{\itbox{L}}\itbox{}} and $\swap$ removes
all sublayers of \ensuremath{\itbox{L}_{{sw}}\itbox{}} from \ensuremath{\itbox{}\overline{\itbox{L}}\itbox{}}, and adds \ensuremath{\itbox{L}} to the
end of \ensuremath{\itbox{}\overline{\itbox{L}}\itbox{}}.\footnote{The symbol $\setminus$ is usually used to
  remove entities from a set, but we informally use it for a sequence
  here.}  The type system checks that \ensuremath{\itbox{L}_{{sw}}\itbox{}} is a \ensuremath{\itbox{swappable}} layer.  
Based on the
above, the rule \rn{RC-With} means that \ensuremath{\itbox{with (new L()) e}} executes
\ensuremath{\itbox{e}} with \ensuremath{\itbox{L}} activated (as the first layer).  The rule \rn{RC-Swap} is
similar; it means that
\ensuremath{\itbox{swap (new L(), L}_{{sw}}\itbox{) e}} executes by deactivating all sublayers of
\ensuremath{\itbox{L}_{{sw}}\itbox{}} and activating layer \ensuremath{\itbox{L}}.  
For example, we can derive:
\[
\infer[\rn{R-With}]{
  \reduceto{\bullet}{\ensuremath{\itbox{with new L1() (with new L2() (with new L3() e))}}}{\ensuremath{\itbox{e{\ensuremath{'}}}}}
}{
  \infer[\rn{R-With}]{
     \reduceto{\ensuremath{\itbox{L1}}}{\ensuremath{\itbox{with new L2() (with new L3() e)}}}{\ensuremath{\itbox{e{\ensuremath{'}}}}}
  }{
     \infer[\rn{R-With}]{
        \reduceto{\ensuremath{\itbox{L1}};\ensuremath{\itbox{L2}}}{\ensuremath{\itbox{with new L3() e}}}{\ensuremath{\itbox{e{\ensuremath{'}}}}}
     }{
        \infer[\rn{R-Invk}]{
           \reduceto{\ensuremath{\itbox{L1}};\ensuremath{\itbox{L2}};\ensuremath{\itbox{L3}}}{\ensuremath{\itbox{new C().m()}}}{\ensuremath{\itbox{e{\ensuremath{'}}}}}
        }{\vdots}
     }
  }
}
\]

The rules \rn{RC-WithArg} and \rn{RC-SwapArg} are for reduction of
expression \ensuremath{\itbox{e}_{l}\itbox{}} that is expected to become a layer instance.  Rules
\rn{RC-WithVal} and \rn{RC-SwapVal} are for final reduction steps of
\ensuremath{\itbox{with}} and \ensuremath{\itbox{swap}} blocks, respectively, that pass the value \ensuremath{\itbox{v}} as it is.  Other
rules for congruence are same as those of ContextFJ: \fsname{} reduction
is call by value but the order of reduction of subexpressions is
unspecified.

\subsection{Type System}

As usual, the role of a type system is to ensure the absence of a
certain class of run-time errors.  Here, they are ``field-not-found''
and ``method-not-found'' errors, including the failure of \ensuremath{\itbox{proceed}},
\ensuremath{\itbox{superproceed}} or \ensuremath{\itbox{super}} calls.

As discussed in the last section, the type system takes information on
activated layers at every program point into account.  We approximate
such information by a set \(\Lambda\)
of layer names, which mean that, for any layer in \(\Lambda\),
an instance of one of its sublayers has to be activated at run time.
This set gives underapproximation of activated layers; other layers
might be activated.  Activated layers are approximated by sets rather
than sequences because the type system is mainly concerned about
access to fields and methods and the order of activated layers does
not influence which fields and methods are accessible.

\begin{sloppypar}
In our type system, a type judgment for an expression is of the form
$\LLGp \ensuremath{\itbox{e}} : \ensuremath{\itbox{T}}$, where $\Gamma$ is a type environment, which records
types of variables, and $\Loc$ stands for where \ensuremath{\itbox{e}} appears, namely, a
method in a class (denoted by \ensuremath{\itbox{C.m}}) or a partial method in a layer (denoted by \ensuremath{\itbox{L.C.m}}).  For example, the
\ensuremath{\itbox{proceed}} call in the body of the partial method \ensuremath{\itbox{People.sayWeather()}} of
layer \ensuremath{\itbox{Thunder}} is typed as follows:
\begin{align*}
  \ensuremath{\itbox{Thunder.People.sayWeather}}; \{\ensuremath{\itbox{Weather}},\ensuremath{\itbox{Thunder}}\} ; \ensuremath{\itbox{this}} : \ensuremath{\itbox{People}} 
  \p \ensuremath{\itbox{proceed()}} : \ensuremath{\itbox{Text}}
\end{align*}
The layer name set $\{\ensuremath{\itbox{Weather}}, \ensuremath{\itbox{Thunder}}\}$ comes from the fact 
that \ensuremath{\itbox{Thunder}} \ensuremath{\itbox{requires}} \ensuremath{\itbox{Weather}}.  \ensuremath{\itbox{Thunder}} is also included
because \ensuremath{\itbox{Thunder}} (or one of its sublayers) is obviously activated when a partial method
defined in this very layer is executed.
\end{sloppypar}

We start with the definitions of two kinds of layer subtyping
discussed in the last section and proceed to functions to look up
method types and typing rules.

\paragraph{Subtyping}

We define subtyping \ensuremath{\itbox{C \(\Leq\) D}} for class types, weak subtyping
\(\ensuremath{\itbox{L}_{1}\itbox{}} \LEQ_w \ensuremath{\itbox{L}_{2}\itbox{}}\) and normal subtyping \(\ensuremath{\itbox{L}_{1}\itbox{}} \LEQ \ensuremath{\itbox{L}_{2}\itbox{}}\) for
layer types by the rules in \figref{CFJ:subtyping}.  Class subtyping
\ensuremath{\itbox{C \(\Leq\) D}} is defined as the reflexive and transitive closure of \ensuremath{\itbox{\(\triangleleft\)}},
just as FJ.  Weak layer subtyping is also the reflexive and transitive
closure of \ensuremath{\itbox{\(\triangleleft\)}}.  We extend it to the relation
\(\Lambda_1 \LEQ_w \Lambda_2\) between layer name sets by
\rn{LSS-Intro}: \(\Lambda_1 \LEQ_w \Lambda_2\) if and only if for
every element in \(\Lambda_2\), there must exist a sublayer of it in
\(\Lambda_1\).  It is used to check activated layers \(\Lambda_1\)
satisfy the requirement \(\Lambda_2\) given by a \ensuremath{\itbox{requires}} clause in
typechecking a layer activation.  Normal subtyping is almost the
reflexive and transitive closure of \ensuremath{\itbox{\(\triangleleft\)}} but there is one additional
condition: for \ensuremath{\itbox{L}_{1}\itbox{}} to be a normal subtype of \ensuremath{\itbox{L}_{2}\itbox{}}, the layers they
\ensuremath{\itbox{require}} must be the same (\rn{LS-Extends}).  Obviously, if
\(\ensuremath{\itbox{L}_{1}\itbox{}} \LEQ \ensuremath{\itbox{L}_{2}\itbox{}}\), then \(\ensuremath{\itbox{L}_{1}\itbox{}} \LEQ_w \ensuremath{\itbox{L}_{2}\itbox{}}\) (but not vice
versa).

\begin{figure}[ht]
  \begin{center}
    \begin{tabular}{c}
      \begin{minipage}{0.47\hsize}
        \begin{flushleft}
          \typicallabel{LS-Extends}
          \noindent
          \fbox{class subtyping $\LEQ$}
          \infrule[CL-Refl]{
          }{
            \ensuremath{\itbox{C}} \LEQ \ensuremath{\itbox{C}}
          }
          \infrule[CL-Trans]{
            \ensuremath{\itbox{C}} \LEQ \ensuremath{\itbox{D}} \andalso
            \ensuremath{\itbox{D}} \LEQ \ensuremath{\itbox{E}}
          }{
            \ensuremath{\itbox{C}} \LEQ \ensuremath{\itbox{E}}
          }
          \infrule[CL-Extends]{
            \ensuremath{\itbox{class  C \(\triangleleft\) D {\char'173}...{\char'175}}} 
          }{
              \ensuremath{\itbox{C}} \LEQ \ensuremath{\itbox{D}}
          }
          \fbox{normal layer subtyping $\LEQ$}
          \infrule[LS-Refl]{
          }{
            \ensuremath{\itbox{L}} \LEQ \ensuremath{\itbox{L}}
          }
          \infrule[LS-Trans]{
            \ensuremath{\itbox{L}_{1}\itbox{}} \LEQ \ensuremath{\itbox{L}_{2}\itbox{}} \andalso
            \ensuremath{\itbox{L}_{2}\itbox{}} \LEQ \ensuremath{\itbox{L}_{3}\itbox{}}
          }{
            \ensuremath{\itbox{L}_{1}\itbox{}} \LEQ \ensuremath{\itbox{L}_{3}\itbox{}}
          }
          \infrule[LS-Base]{
            \ensuremath{\itbox{L \(\triangleleft\) Base}} \andalso
            \ensuremath{\itbox{L req }} \emptyset
          }{
            \ensuremath{\itbox{L}} \LEQ \ensuremath{\itbox{Base}}
          }  
          \infrule[LS-Extends]{
            \ensuremath{\itbox{L}_{1}\itbox{ \(\triangleleft\) L}_{2}\itbox{}} \andalso
            \ensuremath{\itbox{L}_{1}\itbox{ req }} \LSet \\
            \ensuremath{\itbox{L}_{2}\itbox{ req }} \LSet
          }{
            \ensuremath{\itbox{L}_{1}\itbox{}} \LEQ \ensuremath{\itbox{L}_{2}\itbox{}}
          }
          \bigskip
          \smallskip
        \end{flushleft}
      \end{minipage}
      \quad      
      \begin{minipage}{0.47\hsize}
        \begin{flushleft}
          \typicallabel{LSw-Extends}
          \fbox{weak layer subtyping $\LEQ_w$}
          \infrule[LSw-Refl]{
          }{
            \ensuremath{\itbox{L}} \LEQ_w \ensuremath{\itbox{L}}
          }
          \infrule[LSw-Trans]{
            \ensuremath{\itbox{L}_{1}\itbox{}} \LEQ_w \ensuremath{\itbox{L}_{2}\itbox{}} \andalso
            \ensuremath{\itbox{L}_{2}\itbox{}} \LEQ_w \ensuremath{\itbox{L}_{3}\itbox{}}
          }{
            \ensuremath{\itbox{L}_{1}\itbox{}} \LEQ_w \ensuremath{\itbox{L}_{3}\itbox{}}
          }
          \infrule[LSw-Extends]{
            \ensuremath{\itbox{L}_{1}\itbox{ \(\triangleleft\) L}_{2}\itbox{}}
          }{
            \ensuremath{\itbox{L}_{1}\itbox{}} \LEQ_w \ensuremath{\itbox{L}_{2}\itbox{}}
          }
          \fbox{layer set subtyping}
          \infrule[LSS-Intro]{
            \forall \ensuremath{\itbox{L}_{0}\itbox{}} \in \LSet_0. \exists \ensuremath{\itbox{L}_{1}\itbox{}} \in \LSet_1 \text{ s.t. } \ensuremath{\itbox{L}_{1}\itbox{}} \LEQ_w \ensuremath{\itbox{L}_{0}\itbox{}}
          }{
            \LSet_1 \LEQ_w \LSet_0
          }
        \end{flushleft}
      \end{minipage}
    \end{tabular}
  \end{center}
  \caption{\fsname: Subtyping Relations.}
  \label{fig:CFJ:subtyping}
\end{figure}

\paragraph{Method type lookup}

Similarly to $\pmbody$ and $\mbody$, we define two auxiliary functions
$\pmtype$ and $\mtype$ to look up the signature \ensuremath{\itbox{}\overline{\itbox{T}}\itbox{\(\rightarrow\)T}_{0}\itbox{}} (consisting of
argument type \ensuremath{\itbox{}\overline{\itbox{T}}\itbox{}} and a return type \ensuremath{\itbox{T}_{0}\itbox{}}) of a (partial) method.
$\pmtype(\ensuremath{\itbox{m}},\ensuremath{\itbox{C}},\ensuremath{\itbox{L}})$ returns the signature of \ensuremath{\itbox{C.m}} in \ensuremath{\itbox{L}} (or one of its
superlayers).
$\mtype(\ensuremath{\itbox{m}},\ensuremath{\itbox{C}},\LSet_1,\LSet_2)$ returns the type of \ensuremath{\itbox{m}} in \ensuremath{\itbox{C}} under
the assumption that $\LSet_1$ is activated.  The other layer set
$\LSet_2$ ($\supseteq \LSet_1$) is used when the lookup goes on to a
superclass.  If $\LSet_1$ and $\LSet_2$ are the same,
which is mostly the case, we write $\mtype(\ensuremath{\itbox{m}},\ensuremath{\itbox{C}},\LSet_1)$.

\begin{figure}[ht]
  \typicallabel{PMT-PMethod}
  \fbox{$\pmtype(\ensuremath{\itbox{m}},\ensuremath{\itbox{C}},\ensuremath{\itbox{L}}) = \ensuremath{\itbox{}\overline{\itbox{T}}\itbox{\(\rightarrow\)T}_{0}\itbox{}}$}
  \infrule[PMT-Layer]{
    \LT(\ensuremath{\itbox{L}})(\ensuremath{\itbox{C.m}}) = \ensuremath{\itbox{T}_{0}\itbox{ C.m(}\overline{\itbox{T}}\itbox{ }\overline{\itbox{x}}\itbox{){\char'173} return e; {\char'175}}}
  }{
    \pmtype(\ensuremath{\itbox{m}},\ensuremath{\itbox{C}},\ensuremath{\itbox{L}}) = \ensuremath{\itbox{}\overline{\itbox{T}}\itbox{\(\rightarrow\)T}_{0}\itbox{}}
  }
  \infrule[PMT-Super]{
    \LT(\ensuremath{\itbox{L}})(\ensuremath{\itbox{C.m}}) \undf \andalso
    \ensuremath{\itbox{L \(\triangleleft\) L{\ensuremath{'}}}} \andalso
    \pmtype(\ensuremath{\itbox{m}},\ensuremath{\itbox{C}},\ensuremath{\itbox{L{\ensuremath{'}}}}) = \ensuremath{\itbox{}\overline{\itbox{T}}\itbox{\(\rightarrow\)T}_{0}\itbox{}}
  }{
    \pmtype(\ensuremath{\itbox{m}},\ensuremath{\itbox{C}},\ensuremath{\itbox{L}}) = \ensuremath{\itbox{}\overline{\itbox{T}}\itbox{\(\rightarrow\)T}_{0}\itbox{}}
  }
  \fbox{$\mtype(\ensuremath{\itbox{m}},\ensuremath{\itbox{C}},\LSet_1,\LSet_2) = \ensuremath{\itbox{}\overline{\itbox{T}}\itbox{\(\rightarrow\)T}_{0}\itbox{}}$}
  \infrule[MT-Class]{
    \ensuremath{\itbox{class C \(\triangleleft\) D {\char'173}... T}_{0}\itbox{ m(}\overline{\itbox{T}}\itbox{ }\overline{\itbox{x}}\itbox{){\char'173} return e; {\char'175} ...{\char'175}}}
  }{
    \mtype(\ensuremath{\itbox{m}},\ensuremath{\itbox{C}},\LSet_1,\LSet_2) = \ensuremath{\itbox{}\overline{\itbox{T}}\itbox{\(\rightarrow\)T}_{0}\itbox{}}
  }
  \infrule[MT-PMethod]{
    \exists \ensuremath{\itbox{L}} \in \LSet_1.
    \pmtype(\ensuremath{\itbox{m}}, \ensuremath{\itbox{C}}, \ensuremath{\itbox{L}})= \ensuremath{\itbox{}\overline{\itbox{T}}\itbox{\(\rightarrow\)T}_{0}\itbox{}}
  }{
    \mtype(\ensuremath{\itbox{m}},\ensuremath{\itbox{C}},\LSet_1, \LSet_2) = \ensuremath{\itbox{}\overline{\itbox{T}}\itbox{\(\rightarrow\)T}_{0}\itbox{}}
  }
  \infrule[MT-Super]{
    \ensuremath{\itbox{class C \(\triangleleft\) D {\char'173}... }\overline{\itbox{M}}\itbox{ {\char'175}}} \andalso
    \ensuremath{\itbox{m}} \not \in \ensuremath{\itbox{}\overline{\itbox{M}}\itbox{}} \\
    \forall \ensuremath{\itbox{L}} \in \LSet_1.
    \pmtype (\ensuremath{\itbox{m}}, \ensuremath{\itbox{C}}, \ensuremath{\itbox{L}}) \undf \andalso
    \mtype(\ensuremath{\itbox{m}},\ensuremath{\itbox{D}},\LSet_2,\LSet_2) = \ensuremath{\itbox{}\overline{\itbox{T}}\itbox{\(\rightarrow\)T}_{0}\itbox{}}
  }{
    \mtype(\ensuremath{\itbox{m}},\ensuremath{\itbox{C}},\LSet_1,\LSet_2) = \ensuremath{\itbox{}\overline{\itbox{T}}\itbox{\(\rightarrow\)T}_{0}\itbox{}}
  }
  \caption{\fsname: Method Type Lookup functions.}
  \label{fig:CFJ:lookup_method_type}
\end{figure}

These rules by themselves do not define \(\mtype\) as a function,
because different layers may contain partial methods of the same name
with different signatures.  So, precisely speaking, it should rather
be understood as a relation; in a well-typed program, it will behave
as a function, though.

\paragraph{Expression Typing}

As mentioned already, the type judgment for expressions is of the form
$\LLGp \ensuremath{\itbox{e}} : \ensuremath{\itbox{T}}$, read ``\ensuremath{\itbox{e}} is given type \ensuremath{\itbox{T}} under context
$\Gamma$, location $\Loc$ and layer set $\Lambda$''.  In addition to
\ensuremath{\itbox{C.m}} and \ensuremath{\itbox{L.C.m}}, $\Loc$ can be $\bullet$, which means the top-level
(i.e., under execution).  Typing rules are given in
\figref{CFJ:exp_typing}.  We defer typing rules for run-time
expressions $\ensuremath{\itbox{new C(}\overline{\itbox{v}}\itbox{)<D,}\overline{\itbox{L}}\itbox{{\ensuremath{'}},}\overline{\itbox{L}}\itbox{>.m(}\overline{\itbox{e}}\itbox{)}}$ and
$\ensuremath{\itbox{new C(}\overline{\itbox{v}}\itbox{)<D,L,}\overline{\itbox{L}}\itbox{{\ensuremath{'}},}\overline{\itbox{L}}\itbox{>.m(}\overline{\itbox{e}}\itbox{)}}$ to the next section and focus on
expressions that appear class and layer definitions.

\begin{figure}
  \leavevmode
  \fbox{$\LLGp \ensuremath{\itbox{e}} : \ensuremath{\itbox{T}}$}
  \typicallabel{T-Invk} 
  \infrule[T-Var]{
    (\Gamma = \ensuremath{\itbox{}\overline{\itbox{x}}\itbox{:}\overline{\itbox{T}}\itbox{}})
  }{
    \LLGp \ensuremath{\itbox{x}_{i}\itbox{}} : \ensuremath{\itbox{T}_{i}\itbox{}}
  }
  \iffull\else\hfil\fi
  \infrule[T-Field]{
    \LLGp \ensuremath{\itbox{e}_{0}\itbox{}} : \ensuremath{\itbox{C}_{0}\itbox{}} \andalso   
    \fields(\ensuremath{\itbox{C}_{0}\itbox{}}) = \ensuremath{\itbox{}\overline{\itbox{T}}\itbox{ }\overline{\itbox{f}}\itbox{}}
  }{
    \LLGp \ensuremath{\itbox{e}_{0}\itbox{.f}_{i}\itbox{}} : \ensuremath{\itbox{T}_{i}\itbox{}}
  }
  \infrule[T-Invk]{
    \LLGp \ensuremath{\itbox{e}_{0}\itbox{}} : \ensuremath{\itbox{C}_{0}\itbox{}}   \andalso
    \mtype(\ensuremath{\itbox{m}},\ensuremath{\itbox{C}_{0}\itbox{}}, \LSet) = \ensuremath{\itbox{}\overline{\itbox{T}}\itbox{ \(\rightarrow\) T}_{0}\itbox{}} \andalso
    \LLGp \ensuremath{\itbox{}\overline{\itbox{e}}\itbox{}} : \ensuremath{\itbox{}\overline{\itbox{S}}\itbox{}}     \andalso
    \ensuremath{\itbox{}\overline{\itbox{S}}\itbox{}} \LEQ \ensuremath{\itbox{}\overline{\itbox{T}}\itbox{}}
  }{
    \LLGp \ensuremath{\itbox{e}_{0}\itbox{.m(}\overline{\itbox{e}}\itbox{)}} : \ensuremath{\itbox{T}_{0}\itbox{}}
  }
  \infrule[T-New]{
    \fields(\ensuremath{\itbox{C}_{0}\itbox{}}) = \ensuremath{\itbox{}\overline{\itbox{T}}\itbox{ }\overline{\itbox{f}}\itbox{}} \andalso
    \LLGp \ensuremath{\itbox{}\overline{\itbox{e}}\itbox{}} : \ensuremath{\itbox{}\overline{\itbox{S}}\itbox{}} \andalso
    \ensuremath{\itbox{}\overline{\itbox{S}}\itbox{}} \LEQ \ensuremath{\itbox{}\overline{\itbox{T}}\itbox{}}
  }{
    \LLGp \ensuremath{\itbox{new C}_{0}\itbox{(}\overline{\itbox{e}}\itbox{)}} : \ensuremath{\itbox{C}_{0}\itbox{}}
  }
  \infrule[T-NewL]{
  }{
    \LLGp \ensuremath{\itbox{new L}_{0}\itbox{()}} : \ensuremath{\itbox{L}_{0}\itbox{}}
  }  
  \infrule[T-With]{
    \LLGp \ensuremath{\itbox{e}_{l}\itbox{}} : \ensuremath{\itbox{L}} \andalso
    \ensuremath{\itbox{L req }} \LSet' \andalso
    \LSet \LEQ_w \LSet' \andalso
    \Loc; \Lambda \cup \set{\ensuremath{\itbox{L}}}; \Gp \ensuremath{\itbox{e}_{0}\itbox{}} : \ensuremath{\itbox{T}_{0}\itbox{}} 
  }{
    \LLGp \ensuremath{\itbox{with e}_{l}\itbox{ e}_{0}\itbox{}} : \ensuremath{\itbox{T}_{0}\itbox{}}
  }
  \infrule[T-Swap]{
    \LLGp \ensuremath{\itbox{e}_{l}\itbox{}} : \ensuremath{\itbox{L}} \andalso
    \ensuremath{\itbox{L}} \LEQ_w \ensuremath{\itbox{L}_{{sw}}\itbox{}} \andalso \ensuremath{\itbox{L}_{{sw}}\itbox{ swappable}} \andalso
    \ensuremath{\itbox{L req }} \LSet' \\
    \LSet_{rm} = \LSet \setminus \set{\ensuremath{\itbox{L{\ensuremath{'}}}} \mid \ensuremath{\itbox{L{\ensuremath{'}}}} \LEQ_w \ensuremath{\itbox{L}_{{sw}}\itbox{}}} \andalso
    \LSet_{rm} \LEQ_w \LSet' \andalso
    \Loc; \LSet_{rm} \cup \set{\ensuremath{\itbox{L}}}; \Gp \ensuremath{\itbox{e}_{0}\itbox{}} : \ensuremath{\itbox{T}_{0}\itbox{}}
  }{
    \LLGp \ensuremath{\itbox{swap (e}_{l}\itbox{,L}_{{sw}}\itbox{) e}_{0}\itbox{}} : \ensuremath{\itbox{T}_{0}\itbox{}}
  }
  \infrule[T-SuperB]{
    \ensuremath{\itbox{class C \(\triangleleft\) E {\char'173}...{\char'175}}} \andalso 
    \mtype(\ensuremath{\itbox{m{\ensuremath{'}}}},\ensuremath{\itbox{E}}, \emptyset) = \ensuremath{\itbox{}\overline{\itbox{T}}\itbox{ \(\rightarrow\) T}_{0}\itbox{}} \andalso
    \ensuremath{\itbox{C.m}}; \LSet; \Gp \ensuremath{\itbox{}\overline{\itbox{e}}\itbox{}} : \ensuremath{\itbox{}\overline{\itbox{S}}\itbox{}} \andalso
    \ensuremath{\itbox{}\overline{\itbox{S}}\itbox{}} \LEQ \ensuremath{\itbox{}\overline{\itbox{T}}\itbox{}}
  }{
    \ensuremath{\itbox{C.m}}; \LSet; \Gp \ensuremath{\itbox{super.m{\ensuremath{'}}(}\overline{\itbox{e}}\itbox{)}} : \ensuremath{\itbox{T}_{0}\itbox{}}
  }
  \infrule[T-SuperP]{
    \ensuremath{\itbox{class C \(\triangleleft\) E {\char'173}...{\char'175}}} \andalso 
    \ensuremath{\itbox{L req }}\LSet' \andalso 
    \mtype(\ensuremath{\itbox{m{\ensuremath{'}}}},\ensuremath{\itbox{E}}, \LSet' \cup \set{\ensuremath{\itbox{L}}}) = \ensuremath{\itbox{}\overline{\itbox{T}}\itbox{ \(\rightarrow\) T}_{0}\itbox{}}  \\
    \ensuremath{\itbox{L.C.m}}; \LSet; \Gp \ensuremath{\itbox{}\overline{\itbox{e}}\itbox{}} : \ensuremath{\itbox{}\overline{\itbox{S}}\itbox{}} \andalso
    \ensuremath{\itbox{}\overline{\itbox{S}}\itbox{}} \LEQ \ensuremath{\itbox{}\overline{\itbox{T}}\itbox{}}
  }{
    \ensuremath{\itbox{L.C.m}}; \LSet; \Gp \ensuremath{\itbox{super.m{\ensuremath{'}}(}\overline{\itbox{e}}\itbox{)}} : \ensuremath{\itbox{T}_{0}\itbox{}}
  }
  \infrule[T-Proceed]{
    \ensuremath{\itbox{L req }}\LSet' \andalso 
    \mtype(\ensuremath{\itbox{m}},\ensuremath{\itbox{C}},\LSet', \LSet' \cup \set{\ensuremath{\itbox{L}}}) = \ensuremath{\itbox{}\overline{\itbox{T}}\itbox{ \(\rightarrow\) T}_{0}\itbox{}} \andalso
    \ensuremath{\itbox{L.C.m}}; \LSet; \Gp \ensuremath{\itbox{}\overline{\itbox{e}}\itbox{}} : \ensuremath{\itbox{}\overline{\itbox{S}}\itbox{}} \andalso
    \ensuremath{\itbox{}\overline{\itbox{S}}\itbox{}} \LEQ \ensuremath{\itbox{}\overline{\itbox{T}}\itbox{}}
  }{
    \ensuremath{\itbox{L.C.m}}; \LSet; \Gp \ensuremath{\itbox{proceed(}\overline{\itbox{e}}\itbox{)}} : \ensuremath{\itbox{T}_{0}\itbox{}}
  }
  \infrule[T-SuperProceed]{
    \ensuremath{\itbox{L \(\triangleleft\) L{\ensuremath{'}}}}\andalso 
    \pmtype(\ensuremath{\itbox{m}},\ensuremath{\itbox{C}},\ensuremath{\itbox{L{\ensuremath{'}}}}) = \ensuremath{\itbox{}\overline{\itbox{T}}\itbox{ \(\rightarrow\) T}_{0}\itbox{}} \andalso
    \ensuremath{\itbox{L.C.m}}; \LSet; \Gp \ensuremath{\itbox{}\overline{\itbox{e}}\itbox{}} : \ensuremath{\itbox{}\overline{\itbox{S}}\itbox{}} \andalso
    \ensuremath{\itbox{}\overline{\itbox{S}}\itbox{}} \LEQ \ensuremath{\itbox{}\overline{\itbox{T}}\itbox{}}
  }{
    \ensuremath{\itbox{L.C.m}}; \LSet; \Gp \ensuremath{\itbox{superproceed(}\overline{\itbox{e}}\itbox{)}} : \ensuremath{\itbox{T}_{0}\itbox{}}
  }
  \caption{\fsname: Expression typing.}
  \label{fig:CFJ:exp_typing}
\end{figure}

Rules \rn{T-Var}, \rn{T-Field} are easy.  \rn{T-New} and \rn{T-NewL}
are for instance of classes and instance of layers, respectively.  The
rule \rn{T-Invk} is straightforward: the method signature \ensuremath{\itbox{}\overline{\itbox{T}}\itbox{\(\rightarrow\)T}_{0}\itbox{}} is
retrieved from the receiver type \ensuremath{\itbox{C}_{0}\itbox{}} and $\LSet$; the types of the
actual arguments must be subtypes of \ensuremath{\itbox{}\overline{\itbox{T}}\itbox{}}; and the whole expression is
given the method return type \ensuremath{\itbox{T}_{0}\itbox{}}.  The rule \rn{T-With} checks, by
\(\LSet \LEQ_w \LSet'\),
that the layers \ensuremath{\itbox{require}}d by \ensuremath{\itbox{L}}---the type of the layer to be
activated---are already activated and that the body \ensuremath{\itbox{e}_{0}\itbox{}} is well
typed under the assumption that \ensuremath{\itbox{L}} is additionally activated.
\rn{T-Swap} is similar; the set \(\LSet_{rm}\)
stands for the set of layers after deactivation and must be a weak
subtype of the required set \(\LSet'\).
The last four rules are for \ensuremath{\itbox{super}}, \ensuremath{\itbox{proceed}}, and \ensuremath{\itbox{superproceed}}
calls and so they are similar to \rn{T-Invk}.  Differences are in how
the method signature is obtained.  In the rules \rn{T-SuperB} and
\rn{T-SuperP} for a \ensuremath{\itbox{super}} call in a method defined in a class and in
a partial method, respectively, the superclass \ensuremath{\itbox{E}} is given to
\(\mtype\).
Layer names are taken from the \ensuremath{\itbox{requires}} clause instead of
\(\Lambda\)---corresponding
to the fact that the method to be invoked is not affected by \ensuremath{\itbox{with}} or
\ensuremath{\itbox{swap}} surrounding \ensuremath{\itbox{super}} (a class cannot require any layer, hence
the empty set).  In the rule \rn{T-Proceed} for a \ensuremath{\itbox{proceed}} call, the
current class name \ensuremath{\itbox{C}} is used.  Similarly to \rn{T-SuperP}, layer
names are taken from the \ensuremath{\itbox{require}} clause.  The last argument to
\(\mtype\)
is \(\LSet \cup \set{\ensuremath{\itbox{L}}}\)
because a \ensuremath{\itbox{proceed}} call can proceed to a partial method \ensuremath{\itbox{D.m}} (where
\ensuremath{\itbox{D}} is a superclass of \ensuremath{\itbox{C}}) defined in the same layer \ensuremath{\itbox{L}}.  In the
rule \rn{T-SuperProceed}, \(\pmtype\) is used instead of \(\mtype\).

\begin{figure}[ht]
  \typicallabel{T-PMethod}
  \fbox{$\ensuremath{\itbox{M ok in C}}$}
  \infrule[T-Method]{
    \ensuremath{\itbox{C.m}}; \emptyset; \ensuremath{\itbox{}\overline{\itbox{x}}\itbox{}}:\ensuremath{\itbox{}\overline{\itbox{T}}\itbox{}}, \ensuremath{\itbox{this}}:\ensuremath{\itbox{C}} \p \ensuremath{\itbox{e}_{0}\itbox{}} : \ensuremath{\itbox{S}_{0}\itbox{}} \andalso
    \ensuremath{\itbox{S}_{0}\itbox{}} \LEQ \ensuremath{\itbox{T}_{0}\itbox{}}
  }{
    \ensuremath{\itbox{T}_{0}\itbox{ m(}\overline{\itbox{T}}\itbox{ }\overline{\itbox{x}}\itbox{) {\char'173} return e}_{0}\itbox{; {\char'175} ok in C}}
  }
  \noindent\fbox{$\ensuremath{\itbox{PM ok in L}}$}
  \infrule[T-PMethod]{
    \ensuremath{\itbox{L req }}\LSet \andalso
    \ensuremath{\itbox{L.C.m}}; \LSet\cup \set{\ensuremath{\itbox{L}}}; \ensuremath{\itbox{}\overline{\itbox{x}}\itbox{}}:\ensuremath{\itbox{}\overline{\itbox{T}}\itbox{}}, \ensuremath{\itbox{this}}:\ensuremath{\itbox{C}} \p \ensuremath{\itbox{e}_{0}\itbox{}} : \ensuremath{\itbox{S}_{0}\itbox{}} \andalso
    \ensuremath{\itbox{S}_{0}\itbox{}} \LEQ \ensuremath{\itbox{T}_{0}\itbox{}}
  }{
    \ensuremath{\itbox{T}_{0}\itbox{ C.m(}\overline{\itbox{T}}\itbox{ }\overline{\itbox{x}}\itbox{) {\char'173} return e}_{0}\itbox{; {\char'175} ok in L}}
  }
  \fbox{$\ensuremath{\itbox{CL ok}}$}
  \infrule[T-Class]{
    \ensuremath{\itbox{K}} = \ensuremath{\itbox{C(}\overline{\itbox{S}}\itbox{ }\overline{\itbox{g}}\itbox{, }\overline{\itbox{T}}\itbox{ }\overline{\itbox{f}}\itbox{){\char'173} super(}\overline{\itbox{g}}\itbox{); this.}\overline{\itbox{f}}\itbox{=}\overline{\itbox{f}}\itbox{; {\char'175}}}  \\
    \fields(\ensuremath{\itbox{D}}) = \ensuremath{\itbox{}\overline{\itbox{S}}\itbox{ }\overline{\itbox{g}}\itbox{}} \andalso
    \ensuremath{\itbox{}\overline{\itbox{M}}\itbox{ ok in C}}
  }{
    \ensuremath{\itbox{class C \(\triangleleft\) D {\char'173} }\overline{\itbox{T}}\itbox{ }\overline{\itbox{f}}\itbox{; K }\overline{\itbox{M}}\itbox{ {\char'175} ok}}
  }
  \fbox{$\ensuremath{\itbox{LA ok}}$}
  \infrule[T-Layer]{
    \ensuremath{\itbox{L}} \text{ is not sublayer of any \ensuremath{\itbox{swappable}} layer} \andalso \\
    \ensuremath{\itbox{L{\ensuremath{'}} req }} \LSet' \andalso
    \set{\ensuremath{\itbox{}\overline{\itbox{L}}\itbox{}}} \LEQ_w \LSet' \andalso
    \ensuremath{\itbox{}\overline{\itbox{PM}}\itbox{ ok in L}}
  }{
    [\ensuremath{\itbox{swappable}}]\ensuremath{\itbox{ layer L req }\overline{\itbox{L}}\itbox{ \(\triangleleft\) L{\ensuremath{'}} {\char'173} }\overline{\itbox{PM}}\itbox{ {\char'175} ok}}
  }
  \infrule[T-LayerSW]{
    \ensuremath{\itbox{L \(\triangleleft\)}^{\itbox{{+}}}\itbox{ L}_{{sw}}\itbox{}} \andalso
    \ensuremath{\itbox{L}_{{sw}}\itbox{ swappable}} \\
    \ensuremath{\itbox{L{\ensuremath{'}} req }} \LSet' \andalso
    \set{\ensuremath{\itbox{}\overline{\itbox{L}}\itbox{}}} = \LSet' \\
    \ensuremath{\itbox{}\overline{\itbox{PM}}\itbox{ ok in L}} \andalso
    \forall \ensuremath{\itbox{C.m}} \in \set{\ensuremath{\itbox{}\overline{\itbox{PM}}\itbox{}}}. \; \pmtype(\ensuremath{\itbox{m}}, \ensuremath{\itbox{C}}, \ensuremath{\itbox{L}_{{sw}}\itbox{}}) \mbox{ defined} \\
    \neg \exists \ensuremath{\itbox{L}_{2}\itbox{}}  \in\dom(\LT). \ensuremath{\itbox{L}_{2}\itbox{}} \mathrel{\ensuremath{\itbox{req}}} \ensuremath{\itbox{L}}
  }{
    \ensuremath{\itbox{layer L req }\overline{\itbox{L}}\itbox{ \(\triangleleft\) L{\ensuremath{'}} {\char'173} }\overline{\itbox{PM}}\itbox{ {\char'175} ok}}
  }
  \caption{\fsname: Method/Class/Layer typing.}
  \label{fig:CFJ:other_typing}
\end{figure}

\begin{figure}[ht]
\typicallabel{T-Table}
  \fbox{Valid overriding $\noconflict(\ensuremath{\itbox{L}_{1}\itbox{}}, \ensuremath{\itbox{L}_{2}\itbox{}})$, $\override^h(\ensuremath{\itbox{L}},\ensuremath{\itbox{C}})$, $\override^v(\ensuremath{\itbox{C}})$}
  \infrule{
    \forall \ensuremath{\itbox{m}}, \ensuremath{\itbox{C}}, \ensuremath{\itbox{}\overline{\itbox{T}}\itbox{}}, \ensuremath{\itbox{T}_{0}\itbox{}}, \ensuremath{\itbox{}\overline{\itbox{S}}\itbox{}}, \ensuremath{\itbox{S}_{0}\itbox{}}. 
    \mbox{ if }
    \begin{array}[t]{l}
      \LT(\ensuremath{\itbox{L}_{1}\itbox{}})(\ensuremath{\itbox{C.m}}) = \ensuremath{\itbox{T}_{0}\itbox{ m(}\overline{\itbox{T}}\itbox{ }\overline{\itbox{x}}\itbox{){\char'173}...{\char'175}}} \\ \mbox{ and } 
      \LT(\ensuremath{\itbox{L}_{2}\itbox{}})(\ensuremath{\itbox{C.m}}) = \ensuremath{\itbox{S}_{0}\itbox{ m(}\overline{\itbox{S}}\itbox{ }\overline{\itbox{y}}\itbox{){\char'173}...{\char'175}}}, \mbox{then } \ensuremath{\itbox{}\overline{\itbox{T}}\itbox{}},\ensuremath{\itbox{T}_{0}\itbox{}} = \ensuremath{\itbox{}\overline{\itbox{S}}\itbox{}}, \ensuremath{\itbox{S}_{0}\itbox{}}
    \end{array}
  }{
    \noconflict(\ensuremath{\itbox{L}_{1}\itbox{}},\ensuremath{\itbox{L}_{2}\itbox{}})
  }
  \infrule{
    \forall \ensuremath{\itbox{m}}, \ensuremath{\itbox{}\overline{\itbox{T}}\itbox{}}, \ensuremath{\itbox{T}_{0}\itbox{}}, \ensuremath{\itbox{}\overline{\itbox{S}}\itbox{}}, \ensuremath{\itbox{S}_{0}\itbox{}}, \ensuremath{\itbox{}\overline{\itbox{x}}\itbox{}}.
    \mbox{ if }
    \begin{array}[t]{l}
      \LT(\ensuremath{\itbox{L}})(\ensuremath{\itbox{C.m}}) = \ensuremath{\itbox{S}_{0}\itbox{ m(}\overline{\itbox{S}}\itbox{ }\overline{\itbox{x}}\itbox{){\char'173}...{\char'175}}} \\ \mbox{ and }
      \mtype(\ensuremath{\itbox{m}}, \ensuremath{\itbox{C}}, \emptyset, \dom(\LT)) = \ensuremath{\itbox{}\overline{\itbox{T}}\itbox{\(\rightarrow\)T}_{0}\itbox{}}, 
      \mbox{then } \ensuremath{\itbox{}\overline{\itbox{T}}\itbox{}}, \ensuremath{\itbox{T}_{0}\itbox{}} = \ensuremath{\itbox{}\overline{\itbox{S}}\itbox{}}, \ensuremath{\itbox{S}_{0}\itbox{}}
    \end{array}
  }{
    \override^h(\ensuremath{\itbox{L}},\ensuremath{\itbox{C}})
  }
  \infrule{
    \forall \ensuremath{\itbox{m}}, \ensuremath{\itbox{D}}, \ensuremath{\itbox{}\overline{\itbox{T}}\itbox{}}, \ensuremath{\itbox{T}_{0}\itbox{}}, \ensuremath{\itbox{}\overline{\itbox{S}}\itbox{}}, \ensuremath{\itbox{S}_{0}\itbox{}}. 
    \begin{array}[t]{l@{}l}
      \mbox{if }&
      \ensuremath{\itbox{class C \(\triangleleft\) D {\char'173}... S}_{0}\itbox{ m(}\overline{\itbox{S}}\itbox{ x){\char'173}...{\char'175}...{\char'175}}} \\ & \mbox{ and }
      \mtype(\ensuremath{\itbox{m}},\ensuremath{\itbox{D}},\dom(\LT), \dom(\LT)) = \ensuremath{\itbox{}\overline{\itbox{T}}\itbox{\(\rightarrow\)T}_{0}\itbox{}},  \\
      \multicolumn{2}{l}{\mbox{then } \ensuremath{\itbox{}\overline{\itbox{T}}\itbox{}} = \ensuremath{\itbox{}\overline{\itbox{S}}\itbox{}} \mbox{ and } \ensuremath{\itbox{S}_{0}\itbox{}} \LEQ \ensuremath{\itbox{T}_{0}\itbox{}}}
    \end{array}
  }{
    \override^v(\ensuremath{\itbox{C}})
  }

      
  \fbox{$\p (\CT, \LT)$\ensuremath{\itbox{ ok}}}  \fbox{$\p (\CT, \LT, \ensuremath{\itbox{e}}) : \ensuremath{\itbox{T}}$}
  \infrule[T-Table]{
    \forall \ensuremath{\itbox{C}} \in \dom(\CT). \CT(\ensuremath{\itbox{C}})\ensuremath{\itbox{ ok}} \andalso
    \forall \ensuremath{\itbox{L}} \in \dom(\LT). \LT(\ensuremath{\itbox{L}})\ensuremath{\itbox{ ok}} \\
    \forall \ensuremath{\itbox{L}_{1}\itbox{}},\ensuremath{\itbox{L}_{2}\itbox{}} \in \dom(\LT). \noconflict(\ensuremath{\itbox{L}_{1}\itbox{}}, \ensuremath{\itbox{L}_{2}\itbox{}}) 
    \\
    \forall \ensuremath{\itbox{C}} \in \dom(\CT). \ensuremath{\itbox{L}} \in \dom(\LT).
    \override^h(\ensuremath{\itbox{L}},\ensuremath{\itbox{C}}) \andalso
    \forall \ensuremath{\itbox{C}} \in \dom(\CT). \override^v(\ensuremath{\itbox{C}})
  }{
    \p (\CT, \LT)\ensuremath{\itbox{ ok}}
  }
  \infrule[T-Prog]{
    \p (\CT, \LT)\ensuremath{\itbox{ ok}} \andalso
    \bullet; \emptyset; \bullet \p \ensuremath{\itbox{e}} : \ensuremath{\itbox{T}}
  }{
    \p (\CT, \LT, \ensuremath{\itbox{e}}) : \ensuremath{\itbox{T}}
  }
  \caption{\fsname: Program typing.}
  \label{fig:CFJ:program_typing}
\end{figure}

In Igarashi et al.~\cite{DynamicLayer2012contextfj}, in which a
type system for ContextFJ is developed, another layer activation construct called
\ensuremath{\itbox{ensure}} is adopted.  The difference from \ensuremath{\itbox{with}} is that, if an
already activated layer is to be activated, \ensuremath{\itbox{ensure}} does not change
the activated layer sequence, whereas \ensuremath{\itbox{with}} will pull that layer to
the head of the sequence so that partial methods in it are invoked
first.  For example, activating layers \ensuremath{\itbox{L1}}, \ensuremath{\itbox{L2}}, \ensuremath{\itbox{L1}} in this order
results in \ensuremath{\itbox{L1;L2}} with \ensuremath{\itbox{ensure}} but in \ensuremath{\itbox{L2;L1}} with the \ensuremath{\itbox{with}}
statement.  Igarashi et al.\ argue that the rearrangement of layers by
\ensuremath{\itbox{with}} destroys the layer ordering in which interlayer dependency is
respected.  For example, if \ensuremath{\itbox{L2}} \ensuremath{\itbox{require}}s \ensuremath{\itbox{L1}}, then \ensuremath{\itbox{L2;L1}}
violates the \ensuremath{\itbox{require}} clause in the sense that the layers that \ensuremath{\itbox{L2}}
requires do not precede \ensuremath{\itbox{L2}} in the sequence.  So, for simplicity,
Igarashi et al.\ considered only \ensuremath{\itbox{ensure}}, which does not have this
problem.

Our discovery is that, in fact, this anomaly caused by \ensuremath{\itbox{with}} is not
really a problem for type soundness and essentially the same typing
rule works---Our typing rule \rn{T-With} for \ensuremath{\itbox{with}} is indeed very
similar to that for \ensuremath{\itbox{ensure}} in ContextFJ; the only difference is the
use of \(\subseteq\) in the place of weak subtyping $\LEQ_w$
(ContextFJ does not have layer subtyping).  The reason why a layer
sequence like \ensuremath{\itbox{L2;L1}} is not problematic can be explained as follows.
Actually, problematic would be a partial method defined in \ensuremath{\itbox{L2}}
calling another (partial) method, say \ensuremath{\itbox{C.m}}, that exists only in
\ensuremath{\itbox{L1}}---that is, one that is undefined in a base class---via
\ensuremath{\itbox{proceed}}.\footnote{%
  Invoking \texttt{m} via \texttt{this} or \texttt{super} will find \texttt{m} in \texttt{L1}.}  Such a
dangling partial method cannot be executed, however: \ensuremath{\itbox{C.m}} in \ensuremath{\itbox{L1}}
cannot contain \ensuremath{\itbox{proceed}}, which leads to execution of the dangling
partial method, because \ensuremath{\itbox{L1}} is activated first, meaning that \ensuremath{\itbox{L1}}
does not \ensuremath{\itbox{require}} any other layer, but it is assumed here that \ensuremath{\itbox{m}} is
not defined in base class \ensuremath{\itbox{C}}.

\paragraph{Typing for Methods, Partial Methods, Classes, Layers, and Programs}

Typing rules for (partial) methods, layers, and classes and are given
in \figref{CFJ:other_typing}.  The rule \rn{T-Method} is standard.
Readers familiar with FJ may notice that a condition for valid
overriding is missing; it is put in elsewhere--see below.  The rule
\rn{T-PMethod} for a partial method means that the method body \ensuremath{\itbox{e}_{0}\itbox{}}
is typed under the layer set \ensuremath{\itbox{require}}d by this layer.  The rule
\rn{T-Layer} is for layers that are not sublayers of any \ensuremath{\itbox{swappable}}
layer and demands that the \ensuremath{\itbox{requires}} clause of the layer be
\emph{covariant} and all partial methods are well formed.  The rule
\rn{T-LayerSW} is for sublayers of \ensuremath{\itbox{swappable}} layers.  It demands, in
addition to the conditions described in \rn{T-Layer}, that the
\ensuremath{\itbox{requires}} clause of this layer be the same as those of its parent
\ensuremath{\itbox{swappable}} layer, that no partial method be newly introduced, and that
this layer be not required by other layers.  The last condition
requires a global program analysis.

It is worth elaborating the rule \rn{T-LayerSW} in more detail.
First, if the condition \(\set{\ensuremath{\itbox{}\overline{\itbox{L}}\itbox{}}} = \LSet'\)
were \(\set{\ensuremath{\itbox{}\overline{\itbox{L}}\itbox{}}} \LEQ_w \LSet'\)
(as in \rn{T-Layer}), the type system would be unsound.  A
counterexample is below:
\begin{lstlisting}
class C {}
swappable layer L0 { int C.m() { return 0; } }
layer L1 extends L0 {}
layer L2 extends L0 requires L { int C.m() { return proceed(); } }
layer L requires L0 { int C.m() { return proceed(); } }
\end{lstlisting}
Layer \ensuremath{\itbox{L2}} additionally requires \ensuremath{\itbox{L}}, which requires \ensuremath{\itbox{L0}}, a swappable
superclass of \ensuremath{\itbox{L2}}.  The condition \(\set{\ensuremath{\itbox{}\overline{\itbox{L}}\itbox{}}} \LEQ_w \LSet'\)
would be trivially satisfied for \ensuremath{\itbox{L2}} because the \ensuremath{\itbox{requires}} clause of
\ensuremath{\itbox{L0}} is empty.  The partial methods in \ensuremath{\itbox{L2}} and \ensuremath{\itbox{L}} are well formed
because \ensuremath{\itbox{L}} and \ensuremath{\itbox{L0}}, respectively, provide definitions to \ensuremath{\itbox{proceed}}.
Under these classes and layers, the following expression
\begin{lstlisting}
with (new L1()) 
  with (new L())          // fulfills "requires L0"
    swap(L0, new L2())    // fulfills "requires L"
      new C().m()
\end{lstlisting}
is well typed, because \ensuremath{\itbox{L1}}, which is a subclass of \ensuremath{\itbox{L0}}, is activated
before activating \ensuremath{\itbox{L}}, and \ensuremath{\itbox{L}} is activated before activating \ensuremath{\itbox{L2}}.
However, the \ensuremath{\itbox{swap}} expression executed under \(\ensuremath{\itbox{L1}};\ensuremath{\itbox{L}}\)
would get stuck as follows:
\[
\begin{array}{l@{\,}l}
  \ensuremath{\itbox{L1}};\ensuremath{\itbox{L}} \p & \ensuremath{\itbox{swap(L0, new L2()) new C().m()}} \\ & \longrightarrow 
             \ensuremath{\itbox{swap(L0, new L2()) new C<C,L,(L;L2)>().m()}} \\ & \longrightarrow
                \ensuremath{\itbox{swap(L0, new L2()) new C<C,}}\bullet\ensuremath{\itbox{,(L;L2)>().m()}} \\
& \;\not \!\!\longrightarrow
\end{array}
\]
The method invocation would take place under \(\ensuremath{\itbox{L}}; \ensuremath{\itbox{L2}}\),
both of which have \ensuremath{\itbox{C.m}} but the second \ensuremath{\itbox{proceed}} call goes nowhere.

Second, if a subclass of a swappable layer were allowed to define a
new method (which is not defined in the swappable), then the type
system would be unsound, too.  Consider the following classes and
layers.
\pagebreak
\begin{lstlisting}
class C {}
class D extends C {}

swappable layer L0 {}
layer L1 extends L0 {}
layer L2 extends L0 {
  int C.m() { return this.m(); }
  int D.m() { return swap(L, new L2()) super.m(); }
}
\end{lstlisting}
Layer \ensuremath{\itbox{L2}} defines new partial methods \ensuremath{\itbox{C.m}} and \ensuremath{\itbox{D.m}}.
They are well formed: in particular, \ensuremath{\itbox{super.m()}} is well typed 
because \ensuremath{\itbox{L2}} itself provides \ensuremath{\itbox{C.m}}.  The following expression
\begin{lstlisting}
with (new L2()) new D().m()
\end{lstlisting}
is well typed, since \ensuremath{\itbox{D.m}} invoked with \ensuremath{\itbox{L2}} activated.
However, reduction of \ensuremath{\itbox{new D().m()}} under \ensuremath{\itbox{L2}} would get stuck:
\[
\begin{array}{l@{\,}l}
  \ensuremath{\itbox{L2}} \p & \ensuremath{\itbox{new D().m()}} \\
& \longrightarrow \ensuremath{\itbox{swap(L, new L1()) new D<C,L2,L2>().m()}} \\
& \longrightarrow \ensuremath{\itbox{swap(L, new L1()) new D().m()}} \\
& \;\not \!\!\longrightarrow
\end{array}
\]
Since \ensuremath{\itbox{super}} calls are not affected by \ensuremath{\itbox{swap}}, \ensuremath{\itbox{super.m()}} in \ensuremath{\itbox{D.m}}
succeeds but, by the time \ensuremath{\itbox{this.m()}} is executed, \ensuremath{\itbox{L2}} will be swapped
out.

\figref{CFJ:program_typing} is for program typing; a program is well
typed if all classes and layers in $\CT$ and $\LT$, respectively, are
well formed and the main expression \ensuremath{\itbox{e}} is typed (at the top-level
$\bullet$).

The most involved is the rule to check valid method overriding used in
\rn{T-Table}.  The predicate $\noconflict$ means that for two partial
methods of the same (qualified) name must have the same signature.
The predicate $\override^h$ means that, for any partial method, the
overridden method (base method in \ensuremath{\itbox{C}} or partial methods for \ensuremath{\itbox{C}}'s
superclass) must have the same signature.  The predicate $\override^v$
means that a base method can override a (partial) method in its
superclass (or layers modifying it) with a covariant return type.
Note that, unlike Java, checking valid method overriding requires a
whole program because a layer may add a new method to a base class,
one of whose subclass may accidentally define a method of the same
name without knowing of that layer.

\section{Type Soundness} \label{sec:type_soundness}

In this section, we prove type soundness of \fsname{} via subject
reduction and progress~\cite{wright1994syntactic}.  Strictly speaking,
we should present typing rules for run-time expressions first before
stating these properties but, for ease of understanding, we will
reverse the order and start with the statements of the properties.

Since we model the execution of a \ensuremath{\itbox{main}} method starting with no layers
activated, we are mainly interested in the case where \(\Loc\) is
\(\bullet\) and the layer sequence is empty.  However, we have to
strengthen the statements of these properties so that the layer
sequence can be nonempty.  We introduce the notion of
\emph{well-formed layer sets} for this purpose.

We define the relation \(\set{\ensuremath{\itbox{}\overline{\itbox{L}}\itbox{}}}\, \WF\), read ``layer set
\(\set{\ensuremath{\itbox{}\overline{\itbox{L}}\itbox{}}}\) is well formed,'' by the rules in
\figref{CFJ:layer-well-formedness}.  Intuitively, a set of layers is
well-formed if one can obtain the layers by activating them one by one
so that \ensuremath{\itbox{requires}} clauses are satisfied.  We ignore the order of
activation because the \ensuremath{\itbox{with}} statement can change the order of
activated layers by activating an already activated layer again.

\begin{figure}
  \typicallabel{T-InvkA}
  \leavevmode
  \infrule[Wf-Empty]{
  }{
    \emptyset \, \WF
  }
  \infrule[Wf-With]{
    \LSet\, \WF \andalso
    \ensuremath{\itbox{L}_{a}\itbox{ req }} \LSet' \andalso
    \LSet \LEQ_w \LSet'
  }{
    \LSet \cup \set{\ensuremath{\itbox{L}_{a}\itbox{}}}\, \WF
  }
  \infrule[Wf-Swap]{
    \LSet\, \WF \andalso
    \ensuremath{\itbox{L}_{{sw}}\itbox{ swappable}} \andalso \ensuremath{\itbox{L}} \LEQ_w \ensuremath{\itbox{L}_{{sw}}\itbox{}} \andalso
    \ensuremath{\itbox{L req }} \LSet' \\
    \LSet_{rm} = \LSet \setminus \set{\ensuremath{\itbox{L{\ensuremath{'}}}} \mid \ensuremath{\itbox{L{\ensuremath{'}}}} \LEQ_w \ensuremath{\itbox{L}_{{sw}}\itbox{}}} \andalso
    \LSet_{rm} \LEQ_w \LSet'
  }{
    \LSet_{rm} \cup \set{\ensuremath{\itbox{L}}}\, \WF
  }
  \caption{\fsname: Layer set well-formedness.}
  \label{fig:CFJ:layer-well-formedness}
\end{figure}

Aside from layer well-formedness, the statements of subject reduction,
progress, and type soundness are standard:

\begin{theorem}[Subject Reduction]\label{thm:def:subject-reduction}
  Suppose $\p (\CT, \LT)\ensuremath{\itbox{ ok}}$.  If $\bullet;
  \set{\ensuremath{\itbox{}\overline{\itbox{L}}\itbox{}}}; \Gp \ensuremath{\itbox{e}} : \ensuremath{\itbox{T}}$ and \set{\ensuremath{\itbox{}\overline{\itbox{L}}\itbox{}}}\, \WF \ and
  $\reduceto{\ensuremath{\itbox{}\overline{\itbox{L}}\itbox{}}}{\ensuremath{\itbox{e}}}{\ensuremath{\itbox{e{\ensuremath{'}}}}}$, then $\bullet; \set{\ensuremath{\itbox{}\overline{\itbox{L}}\itbox{}}}; \Gp \ensuremath{\itbox{e{\ensuremath{'}}}} :
  \ensuremath{\itbox{S}}$ for some \ensuremath{\itbox{S}} such that $\ensuremath{\itbox{S}} \LEQ \ensuremath{\itbox{T}}$.
\end{theorem}

\begin{theorem}[Progress]\label{thm:def:progress}
  Suppose $\p (\CT, \LT)\ensuremath{\itbox{ ok}}$.  If $\bullet;
  \set{\ensuremath{\itbox{}\overline{\itbox{L}}\itbox{}}}; \bullet \p \ensuremath{\itbox{e}} : \ensuremath{\itbox{T}}$ and \set{\ensuremath{\itbox{}\overline{\itbox{L}}\itbox{}}}\, \WF, then \ensuremath{\itbox{e}} is a
  value or $\reduceto{\ensuremath{\itbox{}\overline{\itbox{L}}\itbox{}}}{\ensuremath{\itbox{e}}}{\ensuremath{\itbox{e{\ensuremath{'}}}}}$ for some \ensuremath{\itbox{e{\ensuremath{'}}}}.
\end{theorem}

\begin{theorem}[Type Soundness]\label{thm:def:soundness}
  If \(\p (\CT, \LT, \ensuremath{\itbox{e}}) : \ensuremath{\itbox{T}}\) and \ensuremath{\itbox{e}} reduces to a normal form under
  the empty set of layers, then the normal form is \ensuremath{\itbox{new S(}\overline{\itbox{v}}\itbox{)}}
  for some \ensuremath{\itbox{}\overline{\itbox{v}}\itbox{}} and \(\ensuremath{\itbox{S}}\) such that \(\ensuremath{\itbox{S}} \LEQ \ensuremath{\itbox{T}}\).
\end{theorem}

\subsection{Typing Rules for Run-time Expressions}
To prove the theorems above, we have to give typing rules for run-time
expressions of the forms \ensuremath{\itbox{new C(}\overline{\itbox{v}}\itbox{)<D,}\overline{\itbox{L}}\itbox{{\ensuremath{'}},}\overline{\itbox{L}}\itbox{>.m(}\overline{\itbox{e}}\itbox{)}} and %
\ensuremath{\itbox{new C(}\overline{\itbox{v}}\itbox{)<D,L,}\overline{\itbox{L}}\itbox{{\ensuremath{'}},}\overline{\itbox{L}}\itbox{>.m(}\overline{\itbox{e}}\itbox{)}}, which are not supposed to appear in a
class/layer table.  The typing rules with the rules for a few
auxiliary judgments are given in \figref{CFJ:runtime_rules}:
\begin{figure}
  \typicallabel{T-InvkA}
  \leavevmode
  \infrule[T-InvkA]{
    \bullet; \LGp \ensuremath{\itbox{new C}_{0}\itbox{(}\overline{\itbox{v}}\itbox{)}} : \ensuremath{\itbox{C}_{0}\itbox{}} \andalso
    \ensuremath{\itbox{C}_{0}\itbox{.m}} \p \ensuremath{\itbox{<D}_{0}\itbox{,}\overline{\itbox{L}}\itbox{{\ensuremath{'}},}\overline{\itbox{L}}\itbox{> ok}}  \andalso \LSet \LEQ_{sw} \set{\ensuremath{\itbox{}\overline{\itbox{L}}\itbox{}}} \\
    \mtype(\ensuremath{\itbox{m}},\ensuremath{\itbox{D}_{0}\itbox{}},\set{\ensuremath{\itbox{}\overline{\itbox{L}}\itbox{{\ensuremath{'}}}}}, \set{\ensuremath{\itbox{}\overline{\itbox{L}}\itbox{}}}) = \ensuremath{\itbox{}\overline{\itbox{T}}\itbox{{\ensuremath{'}}\(\rightarrow\)T}_{0}\itbox{}}  \andalso
    \bullet; \LGp \ensuremath{\itbox{}\overline{\itbox{e}}\itbox{}} : \ensuremath{\itbox{}\overline{\itbox{S}}\itbox{}}  \andalso \ensuremath{\itbox{}\overline{\itbox{S}}\itbox{}} \LEQ \ensuremath{\itbox{}\overline{\itbox{T}}\itbox{{\ensuremath{'}}}}
  }{
    \bullet; \LGp \ensuremath{\itbox{new C}_{0}\itbox{(}\overline{\itbox{v}}\itbox{)<D}_{0}\itbox{,}\overline{\itbox{L}}\itbox{{\ensuremath{'}},}\overline{\itbox{L}}\itbox{>.m(}\overline{\itbox{e}}\itbox{)}} : \ensuremath{\itbox{T}_{0}\itbox{}}
  }
  \infrule[T-InvkAL]{
    \bullet; \LGp \ensuremath{\itbox{new C}_{0}\itbox{(}\overline{\itbox{v}}\itbox{)}} : \ensuremath{\itbox{C}_{0}\itbox{}} \andalso
    \ensuremath{\itbox{C}_{0}\itbox{.m}} \p \ensuremath{\itbox{<D}_{0}\itbox{,(}\overline{\itbox{L}}\itbox{{\ensuremath{'}}{\ensuremath{'}};L}_{0}\itbox{),}\overline{\itbox{L}}\itbox{> ok}}  \andalso
    \LSet \LEQ_{sw} \set{\ensuremath{\itbox{}\overline{\itbox{L}}\itbox{}}} \\
    \ensuremath{\itbox{L}_{0}\itbox{}} \LEQ_w \ensuremath{\itbox{L}_{1}\itbox{}} \andalso
    \pmtype(\ensuremath{\itbox{m}},\ensuremath{\itbox{D}_{0}\itbox{}},\ensuremath{\itbox{L}_{1}\itbox{}}) = \ensuremath{\itbox{}\overline{\itbox{T}}\itbox{{\ensuremath{'}}\(\rightarrow\)T}_{0}\itbox{}} \andalso
    \bullet; \LGp \ensuremath{\itbox{}\overline{\itbox{e}}\itbox{}} : \ensuremath{\itbox{}\overline{\itbox{S}}\itbox{}}  \andalso \ensuremath{\itbox{}\overline{\itbox{S}}\itbox{}} \LEQ \ensuremath{\itbox{}\overline{\itbox{T}}\itbox{{\ensuremath{'}}}}
  }{
    \bullet; \LGp \ensuremath{\itbox{new C}_{0}\itbox{(}\overline{\itbox{v}}\itbox{)<D}_{0}\itbox{,L}_{1}\itbox{,(}\overline{\itbox{L}}\itbox{{\ensuremath{'}}{\ensuremath{'}};L}_{0}\itbox{),}\overline{\itbox{L}}\itbox{>.m(}\overline{\itbox{e}}\itbox{)}} : \ensuremath{\itbox{T}_{0}\itbox{}}
  }
  \infrule[LSSW-Intro]{
    \forall \ensuremath{\itbox{L}_{0}\itbox{}} \in \LSet_0. \exists \ensuremath{\itbox{L}_{1}\itbox{}} \in \LSet_1. (
    \begin{array}[t]{l}
      \ensuremath{\itbox{L}_{1}\itbox{}} \LEQ_w \ensuremath{\itbox{L}_{0}\itbox{}} \text{ or} \\ \exists \ensuremath{\itbox{L}_{2}\itbox{}} \in
      \dom(\LT). \ensuremath{\itbox{L}_{2}\itbox{ swappable}} \text{ and } \ensuremath{\itbox{L}_{0}\itbox{}} \LEQ_w \ensuremath{\itbox{L}_{2}\itbox{}}
      \text{ and } \ensuremath{\itbox{L}_{1}\itbox{}} \LEQ_w \ensuremath{\itbox{L}_{2}\itbox{}})
    \end{array}
  }{
    \LSet_1 \LEQ_{sw} \LSet_0
  }
  \infrule[Wf-Cursor]{
    \ensuremath{\itbox{C}} \LEQ \ensuremath{\itbox{D}} \andalso 
    \set{\ensuremath{\itbox{}\overline{\itbox{L}}\itbox{}_{2}\itbox{}}}\, \WF \andalso \ndp(\ensuremath{\itbox{m}},\ensuremath{\itbox{D}},\ensuremath{\itbox{}\overline{\itbox{L}}\itbox{}_{1}\itbox{}},\ensuremath{\itbox{}\overline{\itbox{L}}\itbox{}_{2}\itbox{}})
  }{
    \ensuremath{\itbox{C.m}} \p \ensuremath{\itbox{<D,}\overline{\itbox{L}}\itbox{}_{1}\itbox{,}\overline{\itbox{L}}\itbox{}_{2}\itbox{> ok}} 
  }
  \infrule[NDP-Class]{
    \ensuremath{\itbox{class C {\char'173}.. C}_{0}\itbox{ m(..){\char'173}..{\char'175} ..{\char'175}}} 
  }{
    \ndp(\ensuremath{\itbox{m}},\ensuremath{\itbox{C}},\ensuremath{\itbox{}\overline{\itbox{L}}\itbox{}_{1}\itbox{}},(\ensuremath{\itbox{}\overline{\itbox{L}}\itbox{}_{1}\itbox{;}\overline{\itbox{L}}\itbox{}_{2}\itbox{}}))
  }  
  \infrule[NDP-Layer]{
    \exists \ensuremath{\itbox{L}_{0}\itbox{}} \in \ensuremath{\itbox{}\overline{\itbox{L}}\itbox{}_{1}\itbox{}}. \ensuremath{\itbox{proceed}} \not \in \pmbody(\ensuremath{\itbox{m}},\ensuremath{\itbox{C}},\ensuremath{\itbox{L}_{0}\itbox{}})
  }{
    \ndp(\ensuremath{\itbox{m}},\ensuremath{\itbox{C}},\ensuremath{\itbox{}\overline{\itbox{L}}\itbox{}_{1}\itbox{}},(\ensuremath{\itbox{}\overline{\itbox{L}}\itbox{}_{1}\itbox{;}\overline{\itbox{L}}\itbox{}_{2}\itbox{}}))
  }
  \infrule[NDP-Super]{
    \ensuremath{\itbox{C \(\triangleleft\) D}} \andalso \ndp(\ensuremath{\itbox{m}},\ensuremath{\itbox{D}},(\ensuremath{\itbox{}\overline{\itbox{L}}\itbox{}_{1}\itbox{;}\overline{\itbox{L}}\itbox{}_{2}\itbox{}}),(\ensuremath{\itbox{}\overline{\itbox{L}}\itbox{}_{1}\itbox{;}\overline{\itbox{L}}\itbox{}_{2}\itbox{}})) \\
  }{
    \ndp(\ensuremath{\itbox{m}},\ensuremath{\itbox{C}},\ensuremath{\itbox{}\overline{\itbox{L}}\itbox{}_{1}\itbox{}},(\ensuremath{\itbox{}\overline{\itbox{L}}\itbox{}_{1}\itbox{;}\overline{\itbox{L}}\itbox{}_{2}\itbox{}}))
  }
  \caption{\fsname: Runtime expression typing.}
  \label{fig:CFJ:runtime_rules}
\end{figure}

\begin{sloppypar}
In the rule \rn{T-InvkA} for \ensuremath{\itbox{new C}_{0}\itbox{(}\overline{\itbox{v}}\itbox{)<D}_{0}\itbox{,}\overline{\itbox{L}}\itbox{{\ensuremath{'}},}\overline{\itbox{L}}\itbox{>.m(}\overline{\itbox{e}}\itbox{)}}, the
premises except for $\ensuremath{\itbox{C}_{0}\itbox{.m}} \p \ensuremath{\itbox{<D}_{0}\itbox{,}\overline{\itbox{L}}\itbox{{\ensuremath{'}},}\overline{\itbox{L}}\itbox{> ok}}$ and
$\LSet \LEQ_{sw} \set{\ensuremath{\itbox{}\overline{\itbox{L}}\itbox{}}}$---they are explained in detail
below---are similar to \rn{T-Invk}.  The method signature is obtained
by using the current cursor \ensuremath{\itbox{<D}_{0}\itbox{,}\overline{\itbox{L}}\itbox{{\ensuremath{'}},}\overline{\itbox{L}}\itbox{>}}.  The rule \rn{T-InvkAL}
for a method invoked by \ensuremath{\itbox{superproceed}} is similar.  One difference is
that the method signature is obtained by using \(\pmtype\); the
receiver is derived from a \ensuremath{\itbox{superproceed}} call that originated from a
superlayer of \ensuremath{\itbox{L}_{0}\itbox{}}, hence \(\ensuremath{\itbox{L}_{0}\itbox{}} \LEQ_{sw} \ensuremath{\itbox{L}_{1}\itbox{}}\).
\end{sloppypar}

The condition $\LSet \LEQ_{sw} \set{\ensuremath{\itbox{}\overline{\itbox{L}}\itbox{}}}$ relates the layer sequence
\ensuremath{\itbox{}\overline{\itbox{L}}\itbox{}} in the cursor and $\LSet$, which intuitively represents the set
of layers activated at this program point.  In many cases,
$\LSet = \set{\ensuremath{\itbox{}\overline{\itbox{L}}\itbox{}}}$ holds but if \ensuremath{\itbox{super}} and \ensuremath{\itbox{proceed}} calls are
surrounded by \ensuremath{\itbox{with}} or \ensuremath{\itbox{swap}}, they can be different.  The relation
$\LEQ_{sw}$ is similar to $\LEQ_w$ but the additional clauses
$\exists \ensuremath{\itbox{L}_{2}\itbox{}} \in \dom(\LT). \ensuremath{\itbox{L}_{2}\itbox{ swappable}}$ and $\ensuremath{\itbox{L}_{0}\itbox{}}
\LEQ_w \ensuremath{\itbox{L}_{2}\itbox{}} \text{ and } \ensuremath{\itbox{L}_{1}\itbox{}} \LEQ_w \ensuremath{\itbox{L}_{2}\itbox{}}$ take into account the
possibility that a layer in $\LSet$ may be activated by swapping
layers in \set{\ensuremath{\itbox{}\overline{\itbox{L}}\itbox{}}} out.

The judgment $\ensuremath{\itbox{C.m}} \p \ensuremath{\itbox{<D,}\overline{\itbox{L}}\itbox{}_{1}\itbox{,}\overline{\itbox{L}}\itbox{}_{2}\itbox{> ok}}$, which means that the
cursor is well formed with respect to method \ensuremath{\itbox{m}} in class \ensuremath{\itbox{C}}, is
defined by \rn{Wf-Cursor}.  It requires that \ensuremath{\itbox{D}} to be a superclass of
\ensuremath{\itbox{C}} and \ensuremath{\itbox{}\overline{\itbox{L}}\itbox{}_{2}\itbox{}} to be well formed.  The last condition
\(\ndp(\ensuremath{\itbox{m}},\ensuremath{\itbox{D}_{0}\itbox{}},\ensuremath{\itbox{}\overline{\itbox{L}}\itbox{{\ensuremath{'}}}},\ensuremath{\itbox{}\overline{\itbox{L}}\itbox{}})\) (standing for ``non-dangling
\ensuremath{\itbox{proceed}}'') intuitively means ``a chain of \ensuremath{\itbox{proceed}} calls from the
given cursor location \ensuremath{\itbox{<D}_{0}\itbox{,}\overline{\itbox{L}}\itbox{{\ensuremath{'}},}\overline{\itbox{L}}\itbox{>}} eventually reaches a (partial)
method that does \emph{not} call \ensuremath{\itbox{proceed}}'' and is defined by the
rules \rn{NDP-Class}, \rn{NDP-Layer} and \rn{NDP-Super}, which are
straightforward.  (Here, ``$\ensuremath{\itbox{proceed}} \not \in \pmbody(\ensuremath{\itbox{m}},\ensuremath{\itbox{C}}, \ensuremath{\itbox{L}_{0}\itbox{}})$''
means that there is no \ensuremath{\itbox{proceed}} calls in the method body obtained by
$\pmbody(\ensuremath{\itbox{m}},\ensuremath{\itbox{C}}, \ensuremath{\itbox{L}_{0}\itbox{}})$.)  This predicate represents an invariant condition
throughout a chain of \ensuremath{\itbox{proceed}} calls and ensures there will not be a dangling \ensuremath{\itbox{proceed}} call.

\subsection{Subject Reduction}

The proof of subject reduction is done by induction on
$\reduceto{\ensuremath{\itbox{}\overline{\itbox{L}}\itbox{}}}{\ensuremath{\itbox{e}}}{\ensuremath{\itbox{e{\ensuremath{'}}}}}$.  Similarly to FJ, one main lemma is the
Substitution Lemma, which is used in the case where \ensuremath{\itbox{e}} is a method
invocation and states substitution of values of types \ensuremath{\itbox{}\overline{\itbox{T}}\itbox{}} for
variables of types \ensuremath{\itbox{}\overline{\itbox{S}}\itbox{}}, where \ensuremath{\itbox{}\overline{\itbox{S}}\itbox{}} are subtypes of \ensuremath{\itbox{}\overline{\itbox{T}}\itbox{}}, in a well
typed term preserves typing.  Another important lemma here is
\lemref{def:substitution-super}, which states substitution for
\ensuremath{\itbox{proceed}}, \ensuremath{\itbox{super}}, and \ensuremath{\itbox{superproceed}} preserves typing.

We state several main lemmas to prove the theorems above; their proofs
as well as other lemmas and proofs are found in Appendix.  We fix
\(\CT\)
and \(\LT\)
and assume \((\CT, \LT)\ensuremath{\itbox{ ok}}\) in the rest of this section.

As usual, adding an unused variable to the type environment preserves
typing (Weakening).  Narrowing usually refers to the property that
replacing the type of a variable in the type environment with its
subtype preserves typing; here, we need narrowing with respect to
(extended) layer set subtyping $\LEQ_{sw}$.  The next lemma states
that a well typed value remains well typed regardless of its typing
context ($\Loc; \LSet; \Gamma$).

\begin{lemma}[Weakening]\label{lem:def:weakening}
  If $\LLGp \ensuremath{\itbox{e}} : \ensuremath{\itbox{T}}$, then $\Loc; \LSet; \Gamma, \ensuremath{\itbox{x}}\colon
  \ensuremath{\itbox{S}} \p \ensuremath{\itbox{e}} : \ensuremath{\itbox{T}}$.
\end{lemma}

\begin{lemma}[Layer Set Narrowing]\label{lem:def:narrowing}
  If $\LLGp \ensuremath{\itbox{e}} : \ensuremath{\itbox{T}}$ and $\LSet' \LEQ_{sw} \LSet$, then $\Loc; \LSet'; \Gp \ensuremath{\itbox{e}} : \ensuremath{\itbox{T}}$.
\end{lemma}

\begin{lemma}[Strengthening for values]\label{lem:def:value-strengthening}
  If $\LLGp \ensuremath{\itbox{v}} : \ensuremath{\itbox{T}}$ then, $\Loc'; \LSet'; \Gamma' \p \ensuremath{\itbox{v}} : \ensuremath{\itbox{T}}$.
\end{lemma}

The statement of the Substitution Lemma is straightforward.

\begin{lemma}[Substitution]\label{lem:def:substitution}
  If $\Loc; \LSet; \Gamma, \ensuremath{\itbox{}\overline{\itbox{x}}\itbox{}} \colon \ensuremath{\itbox{}\overline{\itbox{T}}\itbox{}} \p \ensuremath{\itbox{e}} : \ensuremath{\itbox{T}}$ and $\LLGp
  \ensuremath{\itbox{}\overline{\itbox{v}}\itbox{}} : \ensuremath{\itbox{}\overline{\itbox{S}}\itbox{}}$ and $\ensuremath{\itbox{}\overline{\itbox{S}}\itbox{}} \LEQ \ensuremath{\itbox{}\overline{\itbox{T}}\itbox{}}$, then $\LLGp [\ensuremath{\itbox{}\overline{\itbox{v}}\itbox{}}/\ensuremath{\itbox{}\overline{\itbox{x}}\itbox{}}]\ensuremath{\itbox{e}} :
  \ensuremath{\itbox{S}}$ and $\ensuremath{\itbox{S}} \LEQ \ensuremath{\itbox{T}}$ for some \ensuremath{\itbox{S}}.
\end{lemma}

The next lemma states that substitution for \ensuremath{\itbox{proceed}}, \ensuremath{\itbox{super}}, and
\ensuremath{\itbox{superproceed}} preserves typing.  The first item is for an invocation
of a partial method, which may contain \ensuremath{\itbox{proceed}} and \ensuremath{\itbox{superproceed}}
calls as well as \ensuremath{\itbox{super}} calls; the second is for a base method, which
may contain only \ensuremath{\itbox{super}} calls.  The conditions, which look rather
complicated, correspond to the premises of \rn{T-InvkA} and
\rn{T-InvkAL}.  

\begin{lemma}[Substitution for \ensuremath{\itbox{super}}, \ensuremath{\itbox{proceed}} and \ensuremath{\itbox{superproceed}}]\label{lem:def:substitution-super}
  \
  \begin{enumerate}
  \item 
    \begin{sloppypar}
      If \(\bullet; \LSet; \Gp \ensuremath{\itbox{new C}_{0}\itbox{(}\overline{\itbox{v}}\itbox{)}} : \ensuremath{\itbox{C}_{0}\itbox{}}\) and 
      $\ensuremath{\itbox{L.C.m}}; \LSet; \Gp \ensuremath{\itbox{e}} : \ensuremath{\itbox{T}}$ and
      $\ensuremath{\itbox{C}_{0}\itbox{.m}} \p \ensuremath{\itbox{<C,(}\overline{\itbox{L}}\itbox{{\ensuremath{'}};L{\ensuremath{'}}{\ensuremath{'}}),}\overline{\itbox{L}}\itbox{> ok}}$ and
      \ensuremath{\itbox{C \(\triangleleft\) D}} and \ensuremath{\itbox{L{\ensuremath{'}}{\ensuremath{'}} \(\Leq\)}_{w}\itbox{ L \(\triangleleft\) L{\ensuremath{'}}}} and
      $\LSet \LEQ_{sw} \set{\ensuremath{\itbox{}\overline{\itbox{L}}\itbox{}}}$ and
      $\ensuremath{\itbox{proceed}} \in \ensuremath{\itbox{e}} \implies \ndp(\ensuremath{\itbox{m}},\ensuremath{\itbox{C}},\ensuremath{\itbox{}\overline{\itbox{L}}\itbox{{\ensuremath{'}}}},\ensuremath{\itbox{}\overline{\itbox{L}}\itbox{}})$,
      then $\bullet; \LSet; \Gp S\ensuremath{\itbox{e}} : \ensuremath{\itbox{T}}$
      where
      \[S = 
      \left[\begin{array}{l@{/}l}
          \ensuremath{\itbox{new C}_{0}\itbox{(}\overline{\itbox{v}}\itbox{)<C,}\overline{\itbox{L}}\itbox{{\ensuremath{'}},}\overline{\itbox{L}}\itbox{>.m}} & \ensuremath{\itbox{proceed}}, \\
          \ensuremath{\itbox{new C}_{0}\itbox{(}\overline{\itbox{v}}\itbox{)<D,}\overline{\itbox{L}}\itbox{,}\overline{\itbox{L}}\itbox{>}} &  \ensuremath{\itbox{super}}, \\
          \ensuremath{\itbox{new C}_{0}\itbox{(}\overline{\itbox{v}}\itbox{)<C,L{\ensuremath{'}},(}\overline{\itbox{L}}\itbox{{\ensuremath{'}};L{\ensuremath{'}}{\ensuremath{'}}),}\overline{\itbox{L}}\itbox{>.m}} &  \ensuremath{\itbox{superproceed}}
        \end{array}\right].\]
    \end{sloppypar}

  \item If \(\bullet; \LSet; \Gp \ensuremath{\itbox{new C}_{0}\itbox{(}\overline{\itbox{v}}\itbox{)}} : \ensuremath{\itbox{C}_{0}\itbox{}}\) and 
    $\ensuremath{\itbox{C.m}}; \LSet; \Gp \ensuremath{\itbox{e}} : \ensuremath{\itbox{T}}$ and 
    $\ensuremath{\itbox{C}_{0}\itbox{.m}} \p \ensuremath{\itbox{<C,}\overline{\itbox{L}}\itbox{{\ensuremath{'}},}\overline{\itbox{L}}\itbox{> ok}}$ and
    \ensuremath{\itbox{C \(\triangleleft\) D}} and $\LSet \LEQ_{sw} \set{\ensuremath{\itbox{}\overline{\itbox{L}}\itbox{}}}$, then
    \(\bullet; \LSet; \Gp [\ensuremath{\itbox{new C}_{0}\itbox{(}\overline{\itbox{v}}\itbox{)<D,}\overline{\itbox{L}}\itbox{,}\overline{\itbox{L}}\itbox{>}} / \ensuremath{\itbox{super}}] \ensuremath{\itbox{e}} : \ensuremath{\itbox{T}}\).
  \end{enumerate} 
\end{lemma}

The next two lemmas state method bodies obtained by $\pmbody$ and
$\mbody$ are well typed according to the type information obtained by
$\pmtype$ and $\mtype$, respectively.

\begin{lemma}[Inversion for partial method body]\label{lem:def:pmbody}
  If $\pmbody(\ensuremath{\itbox{m}},\ensuremath{\itbox{C}},\ensuremath{\itbox{L}}) = \ensuremath{\itbox{}\overline{\itbox{x}}\itbox{.e}_{0}\itbox{ in L{\ensuremath{'}}}}$ and $\ensuremath{\itbox{L req }} \LSet$ and
  $\pmtype(\ensuremath{\itbox{m}},\ensuremath{\itbox{C}},\ensuremath{\itbox{L}}) = \ensuremath{\itbox{}\overline{\itbox{T}}\itbox{ \(\rightarrow\) T}_{0}\itbox{}}$, then $\ensuremath{\itbox{L.C.m}}; \LSet \cup
  \set{\ensuremath{\itbox{L}}}; \ensuremath{\itbox{}\overline{\itbox{x}}\itbox{}}:\ensuremath{\itbox{}\overline{\itbox{T}}\itbox{}}, \ensuremath{\itbox{this}}:\ensuremath{\itbox{C}} \p \ensuremath{\itbox{e}_{0}\itbox{}} : \ensuremath{\itbox{S}_{0}\itbox{}}$ for some $\ensuremath{\itbox{S}_{0}\itbox{}}
  \LEQ_w \ensuremath{\itbox{T}_{0}\itbox{}}$.
\end{lemma}

\begin{lemma}[Inversion for method body]\label{lem:def:mbody-mtype}
  \begin{sloppypar}
    Suppose 
    \(\set{\ensuremath{\itbox{}\overline{\itbox{L}}\itbox{}}}\, \WF\) and 
    $\mbody(\ensuremath{\itbox{m}}, \ensuremath{\itbox{C}}, \ensuremath{\itbox{}\overline{\itbox{L}}\itbox{{\ensuremath{'}}}}, \ensuremath{\itbox{}\overline{\itbox{L}}\itbox{}}) = \ensuremath{\itbox{}\overline{\itbox{x}}\itbox{.e}_{0}\itbox{}} \IN \ensuremath{\itbox{C{\ensuremath{'}}}},\ensuremath{\itbox{}\overline{\itbox{L}}\itbox{{\ensuremath{'}}{\ensuremath{'}}}}$ and
    $\mtype(\ensuremath{\itbox{m}}, \ensuremath{\itbox{C}}, \set{\ensuremath{\itbox{}\overline{\itbox{L}}\itbox{{\ensuremath{'}}}}}, \set{\ensuremath{\itbox{}\overline{\itbox{L}}\itbox{}}}) = \ensuremath{\itbox{}\overline{\itbox{T}}\itbox{\(\rightarrow\)T}_{0}\itbox{}}$ and
    $\ndp(\ensuremath{\itbox{m}},\ensuremath{\itbox{C}},\ensuremath{\itbox{}\overline{\itbox{L}}\itbox{{\ensuremath{'}}}},\ensuremath{\itbox{}\overline{\itbox{L}}\itbox{}})$.
    \begin{enumerate}
    \item If $\ensuremath{\itbox{}\overline{\itbox{L}}\itbox{{\ensuremath{'}}{\ensuremath{'}}}} = \ensuremath{\itbox{}\overline{\itbox{L}}\itbox{{\ensuremath{'}}{\ensuremath{'}}{\ensuremath{'}};L}_{0}\itbox{}}$, then 
	  $\ensuremath{\itbox{L}_{0}\itbox{ req }}\LSet$ and 
	  $\ensuremath{\itbox{L}_{0}\itbox{.C{\ensuremath{'}}.m}}; \LSet \cup \set{\ensuremath{\itbox{L}_{0}\itbox{}}}; \ensuremath{\itbox{}\overline{\itbox{x}}\itbox{}}: \ensuremath{\itbox{}\overline{\itbox{T}}\itbox{}}, \ensuremath{\itbox{this}}:\ensuremath{\itbox{C{\ensuremath{'}}}} \p \ensuremath{\itbox{e}_{0}\itbox{}} : \ensuremath{\itbox{U}_{0}\itbox{}}$ and
	  $\ensuremath{\itbox{C}} \LEQ \ensuremath{\itbox{C{\ensuremath{'}}}}$ and 
	  $\ensuremath{\itbox{U}_{0}\itbox{}} \LEQ \ensuremath{\itbox{T}_{0}\itbox{}}$ and
      $\ndp(\ensuremath{\itbox{m}},\ensuremath{\itbox{C{\ensuremath{'}}}},\ensuremath{\itbox{}\overline{\itbox{L}}\itbox{{\ensuremath{'}}{\ensuremath{'}}}},\ensuremath{\itbox{}\overline{\itbox{L}}\itbox{}})$ for some $\LSet$ and \ensuremath{\itbox{U}_{0}\itbox{}}. 

    \item If $\ensuremath{\itbox{}\overline{\itbox{L}}\itbox{{\ensuremath{'}}{\ensuremath{'}}}} = \bullet$, then $\ensuremath{\itbox{C{\ensuremath{'}}.m}}; \emptyset; \ensuremath{\itbox{}\overline{\itbox{x}}\itbox{}}:\ensuremath{\itbox{}\overline{\itbox{T}}\itbox{}},
	  \ensuremath{\itbox{this}}:\ensuremath{\itbox{C{\ensuremath{'}}}} \p \ensuremath{\itbox{e}_{0}\itbox{}} : \ensuremath{\itbox{U}_{0}\itbox{}}$ and 
	  $\ensuremath{\itbox{C}} \LEQ \ensuremath{\itbox{C{\ensuremath{'}}}}$ and 
	  $\ensuremath{\itbox{U}_{0}\itbox{}} \LEQ \ensuremath{\itbox{T}_{0}\itbox{}}$ and
      $\ndp(\ensuremath{\itbox{m}},\ensuremath{\itbox{C{\ensuremath{'}}}},\bullet,\ensuremath{\itbox{}\overline{\itbox{L}}\itbox{}})$ for some \ensuremath{\itbox{U}_{0}\itbox{}}.
    \end{enumerate}
    
  \end{sloppypar}
\end{lemma}

We also need additional lemmas derived from runtime conditions.
Layer-set wellformedness $\LSet\, \WF$ provides two important
properties.  The first states that a well formed layer set is closed
under the \ensuremath{\itbox{requires}} clause and the second that, if method \ensuremath{\itbox{m}} is
found in \ensuremath{\itbox{C}} (under the assumption that $\LSet$ activated) but not in
its direct superclass \ensuremath{\itbox{D}}, then at least one of those methods does not
call \ensuremath{\itbox{proceed}}.  This lemma is used to
prove the next lemma (Lemma \ref{lem:def:wp}), which
derives $\ndp$ for an initial cursor of the form \ensuremath{\itbox{<C,}\overline{\itbox{L}}\itbox{,}\overline{\itbox{L}}\itbox{>}}.  

\begin{lemma}\label{lem:def:layer-set-wf1}]
  If $\LSet \, \WF$, then $\forall \ensuremath{\itbox{L}} \in \LSet, \forall \ensuremath{\itbox{L{\ensuremath{'}}}} \text{
    s.t. } \ensuremath{\itbox{L req L{\ensuremath{'}}}}, \exists \ensuremath{\itbox{L{\ensuremath{'}}{\ensuremath{'}}}} \in \LSet. \ensuremath{\itbox{L{\ensuremath{'}}{\ensuremath{'}}}} \LEQ_w \ensuremath{\itbox{L{\ensuremath{'}}}}$.
\end{lemma}

\begin{lemma}\label{lem:def:layer-set-wf2}
  If $\LSet \, \WF$ and $\mtype(\ensuremath{\itbox{m}},\ensuremath{\itbox{C}},\LSet) \defined$ and
  $\mtype(\ensuremath{\itbox{m}},\ensuremath{\itbox{D}},\LSet) \undf$ and \ensuremath{\itbox{C \(\triangleleft\) D}}, then 
  $(\exists \ensuremath{\itbox{L{\ensuremath{'}}}} \in \LSet.  \ensuremath{\itbox{proceed}}
  \not \in \pmbody(\ensuremath{\itbox{m}},\ensuremath{\itbox{C}},\ensuremath{\itbox{L{\ensuremath{'}}}})) \text{
    or } \mtype(\ensuremath{\itbox{m}},\ensuremath{\itbox{C}},\emptyset, \LSet) \defined$.
\end{lemma}

\begin{lemma}\label{lem:def:wp}
  If $\set{\ensuremath{\itbox{}\overline{\itbox{L}}\itbox{}}}\, \WF$ and
  $\mtype(\ensuremath{\itbox{m}},\ensuremath{\itbox{C}},\set{\ensuremath{\itbox{}\overline{\itbox{L}}\itbox{}}},\set{\ensuremath{\itbox{}\overline{\itbox{L}}\itbox{}}}) = \ensuremath{\itbox{}\overline{\itbox{T}}\itbox{\(\rightarrow\)T}_{0}\itbox{}}$, then
  $\ndp(\ensuremath{\itbox{m}},\ensuremath{\itbox{C}},\ensuremath{\itbox{}\overline{\itbox{L}}\itbox{}},\ensuremath{\itbox{}\overline{\itbox{L}}\itbox{}})$.
\end{lemma}

As stated below, the predicate \(\ndp\) ensures the existence of a method:

\begin{lemma}\label{lem:def:wp_mtype}
  If $\ndp(\ensuremath{\itbox{m}},\ensuremath{\itbox{C}},\ensuremath{\itbox{}\overline{\itbox{L}}\itbox{{\ensuremath{'}}}},\ensuremath{\itbox{}\overline{\itbox{L}}\itbox{}})$ holds for some \ensuremath{\itbox{m}}, \ensuremath{\itbox{C}}, \ensuremath{\itbox{}\overline{\itbox{L}}\itbox{{\ensuremath{'}}}} and
  \ensuremath{\itbox{}\overline{\itbox{L}}\itbox{}}, then $\mtype(\ensuremath{\itbox{m}},\ensuremath{\itbox{C}},\set{\ensuremath{\itbox{}\overline{\itbox{L}}\itbox{{\ensuremath{'}}}}}, \set{\ensuremath{\itbox{}\overline{\itbox{L}}\itbox{}}}) = \ensuremath{\itbox{}\overline{\itbox{T}}\itbox{\(\rightarrow\)T}_{0}\itbox{}}$ for
  some \ensuremath{\itbox{}\overline{\itbox{T}}\itbox{}} and \ensuremath{\itbox{T}_{0}\itbox{}}.
\end{lemma}

\subsection{Progress}

To prove the Progress Theorem, we need the following two lemmas, which
show the existence of a method body from well definedness of $\pmtype$
and $\mtype$.

\begin{lemma}\label{lem:def:pmtype}
  If $\pmtype(\ensuremath{\itbox{m}}, \ensuremath{\itbox{C}}, \ensuremath{\itbox{L}}) = \ensuremath{\itbox{}\overline{\itbox{T}}\itbox{\(\rightarrow\)T}_{0}\itbox{}}$, then there exist \ensuremath{\itbox{}\overline{\itbox{x}}\itbox{}} and
  \ensuremath{\itbox{e}_{0}\itbox{}} and \ensuremath{\itbox{L{\ensuremath{'}}}} ($\neq \ensuremath{\itbox{Base}}$) such that $\pmbody(\ensuremath{\itbox{m}}, \ensuremath{\itbox{C}}, \ensuremath{\itbox{L}}) =
  \ensuremath{\itbox{}\overline{\itbox{x}}\itbox{.e}_{0}\itbox{}} \IN \ensuremath{\itbox{L{\ensuremath{'}}}}$ and the lengths of \ensuremath{\itbox{}\overline{\itbox{x}}\itbox{}} and \ensuremath{\itbox{}\overline{\itbox{T}}\itbox{}} are equal and
  $\ensuremath{\itbox{L}} \LEQ_w \ensuremath{\itbox{L{\ensuremath{'}}}}$.
\end{lemma}

\begin{lemma}\label{lem:def:mtype}
  If $\mtype(\ensuremath{\itbox{m}}, \ensuremath{\itbox{C}}, \set{\ensuremath{\itbox{}\overline{\itbox{L}}\itbox{{\ensuremath{'}}}}}, \set{\ensuremath{\itbox{}\overline{\itbox{L}}\itbox{}}}) = \ensuremath{\itbox{}\overline{\itbox{T}}\itbox{\(\rightarrow\)T}_{0}\itbox{}}$ and \ensuremath{\itbox{}\overline{\itbox{L}}\itbox{{\ensuremath{'}}}}
  is a prefix of \ensuremath{\itbox{}\overline{\itbox{L}}\itbox{}} and \(\set{\ensuremath{\itbox{}\overline{\itbox{L}}\itbox{}}}\, \WF\), then there exist \ensuremath{\itbox{}\overline{\itbox{x}}\itbox{}}
  and \ensuremath{\itbox{e}_{0}\itbox{}} and \ensuremath{\itbox{}\overline{\itbox{L}}\itbox{{\ensuremath{'}}{\ensuremath{'}}}} and \ensuremath{\itbox{C{\ensuremath{'}}}} ($\neq \ensuremath{\itbox{Object}}$) such that
  $\mbody(\ensuremath{\itbox{m}}, \ensuremath{\itbox{C}}, \ensuremath{\itbox{}\overline{\itbox{L}}\itbox{}}, \ensuremath{\itbox{}\overline{\itbox{L}}\itbox{{\ensuremath{'}}}}) = \ensuremath{\itbox{}\overline{\itbox{x}}\itbox{.e}_{0}\itbox{}} \IN \ensuremath{\itbox{C{\ensuremath{'}}}}, \ensuremath{\itbox{}\overline{\itbox{L}}\itbox{{\ensuremath{'}}{\ensuremath{'}}}}$ and the
  lengths of \ensuremath{\itbox{}\overline{\itbox{x}}\itbox{}} and \ensuremath{\itbox{}\overline{\itbox{T}}\itbox{}} are equal and, if \ensuremath{\itbox{}\overline{\itbox{L}}\itbox{{\ensuremath{'}}{\ensuremath{'}}}} is not empty, the last layer name of \ensuremath{\itbox{}\overline{\itbox{L}}\itbox{{\ensuremath{'}}{\ensuremath{'}}}} is not \ensuremath{\itbox{Base}}.
\end{lemma}

\section{Related Work}
\label{sec:relwork}

\paragraph{Foundation of Context-Oriented Programming}
Our work is a direct descendant of Igarashi, Hirschfeld, and
Masuhara~\cite{contextfj2011,DynamicLayer2012contextfj}, where a tiny
COP language ContextFJ is developed and its type system is proved to
be sound.  ContextFJ is not equipped with layer inheritance, layer
subtyping, or first-class layers but allows baseless partial methods to be
declared in the second type system~\cite{DynamicLayer2012contextfj},
in which \ensuremath{\itbox{requires}} declarations are first introduced into COP.

Our swappable layers resemble atomic layers in
ContextL~\cite{costanza2008feature}, in which mutual exclusion between
layers can be specified and activation of an atomic layer
automatically deactivates another layer in conflict.  Our syntax is a
little verbose in that the swappable layer name such as
\texttt{Weather} has to be explicit because a layer may have more than one swappable layer in its superlayers.  It may be a reasonable idea to disallow a sublayer of a swappable layer to be swappable for the sake of syntactic conciseness.

Similarly to our swappable layers for layer deactivation, Kamina et
al.~\cite{KaminaAotaniIgarashi14COP,KaminaAotaniMasuharaIgarashi17} also show
another approach to safe layer deactivation mechanism and formalized
its semantics and type safety with an extension of ContextFJ.  Their
approach is also based on \ensuremath{\itbox{requires}} clauses.  The key idea
is to modify the method lookup so that it searches not only activated
layers but all layers that are required by those activated layers.

Besides block-style layer activation mechanisms as in JCop, there are
other mechanisms such as imperative activation of Subjective-C~\cite{gonzalez2010subjective}, event-based
activation of EventCJ~\cite{kamina2011eventcj}, and implicit activation of Flute~\cite{bainomugisha2012interruptible}.
The original JCop also supports implicit layer
activation~\cite{appeltauer2013declarative}, but currently we omit it
from our formalization.  ServalCJ~\cite{kamina2015generalized}
provides a generalized layer activation mechanism that can treat the 
layer activation mechanisms above uniformly.
Although some of them
\cite{aotani2011featherweight,kamina2016generalized} study formal
semantics, they do not discuss type soundness of languages with
baseless partial methods; e.g., ServalCJ does not support baseless
partial methods.

Clarke and Sergey~\cite{clarke2009semantics} independently formalize a
core language (also called ContextFJ) for context-oriented programming
and develop such a type system. In their type system, each
partial/base method (rather than a layer) is equipped with
dependency information, a set of
the signatures of the methods that it may call.  Dependency information
is very fine-grained but their calculus does
not support class nor layer inheritances.

The JCop compiler transforms a JCop program into a plain Java code
which contains auxiliary classes and methods, constructing a kind of
double dispatch.  Appeltauer et al.~\cite{appeltauer2010layered} discusses
two implementation schemes of JCop's method dispatch mechanism: one
that rely on the translation to plain Java code and the other
that rely on \ensuremath{\itbox{invokedynamic}} of Java 7.

There are several studies to enrich description of relationships
between contexts.  Subjective-C~\cite{gonzalez2010subjective}, an
extension of Objective-C with COP, adopts imperative
context activation with imperative context relationship description,
which supports various kinds of declarations of dependency between
layers, such as implication, requirement, and exclusion.  Context
Petri Nets~\cite{cardozo2013modeling,cardozo2015semantics} is a
context-oriented extension of Petri Nets, and helps formalization of
description of context dependencies in Subjective-C.
$ML_{CoDa}$~\cite{degano2014two,degano2016two} provides two kinds of
components; one for declarative description of context dependencies
and the other for functional computation.  It also provides a type and
effect system and a loading-time verification mechanism that detects
failures in adaptation.

\paragraph{Dynamic Software Product Line}
Software product line (SPL) is a paradigm of industrial software
development that enables to create various variations of software by
mostly reusing common modules.  Programming languages for SPL, such as
Feature-Oriented
Programming~\cite{prehofer1997feature,batory2004scaling} and
Delta-Oriented Programming (DOP)~\cite{schaefer2010delta}, have been
studied.  They provide modules that refine existing classes and
combine them according to a given configuration at compile time or build
time.

Recent studies~\cite{hallsteinsen2008dynamic, rosenmuller2011flexible}
reveal that SPL also needs dynamic reconfiguration of software, and so
dynamic DOP~\cite{damiani2011dynamic,damiani2017core} is proposed.
Dynamic DOP provides mainly three kinds of modules; a delta module for
describing refinement of classes (similar to a layer of COP), a
product-line declaration for describing valid configurations, and a
dynamic reconfiguration graph for replacing heap objects dynamically.
Unlike COP, the composition order of delta modules is determined
uniquely by a given product line declaration; this property is called
\textit{unambiguity}.  A type system of dynamic DOP also ensures that
all valid reconfigurations lead to type-safe products.

\paragraph{Type systems for advanced composition mechanisms of OOP}

There are many type systems proposed for advanced composition
mechanisms such as mixins~\cite{bono1999core,flatt1998classes},
traits~\cite{liquori2008feathertrait,smith2005chai}, open classes (a.k.a.\
inter-type declarations)~\cite{clifton2006multijava}, and
revisers~\cite{chiba2010mostly}.  A common idea is to let programmers
declare dependency between modules as required interfaces; our
\ensuremath{\itbox{requires}} declarations basically follow it.  In most work, however,
composition is done at compile or link time unlike COP languages.  We
think that it is interesting that the same idea works even for dynamic
composition found in COP languages.

Kamina and Tamai~\cite{kamina2004mcjava} propose McJava, in which
mixin-based composition can be deferred to object instantiation.  In
fact, \ensuremath{\itbox{new}} expressions can specify a class and mixins to instantiate
an object.  So, the type of an object also consists of a class name
and a sequence of mixin names.  Whereas composition is per-instance
basis in McJava, it is global in \fsname{}.  However, in McJava,
composition cannot be changed once an object is instantiated.

Drossopoulou et al.~\cite{Drossopoulou2002fickle2} proposed \emph{Fickle}$_{\textrm{I\!I}}$,
a class-based object-oriented language with \emph{dynamic
  reclassification}, which allows an object to change its class at run
time.  Their idea of root classes, which serve as interface, is
similar to our swappable layers; their restriction that state classes
cannot be used as type for fields is similar to ours that a sublayer
of a swappable cannot be \ensuremath{\itbox{require}}d by any other layer.

Bettini et al.~\cite{bettini2013flexible} developed a type system for
\emph{dynamic trait replacement}, which allows methods in an object to
be exchanged at run time.  They introduce the notion of
\emph{replaceable} to describe the signatures of replaceable methods;
a replaceable appears as part of the type of an object and the trait
to replace methods of the object has to provide the methods in that
replaceable.  The roles of replaceables and traits are somewhat similar to
those of swappable layers, which provide interfaces common to swapped
layers, and sublayers of swappable.

Several studies for dynamism of objects in distributed settings exist.
MoMi~\cite{bettini2005momi} presents an approach to having a process
communicate mobile code to other processes in a safe manner; well-typed code
sent from external processes can be merged into local code without
recompilation.  Dynamic class~\cite{johnsen2009dynamic} is a mechanism
to modify classes and a class hierarchy in a type-safe way, where
objects are distributed and long-lived, and so a number of upgrade
operations are performed; a series of upgrade operations is used to
type-check the next upgrade operation.  Unlike these two approaches,
COP realizes its dynamism with its method dispatching mechanism.  It
is quite interesting to consider how COP mechanism works safely in the
above distributed settings.

Although not a type system, Burton and
Sekerinski~\cite{burton2015safety} studies interference problem of
dynamic mixin composition, in which some order of mixin composition
breaks required specification of class methods.  They develop a
refinement calculus in order to formalize dynamic mixin composition.

\section{Concluding Remarks}
\label{sec:discussion}

We have developed a formal type system for a small COP language with
layer inheritance, layer subtyping, layer swapping, and first-class
layers, and shown that the type system is sound with respect to
the operational semantics.  As in previous work, \ensuremath{\itbox{requires}} declarations
are important to guarantee safety in the presence of baseless partial methods.
Subtyping for first-class layers is subtle because there are two kinds
of substitutability.  We have introduced weak subtyping for checking
whether a \ensuremath{\itbox{requires}} clause is satisfied and normal subtyping for
usual substitutability.

In JCop, a layer definition can contain field and (ordinary) method
declarations so that a layer instance can act just like an ordinary
object.  Typechecking accesses to these members of layer instances
is the same as ordinary objects.  If we model fields of layer
instances, we will have to modify the reduction relation so that the
sequence of activated layers consists of layer instances (with their field values) rather than
layer names.

JCop also provides special variable \ensuremath{\itbox{thislayer}}, which can be used in
partial methods and is similar to \ensuremath{\itbox{this}} of classes.  It represents the
layer instance in which the invoked partial method is found at
run time and can be used to access fields and methods of that layer
instance.  In operational semantics, the layer instance would be
substituted for \ensuremath{\itbox{thislayer}}, similarly to \ensuremath{\itbox{this}}.  Typing \ensuremath{\itbox{thislayer}}
is also similar to \ensuremath{\itbox{this}} in the sense that it is given the name of
the layer in which it appears but \ensuremath{\itbox{thislayer}} cannot be used for layer
activation because, at run time, it may be bound to an instance of a
\emph{weak} subtype.

We have not fully investigated the interaction between our type system
with other features in Java, such as concurrency, generics, and
lambda, although we expect most of them are orthogonal.


\paragraph*{Acknowledgments.}
We thank Tomoyuki Aotani, Malte Appeltauer, Robert Hirsch\-feld, and
Tetsuo Kamina for valuable discussions on the subject.
This work was supported in part by Kyoto University
Design School (Inoue) and MEXT KAKENHI Grant Number 23220001 (Igarashi).

\bibliographystyle{abbrv}
\bibliography{local}

\begin{thebibliography}{10}

\bibitem{aotani2011featherweight}
T.~Aotani, T.~Kamina, and H.~Masuhara.
\newblock Featherweight {EventCJ}: a core calculus for a context-oriented
  language with event-based per-instance layer transition.
\newblock In {\em Proceedings of the 3rd International Workshop on
  Context-Oriented Programming}. ACM, 2011.

\bibitem{appeltauer2010layered}
M.~Appeltauer, M.~Haupt, and R.~Hirschfeld.
\newblock Layered method dispatch with invokedynamic: an implementation study.
\newblock In {\em Proceedings of the 2nd International Workshop on
  Context-Oriented Programming}, page~4. ACM, 2010.

\bibitem{appeltauer2012jcop}
M.~Appeltauer and R.~Hirschfeld.
\newblock {\em The JCop language specification: Version 1.0, April 2012}.
\newblock Number~59. Universit{\"a}tsverlag Potsdam, 2012.

\bibitem{appeltauer2013declarative}
M.~Appeltauer, R.~Hirschfeld, and J.~Lincke.
\newblock Declarative layer composition with the {JCop} programming language.
\newblock {\em Journal of Object Technology}, 12, 2013.

\bibitem{bainomugisha2012interruptible}
E.~Bainomugisha, J.~Vallejos, C.~De~Roover, A.~Lombide~Carreton, and
  W.~De~Meuter.
\newblock Interruptible context-dependent executions: a fresh look at
  programming context-aware applications.
\newblock In {\em Proceedings of the ACM international symposium on New ideas,
  new paradigms, and reflections on programming and software}, Onward! 2012,
  pages 67--84. ACM, 2012.

\bibitem{batory2004scaling}
D.~Batory, J.~N. Sarvela, and A.~Rauschmayer.
\newblock Scaling step-wise refinement.
\newblock {\em Software Engineering, IEEE Transactions on}, 30(6):355--371,
  2004.

\bibitem{bettini2005momi}
L.~Bettini, V.~Bono, and B.~Venneri.
\newblock {MoMi: A Calculus for Mobile Mixins}.
\newblock {\em Acta Informatica}, 42(2-3):143 -- 190, 2005.

\bibitem{bettini2013flexible}
L.~Bettini, S.~Capecchi, and F.~Damiani.
\newblock On flexible dynamic trait replacement for {Java}-like languages.
\newblock {\em Science of Computer Programming}, 78(7):907--932, 2013.

\bibitem{bono1999core}
V.~Bono, A.~Patel, and V.~Shmatikov.
\newblock A core calculus of classes and mixins.
\newblock In {\em ECOOP'99--Object-Oriented Programming}, pages 43--66.
  Springer, 1999.

\bibitem{burton2015safety}
E.~Burton and E.~Sekerinski.
\newblock The safety of dynamic mixin composition.
\newblock In {\em Proceedings of the 30th Annual ACM Symposium on Applied
  Computing}, pages 1992--1999. ACM, 2015.

\bibitem{cardozo2013modeling}
N.~Cardozo, S.~Gonz{\'a}lez, K.~Mens, R.~Van Der~Straeten, and T.~D'Hondt.
\newblock Modeling and analyzing self-adaptive systems with context {Petri}
  nets.
\newblock In {\em Proc.\ of TASE}, pages 191--198. IEEE, 2013.

\bibitem{cardozo2015semantics}
N.~Cardozo, S.~Gonz{\'a}lez, K.~Mens, R.~Van Der~Straeten, J.~Vallejos, and
  T.~D'Hondt.
\newblock Semantics for consistent activation in context-oriented systems.
\newblock {\em Information and Software Technology}, 58:71--94, 2015.

\bibitem{chiba2010mostly}
S.~Chiba, A.~Igarashi, and S.~Zakirov.
\newblock Mostly modular compilation of crosscutting concerns by contextual
  predicate dispatch.
\newblock In {\em Proc.\ of the ACM OOPSLA}, pages 539--554, 2010.

\bibitem{clarke2009semantics}
D.~Clarke and I.~Sergey.
\newblock A semantics for context-oriented programming with layers.
\newblock In {\em International Workshop on Context-Oriented Programming},
  page~10. ACM, 2009.

\bibitem{clifton2006multijava}
C.~Clifton, T.~Millstein, G.~T. Leavens, and C.~Chambers.
\newblock {MultiJava}: Design rationale, compiler implementation, and
  applications.
\newblock {\em ACM Trans. Prog. Lang. Syst.}, 28(3):517--575, 2006.

\bibitem{costanza2008feature}
P.~Costanza and T.~D'Hondt.
\newblock Feature descriptions for context-oriented programming.
\newblock In {\em Proc.\ of SPLC (2)}, pages 9--14, 2008.
\newblock Workshop on Dynamic Software Product Lines (DSPL 2008).

\bibitem{damiani2017core}
F.~Damiani, L.~Padovani, I.~Schaefer, and C.~Seidl.
\newblock A core calculus for dynamic delta-oriented programming.
\newblock {\em Acta Informatica}, pages 1--39, 2017.

\bibitem{damiani2011dynamic}
F.~Damiani and I.~Schaefer.
\newblock Dynamic delta-oriented programming.
\newblock In {\em Proceedings of the 15th International Software Product Line
  Conference, Volume 2}, page~34. ACM, 2011.

\bibitem{degano2014two}
P.~Degano, G.-L. Ferrari, and L.~Galletta.
\newblock A two-phase static analysis for reliable adaptation.
\newblock In {\em International Conference on Software Engineering and Formal
  Methods}, pages 347--362. Springer, 2014.

\bibitem{degano2016two}
P.~Degano, G.-L. Ferrari, and L.~Galletta.
\newblock A two-component language for adaptation: Design, semantics and
  program analysis.
\newblock {\em IEEE Transactions on Software Engineering}, 42(6):505--529,
  2016.

\bibitem{Drossopoulou2002fickle2}
S.~Drossopoulou, F.~Damiani, M.~Dezani-Ciancaglini, and P.~Giannini.
\newblock More dynamic object reclassification: Fickle$_\textrm{I\!I}$.
\newblock {\em ACM Trans. Prog. Lang. Syst.}, 24(2):153--191, 2002.

\bibitem{flatt1998classes}
M.~Flatt, S.~Krishnamurthi, and M.~Felleisen.
\newblock Classes and mixins.
\newblock In {\em Proc.\ of the ACM POPL}, pages 171--183. ACM, 1998.

\bibitem{gonzalez2010subjective}
S.~Gonz{\'a}lez, N.~Cardozo, K.~Mens, A.~C{\'a}diz, J.-C. Libbrecht, and
  J.~Goffaux.
\newblock Subjective-{C}.
\newblock In {\em International Conference on Software Language Engineering},
  pages 246--265. Springer, 2010.

\bibitem{hallsteinsen2008dynamic}
S.~Hallsteinsen, M.~Hinchey, S.~Park, and K.~Schmid.
\newblock Dynamic software product lines.
\newblock {\em Computer}, 41(4), 2008.

\bibitem{hirschfeld2008cop}
R.~Hirschfeld, P.~Costanza, and O.~Nierstrasz.
\newblock Context-oriented programming.
\newblock {\em Journal of Object Technology}, 7(3):125--151, 2008.

\bibitem{contextfj2011}
R.~Hirschfeld, A.~Igarashi, and H.~Masuhara.
\newblock {ContextFJ}: A minimal core calculus for context-oriented
  programming.
\newblock In {\em Proc.\ of Foundations of Aspect-Oriented Languages (FOAL)},
  Mar. 2011.

\bibitem{DynamicLayer2012contextfj}
A.~Igarashi, R.~Hirschfeld, and H.~Masuhara.
\newblock A type system for dynamic layer composition.
\newblock In {\em Proc.\ of FOOL}, Oct. 2012.

\bibitem{IgarashiPierceWadler01TOPLAS_FJ}
A.~Igarashi, B.~C. Pierce, and P.~Wadler.
\newblock {Featherweight Java}: A minimal core calculus for {Java} and {GJ}.
\newblock {\em ACM TOPLAS}, 23(3):396--450, 2001.

\bibitem{inoue2015sound}
H.~Inoue and A.~Igarashi.
\newblock A sound type system for layer subtyping and dynamically activated
  first-class layers.
\newblock In {\em Asian Symposium on Programming Languages and Systems}, pages
  445--462. Springer, 2015.

\bibitem{inoue2014towards}
H.~Inoue, A.~Igarashi, M.~Appeltauer, and R.~Hirschfeld.
\newblock Towards type-safe {JCop}: A type system for layer inheritance and
  first-class layers.
\newblock In {\em Proc.\ of the Workshop on Context-Oriented Programming},
  pages 7:1--7:6. ACM, 2014.

\bibitem{johnsen2009dynamic}
E.~B. Johnsen, M.~Kyas, and I.~C. Yu.
\newblock Dynamic classes: Modular asynchronous evolution of distributed
  concurrent objects.
\newblock In {\em {FM} 2009: Formal Methods, Second World Congress, Eindhoven,
  The Netherlands, November 2-6, 2009. Proceedings}, volume 5850 of {\em
  Lecture Notes in Computer Science}, pages 596--611. Springer, 2009.

\bibitem{KaminaAotaniIgarashi14COP}
T.~Kamina, T.~Aotani, and A.~Igarashi.
\newblock On-demand layer activation for type-safe deactivation.
\newblock In {\em Proc.\ of COP'14}, Uppsala, Sweden, July 2014.

\bibitem{kamina2011eventcj}
T.~Kamina, T.~Aotani, and H.~Masuhara.
\newblock Event{CJ}: a context-oriented programming language with declarative
  event-based context transition.
\newblock In {\em Proc. of ACM AOSD}, pages 253--264. ACM, 2011.

\bibitem{kamina2015generalized}
T.~Kamina, T.~Aotani, and H.~Masuhara.
\newblock Generalized layer activation mechanism through contexts and
  subscribers.
\newblock In {\em Proceedings of the 14th International Conference on
  Modularity}, pages 14--28. ACM, 2015.

\bibitem{kamina2016generalized}
T.~Kamina, T.~Aotani, and H.~Masuhara.
\newblock {\em Generalized Layer Activation Mechanism for Context-Oriented
  Programming}, pages 123--166.
\newblock Springer International Publishing, Cham, 2016.

\bibitem{KaminaAotaniMasuharaIgarashi17}
T.~Kamina, T.~Aotani, H.~Masuhara, and A.~Igarashi.
\newblock Method safety mechanism for asynchronous layer deactivation.
\newblock {\em Sci. Comput. Programming}, 156:104--120, Mar. 2018.
\newblock A preliminary version was presented at COP'15.

\bibitem{kamina2004mcjava}
T.~Kamina and T.~Tamai.
\newblock {McJava}--a design and implementation of {Java} with mixin-types.
\newblock In {\em Proc.\ of APLAS}, pages 398--414, 2004.

\bibitem{liquori2008feathertrait}
L.~Liquori and A.~Spiwack.
\newblock Feather{Trait}: A modest extension of {Featherweight Java}.
\newblock {\em ACM Trans. Prog. Lang. Syst.}, 30(2):11, 2008.

\bibitem{prehofer1997feature}
C.~Prehofer.
\newblock Feature-oriented programming: A fresh look at objects.
\newblock In {\em ECOOP'97--Object-Oriented Programming}, pages 419--443.
  Springer, 1997.

\bibitem{rosenmuller2011flexible}
M.~Rosenm{\"u}ller, N.~Siegmund, S.~Apel, and G.~Saake.
\newblock Flexible feature binding in software product lines.
\newblock {\em Automated Software Engineering}, 18(2):163--197, 2011.

\bibitem{schaefer2010delta}
I.~Schaefer, L.~Bettini, V.~Bono, F.~Damiani, and N.~Tanzarella.
\newblock Delta-oriented programming of software product lines.
\newblock {\em Software Product Lines: Going Beyond}, pages 77--91, 2010.

\bibitem{smith2005chai}
C.~Smith and S.~Drossopoulou.
\newblock \emph{Chai}: Traits for java-like languages.
\newblock In {\em {ECOOP} 2005 - Object-Oriented Programming, 19th European
  Conference, Glasgow, UK, July 25-29, 2005, Proceedings}, volume 3586 of {\em
  Lecture Notes in Computer Science}, pages 453--478. Springer, 2005.

\bibitem{wright1994syntactic}
A.~K. Wright and M.~Felleisen.
\newblock A syntactic approach to type soundness.
\newblock {\em Information and computation}, 115(1):38--94, 1994.

\end{thebibliography}
\iffull
\clearpage


\appendix
\newtheorem{lemmaapp}{Lemma}[section]
\newtheorem{theoremapp}{Theorem}[section]

\clearpage
\section{Proofs}
\label{sec:proofs}

We fix \(\CT\) and \(\LT\) and assume \((\CT, \LT)\ensuremath{\itbox{ ok}}\)
throughout this section.

\begin{lemmaapp}\label{lem:sublayer-pmtype}
  If $\pmtype(\ensuremath{\itbox{m}}, \ensuremath{\itbox{C}}, \ensuremath{\itbox{L}_{2}\itbox{}}) = \ensuremath{\itbox{}\overline{\itbox{T}}\itbox{\(\rightarrow\)T}_{0}\itbox{}}$ and $\ensuremath{\itbox{L}_{1}\itbox{}} \LEQ_w \ensuremath{\itbox{L}_{2}\itbox{}}$, 
  then $\pmtype(\ensuremath{\itbox{m}}, \ensuremath{\itbox{C}}, \ensuremath{\itbox{L}_{1}\itbox{}}) = \ensuremath{\itbox{}\overline{\itbox{T}}\itbox{\(\rightarrow\)T}_{0}\itbox{}}$.
\end{lemmaapp}

\begin{proof}
  By induction on $\ensuremath{\itbox{L{\ensuremath{'}}}} \LEQ_w \ensuremath{\itbox{L}}$, using \(\noconflict(\ensuremath{\itbox{L{\ensuremath{'}}}}, \ensuremath{\itbox{L}})\)
  in the case where $\ensuremath{\itbox{L{\ensuremath{'}} \(\triangleleft\) L}}$. \qed
\end{proof}

\begin{lemmaapp}\label{lem:sublayer-mtype}
  If \(\mtype(\ensuremath{\itbox{m}}, \ensuremath{\itbox{C}}, \LSet_1, \LSet_2) = \ensuremath{\itbox{}\overline{\itbox{T}}\itbox{\(\rightarrow\)T}_{0}\itbox{}}\) and
  \(\LSet_3 \LEQ_{sw} \LSet_1\) and \(\LSet_4 \LEQ_{sw} \LSet_2\) and
  \(\LSet_1 \subseteq \LSet_2\) and \(\LSet_3 \subseteq \LSet_4\), then
  \(\mtype(\ensuremath{\itbox{m}}, \ensuremath{\itbox{C}}, \LSet_3, \LSet_4) = \ensuremath{\itbox{}\overline{\itbox{T}}\itbox{\(\rightarrow\)T}_{0}\itbox{}}\).
\end{lemmaapp}

\begin{proof}
  By induction on the derivation of \(\mtype(\ensuremath{\itbox{m}}, \ensuremath{\itbox{C}}, \LSet_1, \LSet_2) = \ensuremath{\itbox{}\overline{\itbox{T}}\itbox{\(\rightarrow\)T}_{0}\itbox{}}\)
  with case analysis on the last rule used.
  \begin{rneqncase}{MT-Class}{
      \ensuremath{\itbox{class C \(\triangleleft\) D {\char'173}... T}_{0}\itbox{ m(}\overline{\itbox{T}}\itbox{ }\overline{\itbox{x}}\itbox{){\char'173} return e; {\char'175} ...{\char'175}}}
    }
    \rn{MT-Class} finishes the case.
  \end{rneqncase}

  \begin{rneqncase}{MT-PMethod}{
      \exists \ensuremath{\itbox{L}_{1}\itbox{}} \in \LSet_1. \pmtype(\ensuremath{\itbox{m}}, \ensuremath{\itbox{C}}, \ensuremath{\itbox{L}_{1}\itbox{}}) = \ensuremath{\itbox{}\overline{\itbox{T}}\itbox{\(\rightarrow\)T}_{0}\itbox{}}
    }
  By $\LSet_3 \LEQ_{sw} \LSet_1$, there exists \(\ensuremath{\itbox{L}_{3}\itbox{}} \in \LSet_3\) such that
  either (1) \(\ensuremath{\itbox{L}_{3}\itbox{}} \LEQ_w \ensuremath{\itbox{L}_{1}\itbox{}}\) or (2) there exists \ensuremath{\itbox{L}} such that
  \ensuremath{\itbox{L swappable}} and \(\ensuremath{\itbox{L}_{1}\itbox{}}, \ensuremath{\itbox{L}_{3}\itbox{}} \LEQ_w \ensuremath{\itbox{L}}\).  In the case (1), 
  \lemref{sublayer-pmtype} and \rn{MT-PMethod} finish the case.
  In the case (2), by \rn{T-LayerSW} and \(\noconflict(\ensuremath{\itbox{L}_{1}\itbox{}},\ensuremath{\itbox{L}_{3}\itbox{}})\),
  it is easy to show \(\pmtype(\ensuremath{\itbox{m}}, \ensuremath{\itbox{C}}, \ensuremath{\itbox{L}_{3}\itbox{}}) = \ensuremath{\itbox{}\overline{\itbox{T}}\itbox{\(\rightarrow\)T}_{0}\itbox{}}\).
  Then, \lemref{sublayer-pmtype} and \rn{MT-PMethod} finish the case.
  \end{rneqncase}
  
  \begin{rneqncase}{MT-Super}{
      \ensuremath{\itbox{class C \(\triangleleft\) D {\char'173}... }\overline{\itbox{M}}\itbox{{\char'175}}} \andalso \ensuremath{\itbox{m}} \not \in \ensuremath{\itbox{}\overline{\itbox{M}}\itbox{}} \\
      \forall \ensuremath{\itbox{L}} \in \LSet_1. \pmtype(\ensuremath{\itbox{m}},\ensuremath{\itbox{C}},\ensuremath{\itbox{L}}) \mbox{ undefined} \andalso
      \mtype(\ensuremath{\itbox{m}}, \ensuremath{\itbox{D}}, \LSet_2, \LSet_2) = \ensuremath{\itbox{}\overline{\itbox{T}}\itbox{\(\rightarrow\)T}_{0}\itbox{}} \\
    }
    If \(\pmtype(\ensuremath{\itbox{m}},\ensuremath{\itbox{C}},\ensuremath{\itbox{L}})\) is undefined for all \(\ensuremath{\itbox{L}} \in \LSet_3\), 
    then the induction hypothesis and \rn{MT-Super} finish the case.
    Otherwise, we have \(\pmtype(\ensuremath{\itbox{m}},\ensuremath{\itbox{C}},\ensuremath{\itbox{L}}) = \ensuremath{\itbox{S}_{0}\itbox{ m(}\overline{\itbox{S}}\itbox{ }\overline{\itbox{x}}\itbox{){\char'173}...{\char'175}}}\) for some 
    \(\ensuremath{\itbox{L}}  \in \LSet_3\).  Then, \(\mtype(\ensuremath{\itbox{m}},\ensuremath{\itbox{C}},\LSet_3, \LSet_4) = \ensuremath{\itbox{}\overline{\itbox{S}}\itbox{\(\rightarrow\)S}_{0}\itbox{}}\) holds
    by \rn{MT-PMethod} and also there exists \ensuremath{\itbox{L{\ensuremath{'}}}} such that $\LT(\ensuremath{\itbox{L{\ensuremath{'}}}})(\ensuremath{\itbox{C.m}}) = \ensuremath{\itbox{S}_{0}\itbox{ m(}\overline{\itbox{S}}\itbox{ }\overline{\itbox{x}}\itbox{){\char'173}...{\char'175}}}$.
    
    By the induction hypothesis, 
    \(\mtype(\ensuremath{\itbox{m}}, \ensuremath{\itbox{D}}, \dom(\LT), \dom(\LT)) = \ensuremath{\itbox{}\overline{\itbox{T}}\itbox{\(\rightarrow\)T}_{0}\itbox{}}\) (since  \(\dom(\LT) \LEQ_w \LSet_2\)).
    By \rn{MT-Super}, \(\mtype(\ensuremath{\itbox{m}}, \ensuremath{\itbox{C}}, \emptyset, \dom(\LT)) = \ensuremath{\itbox{}\overline{\itbox{T}}\itbox{\(\rightarrow\)T}_{0}\itbox{}}\).
    Finally, \ensuremath{\itbox{}\overline{\itbox{S}}\itbox{}}, \ensuremath{\itbox{}\overline{\itbox{S}}\itbox{}_{0}\itbox{}} = \ensuremath{\itbox{}\overline{\itbox{T}}\itbox{}}, \ensuremath{\itbox{T}_{0}\itbox{}} follows from \(\override^h(\ensuremath{\itbox{L{\ensuremath{'}}}}, \ensuremath{\itbox{C}})\), finishing the case. \qed
  \end{rneqncase}
\end{proof}

\begin{lemmaapp}\label{lem:mtype_layer}
  If $\mtype(\ensuremath{\itbox{m}}, \ensuremath{\itbox{C}}, \LSet_1, \LSet_2) = \ensuremath{\itbox{}\overline{\itbox{T}}\itbox{\(\rightarrow\)T}_{0}\itbox{}}$ and $\mtype(\ensuremath{\itbox{m}},
  \ensuremath{\itbox{C}}, \LSet_3, \LSet_4) = \ensuremath{\itbox{}\overline{\itbox{T}}\itbox{{\ensuremath{'}}\(\rightarrow\)T}_{0}\itbox{{\ensuremath{'}}}}$, then \ensuremath{\itbox{}\overline{\itbox{T}}\itbox{}}, \ensuremath{\itbox{T}_{0}\itbox{}} = \ensuremath{\itbox{}\overline{\itbox{T}}\itbox{{\ensuremath{'}}}},
  \ensuremath{\itbox{T}_{0}\itbox{{\ensuremath{'}}}}.
\end{lemmaapp}

\begin{proof}
  By induction on the derivation of $\mtype(\ensuremath{\itbox{m}}, \ensuremath{\itbox{C}}, \LSet_1, \LSet_2) = \ensuremath{\itbox{}\overline{\itbox{T}}\itbox{\(\rightarrow\)T}_{0}\itbox{}}$ 
  with case analysis on the last rule used.
  \begin{rneqncase}{MT-Class}{
      \ensuremath{\itbox{class C \(\triangleleft\) D {\char'173}... T}_{0}\itbox{ m(}\overline{\itbox{T}}\itbox{ }\overline{\itbox{x}}\itbox{){\char'173}...{\char'175} ...{\char'175}}}
    }
    Easy.  Use \(\override^h(\ensuremath{\itbox{L}}, \ensuremath{\itbox{C}})\) for $\ensuremath{\itbox{L}} \in \dom(\LT)$ if $\mtype(\ensuremath{\itbox{m}}, \ensuremath{\itbox{C}}, \LSet_3, \LSet_4) = \ensuremath{\itbox{}\overline{\itbox{T}}\itbox{{\ensuremath{'}}\(\rightarrow\)T}_{0}\itbox{{\ensuremath{'}}}}$ is derived by \rn{MT-PMethod}, in which case there exists \ensuremath{\itbox{L{\ensuremath{'}}}} such that \ensuremath{\itbox{L \(\Leq\)}_{w}\itbox{ L{\ensuremath{'}}}} and \(\LT(\ensuremath{\itbox{L{\ensuremath{'}}}})(\ensuremath{\itbox{C.m}}) = \ensuremath{\itbox{T}_{0}\itbox{{\ensuremath{'}} m(}\overline{\itbox{T}}\itbox{{\ensuremath{'}} }\overline{\itbox{x}}\itbox{) {\char'173}...{\char'175}}}.\)
    (Note that $\mtype(\ensuremath{\itbox{m}}, \ensuremath{\itbox{C}}, \LSet_3, \LSet_4) = \ensuremath{\itbox{}\overline{\itbox{T}}\itbox{{\ensuremath{'}}\(\rightarrow\)T}_{0}\itbox{{\ensuremath{'}}}}$
    cannot be derived by \rn{MT-Super}.)
  \end{rneqncase}
  \begin{rneqncase}{MT-PMethod}{
      \exists \ensuremath{\itbox{L}_{1}\itbox{}} \in \LSet_1. \pmtype(\ensuremath{\itbox{m}}, \ensuremath{\itbox{C}}, \ensuremath{\itbox{L}_{1}\itbox{}})= \ensuremath{\itbox{}\overline{\itbox{T}}\itbox{\(\rightarrow\)T}_{0}\itbox{}}
    }
    There exists \ensuremath{\itbox{L}_{1}\itbox{{\ensuremath{'}}}} such that \ensuremath{\itbox{L}_{1}\itbox{ \(\Leq\)}_{w}\itbox{ L}_{1}\itbox{{\ensuremath{'}}}} and \(\LT(\ensuremath{\itbox{L}_{1}\itbox{{\ensuremath{'}}}})(\ensuremath{\itbox{C.m}}) = \ensuremath{\itbox{T}_{0}\itbox{ m(}\overline{\itbox{T}}\itbox{ }\overline{\itbox{x}}\itbox{) {\char'173}...{\char'175}}}.\) 
    Further case analysis on $\mtype(\ensuremath{\itbox{m}}, \ensuremath{\itbox{C}}, \LSet_3, \LSet_4) = \ensuremath{\itbox{}\overline{\itbox{T}}\itbox{{\ensuremath{'}}\(\rightarrow\)T}_{0}\itbox{{\ensuremath{'}}}}$.
    \begin{rnsubcase}{MT-Class}
      Similar to the above case.
    \end{rnsubcase}
    
    \begin{rneqnsubcase}{MT-PMethod}{
        \exists \ensuremath{\itbox{L}_{2}\itbox{}} \in \LSet_3. \pmtype(\ensuremath{\itbox{m}}, \ensuremath{\itbox{C}}, \ensuremath{\itbox{L}_{2}\itbox{}})= \ensuremath{\itbox{}\overline{\itbox{T}}\itbox{{\ensuremath{'}}\(\rightarrow\)T}_{0}\itbox{{\ensuremath{'}}}}
     }
There exists \ensuremath{\itbox{L}_{2}\itbox{{\ensuremath{'}}}} such that \ensuremath{\itbox{L}_{2}\itbox{ \(\Leq\)}_{w}\itbox{ L}_{2}\itbox{{\ensuremath{'}}}} and \(\LT(\ensuremath{\itbox{L}_{2}\itbox{{\ensuremath{'}}}})(\ensuremath{\itbox{C.m}}) = \ensuremath{\itbox{T}_{0}\itbox{{\ensuremath{'}} m(}\overline{\itbox{T}}\itbox{{\ensuremath{'}} }\overline{\itbox{x}}\itbox{) {\char'173}...{\char'175}}}.\) Then, 
     \(\noconflict(\ensuremath{\itbox{L}_{1}\itbox{{\ensuremath{'}}}},\ensuremath{\itbox{L}_{2}\itbox{{\ensuremath{'}}}})\) finishes the case.
    \end{rneqnsubcase}
    
    \begin{rneqnsubcase}{MT-Super}{
        \ensuremath{\itbox{class C \(\triangleleft\) D {\char'173}... }\overline{\itbox{M}}\itbox{{\char'175}}} & \ensuremath{\itbox{m}} \not \in \ensuremath{\itbox{}\overline{\itbox{M}}\itbox{}} \\
        \forall \ensuremath{\itbox{L}}\in \LSet_3. \pmtype (\ensuremath{\itbox{m}}, \ensuremath{\itbox{C}}, \ensuremath{\itbox{L}}) \undf &
        \mtype(\ensuremath{\itbox{m}}, \ensuremath{\itbox{D}}, \LSet_4, \LSet_4) = \ensuremath{\itbox{}\overline{\itbox{T}}\itbox{{\ensuremath{'}}\(\rightarrow\)T}_{0}\itbox{{\ensuremath{'}}}}
      }
      In this case,
      \(\mtype(\ensuremath{\itbox{m}}, \ensuremath{\itbox{C}}, \emptyset, \dom(\LT)) = \ensuremath{\itbox{}\overline{\itbox{T}}\itbox{{\ensuremath{'}}\(\rightarrow\)T}_{0}\itbox{{\ensuremath{'}}}}\)
      because we can show that 
      \begin{eqnarray*}
        \mtype(\ensuremath{\itbox{m}}, \ensuremath{\itbox{D}}, \LSet_4, \LSet_4) & = &
        \mtype(\ensuremath{\itbox{m}}, \ensuremath{\itbox{D}}, \dom(\LT), \dom(\LT)) \\ & = &
        \mtype(\ensuremath{\itbox{m}}, \ensuremath{\itbox{C}}, \emptyset, \dom(\LT)).
      \end{eqnarray*}
      by \lemref{sublayer-mtype} and \rn{MT-Super}.
      Then, \(\override^h(\ensuremath{\itbox{L}_{1}\itbox{{\ensuremath{'}}}}, \ensuremath{\itbox{C}})\) finishes the case.
    \end{rneqnsubcase}
  \end{rneqncase}

  \begin{rneqncase}{MT-Super}{
        \ensuremath{\itbox{class C \(\triangleleft\) D {\char'173}... }\overline{\itbox{M}}\itbox{{\char'175}}} & \ensuremath{\itbox{m}} \not \in \ensuremath{\itbox{}\overline{\itbox{M}}\itbox{}} \\
        \forall \ensuremath{\itbox{L}}\in \LSet_1. \pmtype (\ensuremath{\itbox{m}}, \ensuremath{\itbox{C}}, \ensuremath{\itbox{L}}) \undf &
        \mtype(\ensuremath{\itbox{m}}, \ensuremath{\itbox{D}}, \LSet_2, \LSet_2) = \ensuremath{\itbox{}\overline{\itbox{T}}\itbox{\(\rightarrow\)T}_{0}\itbox{}}
    }
    Further case analysis on $\mtype(\ensuremath{\itbox{m}}, \ensuremath{\itbox{C}}, \LSet_3, \LSet_4) = \ensuremath{\itbox{}\overline{\itbox{T}}\itbox{{\ensuremath{'}}\(\rightarrow\)T}_{0}\itbox{{\ensuremath{'}}}}$.
    \begin{rnsubcase}{MT-PMethod}
      Similar to the subcase \rn{MT-Super} above.
    \end{rnsubcase}
    \begin{rnsubcase}{MT-Class}
      Cannot happen.
    \end{rnsubcase}
    \begin{rnsubcase}{MT-Super}
      By the induction hypothesis,
      $\mtype(\ensuremath{\itbox{m}}, \ensuremath{\itbox{D}}, \LSet_2, \LSet_2) = \mtype(\ensuremath{\itbox{m}}, \ensuremath{\itbox{D}}, \LSet_4, \LSet_4)$. \qed
    \end{rnsubcase}
  \end{rneqncase}
\end{proof}

\begin{lemmaapp}\label{lem:subtype-fields}
  If $\fields(\ensuremath{\itbox{C}}) = \ensuremath{\itbox{}\overline{\itbox{T}}\itbox{ }\overline{\itbox{f}}\itbox{}}$ and $\ensuremath{\itbox{D}} \LEQ \ensuremath{\itbox{C}}$, then $\fields(\ensuremath{\itbox{D}}) =
  \ensuremath{\itbox{}\overline{\itbox{T}}\itbox{ }\overline{\itbox{f}}\itbox{}} , \ensuremath{\itbox{}\overline{\itbox{S}}\itbox{ }\overline{\itbox{g}}\itbox{}}$ for some \ensuremath{\itbox{}\overline{\itbox{S}}\itbox{}}, \ensuremath{\itbox{}\overline{\itbox{g}}\itbox{}}.
\end{lemmaapp}

\begin{proof}
  By induction on $\ensuremath{\itbox{D}} \LEQ \ensuremath{\itbox{C}}$. \qed
\end{proof}

\begin{lemmaapp}[Weakening, \lemref{def:weakening}] \label{lem:weakening} 
  If $\LLGp \ensuremath{\itbox{e}} : \ensuremath{\itbox{T}}$, then $\Loc; \LSet; \Gamma, \ensuremath{\itbox{x}}\colon
  \ensuremath{\itbox{S}} \p \ensuremath{\itbox{e}} : \ensuremath{\itbox{T}}$.
\end{lemmaapp}

\begin{proof}
  By straightforward induction on $\LLGp \ensuremath{\itbox{e}} : \ensuremath{\itbox{T}}$. \qed
\end{proof}

\begin{lemmaapp}\label{lem:subtype-mtype}
  \sloppy If $\mtype(\ensuremath{\itbox{m}}, \ensuremath{\itbox{C}}, \LSet) = \ensuremath{\itbox{}\overline{\itbox{T}}\itbox{\(\rightarrow\)T}_{0}\itbox{}}$ and $\ensuremath{\itbox{D}} \LEQ \ensuremath{\itbox{C}}$, 
  then $\mtype(\ensuremath{\itbox{m}}, \ensuremath{\itbox{D}}, \LSet) = \ensuremath{\itbox{}\overline{\itbox{T}}\itbox{\(\rightarrow\)S}_{0}\itbox{}}$ and $\ensuremath{\itbox{S}_{0}\itbox{}} \LEQ \ensuremath{\itbox{T}_{0}\itbox{}}$ for
  some \ensuremath{\itbox{S}_{0}\itbox{}}.
\end{lemmaapp}

\begin{proof}
  \sloppy
  By induction on $\ensuremath{\itbox{D}} \LEQ \ensuremath{\itbox{C}}$.
  We show only the case where $\ensuremath{\itbox{D extends C}}$.
  If \ensuremath{\itbox{class D \(\triangleleft\) C {\char'173}... S}_{0}\itbox{ m(}\overline{\itbox{S}}\itbox{ }\overline{\itbox{x}}\itbox{){\char'173} return e; {\char'175} ...{\char'175}}}, then $\mtype(\ensuremath{\itbox{m}},
  \ensuremath{\itbox{D}}, \LSet) = \ensuremath{\itbox{}\overline{\itbox{S}}\itbox{\(\rightarrow\)S}_{0}\itbox{}}$ for some \ensuremath{\itbox{}\overline{\itbox{S}}\itbox{}} by \rn{MT-Class}.  By \(\override^v(\ensuremath{\itbox{D}})\),
  $\ensuremath{\itbox{}\overline{\itbox{S}}\itbox{}} = \ensuremath{\itbox{}\overline{\itbox{T}}\itbox{}}$ and \ensuremath{\itbox{S}_{0}\itbox{ \(\Leq\) T}_{0}\itbox{}}.
  If \(\exists \ensuremath{\itbox{L}} \in \LSet. \pmtype(\ensuremath{\itbox{m}}, \ensuremath{\itbox{D}}, \ensuremath{\itbox{L}})=
  \ensuremath{\itbox{}\overline{\itbox{S}}\itbox{\(\rightarrow\)S}_{0}\itbox{}}\), then
  we have \(\mtype(\ensuremath{\itbox{m}}, \ensuremath{\itbox{D}}, \LSet) = \ensuremath{\itbox{}\overline{\itbox{S}}\itbox{\(\rightarrow\)S}_{0}\itbox{}}\) by \rn{MT-PMethod}.
  By \(\override^h(\ensuremath{\itbox{L}}, \ensuremath{\itbox{D}})\), we get $\ensuremath{\itbox{}\overline{\itbox{S}}\itbox{}} = \ensuremath{\itbox{}\overline{\itbox{T}}\itbox{}}$ and $\ensuremath{\itbox{S}_{0}\itbox{}} = \ensuremath{\itbox{T}_{0}\itbox{}}$.
  Otherwise, $\mtype(\ensuremath{\itbox{m}}, \ensuremath{\itbox{D}}, \LSet) = \ensuremath{\itbox{}\overline{\itbox{T}}\itbox{\(\rightarrow\)T}_{0}\itbox{}}$ by \rn{MT-Super}. \qed
\end{proof}

\begin{lemmaapp}[Narrowing, \lemref{def:narrowing}] \label{lem:narrowing}
  If $\LLGp \ensuremath{\itbox{e}} : \ensuremath{\itbox{T}}$ and $\LSet' \LEQ_{sw} \LSet$, then $\Loc; \LSet'; \Gp \ensuremath{\itbox{e}} : \ensuremath{\itbox{T}}$.
\end{lemmaapp}

\begin{proof}
  By induction on $\LLGp \ensuremath{\itbox{e}} : \ensuremath{\itbox{T}}$.  We show only some representative cases.
  
  \begin{rneqncase}{T-Invk}{
     \ensuremath{\itbox{e}} = \ensuremath{\itbox{e}_{0}\itbox{.m(}\overline{\itbox{e}}\itbox{)}} &
     \LLGp \ensuremath{\itbox{e}_{0}\itbox{}} : \ensuremath{\itbox{C}_{0}\itbox{}} \\
     \mtype(\ensuremath{\itbox{m}}, \ensuremath{\itbox{C}_{0}\itbox{}}, \LSet) = \ensuremath{\itbox{}\overline{\itbox{D}}\itbox{\(\rightarrow\)C}} &
     \LLGp \ensuremath{\itbox{}\overline{\itbox{e}}\itbox{}} : \ensuremath{\itbox{}\overline{\itbox{E}}\itbox{}} &
      \ensuremath{\itbox{}\overline{\itbox{E}}\itbox{}} \LEQ \ensuremath{\itbox{}\overline{\itbox{D}}\itbox{}}
    }
     
    By \lemref{sublayer-mtype}, $\mtype(\ensuremath{\itbox{m}}, \ensuremath{\itbox{C}_{0}\itbox{}}, \LSet') = \ensuremath{\itbox{}\overline{\itbox{D}}\itbox{\(\rightarrow\)C}}$.
    Then, the induction hypothesis and \rn{T-Invk} finish the case.
  \end{rneqncase}

  \begin{rneqncase}{T-With}{
      \ensuremath{\itbox{e}} = \ensuremath{\itbox{with e}_{l}\itbox{ e}_{0}\itbox{}} &
      \LLGp \ensuremath{\itbox{e}_{l}\itbox{}} : \ensuremath{\itbox{L}} &
      \ensuremath{\itbox{L req }} \LSet'' \\
      \LSet \LEQ_w \LSet'' &
      \Loc; \LSet \cup \set{\ensuremath{\itbox{L}}}; \Gp \ensuremath{\itbox{e}_{0}\itbox{}} : \ensuremath{\itbox{T}}
    }

    By \rn{LSSW-Intro}, we have $\LSet' \cup \set{\ensuremath{\itbox{L}}} \LEQ_{sw} \LSet
    \cup \set{\ensuremath{\itbox{L}}}$.  By the induction hypothesis, $\Loc; \LSet' \cup
    \set{\ensuremath{\itbox{L}}}; \Gp \ensuremath{\itbox{e}_{0}\itbox{}} : \ensuremath{\itbox{T}}$.

    It is easy to show that \(\LEQ_{sw}\) is transitive and so
    $\LSet' \LEQ_{sw} \LSet''$.  By the induction hypothesis, we also
    have \(\Loc; \LSet'; \Gp \ensuremath{\itbox{e}_{l}\itbox{}} : \ensuremath{\itbox{L}}\).

    Moreover, in a well-formed program, $\ensuremath{\itbox{L req }} \LSet''$
    means that $\LSet''$ does not contain any sublayer of \ensuremath{\itbox{swappable}}
    layers.  By these facts and \rn{LSS-Intro}, we get $\LSet' \LEQ_w
    \LSet''$.  Then, by \rn{T-With}, $\Loc; \LSet';
    \Gp \ensuremath{\itbox{with e}_{l}\itbox{ e}_{0}\itbox{}} : \ensuremath{\itbox{T}}$, finishing the case.
  \end{rneqncase}

    \begin{rneqncase}{T-Swap}{
      \ensuremath{\itbox{e}} = \ensuremath{\itbox{swap (e}_{l}\itbox{,L}_{{sw}}\itbox{) e}_{0}\itbox{}} &
      \LLGp \ensuremath{\itbox{e}_{l}\itbox{}} : \ensuremath{\itbox{L}} \\
      \ensuremath{\itbox{L}_{{sw}}\itbox{ swappable}} &
      \ensuremath{\itbox{L \(\Leq\)}_{w}\itbox{ L}_{{sw}}\itbox{}} &
      \ensuremath{\itbox{L req }} \LSet'' \\
      \LSet_{rm} = (\LSet \backslash \set{\ensuremath{\itbox{L{\ensuremath{'}}}} \mid \ensuremath{\itbox{L{\ensuremath{'}} \(\Leq\)}_{w}\itbox{ L}_{{sw}}\itbox{}}}) &
      \LSet_{rm} \LEQ_w \LSet'' &
      \Loc; \LSet_{rm} \cup \set{\ensuremath{\itbox{L}}}; \Gp \ensuremath{\itbox{e}_{0}\itbox{}} : \ensuremath{\itbox{T}}
      }
      
    It is easy to show that $(\LSet' \backslash \set{\ensuremath{\itbox{L{\ensuremath{'}}}} \mid
      \ensuremath{\itbox{L{\ensuremath{'}} \(\Leq\)}_{w}\itbox{ L}_{{sw}}\itbox{}}}) \cup \set{\ensuremath{\itbox{L}}} \LEQ_{sw} \LSet_{rm} \cup \set{\ensuremath{\itbox{L}}}$
    from $\LSet' \LEQ_{sw} \LSet$.
    By the induction hypothesis, $\Loc; (\LSet' \backslash \set{\ensuremath{\itbox{L{\ensuremath{'}}}} | \ensuremath{\itbox{L{\ensuremath{'}} \(\Leq\)}_{w}\itbox{ L}_{{sw}}\itbox{}}})
    \cup \set{\ensuremath{\itbox{L}}} ; \Gp \ensuremath{\itbox{e}_{0}\itbox{}} : \ensuremath{\itbox{T}}$.

    By \rn{LSSW-Intro}, we have
    $(\LSet' \backslash \set{\ensuremath{\itbox{L{\ensuremath{'}}}} | \ensuremath{\itbox{L{\ensuremath{'}} \(\Leq\)}_{w}\itbox{ L}_{{sw}}\itbox{}}}) \cup \set{\ensuremath{\itbox{L}}} \LEQ_{sw} \LSet''$.
    Moreover, in a well-formed program, $\ensuremath{\itbox{L req }}
    \LSet''$ means that $\LSet''$ does not have any sublayer of
    \ensuremath{\itbox{swappable}} layers.  By these facts and \rn{LSS-Intro}, we get
    $(\LSet' \backslash \set{\ensuremath{\itbox{L{\ensuremath{'}}}} | \ensuremath{\itbox{L{\ensuremath{'}} \(\Leq\)}_{w}\itbox{ L}_{{sw}}\itbox{}}}) \cup \set{\ensuremath{\itbox{L}}}
    \LEQ_w \LSet''$.
   By the induction hypothesis, we also have
    \(\Loc; \LSet'; \Gp \ensuremath{\itbox{e}_{l}\itbox{}} : \ensuremath{\itbox{L}}\).  Then, by \rn{T-Swap}, $\Loc;
    \LSet'; \Gp \ensuremath{\itbox{swap (e}_{l}\itbox{,L}_{{sw}}\itbox{) e}_{0}\itbox{}} : \ensuremath{\itbox{T}}$, finishing the case. \qed
  \end{rneqncase}
\end{proof}

\begin{lemmaapp}[Strengthening for values, \lemref{def:value-strengthening})] \label{lem:value-strengthening}
  If $\LLGp \ensuremath{\itbox{v}} : \ensuremath{\itbox{T}}$ then, $\Loc'; \LSet'; \Gamma' \p \ensuremath{\itbox{v}} : \ensuremath{\itbox{T}}$.
\end{lemmaapp}

\begin{proof}
  By straightforward induction on $\LLGp \ensuremath{\itbox{v}} : \ensuremath{\itbox{T}}$. \qed
\end{proof}

\begin{lemmaapp}[Substitution, \lemref{def:substitution}] \label{lem:substitution}
  If $\Loc; \LSet; \Gamma, \ensuremath{\itbox{}\overline{\itbox{x}}\itbox{}} \colon \ensuremath{\itbox{}\overline{\itbox{T}}\itbox{}} \p \ensuremath{\itbox{e}} : \ensuremath{\itbox{T}}$ and $\LLGp
  \ensuremath{\itbox{}\overline{\itbox{v}}\itbox{}} : \ensuremath{\itbox{}\overline{\itbox{S}}\itbox{}}$ and $\ensuremath{\itbox{}\overline{\itbox{S}}\itbox{}} \LEQ \ensuremath{\itbox{}\overline{\itbox{T}}\itbox{}}$, then $\LLGp [\ensuremath{\itbox{}\overline{\itbox{v}}\itbox{}}/\ensuremath{\itbox{}\overline{\itbox{x}}\itbox{}}]\ensuremath{\itbox{e}} :
  \ensuremath{\itbox{S}}$ and $\ensuremath{\itbox{S}} \LEQ \ensuremath{\itbox{T}}$ for some \ensuremath{\itbox{S}}.
\end{lemmaapp}

\begin{proof}
  By induction on $\Loc; \LSet; \Gamma, \ensuremath{\itbox{}\overline{\itbox{x}}\itbox{}} \colon \ensuremath{\itbox{}\overline{\itbox{T}}\itbox{}} \p \ensuremath{\itbox{e}} :
  \ensuremath{\itbox{T}}$ with case analysis on the last rule used.  
  We show main cases of \rn{T-With} and \rn{T-Swap}.

  \begin{rneqncase}{T-With}{
      \ensuremath{\itbox{e}} = \ensuremath{\itbox{with e}_{l}\itbox{ e}_{0}\itbox{}} &
      \Loc; \LSet; \Gamma, \ensuremath{\itbox{}\overline{\itbox{x}}\itbox{}} \colon \ensuremath{\itbox{}\overline{\itbox{T}}\itbox{}} \p \ensuremath{\itbox{e}_{l}\itbox{}} : \ensuremath{\itbox{L}} &
      \ensuremath{\itbox{L req }} \LSet' \\
      \LSet \LEQ_w \LSet' &
      \Loc; \LSet \cup \set{\ensuremath{\itbox{L}}}; \Gamma, \ensuremath{\itbox{}\overline{\itbox{x}}\itbox{}} \colon \ensuremath{\itbox{}\overline{\itbox{T}}\itbox{}} \p \ensuremath{\itbox{e}_{0}\itbox{}} : \ensuremath{\itbox{T}}
    }
    By the induction hypothesis, $\Loc; \LSet; \Gamma 
    \p [\ensuremath{\itbox{}\overline{\itbox{v}}\itbox{}}/\ensuremath{\itbox{}\overline{\itbox{x}}\itbox{}}]\ensuremath{\itbox{e}_{l}\itbox{}} : \ensuremath{\itbox{L}_{0}\itbox{}}$ and $\ensuremath{\itbox{L}_{0}\itbox{}} \LEQ \ensuremath{\itbox{L}}$ for some
    \ensuremath{\itbox{L}_{0}\itbox{}}.  By induction on $\ensuremath{\itbox{L}_{0}\itbox{}} \LEQ \ensuremath{\itbox{L}}$, it is easy to show that $\ensuremath{\itbox{L}_{0}\itbox{ req }} \LSet'$.  Since \ensuremath{\itbox{L}_{0}\itbox{ \(\Leq\) L}},
    we also have \ensuremath{\itbox{L}_{0}\itbox{ \(\Leq\)}_{w}\itbox{ L}} and so it is easy to show $\LSet \cup \set{\ensuremath{\itbox{L}_{0}\itbox{}}} \LEQ_w \LSet \cup \set{\ensuremath{\itbox{L}}}$.
    By the induction hypothesis,
    $\Loc; \LSet \cup \set{\ensuremath{\itbox{L}}}; \Gamma \p [\ensuremath{\itbox{}\overline{\itbox{v}}\itbox{}}/\ensuremath{\itbox{}\overline{\itbox{x}}\itbox{}}]\ensuremath{\itbox{e}_{0}\itbox{}} : \ensuremath{\itbox{S}}$ and
    $\ensuremath{\itbox{S}} \LEQ \ensuremath{\itbox{T}}$ for some \ensuremath{\itbox{S}}.  Then, by \lemref{narrowing},
    $\Loc; \LSet \cup \set{\ensuremath{\itbox{L}_{0}\itbox{}}}; \Gamma \p [\ensuremath{\itbox{}\overline{\itbox{v}}\itbox{}}/\ensuremath{\itbox{}\overline{\itbox{x}}\itbox{}}]\ensuremath{\itbox{e}_{0}\itbox{}} : \ensuremath{\itbox{S}}$;  and
    \rn{T-With} finish the case.
  \end{rneqncase}

    \begin{rneqncase}{T-Swap}{
      \ensuremath{\itbox{e}} = \ensuremath{\itbox{swap (e}_{l}\itbox{, L}_{{sw}}\itbox{) e}_{0}\itbox{}} &
      \Loc; \LSet; \Gamma, \ensuremath{\itbox{}\overline{\itbox{x}}\itbox{}} \colon \ensuremath{\itbox{}\overline{\itbox{T}}\itbox{}} \p \ensuremath{\itbox{e}_{l}\itbox{}} : \ensuremath{\itbox{L}} \\
      \ensuremath{\itbox{L}_{{sw}}\itbox{ swappable}} &
      \ensuremath{\itbox{L}} \LEQ_w \ensuremath{\itbox{L}_{{sw}}\itbox{}} &
      \ensuremath{\itbox{L req }} \LSet' \\
      \LSet_{rm} = \LSet \setminus \set{\ensuremath{\itbox{L{\ensuremath{'}}}} \mid \ensuremath{\itbox{L{\ensuremath{'}}}} \LEQ_w \ensuremath{\itbox{L}_{{sw}}\itbox{}}} &
      \LSet_{rm} \LEQ_w \LSet' \\
      \Loc; \LSet_{rm} \cup \set{\ensuremath{\itbox{L}}}; \Gamma, \ensuremath{\itbox{}\overline{\itbox{x}}\itbox{}} \colon \ensuremath{\itbox{}\overline{\itbox{T}}\itbox{}} \p \ensuremath{\itbox{e}_{0}\itbox{}} : \ensuremath{\itbox{T}}
    }
    
    By the induction hypothesis, $\Loc; \LSet; \Gp [\ensuremath{\itbox{}\overline{\itbox{v}}\itbox{}}/\ensuremath{\itbox{}\overline{\itbox{x}}\itbox{}}]\ensuremath{\itbox{e}_{l}\itbox{}} : \ensuremath{\itbox{L}_{0}\itbox{}}$ and $\ensuremath{\itbox{L}_{0}\itbox{}} \LEQ \ensuremath{\itbox{L}}$ for some
    \ensuremath{\itbox{L}_{0}\itbox{}}.  By induction on $\ensuremath{\itbox{L}_{0}\itbox{}} \LEQ \ensuremath{\itbox{L}}$, it is easy to show $\ensuremath{\itbox{L}_{0}\itbox{ req }} \LSet'$.  Since $\ensuremath{\itbox{L}_{0}\itbox{}}
    \LEQ \ensuremath{\itbox{L}}$, we have $\ensuremath{\itbox{L}_{0}\itbox{}} \LEQ_w \ensuremath{\itbox{L}}$ and $\ensuremath{\itbox{L}_{0}\itbox{}} \LEQ_w \ensuremath{\itbox{L}_{{sw}}\itbox{}}$ and $\LSet_{rm} \cup
    \set{\ensuremath{\itbox{L}_{0}\itbox{}}} \LEQ_w \LSet_{rm} \cup \set{\ensuremath{\itbox{L}}}$.  By
    the induction hypothesis, $\Loc; \LSet_{rm} \cup \set{\ensuremath{\itbox{L}}};
    \Gp [\ensuremath{\itbox{}\overline{\itbox{v}}\itbox{}}/\ensuremath{\itbox{}\overline{\itbox{x}}\itbox{}}]\ensuremath{\itbox{e}_{0}\itbox{}} : \ensuremath{\itbox{S}}$ and $\ensuremath{\itbox{S}}
    \LEQ \ensuremath{\itbox{T}}$ for some \ensuremath{\itbox{S}}.  Then, by \lemref{narrowing},
    $\Loc; \LSet_{rm} \cup \set{\ensuremath{\itbox{L}_{0}\itbox{}}}; \Gp [\ensuremath{\itbox{}\overline{\itbox{v}}\itbox{}}/\ensuremath{\itbox{}\overline{\itbox{x}}\itbox{}}]\ensuremath{\itbox{e}_{0}\itbox{}} : \ensuremath{\itbox{S}}$;
    and \rn{T-With} finishes the case. \qed
  \end{rneqncase}
\end{proof}

\begin{lemmaapp}\label{lem:sub-req}
  If \(\ensuremath{\itbox{L}_{1}\itbox{}} \LEQ_w \ensuremath{\itbox{L}_{2}\itbox{}}\) and \(\ensuremath{\itbox{L}_{1}\itbox{ req }} \LSet_1\) and \(\ensuremath{\itbox{L}_{2}\itbox{ req }} \LSet_2\), then
  \(\LSet_1 \LEQ_w \LSet_2\).
\end{lemmaapp}
\begin{proof}
  By induction on \(\ensuremath{\itbox{L}_{1}\itbox{}} \LEQ_w \ensuremath{\itbox{L}_{2}\itbox{}}\).
  Use \rn{T-Layer} in the case for \rn{LSw-Extends}. \qed
\end{proof}

We prove a stronger property than \lemref{def:layer-set-wf1};
in the statement below,
$\mathrel{(\ensuremath{\itbox{\(\Leq\)}_{w}\itbox{}};\ensuremath{\itbox{req}})}$ stands for the composition
of the two relations \ensuremath{\itbox{\(\Leq\)}_{w}\itbox{}} and \ensuremath{\itbox{req}}.

\begin{lemmaapp} \label{lem:layer-set-wf1}
  If $\LSet \, \WF$, then $\forall \ensuremath{\itbox{L}} \in \LSet, \forall \ensuremath{\itbox{L{\ensuremath{'}}}} \text{
    s.t. } \ensuremath{\itbox{L}} \mathrel{(\ensuremath{\itbox{\(\Leq\)}_{w}\itbox{}};\ensuremath{\itbox{req}})} \ensuremath{\itbox{L{\ensuremath{'}}}}, \exists \ensuremath{\itbox{L{\ensuremath{'}}{\ensuremath{'}}}} \in \LSet. \ensuremath{\itbox{L{\ensuremath{'}}{\ensuremath{'}}}} \LEQ_w \ensuremath{\itbox{L{\ensuremath{'}}}}$.  

\end{lemmaapp}

\begin{proof}
  Induction on the derivation of $\LSet \, \WF$.

  \begin{rncase}{Wf-Empty}
    Trivial.
  \end{rncase}

  \begin{rneqncase}{Wf-With}{
      \LSet = \LSet_0 \cup \set{\ensuremath{\itbox{L}_{a}\itbox{}}} \andalso
      \LSet_0 \, \WF \andalso
      \ensuremath{\itbox{L}_{a}\itbox{ req }} \LSet' \andalso
      \LSet_0 \LEQ_w \LSet'
    }
    By the induction hypothesis, we have $\forall \ensuremath{\itbox{L}} \in
    \LSet_0. \forall \ensuremath{\itbox{L{\ensuremath{'}}}} \text{ s.t. } \ensuremath{\itbox{L}} \mathrel{(\ensuremath{\itbox{\(\Leq\)}_{w}\itbox{}};\ensuremath{\itbox{req}})} \ensuremath{\itbox{L{\ensuremath{'}}}}. \exists \ensuremath{\itbox{L{\ensuremath{'}}{\ensuremath{'}}}}
    \in \LSet_0. \ensuremath{\itbox{L{\ensuremath{'}}{\ensuremath{'}}}} \LEQ_w \ensuremath{\itbox{L{\ensuremath{'}}}}$.  By $\LSet_0 \LEQ_w \LSet'$ and \lemref{sub-req}, we have
    $\forall \ensuremath{\itbox{L{\ensuremath{'}}}} \text{ s.t. } \ensuremath{\itbox{L}_{a}\itbox{}} \mathrel{(\ensuremath{\itbox{\(\Leq\)}_{w}\itbox{}};\ensuremath{\itbox{req}})} \ensuremath{\itbox{L{\ensuremath{'}}}}.  \exists \ensuremath{\itbox{L{\ensuremath{'}}{\ensuremath{'}}}}
    \in \LSet_0. \ensuremath{\itbox{L{\ensuremath{'}}{\ensuremath{'}}}} \LEQ_w \ensuremath{\itbox{L{\ensuremath{'}}}}$.  So, $\forall \ensuremath{\itbox{L}} \in
    \LSet. \forall \ensuremath{\itbox{L{\ensuremath{'}}}} \text{ s.t. } \ensuremath{\itbox{L}} \mathrel{(\ensuremath{\itbox{\(\Leq\)}_{w}\itbox{}};\ensuremath{\itbox{req}})} \ensuremath{\itbox{L{\ensuremath{'}}}}. \exists \ensuremath{\itbox{L{\ensuremath{'}}{\ensuremath{'}}}}
    \in \LSet. \ensuremath{\itbox{L{\ensuremath{'}}{\ensuremath{'}}}} \LEQ_w \ensuremath{\itbox{L{\ensuremath{'}}}}$.
  \end{rneqncase}

  \begin{rneqncase}{Wf-Swap}{
      \LSet = \LSet_{rm} \cup \set{\ensuremath{\itbox{L}_{a}\itbox{}}} \andalso
      \LSet_0\, \WF \andalso
      \ensuremath{\itbox{L}_{{sw}}\itbox{ swappable}} \andalso \ensuremath{\itbox{L}_{a}\itbox{}} \LEQ_w \ensuremath{\itbox{L}_{{sw}}\itbox{}} \\
      \ensuremath{\itbox{L}_{a}\itbox{ req }} \LSet_a \andalso
      \LSet_{rm} = \LSet_0 \setminus \set{\ensuremath{\itbox{L{\ensuremath{'}}}} \mid \ensuremath{\itbox{L{\ensuremath{'}}}} \LEQ_w \ensuremath{\itbox{L}_{{sw}}\itbox{}}}  \andalso
    \LSet_{rm} \LEQ_w \LSet_a
    }

    By the induction hypothesis, we have $$\forall \ensuremath{\itbox{L}} \in
    \LSet_0. \forall \ensuremath{\itbox{L{\ensuremath{'}}}} \text{ s.t. } \ensuremath{\itbox{L}} \mathrel{(\ensuremath{\itbox{\(\Leq\)}_{w}\itbox{}};\ensuremath{\itbox{req}})} \ensuremath{\itbox{L{\ensuremath{'}}}}. \exists \ensuremath{\itbox{L{\ensuremath{'}}{\ensuremath{'}}}}
    \in \LSet_0. \ensuremath{\itbox{L{\ensuremath{'}}{\ensuremath{'}}}} \LEQ_w \ensuremath{\itbox{L{\ensuremath{'}}}},$$ and so $$\forall \ensuremath{\itbox{L}} \in
    \LSet_{rm}. \forall \ensuremath{\itbox{L{\ensuremath{'}}}} \text{ s.t. } \ensuremath{\itbox{L}} \mathrel{(\ensuremath{\itbox{\(\Leq\)}_{w}\itbox{}};\ensuremath{\itbox{req}})} \ensuremath{\itbox{L{\ensuremath{'}}}},
    \exists \ensuremath{\itbox{L{\ensuremath{'}}{\ensuremath{'}}}} \in \LSet_0 \text{ s.t. } \ensuremath{\itbox{L{\ensuremath{'}}{\ensuremath{'}}}} \LEQ_w \ensuremath{\itbox{L{\ensuremath{'}}}}.$$
    In fact, we can show that $$\forall \ensuremath{\itbox{L}} \in
    \LSet_{rm}. \forall \ensuremath{\itbox{L{\ensuremath{'}}}} \text{ s.t. } \ensuremath{\itbox{L}} \mathrel{(\ensuremath{\itbox{\(\Leq\)}_{w}\itbox{}};\ensuremath{\itbox{req}})} \ensuremath{\itbox{L{\ensuremath{'}}}},
    \exists \ensuremath{\itbox{L{\ensuremath{'}}{\ensuremath{'}}}} \in (\LSet_{rm} \cup \set{\ensuremath{\itbox{L}_{a}\itbox{}}}) \text{ s.t. } \ensuremath{\itbox{L{\ensuremath{'}}{\ensuremath{'}}}} \LEQ_w \ensuremath{\itbox{L{\ensuremath{'}}}}:$$
    if \(\ensuremath{\itbox{L{\ensuremath{'}}{\ensuremath{'}}}} \in \set{\ensuremath{\itbox{L}_{b}\itbox{}} \mid \ensuremath{\itbox{L}_{b}\itbox{}} \LEQ_w \ensuremath{\itbox{L}}_{sw}}\) for given
    \ensuremath{\itbox{L}} and \ensuremath{\itbox{L{\ensuremath{'}}}}, then it must be the case that
    \(\ensuremath{\itbox{L}}_{sw} \LEQ_w \ensuremath{\itbox{L{\ensuremath{'}}}}\) because \ensuremath{\itbox{L{\ensuremath{'}}}} is required by some weak supertype of \ensuremath{\itbox{L}} and so
    must not be a sublayer of a swappable and that \(\ensuremath{\itbox{L}_{a}\itbox{}} \LEQ_w \ensuremath{\itbox{L{\ensuremath{'}}}}\).
    
    By $\ensuremath{\itbox{L}_{a}\itbox{ req }} \LSet_a$ and $\LSet_{rm} \LEQ_w \LSet_a$,
    we finally have
    $$\forall \ensuremath{\itbox{L}} \in
    \LSet. \forall \ensuremath{\itbox{L{\ensuremath{'}}}} \text{ s.t. } \ensuremath{\itbox{L}} \mathrel{(\ensuremath{\itbox{\(\Leq\)}_{w}\itbox{}};\ensuremath{\itbox{req}})} \ensuremath{\itbox{L{\ensuremath{'}}}},
    \exists \ensuremath{\itbox{L{\ensuremath{'}}{\ensuremath{'}}}} \in (\LSet_{rm} \cup \set{\ensuremath{\itbox{L}_{a}\itbox{}}}) \text{ s.t. } \ensuremath{\itbox{L{\ensuremath{'}}{\ensuremath{'}}}} \LEQ_w \ensuremath{\itbox{L{\ensuremath{'}}}}.$$ \qed
  \end{rneqncase}
\end{proof}

\begin{lemmaapp}[\lemref{def:layer-set-wf2}]\label{lem:layer-set-wf2}
  If $\LSet\, \WF$ and $\mtype(\ensuremath{\itbox{m}},\ensuremath{\itbox{C}},\LSet) \defined$ and
  $\mtype(\ensuremath{\itbox{m}},\ensuremath{\itbox{D}},\LSet) \undf$ and \ensuremath{\itbox{C \(\triangleleft\) D}}, then 
  $(\exists \ensuremath{\itbox{L{\ensuremath{'}}}} \in \LSet.  \ensuremath{\itbox{proceed}}
  \not \in \pmbody(\ensuremath{\itbox{m}},\ensuremath{\itbox{C}},\ensuremath{\itbox{L{\ensuremath{'}}}})) \text{
    or } \mtype(\ensuremath{\itbox{m}},\ensuremath{\itbox{C}},\emptyset, \LSet) \defined$.
\end{lemmaapp}

\begin{proof}
  \sloppy
  We prove by induction on the derivation of $\WF$
  a stronger property:
  If $\LSet \, \WF$ and $\mtype(\ensuremath{\itbox{m}},\ensuremath{\itbox{C}},\LSet) \defined$ and
  $\mtype(\ensuremath{\itbox{m}},\ensuremath{\itbox{D}},\LSet) \undf$ and \ensuremath{\itbox{C \(\triangleleft\) D}}, then 
  $(\exists \ensuremath{\itbox{L{\ensuremath{'}}}} \in \LSet. \ensuremath{\itbox{proceed}} \not \in \pmbody(\ensuremath{\itbox{m}},\ensuremath{\itbox{C}},\ensuremath{\itbox{L{\ensuremath{'}}}})
  \land (\forall \ensuremath{\itbox{L{\ensuremath{'}}{\ensuremath{'}}}}, \ensuremath{\itbox{L{\ensuremath{'}}{\ensuremath{'}}{\ensuremath{'}}}} \text{ s.t. } \ensuremath{\itbox{L{\ensuremath{'}}}} \LEQ_w \ensuremath{\itbox{L{\ensuremath{'}}{\ensuremath{'}}}} \land
                                             \ensuremath{\itbox{L{\ensuremath{'}}{\ensuremath{'}} swappable}} \land
                                             \ensuremath{\itbox{L{\ensuremath{'}}{\ensuremath{'}}{\ensuremath{'}}}} \LEQ_w \ensuremath{\itbox{L{\ensuremath{'}}{\ensuremath{'}}}}.
         \ensuremath{\itbox{proceed}} \not \in \pmbody(\ensuremath{\itbox{m}},\ensuremath{\itbox{C}},\ensuremath{\itbox{L{\ensuremath{'}}{\ensuremath{'}}{\ensuremath{'}}}})))$ 
   or $\mtype(\ensuremath{\itbox{m}},\ensuremath{\itbox{C}},\emptyset, \LSet)$ is defined.

  In what follows, we define predicate
  $\npr(\ensuremath{\itbox{m}},\ensuremath{\itbox{C}},\LSet)$ by $(\exists \ensuremath{\itbox{L{\ensuremath{'}}}} \in \LSet.
  \ensuremath{\itbox{proceed}} \not \in \pmbody(\ensuremath{\itbox{m}},\ensuremath{\itbox{C}},\ensuremath{\itbox{L{\ensuremath{'}}}})
  \land (\forall \ensuremath{\itbox{L{\ensuremath{'}}{\ensuremath{'}}}}, \ensuremath{\itbox{L{\ensuremath{'}}{\ensuremath{'}}{\ensuremath{'}}}} \text{ s.t. } \ensuremath{\itbox{L{\ensuremath{'}}}} \LEQ_w \ensuremath{\itbox{L{\ensuremath{'}}{\ensuremath{'}}}} \land \ensuremath{\itbox{L{\ensuremath{'}}{\ensuremath{'}} swappable}} \land \ensuremath{\itbox{L{\ensuremath{'}}{\ensuremath{'}}{\ensuremath{'}}}} \LEQ_w \ensuremath{\itbox{L{\ensuremath{'}}{\ensuremath{'}}}}. \ensuremath{\itbox{proceed}} \not \in \pmbody(\ensuremath{\itbox{m}},\ensuremath{\itbox{C}},\ensuremath{\itbox{L{\ensuremath{'}}{\ensuremath{'}}{\ensuremath{'}}}})))$ or $
\mtype(\ensuremath{\itbox{m}},\ensuremath{\itbox{C}},\emptyset, \LSet)$ is defined.

  \begin{rncase}{Wf-Empty}
    Trivial.
  \end{rncase}

  \begin{rneqncase}{Wf-With}{
      \LSet = \LSet_0 \cup \set{\ensuremath{\itbox{L}_{a}\itbox{}}} \andalso
      \LSet_0 \, \WF \andalso
      \ensuremath{\itbox{L}_{a}\itbox{ req }} \LSet' \andalso
      \LSet_0 \LEQ_w \LSet'
    }
    If $\mtype(\ensuremath{\itbox{m}},\ensuremath{\itbox{C}},\LSet_0)$ is defined, by the induction
    hypothesis, $\npr(\ensuremath{\itbox{m}},\ensuremath{\itbox{C}},\LSet_0)$ holds.  Since
    $\LSet = \LSet_0 \cup \set{\ensuremath{\itbox{L}_{a}\itbox{}}}$, $\npr(\ensuremath{\itbox{m}},\ensuremath{\itbox{C}},\LSet)$ also
    holds.

    Otherwise, it must be the case that $\mtype(\ensuremath{\itbox{m}},\ensuremath{\itbox{C}},\LSet_0)
    \undf$ and $\mtype(\ensuremath{\itbox{m}},\ensuremath{\itbox{C}},\set{\ensuremath{\itbox{L}_{a}\itbox{}}}) \defined$.  Since $\LSet
    \LEQ_w \LSet_0 \LEQ_w \LSet'$ and neither $\mtype(\ensuremath{\itbox{m}},\ensuremath{\itbox{C}},\LSet_0)$ nor
    $\mtype(\ensuremath{\itbox{m}},\ensuremath{\itbox{D}},\LSet)$ is defined, $\mtype(\ensuremath{\itbox{m}},\ensuremath{\itbox{C}},\LSet',\LSet'
    \cup \set{\ensuremath{\itbox{L}_{a}\itbox{}}})$ is $\undf$.  Then, $\ensuremath{\itbox{proceed}} \not \in
    \pmbody(\ensuremath{\itbox{m}},\ensuremath{\itbox{C}},\ensuremath{\itbox{L}_{a}\itbox{}})$ holds since if the partial method had
    \ensuremath{\itbox{proceed}}, it would contradict the fact that \ensuremath{\itbox{L}_{a}\itbox{}}
    is well-typed (in particular,
    $\mtype(\ensuremath{\itbox{m}},\ensuremath{\itbox{C}},\LSet',\LSet' \cup \set{\ensuremath{\itbox{L}_{a}\itbox{}}})$ would not be
    defined, as opposed to what \rn{T-Proceed} requires).
    If \ensuremath{\itbox{L}_{a}\itbox{}} is a sublayer of swappable layer
    \ensuremath{\itbox{L}_{{sw}}\itbox{}}, for all $\ensuremath{\itbox{L}_{b}\itbox{}} \LEQ_w \ensuremath{\itbox{L}_{{sw}}\itbox{}}$, $\ensuremath{\itbox{proceed}} \not \in
    \pmbody(\ensuremath{\itbox{m}},\ensuremath{\itbox{C}},\ensuremath{\itbox{L}_{b}\itbox{}})$ through the same argument (note that \ensuremath{\itbox{L}_{b}\itbox{ req }}\(\LSet'\)).  Then,
    $\npr(\ensuremath{\itbox{m}},\ensuremath{\itbox{C}},\LSet)$ holds.
  \end{rneqncase}

  \begin{rneqncase}{Wf-Swap}{
      \LSet = \LSet_{rm} \cup \set{\ensuremath{\itbox{L}_{a}\itbox{}}} \andalso
      \LSet_0\, \WF \andalso
      \ensuremath{\itbox{L}_{{sw}}\itbox{ swappable}} \andalso \ensuremath{\itbox{L}_{a}\itbox{}} \LEQ_w \ensuremath{\itbox{L}_{{sw}}\itbox{}} \\
      \ensuremath{\itbox{L}_{a}\itbox{ req }} \LSet_a \andalso
      \LSet_{rm} = \LSet_0 \setminus \set{\ensuremath{\itbox{L{\ensuremath{'}}}} \mid \ensuremath{\itbox{L{\ensuremath{'}}}} \LEQ_w \ensuremath{\itbox{L}_{{sw}}\itbox{}}}  \andalso
      \LSet_{rm} \LEQ_w \LSet_a
    }
    It is easy to show \(\LSet \LEQ_{sw} \LSet_0\) and vice versa.
    By \lemref{sublayer-mtype}, $\mtype(\ensuremath{\itbox{m}}, \ensuremath{\itbox{C}}, \LSet_0)$ is defined
    and $\mtype(\ensuremath{\itbox{m}}, \ensuremath{\itbox{D}}, \LSet_0)$ is undefined.  By the induction hypothesis,
    $\npr(\ensuremath{\itbox{m}},\ensuremath{\itbox{C}},\LSet_0)$, that is,
    either (1) $\mtype(\ensuremath{\itbox{m}},\ensuremath{\itbox{C}},\emptyset, \LSet_0)$ is defined, or (2)
    $(\exists \ensuremath{\itbox{L{\ensuremath{'}}}} \in \LSet_0.
    \ensuremath{\itbox{proceed}} \not \in \pmbody(\ensuremath{\itbox{m}},\ensuremath{\itbox{C}},\ensuremath{\itbox{L{\ensuremath{'}}}})
    \land (\forall \ensuremath{\itbox{L{\ensuremath{'}}{\ensuremath{'}}}}, \ensuremath{\itbox{L{\ensuremath{'}}{\ensuremath{'}}{\ensuremath{'}}}} \text{ s.t. } \ensuremath{\itbox{L{\ensuremath{'}}}} \LEQ_w \ensuremath{\itbox{L{\ensuremath{'}}{\ensuremath{'}}}} \land \ensuremath{\itbox{L{\ensuremath{'}}{\ensuremath{'}} swappable}} \land \ensuremath{\itbox{L{\ensuremath{'}}{\ensuremath{'}}{\ensuremath{'}}}} \LEQ_w \ensuremath{\itbox{L{\ensuremath{'}}{\ensuremath{'}}}}. \ensuremath{\itbox{proceed}} \not \in \pmbody(\ensuremath{\itbox{m}},\ensuremath{\itbox{C}},\ensuremath{\itbox{L{\ensuremath{'}}{\ensuremath{'}}{\ensuremath{'}}}})))$.

    We show \(\npr(\ensuremath{\itbox{m}},\ensuremath{\itbox{C}}, \LSet)\) by case analysis.  In the case
    (1), we have $\mtype(\ensuremath{\itbox{m}},\ensuremath{\itbox{C}},\emptyset, \LSet)$ defined by
    \lemref{sublayer-mtype}.  The case (2) is also easy: if \(\ensuremath{\itbox{L{\ensuremath{'}}}} \in \LSet_{rm}\), then
    \(\ensuremath{\itbox{L{\ensuremath{'}}}} \in \LSet\); otherwise, \(\ensuremath{\itbox{proceed}} \not \in \pmbody(\ensuremath{\itbox{m}},\ensuremath{\itbox{C}},\ensuremath{\itbox{L}_{a}\itbox{}})\) because
    \(\ensuremath{\itbox{L{\ensuremath{'}}}} \LEQ_w \ensuremath{\itbox{L}_{{sw}}\itbox{}}\) and \(\ensuremath{\itbox{L}_{a}\itbox{}} \LEQ_w \ensuremath{\itbox{L}_{{sw}}\itbox{}}\) and \(\ensuremath{\itbox{L}_{{sw}}\itbox{ swappable}}\), hence
    \(\npr(\ensuremath{\itbox{m}}, \ensuremath{\itbox{C}}, \LSet)\). \qed

  \end{rneqncase}
  
\end{proof}

\begin{lemmaapp}[\lemref{def:wp}]\label{lem:wp}
  If $\set{\ensuremath{\itbox{}\overline{\itbox{L}}\itbox{}}}\, \WF$ and
  $\mtype(\ensuremath{\itbox{m}},\ensuremath{\itbox{C}},\set{\ensuremath{\itbox{}\overline{\itbox{L}}\itbox{}}},\set{\ensuremath{\itbox{}\overline{\itbox{L}}\itbox{}}}) = \ensuremath{\itbox{}\overline{\itbox{T}}\itbox{\(\rightarrow\)T}_{0}\itbox{}}$, then
  $\ndp(\ensuremath{\itbox{m}},\ensuremath{\itbox{C}},\ensuremath{\itbox{}\overline{\itbox{L}}\itbox{}},\ensuremath{\itbox{}\overline{\itbox{L}}\itbox{}})$.
\end{lemmaapp}

\begin{proof}
  By induction on the length of \(\ensuremath{\itbox{C}} \LEQ \ensuremath{\itbox{D}} \LEQ \cdots \ensuremath{\itbox{Object}}\).
  The case where the length is zero is trivial.
  \begin{eqncase}{
      \ensuremath{\itbox{C \(\triangleleft\) D}} \andalso
      \mtype(\ensuremath{\itbox{m}},\ensuremath{\itbox{D}},\set{\ensuremath{\itbox{}\overline{\itbox{L}}\itbox{}}}) \undf
    }
    \begin{sloppypar}\noindent
    By $\set{\ensuremath{\itbox{}\overline{\itbox{L}}\itbox{}}} \, \WF$ and \lemref{layer-set-wf2}, we have
    $\mtype(\ensuremath{\itbox{m}},\ensuremath{\itbox{C}},\emptyset, \set{\ensuremath{\itbox{}\overline{\itbox{L}}\itbox{}}})$ is defined or
    $\exists \ensuremath{\itbox{L}_{1}\itbox{}} \in \set{\ensuremath{\itbox{}\overline{\itbox{L}}\itbox{}}}. \ensuremath{\itbox{proceed}} \not \in \pmbody(\ensuremath{\itbox{m}},\ensuremath{\itbox{C}},\ensuremath{\itbox{L}_{1}\itbox{}})$.
    If $\mtype(\ensuremath{\itbox{m}},\ensuremath{\itbox{C}},\emptyset, \set{\ensuremath{\itbox{}\overline{\itbox{L}}\itbox{}}})$ is defined, class \ensuremath{\itbox{C}}
    must have the definition of method \ensuremath{\itbox{m}} since
    $\mtype(\ensuremath{\itbox{m}},\ensuremath{\itbox{D}},\set{\ensuremath{\itbox{}\overline{\itbox{L}}\itbox{}}})$ is undefined, and so \rn{NDP-Class}
    finishes the case.  In the other case, \rn{NDP-Layer} finishes the case.
    \end{sloppypar}
  \end{eqncase}
  
  \begin{eqncase}{
      \ensuremath{\itbox{C \(\triangleleft\) D}} \andalso
      \mtype(\ensuremath{\itbox{m}},\ensuremath{\itbox{D}},\set{\ensuremath{\itbox{}\overline{\itbox{L}}\itbox{}}}) \text{ defined}
    }
    By the induction hypothesis, $\ndp(\ensuremath{\itbox{m}},\ensuremath{\itbox{D}},\ensuremath{\itbox{}\overline{\itbox{L}}\itbox{}},\ensuremath{\itbox{}\overline{\itbox{L}}\itbox{}})$ holds. 
    Then, \rn{NDP-Super} finishes the case. \qed
  \end{eqncase}
\end{proof}

\begin{lemmaapp}[\lemref{def:wp_mtype}]\label{lem:wp_mtype}
  If $\ndp(\ensuremath{\itbox{m}},\ensuremath{\itbox{C}},\ensuremath{\itbox{}\overline{\itbox{L}}\itbox{{\ensuremath{'}}}},\ensuremath{\itbox{}\overline{\itbox{L}}\itbox{}})$, then $\mtype(\ensuremath{\itbox{m}},\ensuremath{\itbox{C}},\set{\ensuremath{\itbox{}\overline{\itbox{L}}\itbox{{\ensuremath{'}}}}}, \set{\ensuremath{\itbox{}\overline{\itbox{L}}\itbox{}}}) = \ensuremath{\itbox{}\overline{\itbox{T}}\itbox{\(\rightarrow\)T}_{0}\itbox{}}$ for
  some \ensuremath{\itbox{}\overline{\itbox{T}}\itbox{}} and \ensuremath{\itbox{T}_{0}\itbox{}}.
\end{lemmaapp}

\begin{proof}
  By induction on $\ndp(\ensuremath{\itbox{m}},\ensuremath{\itbox{C}},\ensuremath{\itbox{}\overline{\itbox{L}}\itbox{{\ensuremath{'}}}},\ensuremath{\itbox{}\overline{\itbox{L}}\itbox{}})$. \qed
\end{proof}

\begin{lemmaapp}\label{lem:inh_pmbody}
  If $\ensuremath{\itbox{L.C.m}}; \LSet; \Gp \ensuremath{\itbox{e}} : \ensuremath{\itbox{T}}$ and $\ensuremath{\itbox{L{\ensuremath{'}} \(\Leq\)}_{w}\itbox{ L}}$, then
  $\ensuremath{\itbox{L{\ensuremath{'}}.C.m}}; \LSet; \Gp \ensuremath{\itbox{e}} : \ensuremath{\itbox{T}}$.
\end{lemmaapp}

\begin{proof}
  Suppose that $\ensuremath{\itbox{L req }} \LSet_0$ and $\ensuremath{\itbox{L{\ensuremath{'}} req }} \LSet_1$.  
  Since \ensuremath{\itbox{L}} and \ensuremath{\itbox{L{\ensuremath{'}}}} are well-formed, $\LSet_1 \LEQ_w \LSet_0$.  We
  proceed by induction on $\ensuremath{\itbox{L.C.m}}; \LSet; \Gp \ensuremath{\itbox{e}} : \ensuremath{\itbox{T}}$.
  We show only main cases.

  \begin{rneqncase}{T-SuperP}{
      \ensuremath{\itbox{e}} = \ensuremath{\itbox{super.m{\ensuremath{'}}(}\overline{\itbox{e}}\itbox{)}} \andalso
      \ensuremath{\itbox{class C \(\triangleleft\) D}} \andalso
      \ensuremath{\itbox{L req }} \LSet_0 \\
      \mtype(\ensuremath{\itbox{m{\ensuremath{'}}}},\ensuremath{\itbox{D}},\LSet_0 \cup \set{\ensuremath{\itbox{L}}}) = \ensuremath{\itbox{}\overline{\itbox{T}}\itbox{\(\rightarrow\)T}} \andalso
      \ensuremath{\itbox{L.C.m}}; \LSet; \Gp \ensuremath{\itbox{}\overline{\itbox{e}}\itbox{}} :
      \ensuremath{\itbox{}\overline{\itbox{S}}\itbox{}} \andalso \ensuremath{\itbox{}\overline{\itbox{S}}\itbox{}} \LEQ \ensuremath{\itbox{}\overline{\itbox{T}}\itbox{}}
    }
    
    Since \ensuremath{\itbox{L{\ensuremath{'}} \(\Leq\)}_{w}\itbox{ L}} and $\LSet_1 \LEQ_w \LSet_0$, we have $\LSet_1
    \cup \set{\ensuremath{\itbox{L{\ensuremath{'}}}}} \LEQ_w \LSet_0 \cup \set{\ensuremath{\itbox{L}}}$.  Then, by
    \lemref{sublayer-mtype}, $\mtype(\ensuremath{\itbox{m{\ensuremath{'}}}}, \ensuremath{\itbox{D}}, \LSet_1 \cup
    \set{\ensuremath{\itbox{L{\ensuremath{'}}}}}) = \ensuremath{\itbox{}\overline{\itbox{T}}\itbox{\(\rightarrow\)T}}$.  The induction hypothesis and
    \rn{T-SuperP} finish the case.
  \end{rneqncase}

  \begin{rneqncase}{T-Proceed}{
      \ensuremath{\itbox{e}} = \ensuremath{\itbox{proceed(}\overline{\itbox{e}}\itbox{)}} \andalso
      \ensuremath{\itbox{L req }} \LSet_0 \andalso
      \mtype(\ensuremath{\itbox{m}},\ensuremath{\itbox{C}},\LSet_0,\LSet_0 \cup \set{\ensuremath{\itbox{L}}}) = \ensuremath{\itbox{}\overline{\itbox{T}}\itbox{\(\rightarrow\)T}} \\
      \ensuremath{\itbox{L.C.m}}; \LGp \ensuremath{\itbox{}\overline{\itbox{e}}\itbox{}} : \ensuremath{\itbox{}\overline{\itbox{S}}\itbox{}} \andalso
      \ensuremath{\itbox{}\overline{\itbox{S}}\itbox{}} \LEQ \ensuremath{\itbox{}\overline{\itbox{T}}\itbox{}}
    }
    Since $\LSet_1 \LEQ_w \LSet_0$, we have $\LSet_1 \cup \set{\ensuremath{\itbox{L{\ensuremath{'}}}}}
    \LEQ_w \LSet_0 \cup \set{\ensuremath{\itbox{L}}}$. Then, by \lemref{sublayer-mtype},
    $\mtype(\ensuremath{\itbox{m}}, \ensuremath{\itbox{C}},\LSet_1,\LSet_1 \cup \set{\ensuremath{\itbox{L{\ensuremath{'}}}}}) = \ensuremath{\itbox{}\overline{\itbox{T}}\itbox{\(\rightarrow\)T}}$.
    The induction hypothesis and \rn{T-Proceed} finish the case.
  \end{rneqncase}

  \begin{rneqncase}{T-SuperProceed}{
      \ensuremath{\itbox{e}} = \ensuremath{\itbox{superproceed(}\overline{\itbox{e}}\itbox{)}} \andalso
      \ensuremath{\itbox{L \(\triangleleft\) L{\ensuremath{'}}{\ensuremath{'}}}}\andalso 
      \pmtype(\ensuremath{\itbox{m}},\ensuremath{\itbox{C}},\ensuremath{\itbox{L{\ensuremath{'}}{\ensuremath{'}}}}) = \ensuremath{\itbox{}\overline{\itbox{T}}\itbox{ \(\rightarrow\) T}} \\
      \ensuremath{\itbox{L.C.m}}; \LSet; \Gp \ensuremath{\itbox{}\overline{\itbox{e}}\itbox{}} : \ensuremath{\itbox{}\overline{\itbox{S}}\itbox{}} \andalso
      \ensuremath{\itbox{}\overline{\itbox{S}}\itbox{}} \LEQ \ensuremath{\itbox{}\overline{\itbox{T}}\itbox{}}      
    }
    We have that for some \ensuremath{\itbox{L{\ensuremath{'}}{\ensuremath{'}}{\ensuremath{'}}}}, \ensuremath{\itbox{L{\ensuremath{'}} \(\triangleleft\) L{\ensuremath{'}}{\ensuremath{'}}{\ensuremath{'}}}}.  Then, $\ensuremath{\itbox{L{\ensuremath{'}}{\ensuremath{'}}{\ensuremath{'}}}} \LEQ_w
    \ensuremath{\itbox{L{\ensuremath{'}}{\ensuremath{'}}}}$ and $\pmtype(\ensuremath{\itbox{m}},\ensuremath{\itbox{C}},\ensuremath{\itbox{L{\ensuremath{'}}{\ensuremath{'}}{\ensuremath{'}}}}) = \ensuremath{\itbox{}\overline{\itbox{T}}\itbox{ \(\rightarrow\) T}}$ by \lemref{sublayer-pmtype}.
    The induction hypothesis and \rn{T-SuperProceed} finish the case. \qed
  \end{rneqncase}    
\end{proof}

\begin{lemmaapp}[Inversion for partial method body, \lemref{def:pmbody}]\label{lem:pmbody}
  If $\pmbody(\ensuremath{\itbox{m}},\ensuremath{\itbox{C}},\ensuremath{\itbox{L}}) = \ensuremath{\itbox{}\overline{\itbox{x}}\itbox{.e}_{0}\itbox{ in L{\ensuremath{'}}}}$ and $\ensuremath{\itbox{L req }} \LSet$ and
  $\pmtype(\ensuremath{\itbox{m}},\ensuremath{\itbox{C}},\ensuremath{\itbox{L}}) = \ensuremath{\itbox{}\overline{\itbox{T}}\itbox{ \(\rightarrow\) T}_{0}\itbox{}}$, then $\ensuremath{\itbox{L.C.m}}; \LSet \cup
  \set{\ensuremath{\itbox{L}}}; \ensuremath{\itbox{}\overline{\itbox{x}}\itbox{}}:\ensuremath{\itbox{}\overline{\itbox{T}}\itbox{}}, \ensuremath{\itbox{this}}:\ensuremath{\itbox{C}} \p \ensuremath{\itbox{e}_{0}\itbox{}} : \ensuremath{\itbox{S}_{0}\itbox{}}$ for some $\ensuremath{\itbox{S}_{0}\itbox{}}
  \LEQ_w \ensuremath{\itbox{T}_{0}\itbox{}}$.
\end{lemmaapp}

\begin{proof}
  By induction on $\pmbody(\ensuremath{\itbox{m}},\ensuremath{\itbox{C}},\ensuremath{\itbox{L}}) = \ensuremath{\itbox{}\overline{\itbox{x}}\itbox{.e}_{0}\itbox{ in L{\ensuremath{'}}}}$.

  \begin{rneqncase}{PMB-Super}{
      \LT(\ensuremath{\itbox{L}})(\ensuremath{\itbox{C.m}}) \undf \andalso
      \ensuremath{\itbox{L \(\triangleleft\) L{\ensuremath{'}}{\ensuremath{'}}}} \andalso
      \pmbody(\ensuremath{\itbox{m}},\ensuremath{\itbox{C}},\ensuremath{\itbox{L{\ensuremath{'}}{\ensuremath{'}}}}) = \ensuremath{\itbox{}\overline{\itbox{x}}\itbox{.e}_{0}\itbox{ in L{\ensuremath{'}}}}
    }
    By $\pmtype(\ensuremath{\itbox{m}},\ensuremath{\itbox{C}},\ensuremath{\itbox{L}}) = \ensuremath{\itbox{}\overline{\itbox{T}}\itbox{ \(\rightarrow\) T}_{0}\itbox{}}$ and \rn{PMT-Super}, it
    must be the case that $\pmtype(\ensuremath{\itbox{m}},\ensuremath{\itbox{C}},\ensuremath{\itbox{L{\ensuremath{'}}{\ensuremath{'}}}}) = \ensuremath{\itbox{}\overline{\itbox{T}}\itbox{ \(\rightarrow\) T}_{0}\itbox{}}$.  By
    the induction hypothesis,
    $$\ensuremath{\itbox{L{\ensuremath{'}}{\ensuremath{'}}.C.m}}; \LSet \cup \set{\ensuremath{\itbox{L{\ensuremath{'}}{\ensuremath{'}}}}}; \ensuremath{\itbox{}\overline{\itbox{x}}\itbox{}}:\ensuremath{\itbox{}\overline{\itbox{T}}\itbox{}}, \ensuremath{\itbox{this}}:\ensuremath{\itbox{C}} \p
    \ensuremath{\itbox{e}_{0}\itbox{}} : \ensuremath{\itbox{S}_{0}\itbox{}}$$ for some $\ensuremath{\itbox{S}_{0}\itbox{}} \LEQ_w \ensuremath{\itbox{T}_{0}\itbox{}}$.
    Lemmas~\ref{lem:narrowing} and \ref{lem:inh_pmbody} finish the
    case.
  \end{rneqncase}
  
  \begin{rneqncase}{PMB-Layer}{
      \LT(\ensuremath{\itbox{L}})(\ensuremath{\itbox{C.m}}) = \ensuremath{\itbox{T}_{0}\itbox{ C.m(}\overline{\itbox{T}}\itbox{ }\overline{\itbox{x}}\itbox{){\char'173} return e; {\char'175}}} \andalso
      \ensuremath{\itbox{L{\ensuremath{'}}}} = \ensuremath{\itbox{L}}
    }
    By \rn{T-PMethod}, it must be the case that
    $$\ensuremath{\itbox{L.C.m}}; \LSet \cup \set{\ensuremath{\itbox{L}}}; \ensuremath{\itbox{}\overline{\itbox{x}}\itbox{}}:\ensuremath{\itbox{}\overline{\itbox{T}}\itbox{}}, \ensuremath{\itbox{this}}:\ensuremath{\itbox{C}} \p \ensuremath{\itbox{e}_{0}\itbox{}} : \ensuremath{\itbox{S}_{0}\itbox{}}$$
    for some \ensuremath{\itbox{S}_{0}\itbox{}} s.t.\ \ensuremath{\itbox{S}_{0}\itbox{ \(\Leq\)}_{w}\itbox{ T}_{0}\itbox{}}, finishing the case. \qed
  \end{rneqncase}  
\end{proof}

\begin{lemmaapp}[Substitution for \ensuremath{\itbox{super}}, \ensuremath{\itbox{proceed}} and \ensuremath{\itbox{superproceed}}, \lemref{def:substitution-super}]
  \label{lem:substitution-super} \
  \begin{enumerate}
  \item 
    \begin{sloppypar}
      If \(\bullet; \LSet; \Gp \ensuremath{\itbox{new C}_{0}\itbox{(}\overline{\itbox{v}}\itbox{)}} : \ensuremath{\itbox{C}_{0}\itbox{}}\) and 
      $\ensuremath{\itbox{L.C.m}}; \LSet; \Gp \ensuremath{\itbox{e}} : \ensuremath{\itbox{T}}$ and
      $\ensuremath{\itbox{C}_{0}\itbox{.m}} \p \ensuremath{\itbox{<C,(}\overline{\itbox{L}}\itbox{{\ensuremath{'}};L{\ensuremath{'}}{\ensuremath{'}}),}\overline{\itbox{L}}\itbox{> ok}}$ and
      \ensuremath{\itbox{C \(\triangleleft\) D}} and \ensuremath{\itbox{L{\ensuremath{'}}{\ensuremath{'}} \(\Leq\)}_{w}\itbox{ L \(\triangleleft\) L{\ensuremath{'}}}} and
      $\LSet \LEQ_{sw} \set{\ensuremath{\itbox{}\overline{\itbox{L}}\itbox{}}}$ and
      $\ensuremath{\itbox{proceed}} \in \ensuremath{\itbox{e}} \implies \ndp(\ensuremath{\itbox{m}},\ensuremath{\itbox{C}},\ensuremath{\itbox{}\overline{\itbox{L}}\itbox{{\ensuremath{'}}}},\ensuremath{\itbox{}\overline{\itbox{L}}\itbox{}})$,
      then $\bullet; \LSet; \Gp S\ensuremath{\itbox{e}} : \ensuremath{\itbox{T}}$
      where
      \[S = 
      \left[\begin{array}{l@{/}l}
          \ensuremath{\itbox{new C}_{0}\itbox{(}\overline{\itbox{v}}\itbox{)<C,}\overline{\itbox{L}}\itbox{{\ensuremath{'}},}\overline{\itbox{L}}\itbox{>.m}} & \ensuremath{\itbox{proceed}}, \\
          \ensuremath{\itbox{new C}_{0}\itbox{(}\overline{\itbox{v}}\itbox{)<D,}\overline{\itbox{L}}\itbox{,}\overline{\itbox{L}}\itbox{>}} &  \ensuremath{\itbox{super}}, \\
          \ensuremath{\itbox{new C}_{0}\itbox{(}\overline{\itbox{v}}\itbox{)<C,L{\ensuremath{'}},(}\overline{\itbox{L}}\itbox{{\ensuremath{'}};L{\ensuremath{'}}{\ensuremath{'}}),}\overline{\itbox{L}}\itbox{>.m}} &  \ensuremath{\itbox{superproceed}}
        \end{array}\right].\]
    \end{sloppypar}

  \item If \(\bullet; \LSet; \Gp \ensuremath{\itbox{new C}_{0}\itbox{(}\overline{\itbox{v}}\itbox{)}} : \ensuremath{\itbox{C}_{0}\itbox{}}\) and 
    $\ensuremath{\itbox{C.m}}; \LSet; \Gp \ensuremath{\itbox{e}} : \ensuremath{\itbox{T}}$ and 
    $\ensuremath{\itbox{C}_{0}\itbox{.m}} \p \ensuremath{\itbox{<C,}\overline{\itbox{L}}\itbox{{\ensuremath{'}},}\overline{\itbox{L}}\itbox{> ok}}$ and
    \ensuremath{\itbox{C \(\triangleleft\) D}} and $\LSet \LEQ_{sw} \set{\ensuremath{\itbox{}\overline{\itbox{L}}\itbox{}}}$, then
    \(\bullet; \LSet; \Gp [\ensuremath{\itbox{new C}_{0}\itbox{(}\overline{\itbox{v}}\itbox{)<D,}\overline{\itbox{L}}\itbox{,}\overline{\itbox{L}}\itbox{>}} / \ensuremath{\itbox{super}}] \ensuremath{\itbox{e}} : \ensuremath{\itbox{T}}\).
  \end{enumerate} 
\end{lemmaapp}

\begin{proof}
  \begin{enumerate}
  \item By induction on $\ensuremath{\itbox{L.C.m}}; \LSet; \Gp \ensuremath{\itbox{e}} : \ensuremath{\itbox{T}}$ with case
    analysis on the last typing rule used.  We show main cases
    below.
    \begin{rncase}{T-SuperB} Cannot happen.
      
    \end{rncase}
    \begin{rneqncase}{T-SuperP}{
        \ensuremath{\itbox{e}} = \ensuremath{\itbox{super.m{\ensuremath{'}}(}\overline{\itbox{e}}\itbox{)}} &
        \mtype(\ensuremath{\itbox{m{\ensuremath{'}}}}, \ensuremath{\itbox{D}}, \LSet'\cup \set{\ensuremath{\itbox{L}}}) = \ensuremath{\itbox{}\overline{\itbox{T}}\itbox{{\ensuremath{'}}\(\rightarrow\)T}} \\
        \ensuremath{\itbox{L.C.m}}; \LSet; \Gp \ensuremath{\itbox{}\overline{\itbox{e}}\itbox{}} : \ensuremath{\itbox{}\overline{\itbox{S}}\itbox{{\ensuremath{'}}}} &
        \ensuremath{\itbox{L req }}\LSet' \andalso
        \ensuremath{\itbox{}\overline{\itbox{S}}\itbox{{\ensuremath{'}}}} \LEQ \ensuremath{\itbox{}\overline{\itbox{T}}\itbox{{\ensuremath{'}}}}
      }
      It suffices to show that $\bullet; \LSet; \Gp %
      \ensuremath{\itbox{new C}_{0}\itbox{(}\overline{\itbox{v}}\itbox{)<D,}\overline{\itbox{L}}\itbox{,}\overline{\itbox{L}}\itbox{>.m{\ensuremath{'}}(}}S\ensuremath{\itbox{}\overline{\itbox{e}}\itbox{)}} : \ensuremath{\itbox{T}}$.
      By assumption, we have \(\bullet; \LSet; \Gp \ensuremath{\itbox{new C}_{0}\itbox{(}\overline{\itbox{v}}\itbox{)}} : \ensuremath{\itbox{C}_{0}\itbox{}}\).
      Next, we show $\ensuremath{\itbox{C}_{0}\itbox{.m}} \p \ensuremath{\itbox{<D,}\overline{\itbox{L}}\itbox{,}\overline{\itbox{L}}\itbox{> ok}}$.
      By $\ensuremath{\itbox{C}_{0}\itbox{.m}} \p \ensuremath{\itbox{<C,}\overline{\itbox{L}}\itbox{{\ensuremath{'}},}\overline{\itbox{L}}\itbox{> ok}}$, 
      we have $\ensuremath{\itbox{C}_{0}\itbox{}} \LEQ \ensuremath{\itbox{C}}$, from which $\ensuremath{\itbox{C}_{0}\itbox{}} \LEQ \ensuremath{\itbox{D}}$ follows, and
      $\set{\ensuremath{\itbox{}\overline{\itbox{L}}\itbox{}}}\, \WF$.  
      By \lemref{layer-set-wf1} and \(\ensuremath{\itbox{L{\ensuremath{'}}{\ensuremath{'}}}} \in \set{\ensuremath{\itbox{}\overline{\itbox{L}}\itbox{}}}\) and
      \(\ensuremath{\itbox{L{\ensuremath{'}}{\ensuremath{'}}}} \LEQ_w \ensuremath{\itbox{L}}\), for any \ensuremath{\itbox{L}_{1}\itbox{}} such that
      \ensuremath{\itbox{L req L}_{1}\itbox{}}, there exists \(\ensuremath{\itbox{L}_{2}\itbox{}} \in \set{\ensuremath{\itbox{}\overline{\itbox{L}}\itbox{}}}\) such that
      \ensuremath{\itbox{L}_{2}\itbox{ \(\Leq\)}_{w}\itbox{ L}_{1}\itbox{}}; so, \(\set{\ensuremath{\itbox{}\overline{\itbox{L}}\itbox{}}} \LEQ_w \LSet'\cup \set{\ensuremath{\itbox{L}}}\).
      Then, by $\mtype(\ensuremath{\itbox{m{\ensuremath{'}}}}, \ensuremath{\itbox{D}}, \LSet'\cup \set{\ensuremath{\itbox{L}}}) = \ensuremath{\itbox{}\overline{\itbox{T}}\itbox{{\ensuremath{'}}\(\rightarrow\)T}}$ and
      \lemref{sublayer-mtype}, we have $\mtype(\ensuremath{\itbox{m{\ensuremath{'}}}}, \ensuremath{\itbox{D}}, \set{\ensuremath{\itbox{}\overline{\itbox{L}}\itbox{}}}) =
      \ensuremath{\itbox{}\overline{\itbox{T}}\itbox{{\ensuremath{'}}\(\rightarrow\)T}}$;
      moreover, by \lemref{wp}, $\ndp(\ensuremath{\itbox{m}}, \ensuremath{\itbox{D}}, \ensuremath{\itbox{}\overline{\itbox{L}}\itbox{}}, \ensuremath{\itbox{}\overline{\itbox{L}}\itbox{}})$.
      So,  $\ensuremath{\itbox{C}_{0}\itbox{.m}} \p \ensuremath{\itbox{<D,}\overline{\itbox{L}}\itbox{,}\overline{\itbox{L}}\itbox{> ok}}$.
      By the induction hypothesis, we have $\bullet; \LSet;
      \Gp S\ensuremath{\itbox{}\overline{\itbox{e}}\itbox{}} : \ensuremath{\itbox{}\overline{\itbox{S}}\itbox{{\ensuremath{'}}}}$ and, by assumption, $\ensuremath{\itbox{}\overline{\itbox{S}}\itbox{{\ensuremath{'}}}} \LEQ \ensuremath{\itbox{}\overline{\itbox{T}}\itbox{{\ensuremath{'}}}}$.  
      Finally, \rn{T-InvkA} finishes the case.

    \end{rneqncase}
	
    \begin{rneqncase}{T-Proceed}{
        \ensuremath{\itbox{e}} = \ensuremath{\itbox{proceed(}\overline{\itbox{e}}\itbox{)}} &
        \mtype(\ensuremath{\itbox{m}}, \ensuremath{\itbox{C}}, \LSet', \LSet' \cup \set{\ensuremath{\itbox{L}}}) = \ensuremath{\itbox{}\overline{\itbox{T}}\itbox{{\ensuremath{'}}\(\rightarrow\)T}} \\
        \ensuremath{\itbox{L.C.m}}; \LSet; \Gp \ensuremath{\itbox{}\overline{\itbox{e}}\itbox{}} : \ensuremath{\itbox{}\overline{\itbox{S}}\itbox{{\ensuremath{'}}}} &
        \ensuremath{\itbox{L req }}\LSet' \andalso
        \ensuremath{\itbox{}\overline{\itbox{S}}\itbox{{\ensuremath{'}}}} \LEQ \ensuremath{\itbox{}\overline{\itbox{T}}\itbox{{\ensuremath{'}}}}
      }
      It suffices to show that
      $\bullet; \LSet; \Gp \ensuremath{\itbox{new C}_{0}\itbox{<C,}\overline{\itbox{L}}\itbox{{\ensuremath{'}},}\overline{\itbox{L}}\itbox{>(}\overline{\itbox{v}}\itbox{).m(}}S\ensuremath{\itbox{}\overline{\itbox{e}}\itbox{)}} : \ensuremath{\itbox{T}}$.
      By assumption, we have \(\bullet; \LSet; \Gp \ensuremath{\itbox{new C}_{0}\itbox{(}\overline{\itbox{v}}\itbox{)}} : \ensuremath{\itbox{C}_{0}\itbox{}}\).
      Since $\ensuremath{\itbox{proceed}} \in \ensuremath{\itbox{e}}$, we have $\ndp(\ensuremath{\itbox{m}},\ensuremath{\itbox{C}},\ensuremath{\itbox{}\overline{\itbox{L}}\itbox{{\ensuremath{'}}}},\ensuremath{\itbox{}\overline{\itbox{L}}\itbox{}})$, from which
      $\ensuremath{\itbox{C}_{0}\itbox{.m}} \p \ensuremath{\itbox{<C,}\overline{\itbox{L}}\itbox{{\ensuremath{'}},}\overline{\itbox{L}}\itbox{> ok}}$ and follow.  By Lemmas~\ref{lem:wp_mtype} and
      \ref{lem:mtype_layer},
      we have $\mtype(\ensuremath{\itbox{m}}, \ensuremath{\itbox{C}}, \set{\ensuremath{\itbox{}\overline{\itbox{L}}\itbox{{\ensuremath{'}}}}}, \set{\ensuremath{\itbox{}\overline{\itbox{L}}\itbox{}}}) = \ensuremath{\itbox{}\overline{\itbox{T}}\itbox{{\ensuremath{'}}\(\rightarrow\)T}}$, too.
      By the induction hypothesis, we have $\bullet; \LSet;
      \Gp S\ensuremath{\itbox{}\overline{\itbox{e}}\itbox{}} : \ensuremath{\itbox{}\overline{\itbox{S}}\itbox{{\ensuremath{'}}}}$ and, by assumption, $\ensuremath{\itbox{}\overline{\itbox{S}}\itbox{{\ensuremath{'}}}} \LEQ \ensuremath{\itbox{}\overline{\itbox{T}}\itbox{{\ensuremath{'}}}}$.  
      Finally, \rn{T-InvkA} finishes the case.
    \end{rneqncase}

    \begin{rneqncase}{T-SuperProceed}{
        \ensuremath{\itbox{e}} = \ensuremath{\itbox{superproceed(}\overline{\itbox{e}}\itbox{)}} &
        \ensuremath{\itbox{L \(\triangleleft\) L{\ensuremath{'}}}} \\
        \pmtype(\ensuremath{\itbox{m}},\ensuremath{\itbox{C}},\ensuremath{\itbox{L{\ensuremath{'}}}}) = \ensuremath{\itbox{}\overline{\itbox{T}}\itbox{{\ensuremath{'}}\(\rightarrow\)T}} \\
        \ensuremath{\itbox{L.C.m}}; \LSet; \Gp \ensuremath{\itbox{}\overline{\itbox{e}}\itbox{}} : \ensuremath{\itbox{}\overline{\itbox{S}}\itbox{{\ensuremath{'}}}} &
        \ensuremath{\itbox{}\overline{\itbox{S}}\itbox{{\ensuremath{'}}}} \LEQ \ensuremath{\itbox{}\overline{\itbox{T}}\itbox{{\ensuremath{'}}}}
      }
      It suffices to show that
      $\bullet; \LSet; \Gp \ensuremath{\itbox{new C}_{0}\itbox{<C,L{\ensuremath{'}},(}\overline{\itbox{L}}\itbox{{\ensuremath{'}};L{\ensuremath{'}}{\ensuremath{'}}),}\overline{\itbox{L}}\itbox{>.m(}}S\ensuremath{\itbox{}\overline{\itbox{e}}\itbox{)}} :
      \ensuremath{\itbox{T}}$.

      By assumption, we have \(\bullet; \LSet; \Gp \ensuremath{\itbox{new C}_{0}\itbox{(}\overline{\itbox{v}}\itbox{)}} : \ensuremath{\itbox{C}_{0}\itbox{}}\)
      and $\ensuremath{\itbox{C}_{0}\itbox{.m}} \p \ensuremath{\itbox{<C,(}\overline{\itbox{L}}\itbox{{\ensuremath{'}};L{\ensuremath{'}}{\ensuremath{'}}),}\overline{\itbox{L}}\itbox{> ok}}$ and  $\ensuremath{\itbox{L{\ensuremath{'}}{\ensuremath{'}}}} \LEQ_w \ensuremath{\itbox{L{\ensuremath{'}}}}$.
      Also, $\pmtype(\ensuremath{\itbox{m}}, \ensuremath{\itbox{C}}, \ensuremath{\itbox{L{\ensuremath{'}}}}) = \ensuremath{\itbox{}\overline{\itbox{T}}\itbox{{\ensuremath{'}}\(\rightarrow\)T}}$, by assumption.
      By the induction hypothesis, we have $\bullet; \LSet;
      \Gp S\ensuremath{\itbox{}\overline{\itbox{e}}\itbox{}} : \ensuremath{\itbox{}\overline{\itbox{S}}\itbox{{\ensuremath{'}}}}$ and, by assumption, $\ensuremath{\itbox{}\overline{\itbox{S}}\itbox{{\ensuremath{'}}}} \LEQ \ensuremath{\itbox{}\overline{\itbox{T}}\itbox{{\ensuremath{'}}}}$.  
      Finally, \rn{T-InvkAL} finishes the case.
    \end{rneqncase}

    \begin{rneqncase}{T-With}{
        \ensuremath{\itbox{e}} = \ensuremath{\itbox{with e}_{l}\itbox{ e}_{0}\itbox{}} &
        \ensuremath{\itbox{L.C.m}}; \LSet; \Gp \ensuremath{\itbox{e}_{l}\itbox{}} : \ensuremath{\itbox{L}} &
        \ensuremath{\itbox{L req }} \LSet_0 \\
        \LSet \LEQ_w \LSet_0 &
        \ensuremath{\itbox{L.C.m}}; \LSet \cup \set{\ensuremath{\itbox{L}}}; \Gp \ensuremath{\itbox{e}_{0}\itbox{}} : \ensuremath{\itbox{T}}
      }
      Since $\LSet \LEQ_{sw} \set{\ensuremath{\itbox{}\overline{\itbox{L}}\itbox{}}}$, we have $\LSet \cup
      \set{\ensuremath{\itbox{L}}} \LEQ_{sw} \set{\ensuremath{\itbox{}\overline{\itbox{L}}\itbox{}}}$ by \rn{LSSW-Intro}.
      By the induction hypothesis,
      \(\bullet; \LSet; \Gp S\ensuremath{\itbox{e}_{l}\itbox{}} : \ensuremath{\itbox{L}}\)
      and \(\bullet; \LSet \cup \set{\ensuremath{\itbox{L}}}; \Gp S\ensuremath{\itbox{e}_{0}\itbox{}} : \ensuremath{\itbox{T}}\).
      \rn{T-With} finishes the case.
    \end{rneqncase}

    \begin{rneqncase}{T-Swap}{
        \ensuremath{\itbox{e}} = \ensuremath{\itbox{swap (e}_{l}\itbox{,L}_{{sw}}\itbox{) e}_{0}\itbox{}} &
        \ensuremath{\itbox{L.C.m}}; \LSet; \Gp \ensuremath{\itbox{e}_{l}\itbox{}} : \ensuremath{\itbox{L}} &
        \ensuremath{\itbox{L req }} \LSet_0 \\
        \ensuremath{\itbox{L}_{{sw}}\itbox{ swappable}} &
        \ensuremath{\itbox{L}} \LEQ_w \ensuremath{\itbox{L}_{{sw}}\itbox{}} \\
        \LSet_{rm} = \LSet \setminus \set{\ensuremath{\itbox{L{\ensuremath{'}}}} \mid \ensuremath{\itbox{L{\ensuremath{'}}}} \LEQ_w \ensuremath{\itbox{L}_{{sw}}\itbox{}}} &
        \LSet_{rm} \LEQ_w \LSet_0 \\
        \ensuremath{\itbox{L.C.m}}; \LSet_{rm} \cup \set{\ensuremath{\itbox{L}}}; \Gp \ensuremath{\itbox{e}_{0}\itbox{}} : \ensuremath{\itbox{T}}
      }

      Since $\LSet \LEQ_{sw} \set{\ensuremath{\itbox{}\overline{\itbox{L}}\itbox{}}}$, we have $\LSet_{rm} \cup
      \set{\ensuremath{\itbox{L}}} \LEQ_{sw} \set{\ensuremath{\itbox{}\overline{\itbox{L}}\itbox{}}}$ by \rn{LSSW-Intro}.  By the
      induction hypothesis, \(\bullet; \LSet; \Gp S\ensuremath{\itbox{e}_{l}\itbox{}} : \ensuremath{\itbox{L}}\) and
      \(\bullet; \LSet_{rm} \cup \set{\ensuremath{\itbox{L}}}; \Gp S\ensuremath{\itbox{e}_{0}\itbox{}} : \ensuremath{\itbox{T}}\).
      \rn{T-Swap} finishes the case.
    \end{rneqncase}


  \item \ By induction on $\ensuremath{\itbox{C.m}}; \LSet; \Gp \ensuremath{\itbox{e}} : \ensuremath{\itbox{T}_{0}\itbox{}}$ with
    case analysis on the last typing rule used.  We show only main
    cases below (note that none of the cases \rn{T-Proceed} and \rn{T-SuperP} and \rn{T-SuperProceed} can happen).
      \begin{rneqncase}{T-SuperB}{
        \ensuremath{\itbox{e}} = \ensuremath{\itbox{super.m{\ensuremath{'}}(}\overline{\itbox{e}}\itbox{)}} &
        \mtype(\ensuremath{\itbox{m{\ensuremath{'}}}}, \ensuremath{\itbox{D}}, \emptyset) = \ensuremath{\itbox{}\overline{\itbox{T}}\itbox{{\ensuremath{'}}\(\rightarrow\)T}_{0}\itbox{}} \\
        \ensuremath{\itbox{C.m}}; \LSet; \Gp \ensuremath{\itbox{}\overline{\itbox{e}}\itbox{}} : \ensuremath{\itbox{}\overline{\itbox{S}}\itbox{{\ensuremath{'}}}} &
        \ensuremath{\itbox{}\overline{\itbox{S}}\itbox{{\ensuremath{'}}}} \LEQ \ensuremath{\itbox{}\overline{\itbox{T}}\itbox{{\ensuremath{'}}}}
      }
      \begin{sloppypar}\noindent
      Let \(S = [\ensuremath{\itbox{new C}_{0}\itbox{(}\overline{\itbox{v}}\itbox{)<D,}\overline{\itbox{L}}\itbox{,}\overline{\itbox{L}}\itbox{>}} / \ensuremath{\itbox{super}}]\).  It suffices
      to show that $\bullet; \LSet; \Gp \ensuremath{\itbox{new C}_{0}\itbox{(}\overline{\itbox{v}}\itbox{)<D,}\overline{\itbox{L}}\itbox{,}\overline{\itbox{L}}\itbox{>.m{\ensuremath{'}}(}}S\ensuremath{\itbox{}\overline{\itbox{e}}\itbox{)}} : \ensuremath{\itbox{T}_{0}\itbox{}}$.
      By assumption, we have \(\bullet; \LSet; \Gp \ensuremath{\itbox{new C}_{0}\itbox{(}\overline{\itbox{v}}\itbox{)}} : \ensuremath{\itbox{C}_{0}\itbox{}}\).
      Next, we show $\ensuremath{\itbox{C}_{0}\itbox{.m}} \p \ensuremath{\itbox{<D,}\overline{\itbox{L}}\itbox{,}\overline{\itbox{L}}\itbox{> ok}}$.
      By $\ensuremath{\itbox{C}_{0}\itbox{.m}} \p \ensuremath{\itbox{<C,}\overline{\itbox{L}}\itbox{{\ensuremath{'}},}\overline{\itbox{L}}\itbox{> ok}}$, 
      we have $\ensuremath{\itbox{C}_{0}\itbox{}} \LEQ \ensuremath{\itbox{C}}$, from which $\ensuremath{\itbox{C}_{0}\itbox{}} \LEQ \ensuremath{\itbox{D}}$ follows, and
      $\set{\ensuremath{\itbox{}\overline{\itbox{L}}\itbox{}}}\, \WF$.  
      By $\mtype(\ensuremath{\itbox{m{\ensuremath{'}}}}, \ensuremath{\itbox{D}}, \emptyset) = \ensuremath{\itbox{}\overline{\itbox{T}}\itbox{{\ensuremath{'}}\(\rightarrow\)T}_{0}\itbox{}}$ and
      \lemref{sublayer-mtype}, we have $\mtype(\ensuremath{\itbox{m{\ensuremath{'}}}}, \ensuremath{\itbox{D}}, \set{\ensuremath{\itbox{}\overline{\itbox{L}}\itbox{}}}) =
      \ensuremath{\itbox{}\overline{\itbox{T}}\itbox{{\ensuremath{'}}\(\rightarrow\)T}_{0}\itbox{}}$;
      moreover, by \lemref{wp}, $\ndp(\ensuremath{\itbox{m}}, \ensuremath{\itbox{D}}, \ensuremath{\itbox{}\overline{\itbox{L}}\itbox{}}, \ensuremath{\itbox{}\overline{\itbox{L}}\itbox{}})$.
      So,  $\ensuremath{\itbox{C}_{0}\itbox{.m}} \p \ensuremath{\itbox{<D,}\overline{\itbox{L}}\itbox{,}\overline{\itbox{L}}\itbox{> ok}}$.
      By the induction hypothesis, we have $\bullet; \LSet;
      \Gp S\ensuremath{\itbox{}\overline{\itbox{e}}\itbox{}} : \ensuremath{\itbox{}\overline{\itbox{S}}\itbox{{\ensuremath{'}}}}$ and, by assumption, $\ensuremath{\itbox{}\overline{\itbox{S}}\itbox{{\ensuremath{'}}}} \LEQ \ensuremath{\itbox{}\overline{\itbox{T}}\itbox{{\ensuremath{'}}}}$.  
      Finally, \rn{T-InvkA} finishes the case.\qed
      \end{sloppypar}
    \end{rneqncase}
  \end{enumerate}
\end{proof}

\begin{lemmaapp}[Inversion for method body, \lemref{def:mbody-mtype}]\label{lem:mbody-mtype}
  \begin{sloppypar}
    Suppose 
    \(\set{\ensuremath{\itbox{}\overline{\itbox{L}}\itbox{}}}\, \WF\) and 
    $\mbody(\ensuremath{\itbox{m}}, \ensuremath{\itbox{C}}, \ensuremath{\itbox{}\overline{\itbox{L}}\itbox{{\ensuremath{'}}}}, \ensuremath{\itbox{}\overline{\itbox{L}}\itbox{}}) = \ensuremath{\itbox{}\overline{\itbox{x}}\itbox{.e}_{0}\itbox{}} \IN \ensuremath{\itbox{C{\ensuremath{'}}}},\ensuremath{\itbox{}\overline{\itbox{L}}\itbox{{\ensuremath{'}}{\ensuremath{'}}}}$ and
    $\mtype(\ensuremath{\itbox{m}}, \ensuremath{\itbox{C}}, \set{\ensuremath{\itbox{}\overline{\itbox{L}}\itbox{{\ensuremath{'}}}}}, \set{\ensuremath{\itbox{}\overline{\itbox{L}}\itbox{}}}) = \ensuremath{\itbox{}\overline{\itbox{T}}\itbox{\(\rightarrow\)T}_{0}\itbox{}}$ and
    $\ndp(\ensuremath{\itbox{m}},\ensuremath{\itbox{C}},\ensuremath{\itbox{}\overline{\itbox{L}}\itbox{{\ensuremath{'}}}},\ensuremath{\itbox{}\overline{\itbox{L}}\itbox{}})$.
    \begin{enumerate}
    \item If $\ensuremath{\itbox{}\overline{\itbox{L}}\itbox{{\ensuremath{'}}{\ensuremath{'}}}} = \ensuremath{\itbox{}\overline{\itbox{L}}\itbox{{\ensuremath{'}}{\ensuremath{'}}{\ensuremath{'}};L}_{0}\itbox{}}$, then 
	  $\ensuremath{\itbox{L}_{0}\itbox{ req }}\LSet$ and 
	  $\ensuremath{\itbox{L}_{0}\itbox{.C{\ensuremath{'}}.m}}; \LSet \cup \set{\ensuremath{\itbox{L}_{0}\itbox{}}}; \ensuremath{\itbox{}\overline{\itbox{x}}\itbox{}}: \ensuremath{\itbox{}\overline{\itbox{T}}\itbox{}}, \ensuremath{\itbox{this}}:\ensuremath{\itbox{C{\ensuremath{'}}}} \p \ensuremath{\itbox{e}_{0}\itbox{}} : \ensuremath{\itbox{U}_{0}\itbox{}}$ and
	  $\ensuremath{\itbox{C}} \LEQ \ensuremath{\itbox{C{\ensuremath{'}}}}$ and 
	  $\ensuremath{\itbox{U}_{0}\itbox{}} \LEQ \ensuremath{\itbox{T}_{0}\itbox{}}$ and
      $\ndp(\ensuremath{\itbox{m}},\ensuremath{\itbox{C{\ensuremath{'}}}},\ensuremath{\itbox{}\overline{\itbox{L}}\itbox{{\ensuremath{'}}{\ensuremath{'}}}},\ensuremath{\itbox{}\overline{\itbox{L}}\itbox{}})$ for some $\LSet$ and \ensuremath{\itbox{U}_{0}\itbox{}}. 

    \item If $\ensuremath{\itbox{}\overline{\itbox{L}}\itbox{{\ensuremath{'}}{\ensuremath{'}}}} = \bullet$, then $\ensuremath{\itbox{C{\ensuremath{'}}.m}}; \emptyset; \ensuremath{\itbox{}\overline{\itbox{x}}\itbox{}}:\ensuremath{\itbox{}\overline{\itbox{T}}\itbox{}},
	  \ensuremath{\itbox{this}}:\ensuremath{\itbox{C{\ensuremath{'}}}} \p \ensuremath{\itbox{e}_{0}\itbox{}} : \ensuremath{\itbox{U}_{0}\itbox{}}$ and 
	  $\ensuremath{\itbox{C}} \LEQ \ensuremath{\itbox{C{\ensuremath{'}}}}$ and 
	  $\ensuremath{\itbox{U}_{0}\itbox{}} \LEQ \ensuremath{\itbox{T}_{0}\itbox{}}$ and
      $\ndp(\ensuremath{\itbox{m}},\ensuremath{\itbox{C{\ensuremath{'}}}},\bullet,\ensuremath{\itbox{}\overline{\itbox{L}}\itbox{}})$ for some \ensuremath{\itbox{U}_{0}\itbox{}}.
    \end{enumerate}
    
  \end{sloppypar}
\end{lemmaapp}

\begin{proof}
  Both 1 and 2 are proved simultaneously by induction on $\mbody(\ensuremath{\itbox{m}}, \ensuremath{\itbox{C}}, \ensuremath{\itbox{}\overline{\itbox{L}}\itbox{{\ensuremath{'}}}}, \ensuremath{\itbox{}\overline{\itbox{L}}\itbox{}}) = \ensuremath{\itbox{}\overline{\itbox{x}}\itbox{.e}_{0}\itbox{}} \IN \ensuremath{\itbox{C{\ensuremath{'}}}}, \ensuremath{\itbox{}\overline{\itbox{L}}\itbox{{\ensuremath{'}}{\ensuremath{'}}}}$.

  \begin{rneqncase}{MB-Class}{
      \ensuremath{\itbox{class C \(\triangleleft\) D {\char'173}... S}_{0}\itbox{ m(}\overline{\itbox{S}}\itbox{ }\overline{\itbox{x}}\itbox{){\char'173}return e}_{0}\itbox{;{\char'175} ...{\char'175}}} \\
      \ensuremath{\itbox{C{\ensuremath{'}}}} = \ensuremath{\itbox{C}} \andalso      
      \ensuremath{\itbox{}\overline{\itbox{L}}\itbox{{\ensuremath{'}}}} = \bullet \andalso
      \ensuremath{\itbox{}\overline{\itbox{L}}\itbox{{\ensuremath{'}}{\ensuremath{'}}}} = \bullet
    }
    By \rn{T-Class}, \rn{T-Method}, \rn{MT-Class}, 
    it must be the case that 
    $$
    \begin{bcpcasearray}
      \ensuremath{\itbox{T}_{0}\itbox{}}, \ensuremath{\itbox{}\overline{\itbox{T}}\itbox{}} = \ensuremath{\itbox{S}_{0}\itbox{}}, \ensuremath{\itbox{}\overline{\itbox{S}}\itbox{}} &
      \ensuremath{\itbox{C.m}}; \emptyset; \ensuremath{\itbox{}\overline{\itbox{x}}\itbox{}}:\ensuremath{\itbox{}\overline{\itbox{T}}\itbox{}}, \ensuremath{\itbox{this}} : \ensuremath{\itbox{C}} \p \ensuremath{\itbox{e}_{0}\itbox{}} : \ensuremath{\itbox{U}_{0}\itbox{}} &
      \ensuremath{\itbox{U}_{0}\itbox{}} \LEQ \ensuremath{\itbox{T}_{0}\itbox{}}
    \end{bcpcasearray}
    $$
    for some \ensuremath{\itbox{U}_{0}\itbox{}}.  We have $\ndp(\ensuremath{\itbox{m}},\ensuremath{\itbox{C}},\bullet,\ensuremath{\itbox{}\overline{\itbox{L}}\itbox{}})$ by \rn{NDP-Class},
    finishing the case.

  \end{rneqncase}

  \begin{rneqncase}{MB-Layer}{
      \pmbody(\ensuremath{\itbox{m}}, \ensuremath{\itbox{C}}, \ensuremath{\itbox{L}_{0}\itbox{}}) = \ensuremath{\itbox{}\overline{\itbox{x}}\itbox{.e}_{0}\itbox{ in L}_{1}\itbox{}} \andalso
      \ensuremath{\itbox{C{\ensuremath{'}}}} = \ensuremath{\itbox{C}} \andalso
      \ensuremath{\itbox{}\overline{\itbox{L}}\itbox{{\ensuremath{'}}{\ensuremath{'}}}} = \ensuremath{\itbox{}\overline{\itbox{L}}\itbox{{\ensuremath{'}}}}
    }
    \begin{sloppypar}\noindent
    By the definition of \pmbody, there exists some \ensuremath{\itbox{L}_{1}\itbox{}} such that
    $\LT(\ensuremath{\itbox{L}_{1}\itbox{}})(\ensuremath{\itbox{C.m}}) = \ensuremath{\itbox{S}_{0}\itbox{ C.m(}\overline{\itbox{S}}\itbox{ }\overline{\itbox{x}}\itbox{){\char'173} return e; {\char'175}}}$ and $\ensuremath{\itbox{L}_{0}\itbox{}} \LEQ_w \ensuremath{\itbox{L}_{1}\itbox{}}$.
    By \rn{T-PMethod}, it must be the case that 
    $$
    \begin{bcpcasearray}
      \ensuremath{\itbox{T}_{0}\itbox{}}, \ensuremath{\itbox{}\overline{\itbox{T}}\itbox{}} = \ensuremath{\itbox{S}_{0}\itbox{}}, \ensuremath{\itbox{}\overline{\itbox{S}}\itbox{}} &
      \ensuremath{\itbox{L}_{1}\itbox{ req }}\LSet_1 &
      \ensuremath{\itbox{L}_{1}\itbox{.C.m}}; \LSet_1 \cup \set{\ensuremath{\itbox{L}_{1}\itbox{}}}; \ensuremath{\itbox{}\overline{\itbox{x}}\itbox{}} : \ensuremath{\itbox{}\overline{\itbox{T}}\itbox{}}, \ensuremath{\itbox{this}} : \ensuremath{\itbox{C}} \p \ensuremath{\itbox{e}_{0}\itbox{}} : \ensuremath{\itbox{U}_{0}\itbox{}} &
      \ensuremath{\itbox{U}_{0}\itbox{}} \LEQ \ensuremath{\itbox{T}_{0}\itbox{}}
    \end{bcpcasearray}
    $$ for some \ensuremath{\itbox{U}_{0}\itbox{}} and \(\LSet_1\).  It is easy to show by induction on $\ensuremath{\itbox{L}_{0}\itbox{}}
    \LEQ_w \ensuremath{\itbox{L}_{1}\itbox{}}$ using \lemref{narrowing} and \rn{T-Layer} and
    \rn{T-LayerSW} that 
     $$\ensuremath{\itbox{L}_{0}\itbox{.C.m}}; \LSet \cup \set{\ensuremath{\itbox{L}_{0}\itbox{}}}; \ensuremath{\itbox{}\overline{\itbox{x}}\itbox{}} : \ensuremath{\itbox{}\overline{\itbox{T}}\itbox{}}, \ensuremath{\itbox{this}} : \ensuremath{\itbox{C}} \p
     \ensuremath{\itbox{e}_{0}\itbox{}} : \ensuremath{\itbox{U}_{0}\itbox{}}$$ for some $\LSet$ such that $\ensuremath{\itbox{L}_{0}\itbox{ req }}\LSet$.
    Finally, we have \ndp(\ensuremath{\itbox{m}},\ensuremath{\itbox{C{\ensuremath{'}}}},\ensuremath{\itbox{}\overline{\itbox{L}}\itbox{{\ensuremath{'}}{\ensuremath{'}}}},\ensuremath{\itbox{}\overline{\itbox{L}}\itbox{}}) by
    assumption, finishing the case.
    \end{sloppypar}
  \end{rneqncase}

  \begin{rneqncase}{MB-Super}{
      \ensuremath{\itbox{}\overline{\itbox{L}}\itbox{{\ensuremath{'}}}} = \bullet \andalso
      \ensuremath{\itbox{class C \(\triangleleft\) D {\char'173} ... }\overline{\itbox{M}}\itbox{ {\char'175}}} &
      \ensuremath{\itbox{m}} \not \in \ensuremath{\itbox{}\overline{\itbox{M}}\itbox{}} \\
      \mbody(\ensuremath{\itbox{m}},\ensuremath{\itbox{D}},\ensuremath{\itbox{}\overline{\itbox{L}}\itbox{}}, \ensuremath{\itbox{}\overline{\itbox{L}}\itbox{}}) = \ensuremath{\itbox{}\overline{\itbox{x}}\itbox{.e}_{0}\itbox{}} \IN \ensuremath{\itbox{C{\ensuremath{'}}}}, \ensuremath{\itbox{}\overline{\itbox{L}}\itbox{{\ensuremath{'}}{\ensuremath{'}}}}
    }
    By \rn{MT-Super}, it must be the case that
    $\mtype(\ensuremath{\itbox{m}},\ensuremath{\itbox{D}},\set{\ensuremath{\itbox{}\overline{\itbox{L}}\itbox{}}}, \set{\ensuremath{\itbox{}\overline{\itbox{L}}\itbox{}}}) = \ensuremath{\itbox{}\overline{\itbox{T}}\itbox{\(\rightarrow\)T}_{0}\itbox{}}$.  By
    \lemref{wp}, we have \ndp(\ensuremath{\itbox{m}},\ensuremath{\itbox{D}},\ensuremath{\itbox{}\overline{\itbox{L}}\itbox{}},\ensuremath{\itbox{}\overline{\itbox{L}}\itbox{}}).  The induction
    hypothesis and transitivity of subtyping finish the case.
  \end{rneqncase}

  \begin{rneqncase}{MB-NextLayer}{
      \ensuremath{\itbox{}\overline{\itbox{L}}\itbox{{\ensuremath{'}}}} = \ensuremath{\itbox{}\overline{\itbox{L}}\itbox{}_{b}\itbox{;L}_{1}\itbox{}} \andalso
      \pmbody(\ensuremath{\itbox{m}},\ensuremath{\itbox{C}},\ensuremath{\itbox{L}_{1}\itbox{}}) \undf \\
      \mbody(\ensuremath{\itbox{m}}, \ensuremath{\itbox{C}}, \ensuremath{\itbox{}\overline{\itbox{L}}\itbox{}_{b}\itbox{}}, \ensuremath{\itbox{}\overline{\itbox{L}}\itbox{}}) = \ensuremath{\itbox{}\overline{\itbox{x}}\itbox{.e}_{0}\itbox{}} \IN \ensuremath{\itbox{C{\ensuremath{'}}}}, \ensuremath{\itbox{}\overline{\itbox{L}}\itbox{{\ensuremath{'}}{\ensuremath{'}}}}
    }
    We show $\ndp(\ensuremath{\itbox{m}},\ensuremath{\itbox{C}},\ensuremath{\itbox{}\overline{\itbox{L}}\itbox{}_{b}\itbox{}},\ensuremath{\itbox{}\overline{\itbox{L}}\itbox{}})$ holds by case analysis
    on $\ndp(\ensuremath{\itbox{m}},\ensuremath{\itbox{C}},\ensuremath{\itbox{}\overline{\itbox{L}}\itbox{}_{b}\itbox{}};\ensuremath{\itbox{L}_{1}\itbox{}},\ensuremath{\itbox{}\overline{\itbox{L}}\itbox{}})$.  The cases \rn{NDP-Super}
    and \rn{NDP-Class} are easy.
    The case  \rn{NDP-Layer} is easy, too:
    since $\pmbody(\ensuremath{\itbox{m}},\ensuremath{\itbox{C}},\ensuremath{\itbox{L}_{1}\itbox{}}) \undf$, by \rn{NDP-Layer}, we have
    $\ndp(\ensuremath{\itbox{m}},\ensuremath{\itbox{C}},\ensuremath{\itbox{}\overline{\itbox{L}}\itbox{}_{b}\itbox{}},\ensuremath{\itbox{}\overline{\itbox{L}}\itbox{}})$.  Since $\pmbody(\ensuremath{\itbox{m}},\ensuremath{\itbox{C}},\ensuremath{\itbox{L}_{1}\itbox{}})$ is undefined
    and $\mtype(\ensuremath{\itbox{m}}, \ensuremath{\itbox{C}}, \set{\ensuremath{\itbox{}\overline{\itbox{L}}\itbox{{\ensuremath{'}}}}}, \set{\ensuremath{\itbox{}\overline{\itbox{L}}\itbox{}}}) = \ensuremath{\itbox{}\overline{\itbox{T}}\itbox{\(\rightarrow\)T}_{0}\itbox{}}$, it
    must be the case that $\mtype(\ensuremath{\itbox{m}}, \ensuremath{\itbox{C}}, \set{\ensuremath{\itbox{}\overline{\itbox{L}}\itbox{}_{b}\itbox{}}}, \set{\ensuremath{\itbox{}\overline{\itbox{L}}\itbox{}}})
    = \ensuremath{\itbox{}\overline{\itbox{T}}\itbox{\(\rightarrow\)T}_{0}\itbox{}}$.  Then, the induction hypothesis finishes the case. \qed
  \end{rneqncase}

\end{proof}

\begin{theoremapp}[Subject Reduction]\label{thm:subject-reduction}
  Suppose $\p (\CT, \LT)\ensuremath{\itbox{ ok}}$.  If $\bullet;
  \set{\ensuremath{\itbox{}\overline{\itbox{L}}\itbox{}}}; \Gp \ensuremath{\itbox{e}} : \ensuremath{\itbox{T}}$ and \set{\ensuremath{\itbox{}\overline{\itbox{L}}\itbox{}}}\, \WF \ and
  $\reduceto{\ensuremath{\itbox{}\overline{\itbox{L}}\itbox{}}}{\ensuremath{\itbox{e}}}{\ensuremath{\itbox{e{\ensuremath{'}}}}}$, then $\bullet; \set{\ensuremath{\itbox{}\overline{\itbox{L}}\itbox{}}}; \Gp \ensuremath{\itbox{e{\ensuremath{'}}}} :
  \ensuremath{\itbox{S}}$ for some \ensuremath{\itbox{S}} such that $\ensuremath{\itbox{S}} \LEQ \ensuremath{\itbox{T}}$.
\end{theoremapp}

\begin{proof}
 By induction on $\reduceto{\ensuremath{\itbox{}\overline{\itbox{L}}\itbox{}}}{\ensuremath{\itbox{e}}}{\ensuremath{\itbox{e{\ensuremath{'}}}}}$ with case analysis on
 the last reduction rule used.  We show only main cases.
 
 \begin{rneqncase}{R-Field}{
     \ensuremath{\itbox{e}} = \ensuremath{\itbox{new C}_{0}\itbox{(}\overline{\itbox{v}}\itbox{).f}_{i}\itbox{}} &
     \fields(\ensuremath{\itbox{C}_{0}\itbox{}}) = \ensuremath{\itbox{}\overline{\itbox{C}}\itbox{ }\overline{\itbox{f}}\itbox{}} &
     \ensuremath{\itbox{e{\ensuremath{'}}}} = \ensuremath{\itbox{v}_{i}\itbox{}}
   }
   By \rn{T-Field} and \rn{T-New}, it must be the case that
   $$
   \begin{bcpcasearray}
     \bullet; \set{\ensuremath{\itbox{}\overline{\itbox{L}}\itbox{}}}; \Gp \ensuremath{\itbox{}\overline{\itbox{v}}\itbox{}} : \ensuremath{\itbox{}\overline{\itbox{D}}\itbox{}} & 
     \ensuremath{\itbox{}\overline{\itbox{D}}\itbox{}} \LEQ \ensuremath{\itbox{}\overline{\itbox{C}}\itbox{}} &
     \ensuremath{\itbox{C}} = \ensuremath{\itbox{C}_{i}\itbox{}}
   \end{bcpcasearray}
   $$
   Then, we have $\bullet; \set{\ensuremath{\itbox{}\overline{\itbox{L}}\itbox{}}}; \Gp \ensuremath{\itbox{v}_{i}\itbox{}} : \ensuremath{\itbox{D}_{i}\itbox{}}$ and $\ensuremath{\itbox{D}_{i}\itbox{}}
   \LEQ \ensuremath{\itbox{C}_{i}\itbox{}}$, finishing the case.
 \end{rneqncase}

  \begin{rneqncase}{R-Invk}{
      \ensuremath{\itbox{e}} = \ensuremath{\itbox{new C}_{0}\itbox{(}\overline{\itbox{v}}\itbox{).m(}\overline{\itbox{w}}\itbox{)}} \\
      \reduceto{\ensuremath{\itbox{}\overline{\itbox{L}}\itbox{}}}{\ensuremath{\itbox{new C}_{0}\itbox{(}\overline{\itbox{v}}\itbox{)<C}_{0}\itbox{,}\overline{\itbox{L}}\itbox{,}\overline{\itbox{L}}\itbox{>.m(}\overline{\itbox{w}}\itbox{)}}}{\ensuremath{\itbox{e{\ensuremath{'}}}}}
    }
    By \rn{T-Invk} and \rn{T-New}, it must be the case that
    $$
    \begin{bcpcasearray}
      \bullet; \set{\ensuremath{\itbox{}\overline{\itbox{L}}\itbox{}}}; \Gp \ensuremath{\itbox{}\overline{\itbox{v}}\itbox{}} : \ensuremath{\itbox{}\overline{\itbox{S}}\itbox{}} &
      \fields(\ensuremath{\itbox{C}_{0}\itbox{}}) = \ensuremath{\itbox{}\overline{\itbox{T}}\itbox{ }\overline{\itbox{f}}\itbox{}} &
      \ensuremath{\itbox{}\overline{\itbox{S}}\itbox{}} \LEQ \ensuremath{\itbox{}\overline{\itbox{T}}\itbox{}}  \\
      \mtype(\ensuremath{\itbox{m}}, \ensuremath{\itbox{C}_{0}\itbox{}}, \set{\ensuremath{\itbox{}\overline{\itbox{L}}\itbox{}}}) = \ensuremath{\itbox{}\overline{\itbox{T}}\itbox{{\ensuremath{'}}\(\rightarrow\)T}} &
      \bullet; \set{\ensuremath{\itbox{}\overline{\itbox{L}}\itbox{}}}; \Gp \ensuremath{\itbox{}\overline{\itbox{w}}\itbox{}} : \ensuremath{\itbox{}\overline{\itbox{S}}\itbox{{\ensuremath{'}}}} &
      \ensuremath{\itbox{}\overline{\itbox{S}}\itbox{{\ensuremath{'}}}} \LEQ \ensuremath{\itbox{}\overline{\itbox{T}}\itbox{{\ensuremath{'}}}} .
    \end{bcpcasearray}
    $$
    By \lemref{wp}, $\ndp(\ensuremath{\itbox{m}},\ensuremath{\itbox{C}_{0}\itbox{}},\ensuremath{\itbox{}\overline{\itbox{L}}\itbox{}},\ensuremath{\itbox{}\overline{\itbox{L}}\itbox{}})$ and so     $\ensuremath{\itbox{C}_{0}\itbox{.m}} \p
    \ensuremath{\itbox{<C}_{0}\itbox{,}\overline{\itbox{L}}\itbox{,}\overline{\itbox{L}}\itbox{> ok}}$ holds.  Since $\set{\ensuremath{\itbox{}\overline{\itbox{L}}\itbox{}}} \LEQ_{sw} \set{\ensuremath{\itbox{}\overline{\itbox{L}}\itbox{}}}$, we have
    $$\bullet; \set{\ensuremath{\itbox{}\overline{\itbox{L}}\itbox{}}}; \Gp \ensuremath{\itbox{new C}_{0}\itbox{(}\overline{\itbox{v}}\itbox{)<C}_{0}\itbox{,}\overline{\itbox{L}}\itbox{,}\overline{\itbox{L}}\itbox{>.m(}\overline{\itbox{w}}\itbox{)}} : \ensuremath{\itbox{T}}$$
    by \rn{T-InvkA}.  By the induction hypothesis, $\bullet;
    \set{\ensuremath{\itbox{}\overline{\itbox{L}}\itbox{}}}; \Gp \ensuremath{\itbox{e{\ensuremath{'}}}} : \ensuremath{\itbox{S}}$ for some $\ensuremath{\itbox{S}} \LEQ \ensuremath{\itbox{T}}$, finishing the
    case.
  \end{rneqncase}

  \begin{rneqncase}{R-InvkP}{
    \ensuremath{\itbox{e}} = \ensuremath{\itbox{new C}_{0}\itbox{(}\overline{\itbox{v}}\itbox{)<C{\ensuremath{'}},}\overline{\itbox{L}}\itbox{{\ensuremath{'}}{\ensuremath{'}},}\overline{\itbox{L}}\itbox{{\ensuremath{'}}>.m(}\overline{\itbox{w}}\itbox{)}} \\
    \mbody(\ensuremath{\itbox{m}}, \ensuremath{\itbox{C{\ensuremath{'}}}}, \ensuremath{\itbox{}\overline{\itbox{L}}\itbox{{\ensuremath{'}}{\ensuremath{'}}}}, \ensuremath{\itbox{}\overline{\itbox{L}}\itbox{{\ensuremath{'}}}}) = \ensuremath{\itbox{}\overline{\itbox{x}}\itbox{.e}_{0}\itbox{}} \IN \ensuremath{\itbox{C{\ensuremath{'}}{\ensuremath{'}}}}, (\ensuremath{\itbox{}\overline{\itbox{L}}\itbox{{\ensuremath{'}}{\ensuremath{'}}{\ensuremath{'}}}};\ensuremath{\itbox{L}_{0}\itbox{}}) \\
    \ensuremath{\itbox{C{\ensuremath{'}}{\ensuremath{'}} \(\triangleleft\) D}} \\
    \ensuremath{\itbox{L}_{0}\itbox{ \(\triangleleft\) L}_{1}\itbox{}} \\
    \ensuremath{\itbox{e{\ensuremath{'}}}} = \left[
      \begin{array}{l@{/}l}
        \ensuremath{\itbox{new C}_{0}\itbox{(}\overline{\itbox{v}}\itbox{)}} & \ensuremath{\itbox{this}}, \\
        \ensuremath{\itbox{}\overline{\itbox{w}}\itbox{}} & \ensuremath{\itbox{}\overline{\itbox{x}}\itbox{}}, \\
        \ensuremath{\itbox{new C}_{0}\itbox{(}\overline{\itbox{v}}\itbox{)<C{\ensuremath{'}}{\ensuremath{'}},}\overline{\itbox{L}}\itbox{{\ensuremath{'}}{\ensuremath{'}}{\ensuremath{'}},}\overline{\itbox{L}}\itbox{{\ensuremath{'}}>.m}} & \ensuremath{\itbox{proceed}}, \\
        \ensuremath{\itbox{new C}_{0}\itbox{(}\overline{\itbox{v}}\itbox{)<D,}\overline{\itbox{L}}\itbox{{\ensuremath{'}},}\overline{\itbox{L}}\itbox{{\ensuremath{'}}>}} & \ensuremath{\itbox{super}}, \\
        \ensuremath{\itbox{new C}_{0}\itbox{(}\overline{\itbox{v}}\itbox{)<C{\ensuremath{'}}{\ensuremath{'}},L}_{1}\itbox{,}}(\ensuremath{\itbox{}\overline{\itbox{L}}\itbox{{\ensuremath{'}}{\ensuremath{'}}{\ensuremath{'}};L}_{0}\itbox{}})\ensuremath{\itbox{,}\overline{\itbox{L}}\itbox{{\ensuremath{'}}>.m}} & \ensuremath{\itbox{superproceed}}
      \end{array}
      \right]\ensuremath{\itbox{e}_{0}\itbox{}}
    }
    By \rn{T-InvkA}, it must be the case that 
    $$
    \begin{bcpcasearray}
      \bullet;\set{\ensuremath{\itbox{}\overline{\itbox{L}}\itbox{}}};\Gp \ensuremath{\itbox{new C}_{0}\itbox{(}\overline{\itbox{v}}\itbox{)}} : \ensuremath{\itbox{C}_{0}\itbox{}} &
      \ensuremath{\itbox{C}_{0}\itbox{.m}} \p \ensuremath{\itbox{<C{\ensuremath{'}},}\overline{\itbox{L}}\itbox{{\ensuremath{'}}{\ensuremath{'}},}\overline{\itbox{L}}\itbox{{\ensuremath{'}}> ok}} &
      \set{\ensuremath{\itbox{}\overline{\itbox{L}}\itbox{}}} \LEQ_{sw} \set{\ensuremath{\itbox{}\overline{\itbox{L}}\itbox{{\ensuremath{'}}}}} \\
      \mtype(\ensuremath{\itbox{m}},\ensuremath{\itbox{C{\ensuremath{'}}}},\set{\ensuremath{\itbox{}\overline{\itbox{L}}\itbox{{\ensuremath{'}}{\ensuremath{'}}}}}, \set{\ensuremath{\itbox{}\overline{\itbox{L}}\itbox{{\ensuremath{'}}}}}) = \ensuremath{\itbox{}\overline{\itbox{T}}\itbox{{\ensuremath{'}}\(\rightarrow\)T}} &
      \bullet;\set{\ensuremath{\itbox{}\overline{\itbox{L}}\itbox{}}};\Gp \ensuremath{\itbox{}\overline{\itbox{w}}\itbox{}} : \ensuremath{\itbox{}\overline{\itbox{S}}\itbox{{\ensuremath{'}}}} &
      \ensuremath{\itbox{}\overline{\itbox{S}}\itbox{{\ensuremath{'}}}} \LEQ \ensuremath{\itbox{}\overline{\itbox{T}}\itbox{{\ensuremath{'}}}}
    \end{bcpcasearray}
    $$
    for some \ensuremath{\itbox{}\overline{\itbox{T}}\itbox{{\ensuremath{'}}}} and \ensuremath{\itbox{}\overline{\itbox{S}}\itbox{{\ensuremath{'}}}}.

    By \lemref{mbody-mtype},
    $$
    \begin{bcpcasearray}
    \ensuremath{\itbox{L}_{0}\itbox{.C{\ensuremath{'}}{\ensuremath{'}}.m}}; \LSet \cup \set{\ensuremath{\itbox{L}_{0}\itbox{}}}; \ensuremath{\itbox{}\overline{\itbox{x}}\itbox{}}:\ensuremath{\itbox{}\overline{\itbox{T}}\itbox{}}, \ensuremath{\itbox{this}}: \ensuremath{\itbox{C{\ensuremath{'}}{\ensuremath{'}}}} \p \ensuremath{\itbox{e}_{0}\itbox{}} : \ensuremath{\itbox{S}} \\
    \ensuremath{\itbox{L}_{0}\itbox{ req }} \LSet \\
    \ensuremath{\itbox{C{\ensuremath{'}}}} \LEQ \ensuremath{\itbox{C{\ensuremath{'}}{\ensuremath{'}}}} \\
    \ensuremath{\itbox{S}} \LEQ \ensuremath{\itbox{T}} \\
    \ndp(\ensuremath{\itbox{m}},\ensuremath{\itbox{C{\ensuremath{'}}{\ensuremath{'}}}},\ensuremath{\itbox{(}\overline{\itbox{L}}\itbox{{\ensuremath{'}}{\ensuremath{'}}{\ensuremath{'}};L}_{0}\itbox{)}},\ensuremath{\itbox{}\overline{\itbox{L}}\itbox{{\ensuremath{'}}}})
    \end{bcpcasearray}
    $$
    and for some \(\LSet\) and \ensuremath{\itbox{S}}.

    By \rn{S-Trans}, $\ensuremath{\itbox{C}_{0}\itbox{}} \LEQ \ensuremath{\itbox{C{\ensuremath{'}}{\ensuremath{'}}}}$.
    From $\ndp(\ensuremath{\itbox{m}},\ensuremath{\itbox{C{\ensuremath{'}}{\ensuremath{'}}}},\ensuremath{\itbox{(}\overline{\itbox{L}}\itbox{{\ensuremath{'}}{\ensuremath{'}}{\ensuremath{'}};L}_{0}\itbox{)}},\ensuremath{\itbox{}\overline{\itbox{L}}\itbox{{\ensuremath{'}}}})$ and $\ensuremath{\itbox{C}_{0}\itbox{.m}} \p \ensuremath{\itbox{<C{\ensuremath{'}},}\overline{\itbox{L}}\itbox{{\ensuremath{'}}{\ensuremath{'}},}\overline{\itbox{L}}\itbox{{\ensuremath{'}}> ok}}$,
    it follows that $\ensuremath{\itbox{C}_{0}\itbox{.m}} \p \ensuremath{\itbox{<C{\ensuremath{'}}{\ensuremath{'}},}}(\ensuremath{\itbox{}\overline{\itbox{L}}\itbox{{\ensuremath{'}}{\ensuremath{'}}{\ensuremath{'}};L}_{0}\itbox{}})\ensuremath{\itbox{,}\overline{\itbox{L}}\itbox{{\ensuremath{'}}> ok}}$.  
    
    By $\set{\ensuremath{\itbox{}\overline{\itbox{L}}\itbox{{\ensuremath{'}}}}}\, \WF$ and \lemref{layer-set-wf1} and $\ensuremath{\itbox{L}_{0}\itbox{}} \in \ensuremath{\itbox{}\overline{\itbox{L}}\itbox{{\ensuremath{'}}}}$ and $\ensuremath{\itbox{L}_{0}\itbox{ req }}\LSet$,
    we have $\forall \ensuremath{\itbox{L}} \in \LSet, \exists \ensuremath{\itbox{L{\ensuremath{'}}}} \in
    \ensuremath{\itbox{}\overline{\itbox{L}}\itbox{{\ensuremath{'}}}}. \ensuremath{\itbox{L{\ensuremath{'}}}} \LEQ_w \ensuremath{\itbox{L}}$.  So, by \rn{LSS-Intro}, we have
    $\set{\ensuremath{\itbox{}\overline{\itbox{L}}\itbox{{\ensuremath{'}}}}} = \set{\ensuremath{\itbox{}\overline{\itbox{L}}\itbox{{\ensuremath{'}}}}} \cup \set{\ensuremath{\itbox{L}_{0}\itbox{}}} \LEQ_w \LSet \cup \set{\ensuremath{\itbox{L}_{0}\itbox{}}}$.  By this fact and
    $\set{\ensuremath{\itbox{}\overline{\itbox{L}}\itbox{}}} \LEQ_{sw} \set{\ensuremath{\itbox{}\overline{\itbox{L}}\itbox{{\ensuremath{'}}}}}$, we get $\set{\ensuremath{\itbox{}\overline{\itbox{L}}\itbox{}}} \LEQ_{sw}
    \LSet \cup \set{\ensuremath{\itbox{L}_{0}\itbox{}}}$.  By \lemref{narrowing},
    $$ 
    \ensuremath{\itbox{L}_{0}\itbox{.C{\ensuremath{'}}{\ensuremath{'}}.m}}; \set{\ensuremath{\itbox{}\overline{\itbox{L}}\itbox{}}}; \ensuremath{\itbox{}\overline{\itbox{x}}\itbox{}}:\ensuremath{\itbox{}\overline{\itbox{T}}\itbox{}}, \ensuremath{\itbox{this}}: \ensuremath{\itbox{C{\ensuremath{'}}{\ensuremath{'}}}} \p \ensuremath{\itbox{e}_{0}\itbox{}} : \ensuremath{\itbox{S}} 
    $$
    By $\ndp(\ensuremath{\itbox{m}},\ensuremath{\itbox{C{\ensuremath{'}}{\ensuremath{'}}}},\ensuremath{\itbox{(}\overline{\itbox{L}}\itbox{{\ensuremath{'}}{\ensuremath{'}}{\ensuremath{'}};L}_{0}\itbox{)}},\ensuremath{\itbox{}\overline{\itbox{L}}\itbox{{\ensuremath{'}}}})$ and the definition of \ndp,
    $\ensuremath{\itbox{proceed}} \in \ensuremath{\itbox{e}_{0}\itbox{}}$ implies $\ndp(\ensuremath{\itbox{m}},\ensuremath{\itbox{C{\ensuremath{'}}}},\ensuremath{\itbox{}\overline{\itbox{L}}\itbox{{\ensuremath{'}}{\ensuremath{'}}{\ensuremath{'}}}},\ensuremath{\itbox{}\overline{\itbox{L}}\itbox{{\ensuremath{'}}}})$.
    Then, by Lemmas~\ref{lem:value-strengthening} and \lemref{substitution-super}(1),
    \begin{eqnarray*}
    \bullet; \set{\ensuremath{\itbox{}\overline{\itbox{L}}\itbox{}}}; \ensuremath{\itbox{}\overline{\itbox{x}}\itbox{}}:\ensuremath{\itbox{}\overline{\itbox{T}}\itbox{}}, \ensuremath{\itbox{this}}: \ensuremath{\itbox{C{\ensuremath{'}}{\ensuremath{'}}}} \p 
    \left[\begin{array}{l@{/}l}
      \ensuremath{\itbox{new C}_{0}\itbox{(}\overline{\itbox{v}}\itbox{)<C{\ensuremath{'}}{\ensuremath{'}},}\overline{\itbox{L}}\itbox{{\ensuremath{'}}{\ensuremath{'}}{\ensuremath{'}},}\overline{\itbox{L}}\itbox{{\ensuremath{'}}>.m}} & \ensuremath{\itbox{proceed}}, \\
      \ensuremath{\itbox{new C}_{0}\itbox{(}\overline{\itbox{v}}\itbox{)<D,}\overline{\itbox{L}}\itbox{{\ensuremath{'}},}\overline{\itbox{L}}\itbox{{\ensuremath{'}}>}} & \ensuremath{\itbox{super}}, \\
      \ensuremath{\itbox{new C}_{0}\itbox{(}\overline{\itbox{v}}\itbox{)<C{\ensuremath{'}}{\ensuremath{'}},L}_{1}\itbox{,}}(\ensuremath{\itbox{}\overline{\itbox{L}}\itbox{{\ensuremath{'}}{\ensuremath{'}}{\ensuremath{'}};L}_{0}\itbox{}})\ensuremath{\itbox{,}\overline{\itbox{L}}\itbox{{\ensuremath{'}}>.m}} & \ensuremath{\itbox{superproceed}}
    \end{array}\right] \ensuremath{\itbox{e}_{0}\itbox{}} : \ensuremath{\itbox{S}} 
    \end{eqnarray*}
    By Lemmas~\ref{lem:value-strengthening}, \ref{lem:weakening} and
    \ref{lem:substitution}, $\bullet; \set{\ensuremath{\itbox{}\overline{\itbox{L}}\itbox{}}}; \Gp \ensuremath{\itbox{e{\ensuremath{'}}}} : \ensuremath{\itbox{S{\ensuremath{'}}}}$ for
    some $\ensuremath{\itbox{S{\ensuremath{'}}}} \LEQ \ensuremath{\itbox{S}}$.  By \rn{S-Trans}, $\ensuremath{\itbox{S{\ensuremath{'}}}} \LEQ \ensuremath{\itbox{T}}$, finishing
    the case.
  \end{rneqncase}

  \begin{rneqncase}{R-InvkSP}{
    \ensuremath{\itbox{e}} = \ensuremath{\itbox{new C}_{0}\itbox{(}\overline{\itbox{v}}\itbox{)<C{\ensuremath{'}},L}_{1}\itbox{,(}\overline{\itbox{L}}\itbox{{\ensuremath{'}}{\ensuremath{'}};L}_{0}\itbox{),}\overline{\itbox{L}}\itbox{{\ensuremath{'}}>.m(}\overline{\itbox{w}}\itbox{)}} \\
    \pmbody(\ensuremath{\itbox{m}}, \ensuremath{\itbox{C{\ensuremath{'}}}},\ensuremath{\itbox{L}_{1}\itbox{}}) = \ensuremath{\itbox{}\overline{\itbox{x}}\itbox{.e}_{0}\itbox{}} \IN \ensuremath{\itbox{L}_{2}\itbox{}}\\
    \ensuremath{\itbox{C{\ensuremath{'}} \(\triangleleft\) D}} \\
    \ensuremath{\itbox{L}_{2}\itbox{ \(\triangleleft\) L}_{3}\itbox{}} \\
    \ensuremath{\itbox{e{\ensuremath{'}}}} = \left[
      \begin{array}{l@{/}l}
        \ensuremath{\itbox{new C}_{0}\itbox{(}\overline{\itbox{v}}\itbox{)}} & \ensuremath{\itbox{this}} \\
        \ensuremath{\itbox{}\overline{\itbox{w}}\itbox{}} & \ensuremath{\itbox{}\overline{\itbox{x}}\itbox{}} \\
        \ensuremath{\itbox{new C}_{0}\itbox{(}\overline{\itbox{v}}\itbox{)<C{\ensuremath{'}}{\ensuremath{'}},}\overline{\itbox{L}}\itbox{{\ensuremath{'}}{\ensuremath{'}},}\overline{\itbox{L}}\itbox{{\ensuremath{'}}>.m}} & \ensuremath{\itbox{proceed}} \\
        \ensuremath{\itbox{new C}_{0}\itbox{(}\overline{\itbox{v}}\itbox{)<D,}\overline{\itbox{L}}\itbox{{\ensuremath{'}},}\overline{\itbox{L}}\itbox{{\ensuremath{'}}>}} & \ensuremath{\itbox{super}}, \\
        \ensuremath{\itbox{new C}_{0}\itbox{(}\overline{\itbox{v}}\itbox{)<C{\ensuremath{'}}{\ensuremath{'}},L}_{3}\itbox{,(}\overline{\itbox{L}}\itbox{{\ensuremath{'}}{\ensuremath{'}};L}_{0}\itbox{),}\overline{\itbox{L}}\itbox{{\ensuremath{'}}>.m}} & \ensuremath{\itbox{superproceed}}
      \end{array}
      \right]\ensuremath{\itbox{e}_{0}\itbox{}}
    }
    By \rn{T-InvkAL}, it must be the case that 
    $$
    \begin{bcpcasearray}
      \bullet;\set{\ensuremath{\itbox{}\overline{\itbox{L}}\itbox{}}};\Gp \ensuremath{\itbox{new C}_{0}\itbox{(}\overline{\itbox{v}}\itbox{)}} : \ensuremath{\itbox{C}_{0}\itbox{}} &
      \ensuremath{\itbox{C.m}} \p \ensuremath{\itbox{<C{\ensuremath{'}},(}\overline{\itbox{L}}\itbox{{\ensuremath{'}}{\ensuremath{'}};L}_{0}\itbox{),}\overline{\itbox{L}}\itbox{{\ensuremath{'}}> ok}} &
      \set{\ensuremath{\itbox{}\overline{\itbox{L}}\itbox{}}} \LEQ_{sw} \set{\ensuremath{\itbox{}\overline{\itbox{L}}\itbox{{\ensuremath{'}}}}} \\
      \ensuremath{\itbox{L}_{0}\itbox{}} \LEQ_w \ensuremath{\itbox{L}_{1}\itbox{}} &
      \pmtype(\ensuremath{\itbox{m}},\ensuremath{\itbox{C{\ensuremath{'}}}},\ensuremath{\itbox{L}_{1}\itbox{}}) = \ensuremath{\itbox{}\overline{\itbox{T}}\itbox{{\ensuremath{'}}\(\rightarrow\)T}} &
      \bullet;\set{\ensuremath{\itbox{}\overline{\itbox{L}}\itbox{}}};\Gp \ensuremath{\itbox{}\overline{\itbox{w}}\itbox{}} : \ensuremath{\itbox{}\overline{\itbox{S}}\itbox{{\ensuremath{'}}}} &
      \ensuremath{\itbox{}\overline{\itbox{S}}\itbox{{\ensuremath{'}}}} \LEQ \ensuremath{\itbox{}\overline{\itbox{T}}\itbox{{\ensuremath{'}}}}
    \end{bcpcasearray}
    $$
    for some \ensuremath{\itbox{}\overline{\itbox{T}}\itbox{{\ensuremath{'}}}} and \ensuremath{\itbox{}\overline{\itbox{S}}\itbox{{\ensuremath{'}}}}.
    Let $\LSet$ be the layer set such that $\ensuremath{\itbox{L}_{1}\itbox{ req }} \LSet$.  By \lemref{pmbody},
    $$
    \ensuremath{\itbox{L}_{1}\itbox{.C{\ensuremath{'}}.m}}; \LSet \cup \set{\ensuremath{\itbox{L}_{1}\itbox{}}}; \ensuremath{\itbox{}\overline{\itbox{x}}\itbox{}}:\ensuremath{\itbox{}\overline{\itbox{T}}\itbox{}}, \ensuremath{\itbox{this}}: \ensuremath{\itbox{C{\ensuremath{'}}}} \p \ensuremath{\itbox{e}_{0}\itbox{}} : \ensuremath{\itbox{S}} 
    $$
    and $\ensuremath{\itbox{S}} \LEQ \ensuremath{\itbox{T}}$ for some \ensuremath{\itbox{S}}.

    Since \(\ensuremath{\itbox{L}_{0}\itbox{}} \LEQ_w \ensuremath{\itbox{L}_{1}\itbox{}}\), \ensuremath{\itbox{L}_{0}\itbox{}} requires all the layers that \ensuremath{\itbox{L}_{1}\itbox{}}
    requires (including \(\LSet\)).  By $\set{\ensuremath{\itbox{}\overline{\itbox{L}}\itbox{{\ensuremath{'}}}}}\, \WF$ and \lemref{layer-set-wf1} and
    $\ensuremath{\itbox{L}_{0}\itbox{}} \in \ensuremath{\itbox{}\overline{\itbox{L}}\itbox{{\ensuremath{'}}}}$, we have
    \(\forall \ensuremath{\itbox{L}} \in \LSet, \exists \ensuremath{\itbox{L{\ensuremath{'}}}} \in \set{\ensuremath{\itbox{}\overline{\itbox{L}}\itbox{{\ensuremath{'}}}}}\) such that
    \(\ensuremath{\itbox{L{\ensuremath{'}}}} \LEQ_w \ensuremath{\itbox{L}}\).  So,
    $\set{\ensuremath{\itbox{}\overline{\itbox{L}}\itbox{{\ensuremath{'}}}}} = \set{\ensuremath{\itbox{}\overline{\itbox{L}}\itbox{{\ensuremath{'}}}}} \cup \set{\ensuremath{\itbox{L}_{0}\itbox{}}} \LEQ_w \LSet \cup
    \set{\ensuremath{\itbox{L}_{1}\itbox{}}}$.  By this and $\set{\ensuremath{\itbox{}\overline{\itbox{L}}\itbox{}}} \LEQ_{sw} \set{\ensuremath{\itbox{}\overline{\itbox{L}}\itbox{{\ensuremath{'}}}}}$, we
    have $\set{\ensuremath{\itbox{}\overline{\itbox{L}}\itbox{}}} \LEQ_{sw} \LSet \cup \set{\ensuremath{\itbox{L}_{1}\itbox{}}}$.  By
    \lemref{narrowing},
    $$ 
    \ensuremath{\itbox{L}_{1}\itbox{.C{\ensuremath{'}}.m}}; \set{\ensuremath{\itbox{}\overline{\itbox{L}}\itbox{}}}; \ensuremath{\itbox{}\overline{\itbox{x}}\itbox{}}:\ensuremath{\itbox{}\overline{\itbox{T}}\itbox{}}, \ensuremath{\itbox{this}}: \ensuremath{\itbox{C{\ensuremath{'}}}} \p \ensuremath{\itbox{e}_{0}\itbox{}} : \ensuremath{\itbox{S}}.
    $$
    By $\ndp(\ensuremath{\itbox{m}},\ensuremath{\itbox{C{\ensuremath{'}}}},\ensuremath{\itbox{(}\overline{\itbox{L}}\itbox{{\ensuremath{'}}{\ensuremath{'}};L}_{0}\itbox{)}},\ensuremath{\itbox{}\overline{\itbox{L}}\itbox{{\ensuremath{'}}}})$ (which follows from $\ensuremath{\itbox{C.m}} \p
    \ensuremath{\itbox{<C{\ensuremath{'}},(}\overline{\itbox{L}}\itbox{{\ensuremath{'}}{\ensuremath{'}};L}_{0}\itbox{),}\overline{\itbox{L}}\itbox{{\ensuremath{'}}> ok}}$) and the definition of \ndp, $\ensuremath{\itbox{proceed}}
    \in \ensuremath{\itbox{e}_{0}\itbox{}}$ implies $\ndp(\ensuremath{\itbox{m}},\ensuremath{\itbox{C{\ensuremath{'}}}},\ensuremath{\itbox{}\overline{\itbox{L}}\itbox{{\ensuremath{'}}{\ensuremath{'}}}},\ensuremath{\itbox{}\overline{\itbox{L}}\itbox{{\ensuremath{'}}}})$ holds.  Then,
    by Lemmas~\ref{lem:value-strengthening}, \ref{lem:weakening}, \ref{lem:substitution}
    and \ref{lem:substitution-super}(1), $\bullet; \set{\ensuremath{\itbox{}\overline{\itbox{L}}\itbox{}}}; \Gp \ensuremath{\itbox{e{\ensuremath{'}}}} :
    \ensuremath{\itbox{S{\ensuremath{'}}}}$ for some $\ensuremath{\itbox{S{\ensuremath{'}}}} \LEQ \ensuremath{\itbox{S}}$.  By \rn{S-Trans}, $\ensuremath{\itbox{S{\ensuremath{'}}}} \LEQ \ensuremath{\itbox{T}}$,
    finishing the case.
  \end{rneqncase}

  \begin{rneqncase}{R-InvkB}{
    \ensuremath{\itbox{e}} = \ensuremath{\itbox{new C}_{0}\itbox{(}\overline{\itbox{v}}\itbox{)<C{\ensuremath{'}},}\overline{\itbox{L}}\itbox{{\ensuremath{'}}{\ensuremath{'}},}\overline{\itbox{L}}\itbox{{\ensuremath{'}}>.m(}\overline{\itbox{w}}\itbox{)}} \\
    \mbody(\ensuremath{\itbox{m}}, \ensuremath{\itbox{C{\ensuremath{'}}}}, \ensuremath{\itbox{}\overline{\itbox{L}}\itbox{{\ensuremath{'}}{\ensuremath{'}}}}, \ensuremath{\itbox{}\overline{\itbox{L}}\itbox{{\ensuremath{'}}}}) = \ensuremath{\itbox{}\overline{\itbox{x}}\itbox{.e}_{0}\itbox{}} \IN \ensuremath{\itbox{C{\ensuremath{'}}{\ensuremath{'}}}}, \bullet \\
    \ensuremath{\itbox{C{\ensuremath{'}}{\ensuremath{'}} \(\triangleleft\) D}} \\
    \ensuremath{\itbox{e{\ensuremath{'}}}} = \left[
      \begin{array}{l@{/}l}
        \ensuremath{\itbox{new C}_{0}\itbox{(}\overline{\itbox{v}}\itbox{)}} & \ensuremath{\itbox{this}} \\
        \ensuremath{\itbox{}\overline{\itbox{w}}\itbox{}} & \ensuremath{\itbox{}\overline{\itbox{x}}\itbox{}} \\
        \ensuremath{\itbox{new C}_{0}\itbox{(}\overline{\itbox{v}}\itbox{)<D,}\overline{\itbox{L}}\itbox{{\ensuremath{'}},}\overline{\itbox{L}}\itbox{{\ensuremath{'}}>}} & \ensuremath{\itbox{super}}
      \end{array}
    \right]\ensuremath{\itbox{e}_{0}\itbox{}}
  }

    By \rn{T-InvkA}, it must be the case that 
    $$
    \begin{bcpcasearray}
      \bullet;\set{\ensuremath{\itbox{}\overline{\itbox{L}}\itbox{}}};\Gp \ensuremath{\itbox{new C}_{0}\itbox{(}\overline{\itbox{v}}\itbox{)}} : \ensuremath{\itbox{C}_{0}\itbox{}} &
      \ensuremath{\itbox{C}_{0}\itbox{.m}} \p \ensuremath{\itbox{<C{\ensuremath{'}},}\overline{\itbox{L}}\itbox{{\ensuremath{'}}{\ensuremath{'}},}\overline{\itbox{L}}\itbox{{\ensuremath{'}}> ok}} &
      \set{\ensuremath{\itbox{}\overline{\itbox{L}}\itbox{}}} \LEQ_{sw} \set{\ensuremath{\itbox{}\overline{\itbox{L}}\itbox{{\ensuremath{'}}}}} \\
      \mtype(\ensuremath{\itbox{m}},\ensuremath{\itbox{C{\ensuremath{'}}}},\set{\ensuremath{\itbox{}\overline{\itbox{L}}\itbox{{\ensuremath{'}}{\ensuremath{'}}}}}, \set{\ensuremath{\itbox{}\overline{\itbox{L}}\itbox{{\ensuremath{'}}}}}) = \ensuremath{\itbox{}\overline{\itbox{T}}\itbox{{\ensuremath{'}}\(\rightarrow\)T}} &
      \bullet;\set{\ensuremath{\itbox{}\overline{\itbox{L}}\itbox{}}};\Gp \ensuremath{\itbox{}\overline{\itbox{w}}\itbox{}} : \ensuremath{\itbox{}\overline{\itbox{S}}\itbox{{\ensuremath{'}}}} &
      \ensuremath{\itbox{}\overline{\itbox{S}}\itbox{{\ensuremath{'}}}} \LEQ \ensuremath{\itbox{}\overline{\itbox{T}}\itbox{{\ensuremath{'}}}}
    \end{bcpcasearray}
    $$
    for some \ensuremath{\itbox{}\overline{\itbox{T}}\itbox{{\ensuremath{'}}}} and \ensuremath{\itbox{}\overline{\itbox{S}}\itbox{{\ensuremath{'}}}}.
    By \lemref{mbody-mtype},
    $$
    \ensuremath{\itbox{C{\ensuremath{'}}{\ensuremath{'}}.m}}; \emptyset; \ensuremath{\itbox{}\overline{\itbox{x}}\itbox{}}:\ensuremath{\itbox{}\overline{\itbox{T}}\itbox{}}, \ensuremath{\itbox{this}}: \ensuremath{\itbox{C{\ensuremath{'}}{\ensuremath{'}}}} \p \ensuremath{\itbox{e}_{0}\itbox{}} : \ensuremath{\itbox{S}} 
    $$
    and $\ensuremath{\itbox{C{\ensuremath{'}}}} \LEQ \ensuremath{\itbox{C{\ensuremath{'}}{\ensuremath{'}}}}$ and $\ensuremath{\itbox{S}} \LEQ \ensuremath{\itbox{T}}$ and 
    $\ndp(\ensuremath{\itbox{m}},\ensuremath{\itbox{C{\ensuremath{'}}{\ensuremath{'}}}},\bullet,\ensuremath{\itbox{}\overline{\itbox{L}}\itbox{{\ensuremath{'}}}})$ for some \ensuremath{\itbox{S}}.
    By \rn{S-Trans}, $\ensuremath{\itbox{C}_{0}\itbox{}} \LEQ \ensuremath{\itbox{C{\ensuremath{'}}{\ensuremath{'}}}}$.
    By \lemref{weakening},
    $$
    \ensuremath{\itbox{C{\ensuremath{'}}{\ensuremath{'}}.m}}; \set{\ensuremath{\itbox{}\overline{\itbox{L}}\itbox{}}}; \ensuremath{\itbox{}\overline{\itbox{x}}\itbox{}}:\ensuremath{\itbox{}\overline{\itbox{T}}\itbox{}}, \ensuremath{\itbox{this}}: \ensuremath{\itbox{C{\ensuremath{'}}{\ensuremath{'}}}} \p \ensuremath{\itbox{e}_{0}\itbox{}} : \ensuremath{\itbox{S}} 
    $$
    By Lemmas~\ref{lem:value-strengthening}, \ref{lem:weakening}, \ref{lem:substitution}
    and \ref{lem:substitution-super}(2), $\bullet; \set{\ensuremath{\itbox{}\overline{\itbox{L}}\itbox{}}}; \Gp \ensuremath{\itbox{e{\ensuremath{'}}}}
    : \ensuremath{\itbox{S{\ensuremath{'}}}}$ for some $\ensuremath{\itbox{S{\ensuremath{'}}}} \LEQ \ensuremath{\itbox{S}}$.  By \rn{S-Trans}, $\ensuremath{\itbox{S{\ensuremath{'}}}} \LEQ
    \ensuremath{\itbox{T}}$, finishing the case.
  \end{rneqncase}

  \begin{rneqncase}{RC-With}{
      \ensuremath{\itbox{e}} = \ensuremath{\itbox{with new L() e}_{0}\itbox{}} &
      \ensuremath{\itbox{e{\ensuremath{'}}}} = \ensuremath{\itbox{with new L() e}_{0}\itbox{{\ensuremath{'}}}} \\
      \with(\ensuremath{\itbox{L}},\ensuremath{\itbox{}\overline{\itbox{L}}\itbox{}}) = \ensuremath{\itbox{}\overline{\itbox{L}}\itbox{{\ensuremath{'}}}} &
      \reduceto{\ensuremath{\itbox{}\overline{\itbox{L}}\itbox{{\ensuremath{'}}}}}{\ensuremath{\itbox{e}_{0}\itbox{}}}{\ensuremath{\itbox{e}_{0}\itbox{{\ensuremath{'}}}}}
    }

    By \rn{T-With}, it must be the case that
    $$
    \begin{bcpcasearray}
      \bullet; \set{\ensuremath{\itbox{}\overline{\itbox{L}}\itbox{}}}; \Gp \ensuremath{\itbox{new L()}} : \ensuremath{\itbox{L}} &
      \ensuremath{\itbox{L req }} \LSet &
      \set{\ensuremath{\itbox{}\overline{\itbox{L}}\itbox{}}} \LEQ_w \LSet &
      \bullet; \set{\ensuremath{\itbox{}\overline{\itbox{L}}\itbox{}}} \cup \set{\ensuremath{\itbox{L}}}; \Gp \ensuremath{\itbox{e}_{0}\itbox{}} : \ensuremath{\itbox{T}}
    \end{bcpcasearray}
    $$
    for some \(\LSet\).  Here,
    $\set{\ensuremath{\itbox{}\overline{\itbox{L}}\itbox{{\ensuremath{'}}}}} = \set{\ensuremath{\itbox{}\overline{\itbox{L}}\itbox{}}} \cup \set{\ensuremath{\itbox{L}}}\, \WF$ by \rn{Wf-With}.  By
    the induction hypothesis,
    $\bullet; \set{\ensuremath{\itbox{}\overline{\itbox{L}}\itbox{}}} \cup \set{\ensuremath{\itbox{L}}}; \Gp \ensuremath{\itbox{e}_{0}\itbox{{\ensuremath{'}}}} : \ensuremath{\itbox{S}}$ for some
    $\ensuremath{\itbox{S}} \LEQ \ensuremath{\itbox{T}}$.  By \rn{T-With},
    $\bullet; \set{\ensuremath{\itbox{}\overline{\itbox{L}}\itbox{}}}; \Gp \ensuremath{\itbox{e}_{0}\itbox{{\ensuremath{'}}}} : \ensuremath{\itbox{S}}$, finishing the case.
  \end{rneqncase}

  \begin{rneqncase}{RC-WithArg}{
      \ensuremath{\itbox{e}} = \ensuremath{\itbox{with e}_{l}\itbox{ e}_{0}\itbox{}} &
      \ensuremath{\itbox{e{\ensuremath{'}}}} = \ensuremath{\itbox{with e}_{l}\itbox{{\ensuremath{'}} e}_{0}\itbox{}} &
      \reduceto{\ensuremath{\itbox{}\overline{\itbox{L}}\itbox{}}}{\ensuremath{\itbox{e}_{l}\itbox{}}}{\ensuremath{\itbox{e}_{l}\itbox{{\ensuremath{'}}}}}
    }
    By \rn{T-With}, it must be the case that
    $$
    \begin{bcpcasearray}
      \bullet; \set{\ensuremath{\itbox{}\overline{\itbox{L}}\itbox{}}}; \Gp \ensuremath{\itbox{e}_{l}\itbox{}} : \ensuremath{\itbox{L}} &
      \ensuremath{\itbox{L req }} \LSet &
      \set{\ensuremath{\itbox{}\overline{\itbox{L}}\itbox{}}} \LEQ_w \LSet &
      \bullet; \set{\ensuremath{\itbox{}\overline{\itbox{L}}\itbox{}}} \cup \set{\ensuremath{\itbox{L}}}; \Gp \ensuremath{\itbox{e}_{0}\itbox{}} : \ensuremath{\itbox{T}}
    \end{bcpcasearray}
    $$
    for some \(\LSet\).  By the induction hypothesis, we have
    $\bullet; \set{\ensuremath{\itbox{}\overline{\itbox{L}}\itbox{}}}; \Gp \ensuremath{\itbox{e}_{l}\itbox{{\ensuremath{'}}}} : \ensuremath{\itbox{L{\ensuremath{'}}}}$ for some \ensuremath{\itbox{L{\ensuremath{'}} \(\Leq\) L}}.  By
    \rn{LS-Extends}, \ensuremath{\itbox{L{\ensuremath{'}}}} and \ensuremath{\itbox{L}} have the same require clause
    $\LSet$.  Since \ensuremath{\itbox{L{\ensuremath{'}} \(\Leq\) L}}, we have $\ensuremath{\itbox{L{\ensuremath{'}}}} \LEQ_w \ensuremath{\itbox{L}}$, and
    $\set{\ensuremath{\itbox{}\overline{\itbox{L}}\itbox{}}} \cup \set{\ensuremath{\itbox{L{\ensuremath{'}}}}} \LEQ_w \set{\ensuremath{\itbox{}\overline{\itbox{L}}\itbox{}}} \cup \set{\ensuremath{\itbox{L}}}$.  By
    \lemref{narrowing} and \rn{T-With},
    $\bullet; \set{\ensuremath{\itbox{}\overline{\itbox{L}}\itbox{}}}; \Gp \ensuremath{\itbox{e{\ensuremath{'}}}} : \ensuremath{\itbox{T}}$.  Reflexivity of \ensuremath{\itbox{\(\Leq\)}}
    finishes the case.
  \end{rneqncase}

  \begin{rneqncase}{R-WithVal}{
      \ensuremath{\itbox{e}} = \ensuremath{\itbox{with new L() v}_{0}\itbox{}} &
      \ensuremath{\itbox{e{\ensuremath{'}}}} = \ensuremath{\itbox{v}_{0}\itbox{}}
    }
    By \rn{T-With}, it must be the case that $\bullet; \set{\ensuremath{\itbox{}\overline{\itbox{L}}\itbox{}}} \cup
    \set{\ensuremath{\itbox{L}}}; \Gp \ensuremath{\itbox{v}_{0}\itbox{}} : \ensuremath{\itbox{T}}$.  By \lemref{value-strengthening},
    $\bullet; \set{\ensuremath{\itbox{}\overline{\itbox{L}}\itbox{}}}; \Gp \ensuremath{\itbox{v}_{0}\itbox{}} : \ensuremath{\itbox{T}}$, finishing the case.
  \end{rneqncase}

  \begin{rneqncase}{RC-Swap}{
      \ensuremath{\itbox{e}} = \ensuremath{\itbox{swap (new L(),L}_{{sw}}\itbox{) e}_{0}\itbox{}} &
      \ensuremath{\itbox{e{\ensuremath{'}}}} = \ensuremath{\itbox{swap (new L(),L}_{{sw}}\itbox{) e}_{0}\itbox{{\ensuremath{'}}}} \\
      \swap(\ensuremath{\itbox{L}},\ensuremath{\itbox{L}_{{sw}}\itbox{}},\ensuremath{\itbox{}\overline{\itbox{L}}\itbox{}}) = \ensuremath{\itbox{}\overline{\itbox{L}}\itbox{{\ensuremath{'}}}} &
      \reduceto{\ensuremath{\itbox{}\overline{\itbox{L}}\itbox{{\ensuremath{'}}}}}{\ensuremath{\itbox{e}_{0}\itbox{}}}{\ensuremath{\itbox{e}_{0}\itbox{{\ensuremath{'}}}}}
    }
    By \rn{T-Swap}, it must be the case that
    $$
    \begin{bcpcasearray}
      \bullet; \set{\ensuremath{\itbox{}\overline{\itbox{L}}\itbox{}}}; \Gp \ensuremath{\itbox{new L()}} : \ensuremath{\itbox{L}} &
      \ensuremath{\itbox{L}_{{sw}}\itbox{ swappable}} &
      \ensuremath{\itbox{L}} \LEQ_w \ensuremath{\itbox{L}_{{sw}}\itbox{}} &
      \ensuremath{\itbox{L req }} \LSet \\
      \LSet_{rm} = \set{\ensuremath{\itbox{}\overline{\itbox{L}}\itbox{}}} \setminus \set{\ensuremath{\itbox{L{\ensuremath{'}}}} \mid \ensuremath{\itbox{L{\ensuremath{'}}}} \LEQ_w \ensuremath{\itbox{L}_{{sw}}\itbox{}}} &
      \LSet_{rm} \LEQ_w \LSet &
      \bullet; \LSet_{rm} \cup \set{\ensuremath{\itbox{L}}}; \Gp \ensuremath{\itbox{e}_{0}\itbox{}} : \ensuremath{\itbox{T}}
    \end{bcpcasearray}
    $$
    for some \ensuremath{\itbox{L}}, \(\LSet\), and \(\LSet_{rm}\).
    Here, $\set{\ensuremath{\itbox{}\overline{\itbox{L}}\itbox{{\ensuremath{'}}}}} = \LSet_{rm} \cup \set{\ensuremath{\itbox{L}}}$. Then, $\set{\ensuremath{\itbox{}\overline{\itbox{L}}\itbox{{\ensuremath{'}}}}}\, 
    \WF$ by \rn{Wf-Swap}.  By the induction hypothesis, $\bullet;
    \LSet_{rm} \cup \set{\ensuremath{\itbox{L}}}; \Gp \ensuremath{\itbox{e}_{0}\itbox{{\ensuremath{'}}}} : \ensuremath{\itbox{S}}$ for some $\ensuremath{\itbox{S}} \LEQ
    \ensuremath{\itbox{T}}$.  By \rn{T-Swap}, $\bullet; \set{\ensuremath{\itbox{}\overline{\itbox{L}}\itbox{}}}; \Gp \ensuremath{\itbox{e}_{0}\itbox{{\ensuremath{'}}}} : \ensuremath{\itbox{S}}$,
    finishing the case.
  \end{rneqncase}

  \begin{rneqncase}{RC-SwapArg}{
      \ensuremath{\itbox{e}} = \ensuremath{\itbox{swap (e}_{l}\itbox{,L}_{{sw}}\itbox{) e}_{0}\itbox{}} &
      \ensuremath{\itbox{e{\ensuremath{'}}}} = \ensuremath{\itbox{swap (e}_{l}\itbox{{\ensuremath{'}},L}_{{sw}}\itbox{) e}_{0}\itbox{}} \\
      \reduceto{\ensuremath{\itbox{}\overline{\itbox{L}}\itbox{}}}{\ensuremath{\itbox{e}_{l}\itbox{}}}{\ensuremath{\itbox{e}_{l}\itbox{{\ensuremath{'}}}}}
    }
    By \rn{T-Swap}, it must be the case that
    $$
    \begin{bcpcasearray}
      \bullet; \set{\ensuremath{\itbox{}\overline{\itbox{L}}\itbox{}}}; \Gp \ensuremath{\itbox{e}_{l}\itbox{}} : \ensuremath{\itbox{L}} &
      \ensuremath{\itbox{L}_{{sw}}\itbox{ swappable}} &
      \ensuremath{\itbox{L}} \LEQ_w \ensuremath{\itbox{L}_{{sw}}\itbox{}} &
      \ensuremath{\itbox{L req }} \LSet \\
      \LSet_{rm} = \set{\ensuremath{\itbox{}\overline{\itbox{L}}\itbox{}}} \setminus \set{\ensuremath{\itbox{L{\ensuremath{'}}}} \mid \ensuremath{\itbox{L{\ensuremath{'}}}} \LEQ_w \ensuremath{\itbox{L}_{{sw}}\itbox{}}} &
      \LSet_{rm} \LEQ_w \LSet &
      \bullet; \LSet_{rm} \cup \set{\ensuremath{\itbox{L}}}; \Gp \ensuremath{\itbox{e}_{0}\itbox{}} : \ensuremath{\itbox{T}}
    \end{bcpcasearray}
    $$
    for some \ensuremath{\itbox{L}}, \(\LSet\), and \(\LSet_{rm}\).
    By the induction hypothesis, we have $\bullet; \set{\ensuremath{\itbox{}\overline{\itbox{L}}\itbox{}}}; \Gp
    \ensuremath{\itbox{e}_{l}\itbox{{\ensuremath{'}}}} : \ensuremath{\itbox{L{\ensuremath{'}}}}$ for some \ensuremath{\itbox{L{\ensuremath{'}} \(\Leq\) L}}.  By \rn{LS-Extends}, \ensuremath{\itbox{L{\ensuremath{'}}}} and
    \ensuremath{\itbox{L}} have the same require clause $\LSet$.  Since \ensuremath{\itbox{L{\ensuremath{'}} \(\Leq\) L}}, we have
    $\ensuremath{\itbox{L{\ensuremath{'}}}} \LEQ_w \ensuremath{\itbox{L}}$, $\ensuremath{\itbox{L{\ensuremath{'}}}} \LEQ_w \ensuremath{\itbox{L}_{{sw}}\itbox{}}$, and $\LSet_{rm} \cup
    \set{\ensuremath{\itbox{L{\ensuremath{'}}}}} \LEQ_w \LSet_{rm} \cup \set{\ensuremath{\itbox{L}}} \LEQ_w \LSet$.  By
    \lemref{narrowing} and \rn{T-Swap}, $\bullet; \set{\ensuremath{\itbox{}\overline{\itbox{L}}\itbox{}}}; \Gp \ensuremath{\itbox{e{\ensuremath{'}}}}
    : \ensuremath{\itbox{T}}$.  Reflexivity of \ensuremath{\itbox{\(\Leq\)}} finishes the case.
  \end{rneqncase}  

  \begin{rncase}{R-SwapVal}
    Similar to Case \rn{R-WithVal}.
  \end{rncase}
  
  \begin{rneqncase}{RC-InvkRecv}{
      \ensuremath{\itbox{e}} = \ensuremath{\itbox{e}_{0}\itbox{.m(}\overline{\itbox{e}}\itbox{)}} &
      \reduceto{\ensuremath{\itbox{}\overline{\itbox{L}}\itbox{}}}{\ensuremath{\itbox{e}_{0}\itbox{}}}{\ensuremath{\itbox{e}_{0}\itbox{{\ensuremath{'}}}}} &
      \ensuremath{\itbox{e{\ensuremath{'}}}} = \ensuremath{\itbox{e}_{0}\itbox{{\ensuremath{'}}.m(}\overline{\itbox{e}}\itbox{)}}
    }
    By \rn{T-Invk}, it must be the case that
    $$
    \begin{bcpcasearray}
      \bullet; \set{\ensuremath{\itbox{}\overline{\itbox{L}}\itbox{}}}; \Gp \ensuremath{\itbox{e}_{0}\itbox{}} : \ensuremath{\itbox{C}_{0}\itbox{}} &
      \mtype(\ensuremath{\itbox{m}}, \ensuremath{\itbox{C}_{0}\itbox{}}, \set{\ensuremath{\itbox{}\overline{\itbox{L}}\itbox{}}}) = \ensuremath{\itbox{}\overline{\itbox{T}}\itbox{\(\rightarrow\)T}} &
      \bullet; \set{\ensuremath{\itbox{}\overline{\itbox{L}}\itbox{}}}; \Gp \ensuremath{\itbox{}\overline{\itbox{e}}\itbox{}} : \ensuremath{\itbox{}\overline{\itbox{S}}\itbox{}} &
      \ensuremath{\itbox{}\overline{\itbox{S}}\itbox{}} \LEQ \ensuremath{\itbox{}\overline{\itbox{T}}\itbox{}}.
    \end{bcpcasearray}
    $$
    for some \ensuremath{\itbox{}\overline{\itbox{T}}\itbox{}} and \ensuremath{\itbox{}\overline{\itbox{S}}\itbox{}}.
    By the induction hypothesis, $\bullet; \set{\ensuremath{\itbox{}\overline{\itbox{L}}\itbox{}}}; \Gp \ensuremath{\itbox{e}_{0}\itbox{{\ensuremath{'}}}} :
    \ensuremath{\itbox{D}_{0}\itbox{}}$ for some $\ensuremath{\itbox{D}_{0}\itbox{}} \LEQ \ensuremath{\itbox{C}_{0}\itbox{}}$.  By \lemref{subtype-mtype},
    $\mtype(\ensuremath{\itbox{m}}, \ensuremath{\itbox{D}_{0}\itbox{}}, \set{\ensuremath{\itbox{}\overline{\itbox{L}}\itbox{}}}) = \ensuremath{\itbox{}\overline{\itbox{T}}\itbox{\(\rightarrow\)S}}$ and $\ensuremath{\itbox{S}} \LEQ \ensuremath{\itbox{T}}$ for
    some \ensuremath{\itbox{S}}.  By \rn{T-Invk}, $\bullet; \set{\ensuremath{\itbox{}\overline{\itbox{L}}\itbox{}}}; \Gp \ensuremath{\itbox{e}_{0}\itbox{{\ensuremath{'}}.m(}\overline{\itbox{e}}\itbox{)}}
    : \ensuremath{\itbox{S}}$, finishing the case.
  \end{rneqncase}

  \begin{rneqncase}{RC-InvkArg}{
      \ensuremath{\itbox{e}} = \ensuremath{\itbox{e}_{0}\itbox{.m(}}\ldots\ensuremath{\itbox{,e}_{i}\itbox{,}}\ldots\ensuremath{\itbox{)}} &
      \reduceto{\ensuremath{\itbox{}\overline{\itbox{L}}\itbox{}}}{\ensuremath{\itbox{e}_{i}\itbox{}}}{\ensuremath{\itbox{e}_{i}\itbox{{\ensuremath{'}}}}} &
      \ensuremath{\itbox{e{\ensuremath{'}}}} = \ensuremath{\itbox{e}_{0}\itbox{.m(}}\ldots\ensuremath{\itbox{,e}_{i}\itbox{{\ensuremath{'}},}}\ldots\ensuremath{\itbox{)}}
    }
    By \rn{T-Invk}, it must be the case that
    $$
    \begin{bcpcasearray}
      \bullet; \set{\ensuremath{\itbox{}\overline{\itbox{L}}\itbox{}}}; \Gp \ensuremath{\itbox{e}_{0}\itbox{}} : \ensuremath{\itbox{C}_{0}\itbox{}} &
      \mtype(\ensuremath{\itbox{m}}, \ensuremath{\itbox{C}_{0}\itbox{}}, \set{\ensuremath{\itbox{}\overline{\itbox{L}}\itbox{}}}) = \ensuremath{\itbox{}\overline{\itbox{T}}\itbox{\(\rightarrow\)T}} &
      \bullet; \set{\ensuremath{\itbox{}\overline{\itbox{L}}\itbox{}}}; \Gp \ensuremath{\itbox{}\overline{\itbox{e}}\itbox{}} : \ensuremath{\itbox{}\overline{\itbox{S}}\itbox{}} &
      \ensuremath{\itbox{}\overline{\itbox{S}}\itbox{}} \LEQ \ensuremath{\itbox{}\overline{\itbox{T}}\itbox{}}.
    \end{bcpcasearray}
    $$
    for some \ensuremath{\itbox{}\overline{\itbox{T}}\itbox{}} and \ensuremath{\itbox{}\overline{\itbox{S}}\itbox{}}.
    By the induction hypothesis, $\bullet; \set{\ensuremath{\itbox{}\overline{\itbox{L}}\itbox{}}}; \Gp \ensuremath{\itbox{e}_{i}\itbox{{\ensuremath{'}}}} :
    \ensuremath{\itbox{S}_{i}\itbox{{\ensuremath{'}}}}$ for some $\ensuremath{\itbox{S}_{i}\itbox{{\ensuremath{'}}}} \LEQ \ensuremath{\itbox{S}_{i}\itbox{}}$.  By \rn{S-Trans}, $\ensuremath{\itbox{S}_{i}\itbox{{\ensuremath{'}}}} \LEQ
    \ensuremath{\itbox{T}_{i}\itbox{}}$.  So, by \rn{T-Invk}, $\bullet; \set{\ensuremath{\itbox{}\overline{\itbox{L}}\itbox{}}}; \Gp \ensuremath{\itbox{e{\ensuremath{'}}}} :
    \ensuremath{\itbox{T}}$, finishing the case.
  \end{rneqncase}

  \begin{rncase}{RC-New, RC-InvkAArg1, RC-InvkAArg2}
    Similar to the case above. \qed
  \end{rncase}
\end{proof}

\begin{lemmaapp}[\lemref{def:pmtype}]\label{lem:pmtype}
  If $\pmtype(\ensuremath{\itbox{m}}, \ensuremath{\itbox{C}}, \ensuremath{\itbox{L}}) = \ensuremath{\itbox{}\overline{\itbox{T}}\itbox{\(\rightarrow\)T}_{0}\itbox{}}$, then there exist \ensuremath{\itbox{}\overline{\itbox{x}}\itbox{}} and
  \ensuremath{\itbox{e}_{0}\itbox{}} and \ensuremath{\itbox{L{\ensuremath{'}}}} ($\neq \ensuremath{\itbox{Base}}$) such that $\pmbody(\ensuremath{\itbox{m}}, \ensuremath{\itbox{C}}, \ensuremath{\itbox{L}}) =
  \ensuremath{\itbox{}\overline{\itbox{x}}\itbox{.e}_{0}\itbox{}} \IN \ensuremath{\itbox{L{\ensuremath{'}}}}$ and the lengths of \ensuremath{\itbox{}\overline{\itbox{x}}\itbox{}} and \ensuremath{\itbox{}\overline{\itbox{T}}\itbox{}} are equal and
  $\ensuremath{\itbox{L}} \LEQ_w \ensuremath{\itbox{L{\ensuremath{'}}}}$.
\end{lemmaapp}

\begin{proof}
  By induction on $\pmtype(\ensuremath{\itbox{m}}, \ensuremath{\itbox{C}}, \ensuremath{\itbox{L}}) = \ensuremath{\itbox{}\overline{\itbox{T}}\itbox{\(\rightarrow\)T}_{0}\itbox{}}$.

  \begin{rneqncase}{PMT-Layer}{
      $\LT(\ensuremath{\itbox{L}})(\ensuremath{\itbox{C.m}}) = \ensuremath{\itbox{T}_{0}\itbox{ C.m(}\overline{\itbox{T}}\itbox{ }\overline{\itbox{x}}\itbox{){\char'173} return e; {\char'175}}}$
    }
    By \rn{T-PMethod}, the lengths of \ensuremath{\itbox{}\overline{\itbox{T}}\itbox{}} and \ensuremath{\itbox{}\overline{\itbox{x}}\itbox{}} are equal.  $\ensuremath{\itbox{L}}
    \LEQ_w \ensuremath{\itbox{L}}$ by Reflexivity of $\LEQ_w$.  Then, \rn{PMB-Layer}
    finishes the case.
  \end{rneqncase}

  \begin{rneqncase}{PMT-Super}{
      $\LT(\ensuremath{\itbox{L}})(\ensuremath{\itbox{C.m}}) \undf$ \andalso
      \ensuremath{\itbox{L \(\triangleleft\) L{\ensuremath{'}}}} \andalso
      \pmtype(\ensuremath{\itbox{m}},\ensuremath{\itbox{C}},\ensuremath{\itbox{L{\ensuremath{'}}}}) = \ensuremath{\itbox{}\overline{\itbox{T}}\itbox{\(\rightarrow\)T}_{0}\itbox{}}
    }
    The induction hypothesis and \rn{PMB-Layer} and \rn{LSw-Extends}
    and \rn{LSw-Extends} finish the case. \qed
  \end{rneqncase}  
  
\end{proof}

\begin{lemmaapp}[\lemref{def:mtype}]\label{lem:mtype} 

  If $\mtype(\ensuremath{\itbox{m}}, \ensuremath{\itbox{C}}, \set{\ensuremath{\itbox{}\overline{\itbox{L}}\itbox{{\ensuremath{'}}}}}, \set{\ensuremath{\itbox{}\overline{\itbox{L}}\itbox{}}}) = \ensuremath{\itbox{}\overline{\itbox{T}}\itbox{\(\rightarrow\)T}_{0}\itbox{}}$ and \ensuremath{\itbox{}\overline{\itbox{L}}\itbox{{\ensuremath{'}}}}
  is a prefix of \ensuremath{\itbox{}\overline{\itbox{L}}\itbox{}} and \(\set{\ensuremath{\itbox{}\overline{\itbox{L}}\itbox{}}}\, \WF\), then there exist \ensuremath{\itbox{}\overline{\itbox{x}}\itbox{}}
  and \ensuremath{\itbox{e}_{0}\itbox{}} and \ensuremath{\itbox{}\overline{\itbox{L}}\itbox{{\ensuremath{'}}{\ensuremath{'}}}} and \ensuremath{\itbox{C{\ensuremath{'}}}} ($\neq \ensuremath{\itbox{Object}}$) such that
  $\mbody(\ensuremath{\itbox{m}}, \ensuremath{\itbox{C}}, \ensuremath{\itbox{}\overline{\itbox{L}}\itbox{}}, \ensuremath{\itbox{}\overline{\itbox{L}}\itbox{{\ensuremath{'}}}}) = \ensuremath{\itbox{}\overline{\itbox{x}}\itbox{.e}_{0}\itbox{}} \IN \ensuremath{\itbox{C{\ensuremath{'}}}}, \ensuremath{\itbox{}\overline{\itbox{L}}\itbox{{\ensuremath{'}}{\ensuremath{'}}}}$ and the
  lengths of \ensuremath{\itbox{}\overline{\itbox{x}}\itbox{}} and \ensuremath{\itbox{}\overline{\itbox{T}}\itbox{}} are equal and, if \ensuremath{\itbox{}\overline{\itbox{L}}\itbox{{\ensuremath{'}}{\ensuremath{'}}}} is not empty, the last layer name of \ensuremath{\itbox{}\overline{\itbox{L}}\itbox{{\ensuremath{'}}{\ensuremath{'}}}} is not \ensuremath{\itbox{Base}}.
\end{lemmaapp}

\begin{proof}
  By lexicographic induction on $\mtype(\ensuremath{\itbox{m}}, \ensuremath{\itbox{C}}, \set{\ensuremath{\itbox{}\overline{\itbox{L}}\itbox{{\ensuremath{'}}}}}, \set{\ensuremath{\itbox{}\overline{\itbox{L}}\itbox{}}}) =
  \ensuremath{\itbox{}\overline{\itbox{T}}\itbox{\(\rightarrow\)T}_{0}\itbox{}}$ and the length of \ensuremath{\itbox{}\overline{\itbox{L}}\itbox{{\ensuremath{'}}}}.

  \begin{eqncase}{
      \ensuremath{\itbox{}\overline{\itbox{L}}\itbox{{\ensuremath{'}}}} = \bullet \andalso
      \ensuremath{\itbox{class C \(\triangleleft\) D {\char'173}... S}_{0}\itbox{ m(}\overline{\itbox{S}}\itbox{ }\overline{\itbox{x}}\itbox{){\char'173} return e}_{0}\itbox{; {\char'175} ...{\char'175}}}
    }
    By \rn{MT-Class}, it must be the case that $\ensuremath{\itbox{}\overline{\itbox{T}}\itbox{}}, \ensuremath{\itbox{T}_{0}\itbox{}} = \ensuremath{\itbox{}\overline{\itbox{S}}\itbox{}},
    \ensuremath{\itbox{S}_{0}\itbox{}}$ and the lengths of \ensuremath{\itbox{}\overline{\itbox{S}}\itbox{}} and \ensuremath{\itbox{}\overline{\itbox{x}}\itbox{}} are equal.  Then, by \rn{MB-Class},
    $\mbody(\ensuremath{\itbox{m}}, \ensuremath{\itbox{C}},\bullet, \ensuremath{\itbox{}\overline{\itbox{L}}\itbox{}}) = \ensuremath{\itbox{}\overline{\itbox{x}}\itbox{.e}_{0}\itbox{}} \IN \ensuremath{\itbox{C}}, \bullet$, finishing the case.
  \end{eqncase}

  \begin{eqncase}{
      \ensuremath{\itbox{}\overline{\itbox{L}}\itbox{{\ensuremath{'}}}} = \bullet \andalso
      \ensuremath{\itbox{class C \(\triangleleft\) D {\char'173}... }\overline{\itbox{M}}\itbox{{\char'175}}} \andalso
      \ensuremath{\itbox{m}} \not \in \ensuremath{\itbox{}\overline{\itbox{M}}\itbox{}}
    }
    It must be the case that $\mtype(\ensuremath{\itbox{m}},\ensuremath{\itbox{C}},\set{\ensuremath{\itbox{}\overline{\itbox{L}}\itbox{{\ensuremath{'}}}}},\set{\ensuremath{\itbox{}\overline{\itbox{L}}\itbox{}}}) =
    \ensuremath{\itbox{}\overline{\itbox{T}}\itbox{\(\rightarrow\)T}_{0}\itbox{}}$ is derived by \rn{MT-Super} and $\mtype(\ensuremath{\itbox{m}}, \ensuremath{\itbox{D}},
    \set{\ensuremath{\itbox{}\overline{\itbox{L}}\itbox{}}}, \set{\ensuremath{\itbox{}\overline{\itbox{L}}\itbox{}}}) = \ensuremath{\itbox{}\overline{\itbox{T}}\itbox{\(\rightarrow\)T}_{0}\itbox{}}$.  The induction hypothesis
    and \rn{MB-Super} finish the case.
  \end{eqncase}
  
  \begin{eqncase}{
      \ensuremath{\itbox{}\overline{\itbox{L}}\itbox{{\ensuremath{'}}}} = \ensuremath{\itbox{}\overline{\itbox{L}}\itbox{{\ensuremath{'}}{\ensuremath{'}}{\ensuremath{'}}}}, \ensuremath{\itbox{L}_{0}\itbox{}} &
      \pmtype(\ensuremath{\itbox{m}},\ensuremath{\itbox{C}},\ensuremath{\itbox{L}_{0}\itbox{}}) = \ensuremath{\itbox{}\overline{\itbox{T}}\itbox{\(\rightarrow\)T}_{0}\itbox{}}
    }
    By \lemref{pmtype} and \rn{MB-Layer}.
  \end{eqncase}

  \begin{eqncase}{
      \ensuremath{\itbox{}\overline{\itbox{L}}\itbox{{\ensuremath{'}}}} = \ensuremath{\itbox{}\overline{\itbox{L}}\itbox{{\ensuremath{'}}{\ensuremath{'}}{\ensuremath{'}}}}; \ensuremath{\itbox{L}_{0}\itbox{}} &
      \pmtype(\ensuremath{\itbox{m}},\ensuremath{\itbox{C}},\ensuremath{\itbox{L}_{0}\itbox{}}) \undf
    }
    Since $\pmtype(\ensuremath{\itbox{m}},\ensuremath{\itbox{C}},\ensuremath{\itbox{L}_{0}\itbox{}}) \undf$, it must be the case that
    $\mtype(\ensuremath{\itbox{m}}, \ensuremath{\itbox{C}}, \set{\ensuremath{\itbox{}\overline{\itbox{L}}\itbox{{\ensuremath{'}}{\ensuremath{'}}{\ensuremath{'}}}}}, \set{\ensuremath{\itbox{}\overline{\itbox{L}}\itbox{}}}) = \ensuremath{\itbox{}\overline{\itbox{T}}\itbox{\(\rightarrow\)T}_{0}\itbox{}}$.  By the
    induction hypothesis, there exist \ensuremath{\itbox{}\overline{\itbox{x}}\itbox{}} and \ensuremath{\itbox{e}_{0}\itbox{}} and \ensuremath{\itbox{}\overline{\itbox{L}}\itbox{{\ensuremath{'}}{\ensuremath{'}}}} and
    \ensuremath{\itbox{C{\ensuremath{'}}}} ($\neq \ensuremath{\itbox{Object}}$) such that $\mbody(\ensuremath{\itbox{m}}, \ensuremath{\itbox{C}}, \ensuremath{\itbox{}\overline{\itbox{L}}\itbox{{\ensuremath{'}}{\ensuremath{'}}{\ensuremath{'}}}}, \ensuremath{\itbox{}\overline{\itbox{L}}\itbox{{\ensuremath{'}}}})
    = \ensuremath{\itbox{}\overline{\itbox{x}}\itbox{.e}_{0}\itbox{}} \IN \ensuremath{\itbox{C{\ensuremath{'}}}}, \ensuremath{\itbox{}\overline{\itbox{L}}\itbox{{\ensuremath{'}}{\ensuremath{'}}}}$ and the lengths of \ensuremath{\itbox{}\overline{\itbox{x}}\itbox{}} and \ensuremath{\itbox{}\overline{\itbox{T}}\itbox{}} are
    equal.
    It follows that $\pmbody(\ensuremath{\itbox{m}},\ensuremath{\itbox{C}},\ensuremath{\itbox{L}_{0}\itbox{}})$ is $\undf$ from
    $\pmtype(\ensuremath{\itbox{m}},\ensuremath{\itbox{C}},\ensuremath{\itbox{L}_{0}\itbox{}}) \undf$.  \rn{MB-NextLayer} finishes the case.  \qed
  \end{eqncase}
\end{proof}

\begin{theoremapp}[Progress]\label{thm:progress}
  Suppose $\p (\CT, \LT)\ensuremath{\itbox{ ok}}$.  If $\bullet;
  \set{\ensuremath{\itbox{}\overline{\itbox{L}}\itbox{}}}; \bullet \p \ensuremath{\itbox{e}} : \ensuremath{\itbox{T}}$ and \set{\ensuremath{\itbox{}\overline{\itbox{L}}\itbox{}}}\, \WF, then \ensuremath{\itbox{e}} is a
  value or $\reduceto{\ensuremath{\itbox{}\overline{\itbox{L}}\itbox{}}}{\ensuremath{\itbox{e}}}{\ensuremath{\itbox{e{\ensuremath{'}}}}}$ for some \ensuremath{\itbox{e{\ensuremath{'}}}}.
\end{theoremapp}

\begin{proof}
  By induction on $\bullet; \set{\ensuremath{\itbox{}\overline{\itbox{L}}\itbox{}}}; \bullet \p \ensuremath{\itbox{e}} : \ensuremath{\itbox{T}}$ with case analysis on the last
  typing rule used.

  \begin{rncase}{T-Var, T-Super, T-Proceed, T-SuperProceed}
    Cannot happen.
  \end{rncase}

  \begin{rneqncase}{T-Field}{
      \ensuremath{\itbox{e}} = \ensuremath{\itbox{e}_{0}\itbox{.f}_{i}\itbox{}} &
      \bullet; \set{\ensuremath{\itbox{}\overline{\itbox{L}}\itbox{}}}; \bullet \p \ensuremath{\itbox{e}_{0}\itbox{}} : \ensuremath{\itbox{C}_{0}\itbox{}} &
      \fields(\ensuremath{\itbox{C}_{0}\itbox{}}) = \ensuremath{\itbox{}\overline{\itbox{T}}\itbox{ }\overline{\itbox{f}}\itbox{}} &
      \ensuremath{\itbox{C}} = \ensuremath{\itbox{C}_{i}\itbox{}}
    }
    By the induction hypothesis, either \ensuremath{\itbox{e}_{0}\itbox{}} is a value or there exists
    \ensuremath{\itbox{e}_{0}\itbox{{\ensuremath{'}}}} such that $\reduceto{\ensuremath{\itbox{}\overline{\itbox{L}}\itbox{}}}{\ensuremath{\itbox{e}_{0}\itbox{}}}{\ensuremath{\itbox{e}_{0}\itbox{{\ensuremath{'}}}}}$.  In the latter
    case, \rn{RC-Field} finishes the case.  In the former case where
    \ensuremath{\itbox{e}_{0}\itbox{}} is a value, by \rn{T-New}, we have
    $$
    \begin{bcpcasearray}
      \ensuremath{\itbox{e}_{0}\itbox{}} = \ensuremath{\itbox{new C}_{0}\itbox{(}\overline{\itbox{v}}\itbox{)}} &
      \bullet; \set{\ensuremath{\itbox{}\overline{\itbox{L}}\itbox{}}}; \bullet \p \ensuremath{\itbox{}\overline{\itbox{v}}\itbox{}} : \ensuremath{\itbox{}\overline{\itbox{S}}\itbox{}} &
      \ensuremath{\itbox{}\overline{\itbox{S}}\itbox{}} \LEQ \ensuremath{\itbox{}\overline{\itbox{T}}\itbox{}}.
    \end{bcpcasearray}
    $$
    So, we have $\reduceto{\ensuremath{\itbox{}\overline{\itbox{L}}\itbox{}}}{\ensuremath{\itbox{e}}}{\ensuremath{\itbox{v}_{i}\itbox{}}}$, finishing the case.
  \end{rneqncase}

  \begin{rneqncase}{T-Invk}{
      \ensuremath{\itbox{e}} = \ensuremath{\itbox{e}_{0}\itbox{.m(}\overline{\itbox{e}}\itbox{)}} &
      \bullet; \set{\ensuremath{\itbox{}\overline{\itbox{L}}\itbox{}}}; \bullet \p \ensuremath{\itbox{e}_{0}\itbox{}} : \ensuremath{\itbox{C}_{0}\itbox{}} \\
      \mtype(\ensuremath{\itbox{m}}, \ensuremath{\itbox{C}_{0}\itbox{}}, \set{\ensuremath{\itbox{}\overline{\itbox{L}}\itbox{}}}) = \ensuremath{\itbox{}\overline{\itbox{T}}\itbox{\(\rightarrow\)T}} &
      \bullet; \set{\ensuremath{\itbox{}\overline{\itbox{L}}\itbox{}}}; \bullet \p \ensuremath{\itbox{}\overline{\itbox{e}}\itbox{}} : \ensuremath{\itbox{}\overline{\itbox{S}}\itbox{}} &
      \ensuremath{\itbox{}\overline{\itbox{S}}\itbox{}} \LEQ \ensuremath{\itbox{}\overline{\itbox{T}}\itbox{}}
    }
    By the induction hypothesis, there exist $i \geq 0$ and \ensuremath{\itbox{e}_{i}\itbox{{\ensuremath{'}}}}
    such that $\reduceto{\ensuremath{\itbox{}\overline{\itbox{L}}\itbox{}}}{\ensuremath{\itbox{e}_{i}\itbox{}}}{\ensuremath{\itbox{e}_{i}\itbox{{\ensuremath{'}}}}}$, in which case
    \rn{RC-InvkRecv} or \rn{RC-InvkArg} finishes the case, or all
    \ensuremath{\itbox{e}_{i}\itbox{}}'s are values \ensuremath{\itbox{v}_{0}\itbox{}}, \ensuremath{\itbox{}\overline{\itbox{v}}\itbox{}}.  Then, by \rn{T-New}, $\ensuremath{\itbox{v}_{0}\itbox{}} =
    \ensuremath{\itbox{new C}_{0}\itbox{(}\overline{\itbox{w}}\itbox{)}}$ for some values \ensuremath{\itbox{}\overline{\itbox{w}}\itbox{}}.  By \lemref{mtype}, there
    exist \ensuremath{\itbox{}\overline{\itbox{x}}\itbox{}}, \ensuremath{\itbox{e}_{0}\itbox{{\ensuremath{'}}}}, \ensuremath{\itbox{}\overline{\itbox{L}}\itbox{{\ensuremath{'}}{\ensuremath{'}}}} and \ensuremath{\itbox{C{\ensuremath{'}}}} ($\neq \ensuremath{\itbox{Object}}$) such that
    $\mbody(\ensuremath{\itbox{m}}, \ensuremath{\itbox{C}_{0}\itbox{}}, \ensuremath{\itbox{}\overline{\itbox{L}}\itbox{}}, \ensuremath{\itbox{}\overline{\itbox{L}}\itbox{}}) = \ensuremath{\itbox{}\overline{\itbox{x}}\itbox{.e}_{0}\itbox{}} \IN \ensuremath{\itbox{C{\ensuremath{'}}}}, \ensuremath{\itbox{}\overline{\itbox{L}}\itbox{{\ensuremath{'}}{\ensuremath{'}}}}$ and
    the lengths of \ensuremath{\itbox{}\overline{\itbox{x}}\itbox{}} and \ensuremath{\itbox{}\overline{\itbox{T}}\itbox{}} are the same.  Since $\ensuremath{\itbox{C{\ensuremath{'}}}} \neq
    \ensuremath{\itbox{Object}}$, there exists \ensuremath{\itbox{D{\ensuremath{'}}}} such that \ensuremath{\itbox{class C{\ensuremath{'}} \(\triangleleft\) D{\ensuremath{'}} {\char'173}...{\char'175}}}.  We
    have two subcases here depending on whether \ensuremath{\itbox{}\overline{\itbox{L}}\itbox{{\ensuremath{'}}{\ensuremath{'}}}} is empty or
    not.  We will show the case where \ensuremath{\itbox{}\overline{\itbox{L}}\itbox{{\ensuremath{'}}{\ensuremath{'}}}} is not empty; the other
    case is similar.  Let $\ensuremath{\itbox{}\overline{\itbox{L}}\itbox{{\ensuremath{'}}{\ensuremath{'}}}} = \ensuremath{\itbox{}\overline{\itbox{L}}\itbox{{\ensuremath{'}}{\ensuremath{'}}{\ensuremath{'}}}};\ensuremath{\itbox{L}_{0}\itbox{}}$ for some \ensuremath{\itbox{}\overline{\itbox{L}}\itbox{{\ensuremath{'}}{\ensuremath{'}}{\ensuremath{'}}}}.
    Since $\ensuremath{\itbox{L}_{0}\itbox{}} \neq \ensuremath{\itbox{Base}}$, there
    exists \ensuremath{\itbox{L}_{1}\itbox{}} such that $\ensuremath{\itbox{layer L}_{0}\itbox{ \(\triangleleft\) L}_{1}\itbox{ {\char'173}...{\char'175}}}$.  Then, the expression
    $$
    \ensuremath{\itbox{e{\ensuremath{'}}}} = \left[
      \begin{array}{l@{/}l}
        \ensuremath{\itbox{new C}_{0}\itbox{(}\overline{\itbox{w}}\itbox{)}} & \ensuremath{\itbox{this}} \\
        \ensuremath{\itbox{}\overline{\itbox{v}}\itbox{}} & \ensuremath{\itbox{}\overline{\itbox{x}}\itbox{}} \\
        \ensuremath{\itbox{new C}_{0}\itbox{(}\overline{\itbox{w}}\itbox{)<C{\ensuremath{'}},}\overline{\itbox{L}}\itbox{{\ensuremath{'}}{\ensuremath{'}}{\ensuremath{'}},}\overline{\itbox{L}}\itbox{>.m}} & \ensuremath{\itbox{proceed}} \\
        \ensuremath{\itbox{new C}_{0}\itbox{(}\overline{\itbox{w}}\itbox{)<D{\ensuremath{'}},}\overline{\itbox{L}}\itbox{,}\overline{\itbox{L}}\itbox{>}} & \ensuremath{\itbox{super}} \\
        \ensuremath{\itbox{new C}_{0}\itbox{(}\overline{\itbox{w}}\itbox{)<C{\ensuremath{'}},L}_{1}\itbox{,}\overline{\itbox{L}}\itbox{,}\overline{\itbox{L}}\itbox{>}} & \ensuremath{\itbox{superproceed}}
      \end{array}
      \right]\ensuremath{\itbox{e}_{0}\itbox{{\ensuremath{'}}}}
    $$
    is well defined (note that the lengths of \ensuremath{\itbox{}\overline{\itbox{x}}\itbox{}} and \ensuremath{\itbox{}\overline{\itbox{v}}\itbox{}} are equal).
    Then, by \rn{R-InvkP} and \rn{R-Invk}, $\reduceto{\ensuremath{\itbox{}\overline{\itbox{L}}\itbox{}}}{\ensuremath{\itbox{e}}}{\ensuremath{\itbox{e{\ensuremath{'}}}}}$.
  \end{rneqncase}

  \begin{rneqncase}{T-New}{
      \ensuremath{\itbox{e}} = \ensuremath{\itbox{new C(}\overline{\itbox{e}}\itbox{)}} &
      \fields(\ensuremath{\itbox{C}}) = \ensuremath{\itbox{}\overline{\itbox{T}}\itbox{ }\overline{\itbox{f}}\itbox{}} &
      \bullet; \set{\ensuremath{\itbox{}\overline{\itbox{L}}\itbox{}}}; \bullet \p \ensuremath{\itbox{}\overline{\itbox{e}}\itbox{}} : \ensuremath{\itbox{}\overline{\itbox{S}}\itbox{}} &
      \ensuremath{\itbox{}\overline{\itbox{S}}\itbox{}} \LEQ \ensuremath{\itbox{}\overline{\itbox{T}}\itbox{}}
    }
    By the induction hypothesis, either (1) \ensuremath{\itbox{}\overline{\itbox{e}}\itbox{}} are all values, in
    which case \ensuremath{\itbox{e}} is also a value; or (2) there exists $i$ and \ensuremath{\itbox{e}_{i}\itbox{{\ensuremath{'}}}}
    such that $\reduceto{\ensuremath{\itbox{}\overline{\itbox{L}}\itbox{}}}{\ensuremath{\itbox{e}_{i}\itbox{}}}{\ensuremath{\itbox{e}_{i}\itbox{{\ensuremath{'}}}}}$, in which case
    \rn{RC-New} finishes the case.
  \end{rneqncase}

  \begin{rncase}{T-NewL}
    Trivial.
  \end{rncase}

  \begin{rneqncase}{T-With}{
      \ensuremath{\itbox{e}} = \ensuremath{\itbox{with e}_{l}\itbox{ e}_{0}\itbox{}} &
      \bullet; \set{\ensuremath{\itbox{}\overline{\itbox{L}}\itbox{}}}; \bullet \p \ensuremath{\itbox{e}_{l}\itbox{}} : \ensuremath{\itbox{L}} &
      \bullet; \set{\ensuremath{\itbox{}\overline{\itbox{L}}\itbox{}}} \cup \set{\ensuremath{\itbox{L}}}; \bullet \p \ensuremath{\itbox{e}_{0}\itbox{}} : \ensuremath{\itbox{T}} \\
      \ensuremath{\itbox{L req }}\LSet &
      \set{\ensuremath{\itbox{}\overline{\itbox{L}}\itbox{}}} \LEQ_w \LSet
    }
    By the induction hypothesis, either \ensuremath{\itbox{e}_{l}\itbox{}} is not a value, in which
    case \rn{RC-WithArg} finishes the case; or \ensuremath{\itbox{e}_{0}\itbox{}} is a value, in
    which case \rn{R-WithVal} finishes the case; or there exists
    \ensuremath{\itbox{e}_{0}\itbox{{\ensuremath{'}}}} such that $\reduceto{\with(\ensuremath{\itbox{L}},\ensuremath{\itbox{}\overline{\itbox{L}}\itbox{}})}{\ensuremath{\itbox{e}_{0}\itbox{}}}{\ensuremath{\itbox{e}_{0}\itbox{{\ensuremath{'}}}}}$, in
    which case \rn{RC-With} finishes the case (notice that
    \(\set{\with(\ensuremath{\itbox{L}}, \ensuremath{\itbox{}\overline{\itbox{L}}\itbox{}})}\, \WF\), by \rn{Wf-With}).
  \end{rneqncase}

  \begin{rneqncase}{T-Swap}{
      \ensuremath{\itbox{e}} = \ensuremath{\itbox{swap (e}_{l}\itbox{,L}_{{sw}}\itbox{) e}_{0}\itbox{}} &
      \bullet; \set{\ensuremath{\itbox{}\overline{\itbox{L}}\itbox{}}}; \bullet \p \ensuremath{\itbox{e}_{l}\itbox{}} : \ensuremath{\itbox{L}} \\
      \ensuremath{\itbox{L}_{{sw}}\itbox{ swappable}} &
      \ensuremath{\itbox{L}} \LEQ_w \ensuremath{\itbox{L}_{{sw}}\itbox{}} &
      \ensuremath{\itbox{L req }} \LSet' \\
      \LSet_{rm} = \set{\ensuremath{\itbox{}\overline{\itbox{L}}\itbox{}}} \setminus \set{\ensuremath{\itbox{L{\ensuremath{'}}}} \mid \ensuremath{\itbox{L{\ensuremath{'}}}} \LEQ_w \ensuremath{\itbox{L}_{{sw}}\itbox{}}} &
      \LSet_{rm} \LEQ_w \LSet' &
      \Loc; \LSet_{rm} \cup \set{\ensuremath{\itbox{L}}}; \Gp \ensuremath{\itbox{e}_{0}\itbox{}} : \ensuremath{\itbox{T}_{0}\itbox{}}
    }
    By the induction hypothesis, either \ensuremath{\itbox{e}_{l}\itbox{}} is not a value, in which
    case \rn{RC-SwapArg} finishes the case; or \ensuremath{\itbox{e}_{0}\itbox{}} is a value, in
    which case \rn{R-SwapVal} finishes the case; or there exists
    \ensuremath{\itbox{e}_{0}\itbox{{\ensuremath{'}}}} such that
    $\reduceto{\swap(\ensuremath{\itbox{L}},\ensuremath{\itbox{L}_{{sw}}\itbox{}},\ensuremath{\itbox{}\overline{\itbox{L}}\itbox{}})}{\ensuremath{\itbox{e}_{0}\itbox{}}}{\ensuremath{\itbox{e}_{0}\itbox{{\ensuremath{'}}}}}$, in which
    case \rn{RC-Swap} finishes the case (notice that, by \rn{Wf-Swap},
    \(\set{\swap(\ensuremath{\itbox{L}},\ensuremath{\itbox{L}_{{sw}}\itbox{}},\ensuremath{\itbox{}\overline{\itbox{L}}\itbox{}})}\, \WF\)).
  \end{rneqncase}

  \begin{rncase}{T-InvkA}
    Similar to the case for \rn{T-Invk}.
  \end{rncase}

  \begin{rneqncase}{T-InvkAL}{
      \multicolumn{4}{@{}l}{\ensuremath{\itbox{e}} = \ensuremath{\itbox{new C}_{0}\itbox{(}\overline{\itbox{v}}\itbox{)<D}_{0}\itbox{,L}_{1}\itbox{,(}\overline{\itbox{L}}\itbox{{\ensuremath{'}}{\ensuremath{'}};L}_{0}\itbox{),}\overline{\itbox{L}}\itbox{{\ensuremath{'}}>.m(}\overline{\itbox{e}}\itbox{)}}} \\
      \bullet; \set{\ensuremath{\itbox{}\overline{\itbox{L}}\itbox{}}}; \bullet \p \ensuremath{\itbox{new C}_{0}\itbox{(}\overline{\itbox{v}}\itbox{)}} : \ensuremath{\itbox{C}_{0}\itbox{}} &
      \ensuremath{\itbox{C}_{0}\itbox{.m}} \p \ensuremath{\itbox{<D}_{0}\itbox{,L}_{1}\itbox{,(}\overline{\itbox{L}}\itbox{{\ensuremath{'}}{\ensuremath{'}};L}_{0}\itbox{),}\overline{\itbox{L}}\itbox{{\ensuremath{'}}> ok}} \\
      \set{\ensuremath{\itbox{}\overline{\itbox{L}}\itbox{}}} \LEQ_{sw} \set{\ensuremath{\itbox{}\overline{\itbox{L}}\itbox{{\ensuremath{'}}}}} &
      \ensuremath{\itbox{L}_{0}\itbox{}} \LEQ_w \ensuremath{\itbox{L}_{1}\itbox{}} \\
      \pmtype(\ensuremath{\itbox{m}},\ensuremath{\itbox{D}_{0}\itbox{}},\ensuremath{\itbox{L}_{1}\itbox{}}) = \ensuremath{\itbox{}\overline{\itbox{T}}\itbox{{\ensuremath{'}}\(\rightarrow\)T}_{0}\itbox{}} \\
      \bullet; \set{\ensuremath{\itbox{}\overline{\itbox{L}}\itbox{}}}; \bullet \p \ensuremath{\itbox{}\overline{\itbox{e}}\itbox{}} : \ensuremath{\itbox{}\overline{\itbox{S}}\itbox{{\ensuremath{'}}}} &
      \ensuremath{\itbox{}\overline{\itbox{S}}\itbox{{\ensuremath{'}}}} \LEQ \ensuremath{\itbox{}\overline{\itbox{T}}\itbox{{\ensuremath{'}}}}
    }
    By the induction hypothesis, either (1) there exists $i \geq 1$ and \ensuremath{\itbox{e}_{i}\itbox{{\ensuremath{'}}}}
    such that $\reduceto{\ensuremath{\itbox{}\overline{\itbox{L}}\itbox{}}}{\ensuremath{\itbox{e}_{i}\itbox{}}}{\ensuremath{\itbox{e}_{i}\itbox{{\ensuremath{'}}}}}$, in which case
    \rn{RC-InvkArg} finishes the case, or (2) all \ensuremath{\itbox{e}_{i}\itbox{}}'s are values
    \ensuremath{\itbox{}\overline{\itbox{w}}\itbox{}}.  Then, by \lemref{pmtype}, there exist \ensuremath{\itbox{}\overline{\itbox{x}}\itbox{}},
    \ensuremath{\itbox{e}_{0}\itbox{{\ensuremath{'}}}} and \ensuremath{\itbox{L}_{2}\itbox{}} ($\neq \ensuremath{\itbox{Base}}$) such that $\pmbody(\ensuremath{\itbox{m}}, \ensuremath{\itbox{D}_{0}\itbox{}},
    \ensuremath{\itbox{L}_{1}\itbox{}}) = \ensuremath{\itbox{}\overline{\itbox{x}}\itbox{.e}_{0}\itbox{}} \IN \ensuremath{\itbox{L}_{2}\itbox{}}$ and the lengths of \ensuremath{\itbox{}\overline{\itbox{x}}\itbox{}} and \ensuremath{\itbox{}\overline{\itbox{T}}\itbox{{\ensuremath{'}}}} are
    the same.  Since $\ensuremath{\itbox{L}_{2}\itbox{{\ensuremath{'}}}} \neq \ensuremath{\itbox{Base}}$, there exists \ensuremath{\itbox{L}_{3}\itbox{}} such
    that \ensuremath{\itbox{layer L}_{2}\itbox{ \(\triangleleft\) L}_{3}\itbox{ {\char'173}...{\char'175}}}.

    By Sanity Condition (8), \ensuremath{\itbox{D}_{0}\itbox{}} is not \ensuremath{\itbox{Object}} and there
    exists \ensuremath{\itbox{E}_{0}\itbox{}} such that \ensuremath{\itbox{class D}_{0}\itbox{ \(\triangleleft\) E}_{0}\itbox{ {\char'173}...{\char'175}}}.  Then, the
    expression
    $$
    \ensuremath{\itbox{e{\ensuremath{'}}}} = \left[
      \begin{array}{l@{/}l}
        \ensuremath{\itbox{new C}_{0}\itbox{(}\overline{\itbox{v}}\itbox{)}} & \ensuremath{\itbox{this}} \\
        \ensuremath{\itbox{}\overline{\itbox{w}}\itbox{}} & \ensuremath{\itbox{}\overline{\itbox{x}}\itbox{}} \\
        \ensuremath{\itbox{new C}_{0}\itbox{(}\overline{\itbox{v}}\itbox{)<D}_{0}\itbox{,}\overline{\itbox{L}}\itbox{{\ensuremath{'}}{\ensuremath{'}},}\overline{\itbox{L}}\itbox{>.m}} & \ensuremath{\itbox{proceed}} \\
        \ensuremath{\itbox{new C}_{0}\itbox{(}\overline{\itbox{v}}\itbox{)<E}_{0}\itbox{,}\overline{\itbox{L}}\itbox{,}\overline{\itbox{L}}\itbox{>}} & \ensuremath{\itbox{super}} \\
        \ensuremath{\itbox{new C}_{0}\itbox{(}\overline{\itbox{v}}\itbox{)<D}_{0}\itbox{,L}_{3}\itbox{,(}\overline{\itbox{L}}\itbox{{\ensuremath{'}}{\ensuremath{'}};L}_{0}\itbox{),}\overline{\itbox{L}}\itbox{>.m}} & \ensuremath{\itbox{superproceed}}
      \end{array}
      \right]\ensuremath{\itbox{e}_{0}\itbox{{\ensuremath{'}}}}
    $$
    is well defined (note that the lengths of \ensuremath{\itbox{}\overline{\itbox{x}}\itbox{}} and \ensuremath{\itbox{}\overline{\itbox{v}}\itbox{}} are equal).
    Then, by \rn{R-InvkSP}, $\reduceto{\ensuremath{\itbox{}\overline{\itbox{L}}\itbox{}}}{\ensuremath{\itbox{e}}}{\ensuremath{\itbox{e{\ensuremath{'}}}}}$. \qed

  \end{rneqncase}
  
\end{proof}

\begin{theoremapp}[Type Soundness]\label{thm:soundness}
  If \(\p (\CT, \LT, \ensuremath{\itbox{e}}) : \ensuremath{\itbox{T}}\) and \ensuremath{\itbox{e}} reduces to a normal form under
  the empty set of layers, then the normal form is \ensuremath{\itbox{new S(}\overline{\itbox{v}}\itbox{)}}
  for some \ensuremath{\itbox{}\overline{\itbox{v}}\itbox{}} and \(\ensuremath{\itbox{S}}\) such that \(\ensuremath{\itbox{S}} \LEQ \ensuremath{\itbox{T}}\).
\end{theoremapp}

\begin{proof}
  By \rn{T-Prog} and Theorems~\ref{thm:subject-reduction} and \ref{thm:progress}.
  \qed
 


\end{proof}

\fi
\end{document}